\documentclass[reprint,aps,superscriptaddress,twocolumn,showpacs,longbibliography]{revtex4-2}

\usepackage{amsmath,amsfonts,enumerate,amsthm,amssymb,bm,bbm,mathtools,nicefrac}
\usepackage[colorlinks=true,citecolor=blue,linkcolor=magenta]{hyperref}
\usepackage[capitalize,nameinlink,noabbrev]{cleveref}
\usepackage{multirow}
\usepackage{bbold}
\usepackage{physics,braket}
\usepackage{soul}
\usepackage[caption = false]{subfig}
\usepackage[skip=5pt,figurewithin=none]{caption}
\usepackage[normalem]{ulem}
\usepackage{scrextend}
\usepackage[usenames,dvipsnames]{xcolor}
\usepackage{apptools}
\usepackage{graphicx}
\usepackage{siunitx}
\usepackage{booktabs}
\usepackage{ragged2e}
\usepackage{fontawesome}

\usepackage{colortbl} 
\usepackage{float}
\makeatletter
\let\newfloat\newfloat@ltx
\makeatother
\usepackage{algorithm,algorithmicx,algpseudocode}
\usepackage{tikz}
\usetikzlibrary{quantikz2}
\usepackage{comment}
\usepackage{lipsum}

\usepackage{mathtools}
\usepackage{tikzsymbols}

\usepackage{multirow,diagbox} 
\usepackage{svg}
\usepackage{etoolbox} 
\usepackage{pifont,dsfont} 

\usepackage[shortlabels]{enumitem}
\usepackage{tabularx,xifthen,afterpage}
\usepackage{setspace,adjustbox} 
\usepackage{array}
\newcolumntype{Y}{>{\centering\arraybackslash}p{3cm}}

\newcommand{\leaves}{%
  \mathrel{%
    \begin{tikzpicture}[baseline=-0.5ex, line width = 0.7]
      \draw (0,-0.2)--(0,0.2);
      \foreach \angle in {150, 210, 190, 170} {
        \draw (0,0)--++(\angle:0.3);
      }
    \end{tikzpicture}%
  }
}

\mathchardef\ordinarycolon\mathcode`\:
\mathcode`\:=\string"8000
\begingroup \catcode`\:=\active
  \gdef:{\mathrel{\mathop\ordinarycolon}}
\endgroup
\newcommand*{\NC}{\mathsf{NC}}

\newcommand*{\AC}{\mathsf{AC}}
\newcommand*{\TC}{\mathsf{TC}}
\newcommand*{\QNC}{\mathsf{QNC}}
\newcommand*{\MOD}{\mathsf{mod}}

\newcommand{\AND}{\mathsf{AND}}
\newcommand{\OR}{\mathsf{OR}}
\newcommand{\NOT}{\mathsf{NOT}}
\newcommand{\BT}{\mathsf{bPTF}}

\newcommand{\BTC}{\mathsf{bPTC}}

\newcommand{\DT}{\mathsf{DT}}

\newcommand{\Supp}{\mathsf{Supp}}
\newcommand{\Dec}{\textsc{dec}}

\newcommand{\Rec}{\textsc{rec}}
\newcommand{\Rep}{\textsc{rep}}
\newcommand{\Red}{\textsc{red}}
\newcommand{\Prep}{\textsc{prep}}
\newcommand{\poly}{\mathrm{poly}}

\DeclareMathOperator*{\Expec}{\mathbb{E}}

\newcommand{\dist}{\mathrm{dist}}
\newcommand{\diam}{\mathrm{diam}}
\newcommand{\sur}{\mathrm{Sur}}
\newcommand{\bulk}{\mathrm{Bulk}}
\theoremstyle{remark} 

\newtheorem{theorem}{Theorem}
\newtheorem{remark}{Remark}
\newtheorem{lemma}{Lemma}

\newtheorem{corollary}{Corollary}

\newtheorem{definition}{Definition}
\newtheorem{fact}{Fact}

\AtAppendix{\counterwithin{theorem}{section}}
\AtAppendix{\counterwithin{corollary}{section}}
\AtAppendix{\counterwithin{lemma}{section}}
\AtAppendix{\counterwithin{proposition}{section}}
\AtAppendix{\counterwithin{definition}{section}}
\AtAppendix{\counterwithin{remark}{section}}
\AtAppendix{\counterwithin{conjecture}{section}}

\definecolor{ssoftgray}{gray}{0.90}
\definecolor{softgray}{gray}{0.95}
\definecolor{softblue}{RGB}{224, 242, 255}
\definecolor{softorange}{RGB}{255, 244, 229}
\definecolor{softgreen}{RGB}{240, 255, 240}
\definecolor{softyellow}{RGB}{255, 249, 219}

\theoremstyle{definition} 

\begin{document}

\title{Unconditional advantage of noisy qudit quantum circuits over biased threshold circuits in constant depth}

\author{Michael de Oliveira~\small\faEnvelopeO}
\email{michaeldeoliveira848@gmail.com}
\affiliation{Hon Hai (Foxconn) Quantum Computing Research Center, Taipei City, Taiwan}
\affiliation{International Iberian Nanotechnology Laboratory, Braga, Portugal} 
\affiliation{LIP6, Sorbonne Université, Paris, France}
\affiliation{INESC TEC, Porto, Portugal}

\author{Sathyawageeswar Subramanian~\small\faEnvelopeO}
\email{sathynius2@gmail.com}
\affiliation{Department of Computer Science and Technology, University of Cambridge, Cambridge, United Kingdom}

\author{Leandro Mendes}
\affiliation{Hon Hai (Foxconn) Quantum Computing Research Center, Taipei City, Taiwan}

\author{Min-Hsiu Hsieh~\small\faEnvelopeO}
\email{minhsiuh@gmail.com}
\affiliation{Hon Hai (Foxconn) Quantum Computing Research Center, Taipei City, Taiwan}

\begin{abstract}
The rapid evolution of quantum devices fuels concerted efforts to experimentally establish quantum advantage over classical computing. Many demonstrations of quantum advantage, however, rely on computational assumptions and face verification challenges. Furthermore, steady advances in classical algorithms and machine learning make the issue of provable, practically demonstrable quantum advantage a moving target. In this work, we unconditionally demonstrate that parallel quantum computation can exhibit greater computational power than previously recognized. We prove that polynomial-size biased threshold circuits of constant depth---which model neural networks with tunable expressivity---fail to solve certain problems solvable by small constant-depth quantum circuits with local gates, for values of the bias that allow quantifiably large computational power. Additionally, we identify a family of problems that are solvable in constant depth by a universal quantum computer over prime-dimensional qudits with bounded connectivity, but remain hard for polynomial-size biased threshold circuits. We thereby bridge the foundational theory of non-local games in higher dimensions with computational advantage on emerging devices operating on a wide range of physical platforms. Finally, we show that these quantum advantages are robust to noise across all prime qudit dimensions with all-to-all connectivity, enhancing their practical appeal.
\end{abstract}

\maketitle

\section{Introduction}
Quantum technologies, particularly quantum computing, have recently made significant progress. This includes a continuous increase in the number of qubits/qudits on quantum processors \cite{kim2023evidence}, notable reductions in error rates for native operations \cite{Morgado21,maldonado2022error}, and extended coherence times \cite{wang2017single} across various hardware platforms. Additionally, these technical advances have culminated in several notable breakthroughs in error-correction experiments, marking significant progress toward early fault-tolerant quantum computation \cite{bluvstein2023logical,acharya2024quantum,andersion24}. However, achieving the theoretical computational advantages promised by landmark quantum algorithms, such as integer factorization or search, remains limited by the substantial resource requirements, which current quantum hardware is still far from meeting \cite{beverland2022assessing,Gidney2021howtofactorbit,scherer2017concrete}.

This state of affairs raises the question of understanding whether the evolving small and faulty quantum devices could still support quantum advantages that are realizable with low resource requirements, over comparable classical devices (or computational models). In particular, a series of pioneering studies showed that under plausible complexity theoretic assumptions---e.g.\ that the polynomial hierarchy does not collapse to the third level---certain classes of quantum circuits are exponentially hard to simulate classically while being more resource-efficient and hardware-friendly. Conditional hardness results of this kind include IQP circuits \cite{bremner2011classical}, Boson sampling experiments \cite{aaronson2011computational}, and random circuit sampling \cite{bouland2019complexity}. On the other hand, the reliance on complexity-theoretic assumptions and the challenge of addressing noise in current quantum experiments reveal a significant gap between theoretical models and practical implementation, complicating the verification of correctness. Moreover, advancements in classical algorithms, simulation techniques, and artificial intelligence continue to escalate the race to achieve quantum advantage beyond classical capabilities.

Rather than comparing the capabilities of shallow-depth quantum circuits to general models of classical computation, recent research has shifted towards comparing them to their classical shallow-depth counterparts, highlighting the potential for unconditional quantum advantages on near-term hardware. The seminal work of Bravyi et al.~\cite{Bravyi17} demonstrated a search problem that can be solved by constant-depth quantum circuits using only 2-qubit Clifford gates but which no constant-depth classical circuit with bounded fan-in gates can solve, without relying on any complexity-theoretic assumptions. This first \textit{unconditional} separation between shallow-depth quantum and classical circuits sparked renewed interest in the field. It underscored the potential for significant computational and practical advantages in \emph{parallel} quantum computations. In both classical and quantum computing, circuit width is a good measure of parallelism (e.g.\ the number of processors) and depth is a good measure of runtime, with constant-depth circuits capturing the computations that can be performed by a polynomial number 
of processors running for a constant amount of time. Several extensions of Bravyi et al.'s result followed, enhancing the separation to average-case hardness \cite{le2019average}, introducing noise resilience \cite{bravyi2020quantum}, demonstrating that the quantum advantage extends to larger constant-depth logical circuits with unbounded fan-in Boolean gates \cite{Watts19}, and identifying problems of greater practical interest \cite{Briet23}. 

While this prior work has established that parallel quantum computations can show advantages over comparable parallel classical computations, these advantages have only been demonstrated for classical circuit classes with limited practical applicability. To broaden the scope of quantum advantage, it is thus essential to explore circuit classes beyond those previously considered. A prime candidate for this exploration is circuits that are allowed to use \textit{threshold} gates, which output one if the Hamming weight of the input string meets or exceeds a threshold $k$, and zero otherwise, mirroring the behavior of a Heaviside step function. Constant-depth circuits comprising such threshold gates belong to the complexity class $\TC^0$ \cite{Parberry1988}. A point of interest is that these circuits serve as a canonical theoretical model for vanilla neural networks \cite{siu1991power,minsky1969introduction,muroga1971threshold,baldi2019polynomial}, and have even been useful in obtaining mathematical results about the expressivity of transformer architectures that underlie large language models such as ChatGPT \cite{Merrill2022,Merrill2023}. A pertinent question, which could provide foundational insights into parallel quantum computations and the computational potential of quantum machine learning models, is whether quantum circuits can outperform threshold circuits in constant depth.

Here, we extend the scope of provable quantum advantages in parallel computation to more advanced and potentially practical classical parallel models for all qudits of prime dimensions. In particular, our results focus on a classical circuit class that captures the computational power of \textit{threshold} circuits in a well-defined and measurable way through a parameterized bias, allowing us to quantify the potential extent of quantum advantage over these models. By incorporating realistic noise models into shallow quantum circuits, we demonstrate that these quantum advantages are robust to noise, preserving their feasibility under practical conditions. Finally, we quantify the resource costs for demonstrating quantum advantage, presenting a hierarchy of quantum advantage experiments and their associated resource requirements.

\section{Results}\label{sec_results}

\subsection{Background}

In this paper, we consider classes of shallow-depth quantum circuits---circuits of constant depth independent of the input length, using a polynomial number of gates that have bounded fan-in (i.e., each gate has a fixed, constant number of input and output wires) and are drawn from a finite, qudit universal quantum gate set. Over qubits, this class is commonly referred to as $\QNC^0$. Throughout this paper, we refer to qudit systems of prime dimension $p$ as `qupits' (with the associated state space $\mathbb{C}^p$). Concurrently, we will compare shallow quantum circuits, as efficient and hardware-friendly as possible, with circuits in the $\BTC^0(k)$ class described by constant depth classical circuits composed of unbounded fan-in gates that compute \textit{biased polynomial threshold functions} (PTFs), which may be non-linear but are constrained by a bias parameter $k$ \cite{Kumar23}.  A polynomial threshold function with bias $k$ is defined as follows,
\begin{align}
f(x)=\begin{cases}
P(x), &\sum_{i=1}^n x_i \leq k\\
1, &\sum_{i=1}^n x_i > k
\end{cases}; \, \, \, \hfill,
\end{align}
\noindent  with $P:\mathbb{F}_2^n\to\mathbb{F}_2$ being a polynomial over $\mathbb{F}_2=\{0,1\}$ and $k$ restricting the maximum degree of $P$. We note that the common definition in the literature takes $P:\mathbb{R}^n_2\to\mathbb{R}$, and sets $f(x)=\frac12\left(1+\mathrm{sgn}(P(x)\right)$. On the other hand here we are interested in polynomials of degree at most $k$ over $\mathbb{F}_2^n$, with threshold behaviour determined by the bias parameter $k$.

This parameter constrains their behavior to function as the Boolean $\OR$ when the input string has a Hamming weight of at least $k$. Conversely, the behavior can be inverted, in which case the bias parameter constrains it to function as the Boolean $\AND$, applying to input strings with a Hamming weight of at most $n-k$. For a constant, non-zero $k$ (i.e., $k = \mathcal{O}(1)$), this class corresponds to Boolean circuits with unbounded fan-in $\AND$, $\OR$, and $\NOT$ gates of constant depth, known as the circuit class $\AC^0$ \cite{AroraBarak}. Notably, $\AC^0$ represents the largest classical circuit class for which unconditional quantum advantages have been established \cite{Watts19}. When $k = \Theta(n)$, it encompasses the strictly larger $\TC^0$ circuit class.

\begin{figure}[htbp]
    \centering
    \includegraphics[scale=0.38]{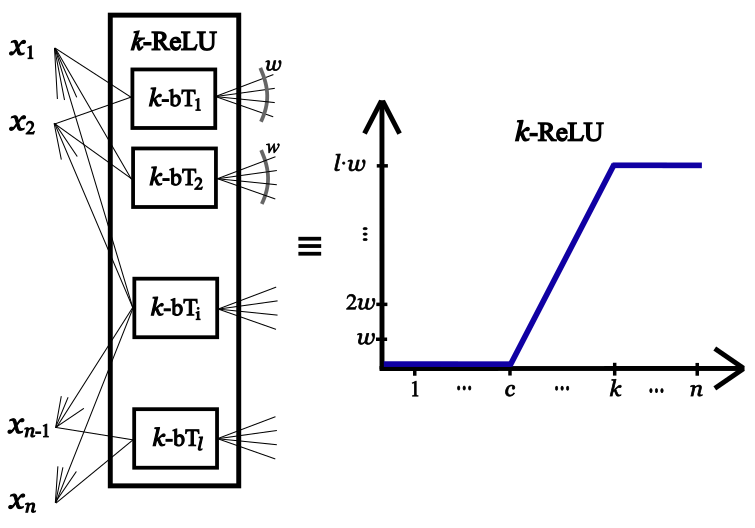}
    \caption{\justifying Representation of a $k$-$\mathsf{ReLU}$ gate within $\BTC^0(k)$, constructed using multiple biased threshold gates ($k$-$\mathsf{bT}$). This gate is equivalent to a $\mathsf{ReLU}$ gate, defined as $f(x) = \max\{0, x - c\}$ (where the center is shifted from $0$ to $c$), up to an integer precision $w$ for any input string with a Hamming weight bounded by $k$. Our scheme considers $\mathsf{ReLU}:\{0, 1\}^n \mapsto \{0, n - c\}$, which takes $n$-bit strings as input and interprets their Hamming weight as the input $x$.}
      \label{fig:3b}
\end{figure}

For intermediate scalings, this class provides access to unbounded fan-in gates with biased yet non-trivial activation regions, see \cref{fig:3b} and refer to SI Sec.\ D2 for our decomposition algorithm, which translates $k$ biased activation functions into $\BTC^0(k)$ circuits. Furthermore, it includes majority gates with small fan-in, capturing some of the computational power of $\TC^0$. Notably, even a single biased threshold gate with bias $k=\omega(\log n)$ is known to require Boolean circuits of superpolynomial size (i.e., $\Omega(n^{\mathrm{polylog}(n)})$) using unbounded fan-in $\AND$, $\OR$, and $\NOT$ gates to simulate it \cite{Kumar23}. From this, and given that bounded fan-in constant-depth classical circuits are described by the $\NC^0$ circuit class, we derive the following sequence of inclusions for classical parallel computational classes:
\begin{equation}\label{eq:cont}
\NC^0 \subsetneq \AC^0 \subsetneq \BTC^0(k)\text{ for }k=\omega(\log n).
\end{equation}

Thus, we highlight that biased polynomial threshold circuits of constant depth provide a valuable framework for investigating quantum computational power, especially by analyzing the impact of various bias parameters and their practical relevance (see \cref{fig:3a}).

\begin{figure*}[htbp]
    \centering
     \includegraphics[scale=0.53]{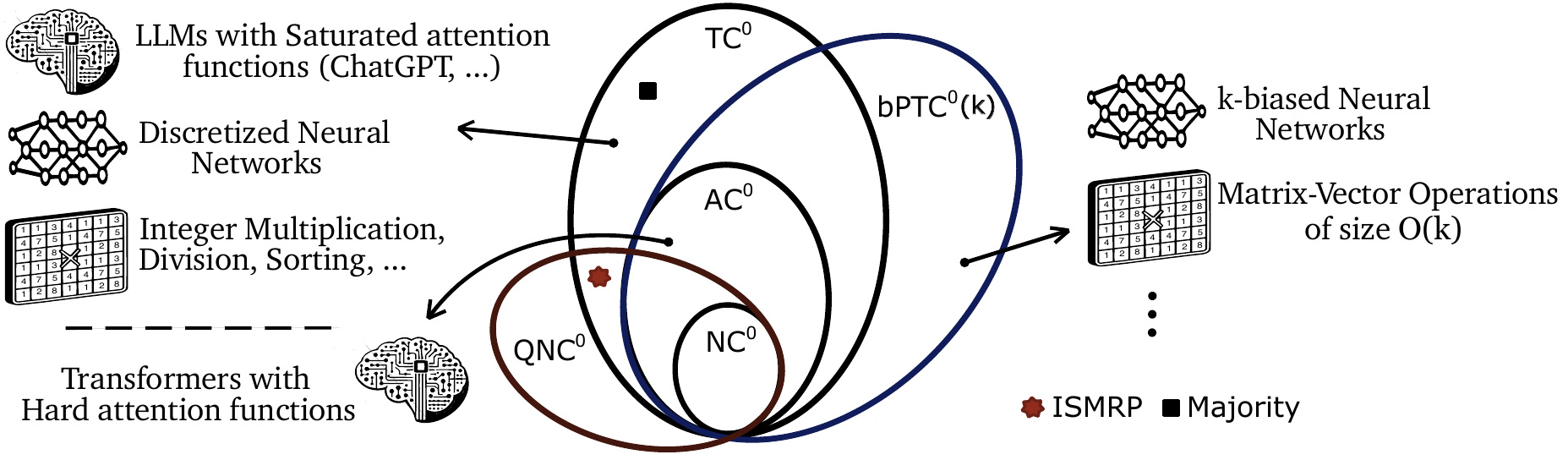}
    \caption{\justifying \justifying Relationships between the key classical and quantum circuit (complexity) classes considered. Notably, $\mathsf{Majority}$ is in $\TC^0$ but not $\AC^0$. The ISMR family, introduced here, separates constant-depth quantum circuits from biased polynomial threshold circuits (bias $k = \omega(\log n)$). The latter class contains $\AC^0$, and can solve $\NC^1$-complete problems for super-logarithmic biases and input lengths, suggesting non-trivial overlap with $\TC^0$. We also note computational tasks and models of significant practical value, such as neural networks. For example, Large Language Models (LLMs) that use the transformer architecture can be simulated by $\AC^0$ circuits when the attention mechanism is limited in certain natural ways \cite{strobl2024formal, Merrill2022}. Meanwhile, $\TC^0$, the standard for modeling discretized neural networks \cite{siu1991power, bertoni2002structural}, can simulate LLMs with realistic constraints on variable precision and autoregression \cite{Merrill2023,vsima2003general,parekh2018constant}, in the absence of more complicated elements such as feedback loops. The biased polynomial threshold circuits we study can simulate neural network variants and approximate activation functions controlled by the bias parameter $k$.}
    \label{fig:3a}
\end{figure*}

Finally, we focus on relational or search problems, where the inputs $x$ are $n$-bit strings, and the outputs $y$ are $m$-bit strings. Formally, we have valid input-output pairs $(x, y) \in \mathcal{R}$ for some relation $\mathcal{R} \subset \{0,1\}^n \times \{0,1\}^m$. As in previous results, these relations will have non-local games embedded into them. In this paper, we introduce the family of qudit $\mathsf{XOR}$ non-local games, designated Modular XOR games, described in \cref{fig:Xor}. Building on this class of multi-party non-local games, we introduce a corresponding family of relational problems, that we term Inverted Strict Modular Relation Problems (ISMRP). For any prime $p$, the ISMR problem $\mathcal{R}_{p}$ is defined as follows. Given an input $x\in\mathbb{F}_2^n$ such that $|x|\!\!\! \mod{p} = 0$, the goal is to output a string from the set
\begin{equation}
\mathcal{R}_{p} (x) := \left \{ y\ \Big |\ y \in \mathbb{F}_2^{m} :\ |y| = -\left(\frac{|x|}{p}\right)\!\!\!\! \mod{p}
\right \},
\end{equation}
where $|z|=\sum_{i=1}^n z_i$ is the $\ell_1$-norm for any string $z\in \mathbb{F}_p^{m}$, which is equal to the Hamming weight in the case of bitstrings. We typically choose the output length $m$ to be slightly larger than $n$.

These problems possess some intriguing structural properties, as the set of valid output strings is determined entirely by the modular residue of the input's $\ell_1$-norm $\mod{p}$. To describe how well a circuit can solve this search problem, we employ a correlation measure suited to its modular nature:
\begin{equation}
\mathsf{Corr}_{\mathcal{D}}(f,g)=\Expec_{x \sim \mathcal{D}}\left[\mathsf{Re}\left(e^{i\frac{2 \pi |f(x)|-|g(x)|}{p}}\right) \right].
\end{equation}
Here $\mathcal{D}$ is a distribution over input strings, and $f,g$ are $\mathbb{F}_p^m$-valued functions. Our notion of correlation for ISMR problems extends standard correlation for Boolean functions to mappings from $\{0, 1\}^n$ to cyclic groups $(\mathbb{F}_p,+)$, measuring the deviation between the $\ell_1$-norm of the true output $|f(x)|$ and an estimate $|g(x)|$. Perfect alignment, with $|f(x)|-|g(x)| \equiv 0\ (\mathsf{mod}\ p)$, maximizes correlation, while other values decrease correlation, with penalties growing as deviations approach $\frac{p - 1}{2}$. Thus, it accounts for ‘how wrong’ an output is.

\subsection{Higher dimensional qudits}\label{sec_qupits}
Qudit-based quantum computation has generated significant interest in recent years, harnessing multidimensional quantum states that are more common in nature compared to two-level systems. They offer greater accuracy and efficiency in information storage and processing, along with improved noise resilience \cite{chi2022programmable,ringbauer2022universal}. For example, complex entangled states such as multidimensional Greenberger-Horne-Zeilinger (GHZ) states and cluster states have demonstrated higher noise robustness compared to their qubit counterparts \cite{reimer2019high}. Additionally, algorithmic adaptations for qudits have empirically been observed to offer advantages, particularly in quantum simulations of complex systems, where qudits serve as a natural simulation platform. These capabilities have driven the development of qudit-based computing models across various hardware platforms. Moreover, the use of qudits could facilitate quantum advantage experiments that, in turn, may also serve as critical benchmarks for evaluating future quantum devices.

This motivates our first contribution, demonstrating separations in computational power between quantum circuits operating over prime dimensional qudits and biased polynomial threshold circuits of constant depth.

\begin{theorem}
\label{thminf:HigherDim}
\normalfont
For every prime $p$ and large enough $n$, the search problem $\mathcal{R}_p$ with $n$-bit inputs and $\mathcal{O}(n\log (n)^{p-1})$-bit outputs admits a constant-depth quantum circuit consisting of local gates over $\mathcal{O}(n)$ `qupits', consisting of $o(n^2)$ gates (i.e.\ of sub-quadratic size), that has constant correlation with $\mathcal{R}_p$. In contrast, \textnormal{any} polynomial-size biased polynomial threshold circuit of depth $d$ and bias $k=n^{1/(5d)}$, with access to random bits, has exponentially small correlation $\exp\left(-\Omega \left(n^{3/5 - \mathcal{O}(1)}\right)\right)$  with $\mathcal{R}_p$.
\end{theorem}


As noted, biased polynomial threshold circuits serve as a useful theoretical model for neural networks, and the extent of expressivity and computational power modeled by such circuits is controlled by the bias parameter. Our results demonstrate that biased threshold circuits with bias parameter values satisfying $k=\mathcal{O}(n^{1/d})$, require superpolynomial size to solve a relational problem that small quantum shallow-depth circuits over qupits  solve efficiently. This strengthens the provable quantum advantages over classical models with significantly more computational power than those in prior studies \cite{Watts19,caha2023colossal,Bravyi17} (see also \cref{eq:cont}). Crucially, we prove our bounds for the maximum bias parameter for the search problems considered---any larger bias $k$ would enable these circuits to efficiently compute functions like parity and solve the problem with high probability using linear-sized circuits. Thus, our results push the possible quantum advantage against to this class of circuits, in terms of bias, to its theoretical limit.


Previous studies have suggested that qudit non-local games with significant gaps in quantum versus classical winning probabilities could imply computational separations between classical and quantum circuit classes \cite{Sivert21}, though such separations were not explicitly demonstrated. In \cref{thminf:HigherDim}, the family of relational problems $\mathcal{R}_p$ for each prime dimension correspond to a family of Modular $\mathsf{XOR}$ non-local games that we define (\cref{fig:Xor}). For these non-local games, we constructively prove that classical strategies achieve, at best, exponentially lower success probabilities than quantum strategies as the number of parties increases. The translation to quantum computational advantage is then achieved using shallow-depth quantum circuits in each prime qudit dimension over a corresponding standard, minimal \textit{finite} gate set, with the precise instantiation of these circuits provided in the Supplementary Information (SI) Sec.\ B2. This approach broadens quantum advantages to an infinite family of problems and demonstrates the feasibility of implementing constant-depth realizations of such advantages on existing quantum hardware. 

\begin{figure}[htbp]
    \centering
    \includegraphics[scale=0.26]{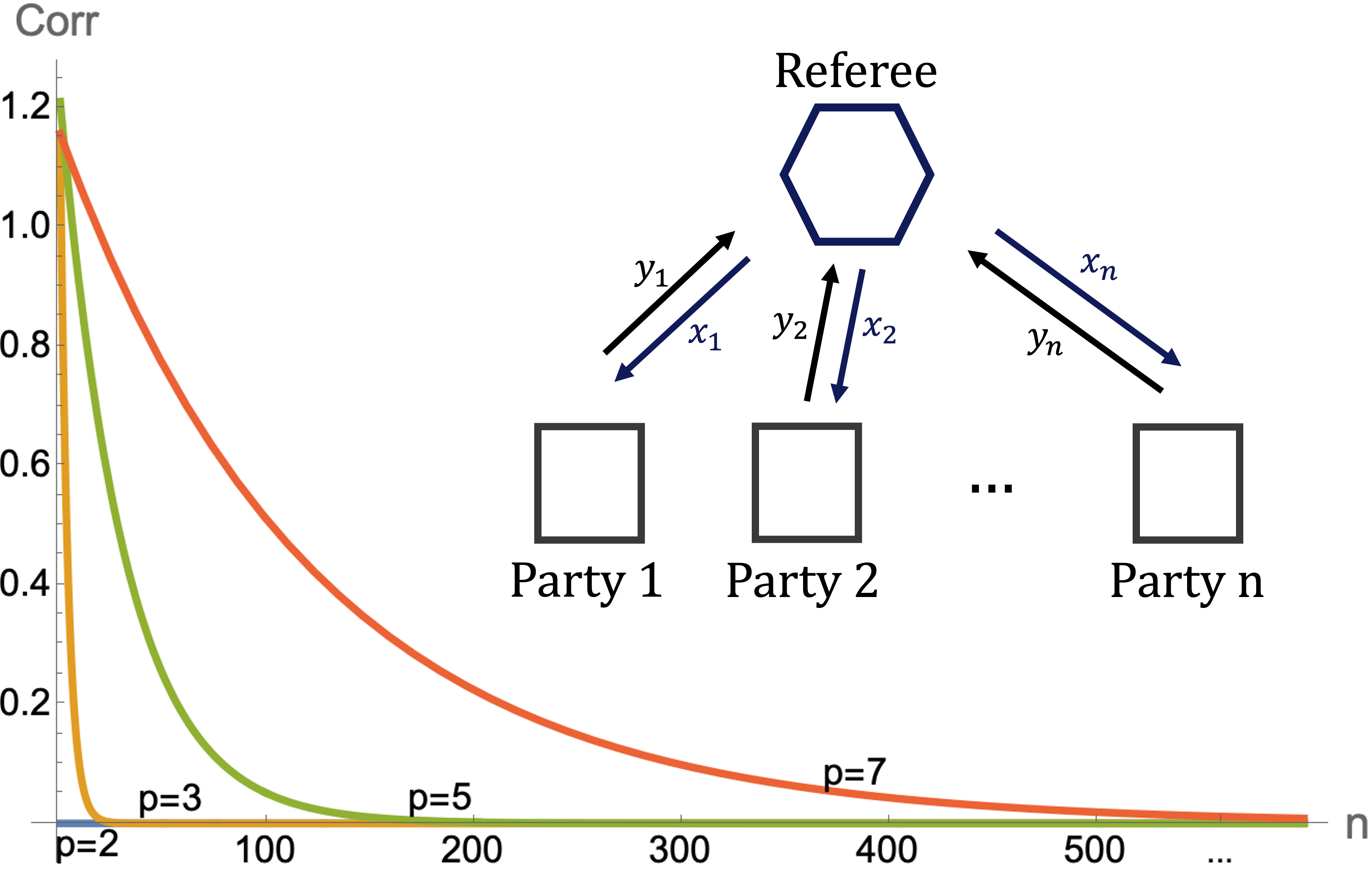}
    \caption{\justifying Representation of the $n$-party Modular $\mathsf{XOR}$ games $\mathcal{G}_p$. Each party $P_i$ receives an input $x_i \in \mathbb{F}_p$ from a referee, forming a combined input string $x$ with an $\ell_1$-norm $|x| = kp$, where $k \in \mathbb{N}$. Without further communication, each party responds with $y_i \in \mathbb{F}_p$, resulting in a collective output $y = (y_1, y_2, \ldots, y_n)$. The players win if the output $y$ has $\ell_1$-norm equal to the additive inverse of $k$ modulo $p$, i.e., $|y| = -k \mod p$. For context, $\mathcal{G}_2$ relates to the Mermin-Peres non-local game used in prior works \cite{Bravyi17,Watts19}. The plot illustrates the upper bounds on the (winning) correlation of classical strategies for $\mathcal{G}_2$, $\mathcal{G}_3$, $\mathcal{G}_5$, and $\mathcal{G}_7$, whereas the optimal quantum strategy attains a maximum correlation of 1 for all $\mathcal{G}_p$.}
    \label{fig:Xor}
\end{figure}


We also emphasize that this separation is achieved in the setting of average-case hardness with uniform input distributions over binary inputs.  To do this, our method elegantly manages \textit{non-uniform input distributions} for qupit non-local games, handling the necessary dit-to-bit mappings to ensure that we use uniform (rather than \textit{biased}) random restrictions for the biased polynomial threshold circuits. This approach allows us to obtain average-case $\BTC^0(k)$ hardness with respect to the uniform distribution on binary inputs for the search problems that demonstrate the separation. Additionally, it avoids the need to compute explicit success probability bounds for each game. In particular, we show that for our family of games parameterized by a prime $p$, the correlation of any classical strategy with a winning answer decreases exponentially with the number of parties in the game (see \cref{methods_A}). This correlation decay alone is sufficient to achieve comparable separations. We thus sidestep the obstacle of explicitly computing bounds on winning \textit{probabilities}: this has been emphasised in prior research, which predominantly relied on established quantum-classical distinctions in the winning probabilities for non-local games \cite{Bravyi17,Watts19}. Few developments beyond the conventional qubit setting were considered or achieved previously \cite{lawrence2017mermin}.  Finally, our search problems $\mathcal{R}_p$ remain binary, aligning naturally with Boolean circuits. They do not impose hardness by making clever use of arithmetic operations over other prime fields, making traditional Razborov-Smolensky type lower bounds inapplicable to this classical Boolean gate setting.


For the quantum solutions, we have generalized a method that translates optimal quantum strategies for these non-local games into small constant-depth quantum circuits that address the corresponding qupit search problems. We specifically leverage the use of generalized GHZ states in optimal strategies and the ability to generate multipartite entanglement between non-neighbouring qupits in constant depth. We have derived local unitary (LU)-equivalent generalized GHZ states and introduced a technique to efficiently describe the support of standard measurement outcomes of these states, factoring in phase dependencies. This enables us to use them in combination with an additional correction function that is computable in constant depth, approximating the outcome of the optimal quantum strategy. Moreover, it facilitates the determination of the success probability, highlighting the quantum circuits' superiority over classical approaches.


To understand the potential advantages of parallel quantum computation, we first need a candidate problem with an efficient quantum solution. A significant part of the challenge subsequently lies in proving lower bounds on the resources (for example, the circuit size or depth) required by the corresponding classical parallel models to solve this problem. Our approach introduces a novel multi-output multi-switching lemma that serves as a reduction tool, breaking down multi-output biased polynomial threshold circuits into simpler classical computational models, such as decision trees, to which locality and light-cone arguments can be applied (see \cref{methods_A}).  In \cref{thminf:HigherDim}, this new lemma plays a pivotal role by simplifying these circuits into forms that reveal their locality, directly linking their performance to the optimal classical strategies for the corresponding non-local games. More broadly, as these tools address not only binary decision problems but also search-type problems, this technique may hold independent significance, as neural networks are fundamentally sequence-to-sequence models. Therefore, this approach can be applied to relational problems beyond those examined here. Furthermore, parameterization by the bias parameter $k$ provides a continuum of classical computational power levels, allowing for broader applicability. In particular, we use this framework to clarify the relationship between hardware requirements and the potential for quantum advantage.


\subsection{Qubits}\label{sec_qubit}

Among all prime qudit dimensions, only two—$p=2$ and $p=3$—allow for an intuitive translation from quantum lower and classical upper bounds on correlations to corresponding bounds on success probabilities. Specifically, our results show that polynomial-size biased polynomial threshold circuits (of small bias) cannot solve the $\mathcal{R}_2$ and $\mathcal{R}_3$ problems with success probabilities appreciably larger than $1/2$ and $1/3$, both of which correspond to random guessing. In contrast, quantum circuits can solve the qubit case exactly, while in the qutrit case, the problem is solvable with a probability larger than and bounded away from $1/2$ by a constant. 

In the qubit case, we can in fact demonstrate something stronger: that if classical circuits are required to solve $\mathcal{R}_2$ exactly—that is, to produce a valid output string with certainty for all inputs—quantum advantages or separations of even greater magnitude are attained.

\begin{theorem} \label{thminf:qubit_exact}
\normalfont
For large enough $n$, the search problem $\mathcal{R}_2$ with $n$-bit inputs and $\mathcal{O}(n\log n)$-bit outputs can be solved by a constant-depth qubit quantum circuit with $o(n^2)$ gates, which on any valid input $x$ outputs $y$ such that $(x,y)\in \mathcal{R}_2$ with certainty. In contrast, the size $s$ of any depth-$d$ $k$-biased polynomial threshold circuit, with access to random bits, that computes any valid $y$ is lower bounded as in \cref{tab:qubit_exact}.
\begin{table}[!hbtp]
\begin{center}
\renewcommand{\arraystretch}{1.5}
\begin{tabular}{|c|c|}
\hline
 $\BTC^0(k)/\mathsf{rpoly}$ & \textnormal{\textbf{Exact hardness}}    \\  
\hline\hline
$k=\mathcal{O}(1)$  & $s=\Omega\left(\exp\left({\left(\frac{\sqrt{n}}{(\log n)^{3/2+\mathcal{O}(1)}}\right)^{\frac{1}{d-1}}}\right)\right)$  \\
\hline
$k=n^{1/(5d)}$ & $s=\Omega\left(\exp \left (\left( \frac{n^{3/10}}{\left(\log{n}\right)^{3/2+\mathcal{O}(1)}} \right)^{\frac{1}{d-1}}\right )\right)$ \\
\hline
\end{tabular}
 \caption{\justifying Lower bounds on the size of $\BTC^0(k)$ circuits solving the $\mathcal{R}_2$ problem, for different regimes of the value of $k$. 
 }
\label{tab:qubit_exact}
\end{center}
\end{table}
\end{theorem}

We remark that this exact-case analysis is well-motivated by the fact that the quantum circuit in the qubit case is actually a deterministic solution to $\mathcal{R}_2$, and therefore it is fair to compare it with classical deterministic circuits---to draw a complexity theoretic analogy, this is similar to comparing $\mathsf{EQP}$ and $\mathsf{P}$, as opposed to the more common comparison of $\mathsf{BQP}$ and $\mathsf{BPP}$. We also note that for $k=\mathcal{O}(1)$, constant-depth $k$-biased polynomial threshold circuits are equivalent to constant-depth logical circuits with unbounded fan-in Boolean gates (i.e. $\AC^0$), and so our bounds in the exact-case augment and improve on prior work \cite{Watts19}.

In proving \cref{thminf:qubit_exact}, we have developed a technique to capture the classical-quantum circuit separation using the \textit{algebraic normal form}, a standard representation for Boolean functions (see \cref{M_qubit_sec}). Importantly, this approach avoids the need for non-local or contextual games \cite{Sivert21}, which have traditionally been essential in prior work, while also showing that previous quantum advantages might benefit from new techniques to improve resource requirements and bring them closer to practical quantum demonstrations. 

We thus derive two lower bounds on the classical circuit size in the qubit setting. The most robust demonstration of quantum advantage arises in the setting of average-case hardness, where deviations of the classical success probability away from random guessing can be directly bounded. However, our exact-case hardness analysis reveals a quantum advantage of a greater magnitude, establishing higher resource requirements for classical circuits to match the performance of the quantum circuits they must compete with. This suggests that quantum advantage experiments against shallow-depth classical circuits could be achieved with fewer resources.

\subsection{Noise-robustness}\label{sec_noise}

Finally, quantum systems are unavoidably affected by noise, and error correction is a much more complex process in the quantum realm, bringing into question the \textit{robustness} of a computational advantage under noise. This question is of significance to both theory and practice, especially as we navigate the NISQ era. Even for very powerful quantum computational models, the introduction of noise often dramatically diminishes computational advantages that they may offer over their classical counterparts: for instance, recent work shows that even small constant error rates result in a collapse of the power of multi-prover interactive proofs where the provers share entanglement ($\mathsf{MIP}^*$) from $\mathsf{RE}$ to multi-prover interactive proofs without shared entanglement ($\mathsf{MIP}$) \cite{Dong24MIPnoise}.
Our third main result is to prove that all our separations are robust to noise: even if all steps of the quantum computation are affected by local stochastic noise, there is a family of modified relation problems that these noisy quantum circuits, when provided with logical magic states, can solve, but is hard for noiseless $\BTC^0(k)$ circuits. 

\begin{theorem} \label{thminf:noisy}
\normalfont
For every prime $p$ and large enough $n$, there exists a search problem $\mathfrak{R}_p$ with $n$-bit inputs and $\mathcal{O}(n\cdot \mathsf{poly}(\log n))$-bit outputs, such that for local stochastic noise with physical error rate below a constant threshold, there is a noise-resilient constant-depth quantum circuit over qupits with local gates and all-to-all connectivity, equipped with logical $\ket{T^{1/p}}$-magic states, that has constant correlation with $\mathfrak{R}_p$. In contrast, \textnormal{any} (even noiseless) depth-$d$ polynomial-size biased polynomial threshold circuit with bias $k=n^{1/(5d)}$ has exponentially small correlation $\exp\left(-\Omega \left(n^{3/5 - \mathcal{O}(1)}\right)\right)$ with $\mathfrak{R}_p$.
\end{theorem}
The search problem $\mathfrak{R}_p$ is defined using the ISMR problem $\mathcal{R}_p$, accounting for quantum error correction (see SI Sec.\ C). 

Local stochastic noise is a standard model used in quantum error correction, favored due to its ability to account for gate-level noise, noisy input state preparation, and noisy measurements while allowing for (weakly) non-local errors. Locality means that the probability of any given Pauli error decays exponentially with the number of sites it affects non-trivially. Additionally, this model effectively captures fabrication faults and aligns with standard physical descriptions of noise, wherein errors become exponentially less probable as their weight increases \cite{bombin2016resilience}. \cref{thminf:noisy} proves that our computational separations are robust to the presence of qupit generalized local stochastic noise in the quantum circuits.

This result improves upon prior work in three main ways. First, it extends unconditional separations between noisy, shallow-depth quantum circuits and classical circuit classes. In particular, this is achieved for all prime qudit dimensions and against the largest classical circuit class to date. Notably, while most qudit dimensions require logical $T$-type magic states as advice in the general formulation of \cref{thminf:noisy}, we have also demonstrated a noise-resilient quantum advantage using Clifford circuits over qubits that do not rely on such advice. Additionally, as a corollary, this establishes separations against $\NC^0$ for each prime qudit dimension, potentially enabling near-term quantum advantage experiments due to the reduced resource requirements for this class and the favorable error-resilience properties of qudits.

Second, our approach extends the framework for noise-robust quantum advantages beyond the qubit Clifford model introduced in Refs.~\cite{bravyi2020quantum}, by addressing qudit non-Clifford gates. Specifically, we demonstrate that for a particular CSS-type error correction code, this extension is achievable through the use of logical $T$-type magic states, which can themselves be affected by local stochastic noise, along with qudit magic-state injection protocols (see \cref{M_noise}). The need for these more complex quantum circuits arises from the inability to violate Bell inequalities with stabilizer states for any qupit dimension beyond qubits, as shown in \cite{gross2006hudson,howard2013quantum,meyer2024bell}, which also suggests that the same limitation extends to quantum-classical separations in circuit complexity. Additionally, as part of \cref{thminf:noisy}, we design new quantum circuits in the form of non-adaptive Clifford circuits with input-independent advice states to solve ISMR problems fault-tolerantly. These quantum circuits essentially give rise to our definition of the $\mathfrak{R}_p$ search problems.

Third, our work extends shallow-depth computational separations and error-correction mechanisms across arbitrary prime qudit dimensions, demonstrating their robustness. Previous research used the minimum weight perfect matching decoder, which performs poorly in higher dimensions. By using a different decoder, we show that we can still perform corrections and recover the desired states with exponentially high confidence. Specifically, we illustrate how the qupit surface code, when combined with the hard decision renormalization decoder, supports fault-tolerant implementation of the necessary quantum circuits. This advancement includes the development of single-shot logical state preparation for qupits. We have also extended the 3D block construction from \cite{Raussendorf_2005} to higher dimensions, showing that a particular measurement pattern yields a reduced state corresponding to a logical GHZ$_2$ state, up to local Clifford corrections.

\subsection{Resource estimates}\label{subsec:resource-estim}

When testing computational advantages with physical implementations, it is essential to pinpoint the circuit depth $d$ and input size $n$ (i.e., the number of input qubits) where a transition in circuit size occurs. Specifically, at what depths and input sizes does quantum advantage emerge? To address this, we estimate these values by solving for the points where our  new asymptotic lower bounds on the size of the best classical circuits match our corresponding quantum upper bounds. In \cref{tab:improved_table_spacing}, we present the input size $n$ for a given depth $d$, corresponding to the shallowest quantum circuits (with all-to-all connectivity) in each qudit dimension that outperform their classical counterparts in solving the respective ISMR problems.

For context, the transition point for Shor's factoring algorithm is estimated to be $\sim$1,700 qubits, $10^{36}$ Toffoli gates, and a circuit depth of $10^{25}$ \cite{chevignard2024reducing}, while for the HHL quantum matrix inversion algorithm it is roughly $10^8$ qubits and a depth of $10^{29}$ \cite{scherer2017concrete}. 

In comparison, recent advances in quantum hardware have prioritized scaling up the number of qubits over extending coherence times, leading to a greater emphasis on shallower quantum circuits \cite{lubinski2023application,bluvstein2023logical}. Notably, the quantum advantages over constant-depth classical circuits with bounded fan-in gates, as studied in this work, require only thousands of qudits across various dimensions. These setups can demonstrate classical intractability in tasks such as Bell violations \cite{Shalm15,Rauch18}. Classical circuits for these problems must have a depth of at least $d=\Omega(\log n)$, showing a clear quantum advantage when a quantum circuit solves the same problem at a strictly smaller depth. In practical scenarios with noise, the quantum circuit depth may increase by a constant factor, while the minimal classical circuit depth remains at least $d=\Omega(\frac{\log n}{\log\log n})$.  These noise-driven increases in the input size and other parameters required to observe quantum advantage is still significantly smaller, in terms of total resource counts, than what is required in other previous quantum advantage demonstrations.

Progressing up the hierarchy of computational power toward demonstrations of unconditional quantum advantage, the challenge lies in outperforming larger classical constant-depth circuit classes, such as biased polynomial threshold circuits. Achieving such quantum advantages would require significantly greater quantum resources, yet could still compare favorably to other \textit{conditional} quantum advantage demonstrations in noise-free settings. These quantum circuits maintain the advantage of shallow depth, potentially making them more practical for near-term quantum hardware. However, comparisons between noise-affected quantum circuits and noise-free classical circuits often demand unrealistic resources, exposing a fundamental imbalance in such analyses. Classical computing benefits from decades of refinement and classical error-correcting codes may also be necessary to achieve exponentially high efficiency rates akin to those expected of error-corrected quantum circuits \cite{gal2012tight}. Although noise levels in quantum and classical systems are unlikely to converge \cite{wang2024fault,acharya2024quantum,pinheiro09}, resource estimates that account for noise on both sides—combined with advances in quantum hardware—are anticipated to bring these comparisons closer to what is expected from noise-free quantum models.

Beyond the estimates in  \cref{tab:improved_table_spacing}, our bounds extend to all prime qudit dimensions and explore quantum circuits with varying hardware connectivities. In noise-resilient scenarios, all-to-all connectivity is required, while in noise-free comparisons, architectures can range from $p+1$-dimensional configurations for qudit dimension $p$ to full all-to-all connectivities. Among these, our estimates represent the lowest obtained for equivalent unconditional quantum separations, with the qubit case achieving the smallest resource requirements due to our analysis of exact-case hardness bounds. As mentioned before, while this setting reflects a less robust form of quantum advantage it nevertheless significantly reduces resource requirements, bringing theoretical predictions closer to the capabilities of current quantum devices. These estimates could be improved for all qudit dimensions by requiring classical circuits to better match the performance of the quantum circuits, and by adding connectivity restrictions to classical circuits. While these constraints are not formally part of the definition of shallow-depth circuit classes such as $\NC^0$, they reflect the limitations of realistic classical hardware and could lower the input size needed to demonstrate quantum advantage \cite{bharti2023power}. More generally, tighter lower-bound techniques and the discovery of computational problems with greater quantum advantages could further refine these estimates.

Finally, our estimates indicate that these quantum advantage experiments could serve as powerful quantum benchmarks, providing a systematic framework for evaluating and comparing computational capabilities. To explore larger quantum computational advantages, one could test classical circuits with larger fan-in gates, establishing new benchmarks and creating a structured "ladder of quantum advantages" to assess increasingly stronger computational separations. This hierarchy can be expanded by examining the ability of shallow quantum circuits to outperform more advanced classical circuit classes, such as the biased polynomial threshold circuits analyzed in this work.  Adjusting the bias parameter within these circuits further allows one to tune or amplify potential quantum advantages offered by specific quantum circuits and architectures. Notably, the levels within both hierarchies are separated by small multiplicative factors, making them an effective tool for benchmarking hardware improvements and guiding steady progress in quantum computing.

\begin{table*}[htbp]
    \centering
    \renewcommand{\arraystretch}{1.4} 
   \begin{tabular}{c@{\hspace{0.5cm}}ccccc@{\hspace{0.5cm}}  cccc}
        \toprule
        & \multicolumn{4}{c}{$\NC^0$ Regime} & & \multicolumn{4}{c}{$\BTC^0(k)$ Regime} \\
        \cmidrule(lr){2-5} \cmidrule(lr){7-10}
        Quantum circuit type & $\mathsf{Min}\ F^*$ & $F=2$ & $\cdots$ & $F=8$ & & $k=1$ & $k=2$ & $\cdots$  & $k=n^{1/(5d)}$ \\
        \midrule
        Qubits & \cellcolor{softgray} 2540 & \cellcolor{softgray} 9364 & $\cdots$ & \cellcolor{softgray}\num{1.6e8} & & \cellcolor{softgray}\num{4.3e13} & \cellcolor{softgray}\num{5.5e16} & $\cdots$  & \cellcolor{softgray}\num{5.0e26} \\
        \cellcolor{softorange} Noisy Qubits & - & \cellcolor{softgray}\num{7.0e10} & $\cdots$ &\cellcolor{softgray} \num{1.5e23} & & \cellcolor{softgray}\num{1.8e38} &\cellcolor{softgray} \num{8.0e44} & $\cdots$  &\cellcolor{softgray} \num{1.1e75} \\
        \midrule
         Qutrits & 125162 & 1952660 & $\cdots$ & \num{1.3e10} & & \num{5.5e14} & \num{1.1e18} & $\cdots$  & \num{2.3e30} \\
        \cellcolor{softorange} Noisy Qutrits & - & \num{3.5e11} & $\cdots$ & \num{9.0e30} & & \num{3.5e52} & \num{9.0e60} & $\cdots$  & \num{3.4e101} \\
        \midrule
         Ququints & \num{1.0e8} & \num{1.5e9} & $\cdots$ & \num{8.0e12} & & \num{9.5e14} & \num{1.7e18} & $\cdots$  & \num{5.0e30} \\
        \cellcolor{softorange} Noisy Ququints & - & \num{3.0e14} & $\cdots$ & \num{7.0e33} & & \num{5.5e52} & \num{1.4e61} & $\cdots$  & \num{7.0e101} \\
        \bottomrule
    \end{tabular}
    \caption{\justifying Estimates of the input sizes $n$ required to demonstrate quantum advantage using constant-depth quantum circuits in both noise-free and noise-affected settings, based on the quantum upper and classical lower bounds determined in \cref{sec_qupits}, \cref{sec_qubit} and \cref{sec_noise}. For the noise-free case, we consider the depth of four quantum circuits solving the $\mathcal{R}_p$ problems. These circuits, featuring $2n$ gates and all-to-all connectivity, generate the shortest possible solution strings, creating harder instances for classical circuits to replicate. The classical lower bounds for each problem are obtained from and depend on the deviation between the optimal classical and quantum winning strategies for the XOR non-local games $\mathcal{G}_p$. We also examine the minimal fan-in ($F^*$) scenario by comparing quantum circuits with classical circuits of equal locality—that is, having the same fan-in as the quantum circuits in each layer. This yields the lowest resource estimates for direct comparisons. Additionally, we analyze exact-case hardness bounds for the qubit setting to establish lower resource estimates. For the remaining qudit dimensions, we rely on average-case hardness to derive comparable estimates. \\  $\phantom{+ d}$In the noisy setting, we analyze depth-9 error-corrected quantum circuits for qubits and depth-11 circuits for qudits. The latter qudit circuits include additional logical operations implemented in a noise-resilient manner, requiring noise-tolerant versions based on the qudit surface code. We assume a code distance of order $\log(n)$, as no specific error threshold is defined for the local stochastic noise. This threshold is left as a parameter for further investigation and potential alignment with quantum hardware advancements. }
    \label{tab:improved_table_spacing}
\end{table*}

\section{Discussion}

In this paper, we advance the growing body of work demonstrating unconditional separations between the computational power of classical and quantum shallow-depth circuits, focusing on extending such results to hold against the largest classical circuit classes studied to date in this context. Specifically, we show that small local shallow-depth quantum circuits can efficiently solve search problems that polynomial-size circuits of $k$-biased threshold gates fail to solve with significant probability, even on average, and even for appreciably large bias parameter $k$.


We have developed a family of non-local games for qupits of each prime dimension and utilized the difference in winning probabilities between classical and quantum strategies to demonstrate that the computational separation established for constant-depth qubit circuits extends to constant-depth quantum circuits over higher-dimensional quantum systems. This, combined with the fact that these circuits are defined over a standard, minimal finite qupit gate set, allows for constant-depth realizations using the elementary operations available on standard quantum computing devices and fault-tolerant implementations. These explicit quantum advantages clarify the theoretical landscape and have practical relevance, as many quantum computing platforms naturally operate in higher dimensions \cite{ringbauer2022universal, cuadra22}. Furthermore, we hope our estimates inspire further advances in proof techniques and parameter optimization, narrowing the gap between theoretical predictions and the capabilities of near-term quantum devices toward achieving experimental quantum advantages.

Considering an infinite gate set, such as $\{\text{all single-qubit gates},\ \mathsf{CNOT}\}$ as assumed in Refs.~\cite{watts2023unconditional, Briet23, tani16} would certainly enable realizing all our qudit circuits using qubit circuits in a noise-free setting \cite{Barenco95, Reck94}. However, these realizations are not feasible in real hardware in constant depth \cite{hu2020space} or in a fault-tolerant manner. Thus, in a more realistic context, any finite minimal gate set defined over a specific qudit dimension would require decomposition into its native gates if used to solve one of the ISMR problems from another prime dimension.  In this regard, employing Solovay-Kitaev-type decompositions, these gates would likely necessitate log-depth decompositions to achieve suitable approximations. Therefore, under this hardware-realistic definition of constant-depth qudit quantum circuits, we conjecture that for each prime $p$, there exists a relational problem---in particular, our ISMR problem $\mathcal{R}_p$---that cannot be solved by a constant-depth quantum circuit using a minimal universal gate set for qudits of any dimension $q\neq p$, but can be solved by such a circuit using a gate set for dimension $p$. We propose that the ISMR problems could thus play a similar role for qudit constant-depth quantum circuit classes as Razborov-Smolensky-type modular problems do for the classical $\AC^0[p]$ circuit classes. If validated, this conjecture could reveal important aspects of parallel quantum computations that depend on specific system dimensions and might also guide hardware manufacturers in considering quantum system dimensions beyond qubits.

Having initiated the consideration of error-corrected qudit circuits for unconditional quantum separations, we speculate whether prior work on high-dimensional error-correcting codes with improved parameters \cite{Anwar_2014} could facilitate experimental demonstrations. We have reiterated the interest of magic states in the constant-depth fault-tolerant regimes \cite{Mezher20,Paletta2024robustsparseiqp}, and remark that we believe that magic state factories are unlikely to be parallelizable to the extent of being realizable by constant-depth quantum circuits. It is hence of great interest to understand what the simplest classical circuit class capable of simulating these processes is, as well as introduce the capacity for adaptive operations during error-corrected circuit execution. This offers potential avenues to extend unconditional quantum-classical separations beyond what is currently known and lift larger conditional separations to the simplest fault-tolerant quantum circuit classes \cite{yoganathan2019quantum}. In parallel to such complexity theoretic questions, it would be important to consider more complex and alternative noise models \cite{hasegawa21} to better align with specific practical quantum computing architectures. In the same vein, it is also of interest to understand if there is a complexity phase transition for some values of the bias parameter $k$ and the noise strength.

Finally, our computational separations could highlight potential advantages of quantum over classical machine learning models. Quantum machine learning has often been shown to outperform classical approaches by enabling the encoding of classically hard problems or leveraging quantum phenomena such as non-locality and contextuality. Some of these advantages can also be demonstrated using metrics that are standard for these learning models, such as the Kullback-Leibler divergence \cite{zhang2024quantum}. Our results complement this body of work. Additionally, biased polynomial threshold circuits naturally model a broad range of neural networks, including certain transformer architectures central to modern LLMs, leading us to conjecture that they support attention mechanisms beyond $\AC^0$ and likely beyond $\TC^0$ (for appropriate values of the bias). Consequently, we are optimistic that further work extending our line of investigation can reveal novel quantum advantages over classical machine learning models in these settings.

Moreover, a key challenge beyond proving quantum advantage is the difficulty of learning effective quantum solutions to various problems. In quantum machine learning, it has been demonstrated that models that are easy to train can often be classically simulated efficiently \cite{cerezo2023does}, thus limiting their quantum advantage. Conversely, more complex models that could potentially offer quantum benefits are theoretically more challenging to train. Therefore, identifying models that are well-suited for practical quantum advantage is crucial. Our research supports the notion that even simple and potentially easy-to-learn quantum circuits can outperform their classical counterparts significantly, reinforcing the idea that efficiently learnable quantum circuits can nevertheless provide practical advantages in information processing tasks \cite{anschuetz2024arbitrary, huang2024learning, gao2017efficient, anschuetz2023interpretable, zhang2024quantum, bowles2023contextuality, abbas2021power, Du21}.

\section{Methods}

\label{sec:proof-overview}
We now provide a concise overview of the proofs our main results. We aim to strike a balance between rigor and accessibility, and give an intuitive discussion of both the classical and quantum techniques employed. We will also highlight the key technical improvements we have made over existing work in this area.

\subsection{Proof of \texorpdfstring{\cref{thminf:HigherDim}:}{} Separations for higher dimensions (qupits)}\label{methods_A}
We introduce an infinite family of ISMR problems that generalize the parity halving problem introduced by Ref.~\cite{Watts19}, and the hidden linear function problem studied by Ref.~\cite{Bravyi17} reduces to the latter. The ISMRPs enable us to extend prior work on unconditional quantum advantage in two complementary directions: from qubits to all prime dimensional qudits, and beyond classical unbounded fan-in AND-OR circuits to the strictly more powerful class of biased polynomial threshold circuits. 

The proof has two parts. First, we show that multi-output biased polynomial threshold circuits of polynomial size have poor correlation with the ISMR outputs as the input size grows. Second, we demonstrate that quantum circuits using qupits can solve all instances of the ISMR problems with a constant positive correlation (depending only on $p$), regardless of input size.

The key challenge is to establish upper bounds for the correlation with which biased polynomial threshold circuits of size $s$ and depth $d$ can solve the ISMR problems on typical inputs of length $n$. This is the most technical part of the proof, as no previous lower-bound techniques existed for multi-output biased polynomial threshold circuits of constant-depth. To address this, and constrain the computational power of multi-output biased polynomial threshold circuits, we introduce a new multi-output multi-switching lemma that can reduce these circuits to decision forests. In this context, a decision forest consists of a single global decision tree with a set of decision trees at each one of its leaves, each computing a single output bit of the original circuit. 

\begin{lemma}\label{leminf:GCred2}
\normalfont
Let $f$ be a $k$-biased polynomial threshold circuit with $m$ output bits and $n$ input bits, of size $s$ and depth $d$. Then, there is a random restriction whose probability $p$ depends on $s$, $d$, $k$, $t$ and $q$, that reduces $f$ to a decision forest with global decision tree depth $2t-2$, and maximum depth $q$ of the decision trees at its leaves with probability at least $1-s\cdot 2^{-t+k}$.
\end{lemma}

By applying this lemma and selectively fixing variables in the global decision tree, we reduce the biased polynomial threshold circuits to $m$ independent decision trees, each computing a single outcome bit. This reduction limits the complexity of the circuit, while the problem retains its structure and hardness over the variables left unfixed by these restrictions. Further exploiting this asymmetry using locality arguments such as lightcone techniques on the final decision trees allows us to relate the efficiency of the initial circuits in solving the $\mathcal{R}_p$ problem to that of classical strategies in solving a non-local game embedded within it over the remaining variables. For example, in the qubit case, Mermin’s multi-player game \cite{Mermin90} is implicitly integrated into $\mathcal{R}_2$ over the variables preserved by the random restriction. Thus, for the qubit case, our new multi-output multi-switching lemma immediately shows that, for large $n$, any biased polynomial threshold circuit of depth $d\geq 4$, size $s \leq \exp(n^{1/(2d-2)})$ and $k\leq n^{1/(5d)}$  solves $\mathcal{R}_{2}$ with a success probability bounded by 
\begin{equation}\label{lemma:qubitave}
\frac{1}{2} +\exp\Big(-\Omega\Big(\frac{n^{2-o(1)}}{m^{1+o(1)}(k\cdot \log(s))^{2d}}\Big )\Big).
\end{equation}

However, in the higher-dimensional setting, many essential technical tools were previously undeveloped. To the best of the authors' knowledge, this is the first time the family of modular $\mathsf{XOR}$ non-local games related to the respective ISMR problems has been defined. Consequently, we need to determine upper bounds on the efficiencies of any classical strategies, whereas previous work relied on established non-local games with known bounds on optimal winning strategies.

\begin{lemma}\label{inflem:entire_breakp}
\normalfont
Let $w$ be any local probabilistic classical strategy that wins the Modular $\mathsf{XOR}$ non-local game $\mathcal{G}_p$ with $n$ parties exchanging messages that are values in $\mathbb{F}_p$. For inputs drawn according to the uniform distribution over strings of length $n(p-1)$ with $\ell_1$-norm satisfying $\big(\sum_{i=1}^n x_i\big)\  \MOD\ p=0$ and a fixed binary to base-p encoding, the maximal correlation is bounded by $\mathsf{Corr}\left(w, \mathcal{G}_p\right)  \leq (c_p)^{\frac{n}{p-1}}$, for a constant $c_p \in (0, 1)$. 
\end{lemma}

To establish these bounds for the non-local games $\mathcal{G}_p$—crucial for applying uniform random restrictions in the proof and deriving average-case hardness results—we consider encodings and decodings between uniform distributions over $\mathbb{F}_2^n$ and non-uniform distributions over $\mathbb{F}_p^n$. This approach is necessary because these non-local games are naturally defined as mappings of the form $\mathbb{F}_p^n\mapsto \mathbb{F}_p^m$. Additionally, these bounds rely on selecting an encoding that introduces a linear bias in the non-local games, as achieved by our chosen encoding, which can then be incorporated into our proof technique to establish tight upper bounds on the success probabilities of classical winning strategies (see SI Sec.\ B2).

Finally, by employing the same sequence of techniques—starting with the multi-output multi-switching lemma, followed by lightcone arguments, and considering blocks of bits that represent the dits of the non-local game, while leveraging the fact that these dits do not exhibit any specific structure assumed in previous works (e.g., \cite{caha2023colossal})—we relate the efficiency of the biased polynomial threshold circuit, as well as the intermediate $\NC^0$ circuits, to the upper bounds on optimal classical strategies. Specifically, we show that for sufficiently large $n$ and $q \in \mathbb{N}_{>0}$, any biased polynomial threshold circuit $C$ of depth $d\geq 4$, size $s \leq \exp(n^{1/(2d-2)})$, bias parameter $k\leq n^{1/(5d)}$ and access to random strings $\mathrm{rpoly}\in \mathbb{F}_2^{\mathrm{poly}(n)}$ solves $\mathcal{R}_{p}$ with correlation bounded by
\begin{equation}
\mathsf{Corr} \left(C,\mathcal{R}_p \right)= \exp \left (- \Omega\left  (\frac{n^{2-o(1)}}{m^{1+2/q}\log(s)^{2d-2}k^{2d}} \right )  \right ),
\end{equation}
\noindent for a uniform input distribution over strings $x\in\mathbb{F}_2^n$ that satisfy $|x|\ \MOD\ p= 0$.

To establish the quantum lower bound on correlation for the ISMR problems, we construct circuits that solve these problems with constant positive correlation by translating optimal quantum strategies for modular $\mathsf{XOR}$ non-local games into constant-depth quantum circuits that produce equivalent output strings. This approach first generates a qudit generalized ``poor man's cat state'', which is LU-equivalent state to qudit generalized GHZ states, essential for optimal quantum strategies. These states serve as resources, and we show that a random instance from this class of states can be generated using constant-depth qupit circuits, along with a string $z\in \mathbb{F}_p^n$ that suffices to define the operator mapping the state to the respective $p$-dimensional GHZ state.

Using these resource states, we apply the rotations and measurements defined by the optimal quantum winning strategies to produce an output string. However, higher-dimensional poor man's cat states introduce undesired phase factors from multiple inner products between the state-defining string $z$ and the input $x$, altering the expected measurement outcomes in subtle ways. To address this, we first present a concise representation of the output string's support based on the input and random strings defining the states. We then show that computing at least one of these inner products and incorporating it into the output string enhances efficiency in solving the original problems, while also demonstrating that a constant-depth classical circuit cannot compute at least one of these terms. Thus, we conclude that the quantum circuits solve these problems with constant, input-independent correlation. More formally, we prove that a constant-depth quantum circuit $C_Q$ solves the ISMR problem $\mathcal{R}_p$ on a uniformly random input from $\mathbb{F}_2^n$ that satisfies the condition $\sum_{i=1}^n x_i\ \MOD\ p = 0$  with a high correlation, namely 
\begin{equation}
\label{eq:correlation-4A}
    \mathsf{Corr} \left(C_Q,\mathcal{R}_p \right) = \frac{p-1}{p^2}.
\end{equation}
We remark that also for the qutrit case, the quantum advantage can be expressed through the success probability for the $\mathcal{R}_3$ problem, as the correlation function relates directly to success probabilities in both qubit and qutrit scenarios.

In summary, quantum advantages are achieved with qupit constant-depth circuits across various finite-dimensional connectivities. Each problem $\mathcal{R}_p$ can be addressed with qupit circuits of dimension $p$ and $p$-dimensional connectivity. However, to optimize these quantum advantages, careful consideration of parameters is essential. Specifically, the interplay between the bias parameter $k$, circuit size, and the dimensionality of constant-depth quantum circuits must be optimized. For instance, in the Measurement-Based Quantum Computation (MBQC) paradigm, all-to-all connectivity allows for shallower circuits and improved resource estimates, as shown in \cref{subsec:resource-estim}.

 \subsection{Proof of \texorpdfstring{\cref{thminf:qubit_exact}:}{} Qubit Exact-case hardness}\label{M_qubit_sec}

As mentioned before, since the quantum circuits actually solve the search problem with certainty, it is fair to ask what the hardness of an analogous classical exact solution is, exploring the boundary between deterministic and probabilistic circuits.
We use the term \textit{exact-case} to describe the ability of the circuit to solve the problem with certainty for all inputs. This differs from vanilla worst-case hardness, wherein the circuit is only required to succeed with a high enough (constant) success probability. In particular, from the average-case correlation bound in Ref.~\cite{Watts19}, one can determine a lower bound on the size of an $\AC^0$ circuit that solves this problem in the exact setting. A similar result follows for $\BTC^0(k)$ from our average-case hardness result referenced in \cref{lemma:qubitave}. However, these bounds are not tight, as we have demonstrated with our exact-case hardness result, which implies that even larger $\AC^0$ and $\BTC^0(k)$ circuits are required to solve the problem exactly. Focusing on deterministic classical circuits in this way has the benefit of revealing quantum-classical advantages at input sizes that are orders of magnitude smaller. 

To achieve this, we develop a deeper combinatorial understanding of the $\mathcal{R}_2$ problem. We observe first that the $\mathsf{XOR}$ of all the output bits is always equal to a fixed Boolean function, namely $\mathsf{LSB}:\{0,1\}^n\mapsto\{0,1\}$ which outputs the second least significant bit of the binary representation of the Hamming weight of the input $|x|$. In addition, each output bit produced by a $\BTC^0(k)$ circuit can be viewed as a distinct Boolean function $f_i:\{0,1\}^n\mapsto \{0,1\}$. By examining these functions in their algebraic normal form (ANF)—which represents each function as a polynomial over the field $\mathbb{F}_2$—we find that the $\mathsf{XOR}$ of the ANF of all these Boolean functions must match the ANF of the function $\mathsf{LSB}(x)$.

When the input distribution is supported over all the bit strings in $\mathbb{F}_2^n$ we would be able to use the ANF of the $\mathsf{LSB}$ function directly. However, $\mathcal{R}_2^m$ is a \textit{promise problem}, in that only strings of even parity are considered valid inputs. Thus, we modify the discussion above to work for \textit{partial} (or ``partially defined'') Boolean functions. To do this, we first prove a property about the $\mathsf{ANF}$ for all the (exponentially many) Boolean functions that equal the LSB function on our domain of interest. We do this by showing that all these functions require, in their $\mathsf{ANF}$ representation, at least $\Omega(n^2)$ degree-two terms. Secondly, we examine the capacity of decision trees, to which $\BTC^0(k)$ circuits can be reduced under random restrictions, in generating terms of degree two. As an illustrative example of how the tree depth relates to the $\mathsf{ANF}$ of Boolean functions, consider the parity function: any decision tree computing the parity of $n$ bits must have a depth of at least $n$. This ensures that the tree can produce all the degree one terms of the parity function's $\mathsf{ANF}$. See SI Sec.\ B1 for more details.

Finally, these two components, in conjunction with our switching lemma (\cref{leminf:GCred2}) that reduces $\BTC^0(k)$ circuits to decision trees, allow us to determine a minimal depth of the decision trees that directly translate to the minimal size for this class of circuits: for sufficiently large $n$, any $\BTC^0(k)/\mathsf{rpoly}$ circuit depth $d\geq 4$ and $k\leq n^{1/(5d)}$ that solves $\mathcal{R}_{2}^m$ must have size at least
\begin{equation}
    s \geq \exp\left(\widetilde{\mathcal{O}}\left({\big(n k^{-d} m^{-1/2}\big)^{1/(d-1)}}\right)\right).
\end{equation} 

Together with our quantum circuits for $\mathcal{R}_{2}^m$, this lower bound completes our proof of \cref{thminf:qubit_exact}. 

\subsection{Proof of \texorpdfstring{\cref{thminf:noisy}:}{Theorem 1.4:} Noise-resilient quantum advantage}\label{M_noise}

At a high level, we define a set of problems related to noise-tolerant, constant-depth quantum circuits for solving ISMR problems. These problems are designed to demonstrate the noise robustness of our quantum advantages. Specifically, we show that the ISMR problem can be reduced to a problem in our new family when an $\AC^0$ circuit can decode the output of these noise-tolerant constant-depth quantum circuits. This reduction extends the hardness results from the noiseless case to the noisy, achieving noise-robust separations with correlation bounds similar to those in \cref{thminf:HigherDim}.

We tackle two main technical challenges that go beyond previous noise-robust separations explored in Refs.~\cite{bravyi2020quantum,grier2021interactivenoisy}. First, for qudits of dimension $p\geq 3$, we must handle \textit{non-Clifford}  circuits, requiring us to show that these circuits can also be implemented fault-tolerantly—a new approach beyond prior work on unconditional separations between quantum and classical circuits. Second, we generalize quantum error-correction techniques used in Refs.~\cite{bravyi2020quantum}, including single-shot state preparation, decoding, and transversal constant-depth gate execution, to prime qudit dimensions while integrating methods for fault-tolerant non-Clifford circuits.

We address the first challenge by introducing logical advice states over qupits, enabling fault-tolerant non-Clifford operations through a qupit magic state injection protocol. We demonstrate that using a CSS-type error correction code that meets specific conditions (see SI Sec.\ A4), we can implement non-adaptive qupit Clifford circuits fault-tolerantly within the constant depth, for arbitrary prime dimensions, utilizing an advice state. Moreover, we show that for a code distance $l=\mathcal{O}(\poly\log n)$, if the circuit and advice are affected by local stochastic noise, respectively, $\mathcal{E}\sim\mathcal{N}(\varrho)$ and $\mathcal{E}_{A}\sim \mathcal{N}(\rho)$, then whenever both the physical error rates $\varrho$ and $\rho$ are below a threshold value scaling as $p_{th}=2^{-2^{\mathcal{O}(d)}}$, for any input $x\in \mathbb{F}^{n}$ we can show that $\Pr[\Dec^*(\mathcal{EC}(x))= C(x)]>0.99$, with $\Dec^*$ being the combined correction and decoding operation needed to retrieve the logical outcome from the encoded state generated by the fault-tolerant implementation $\mathcal{EC}$.

Our proof consists of two parts. First, we demonstrate that the qupit surface code meets all the necessary conditions for fault-tolerant implementation, specifically ensuring that it supports single-shot state preparation, transversal gate implementation in constant depth, and single-shot information retrieval using the selected decoder. Making these steps precise in higher-dimensional constant-depth quantum circuits requires new insights that extend beyond prior work. For the second part of our proof, we design new constant-depth non-adaptive Clifford circuits operating over qupits with input-independent advice states, capable of solving the ISMR problems.

\paragraph{Noise-resilient qupit Clifford circuits with quantum advice.} 
In higher dimensions, errors do not simply manifest as defects at the endpoints of the qupit lattice, as they do in the case of qubits. Instead, for every error, defects are likely to be scattered throughout the lattice. This distribution of defects motivates our use of the hard-decision renormalization group (HDRG) decoder, which has been shown to have good error-correction properties beyond the qubit case, overcoming issues associated with the minimum weight perfect matching decoder \cite{watson15,anwar2014towards}. In addition, for our setting, the information retrieval property necessitates that the code and decoder accurately perform logical $\overline{Z}$ measurements even under noisy conditions. To address this, we extend a result from Ref.~\cite{Bravyi_2013} to qupit generalized local stochastic noise. We demonstrate that the HDRG decoder’s probability of failure decreases exponentially as the qupit lattice size increases. Consequently, the HDRG decoder also reliably yields the correct outcome for logical measurements, provided the physical error rate does not exceed a threshold. Specifically, we establish that $\mathrm{Pr}[\mathrm{Success}] \geq 1 - \exp\left(-\Omega(m^{\eta})\right)$, where $m$ is the surface code distance and $\eta$ is a constant.

Although the transversal implementation of qudit Clifford operations in constant-depth, follows from the work of Ref.~\cite{Moussa_2016}, achieving single-shot state preparation in the qupit surface code is more complex. We show that single-shot state preparation can be achieved using the qudit surface code and the HDRG decoder for qupits of dimension $p\geq 2$. More precisely, we show that the 3D block construction from Ref.~\cite{Raussendorf_2005} preparing logical Bell pairs can be adapted for single-shot state preparation of $\mathsf{GHZ}_2$ qupit states, a capability not previously demonstrated. This involves defining functions necessary for the single-shot state preparation process, typically categorized as recovery and repair. The recovery function entails applying operations to retrieve the correct state from the randomness inherent in the noise-free process, while the repair function corrects the states based on the effects of noise that may occur during the described noise-free state preparation circuit and the recovery function.

The recovery function is derived directly from our generalization of the state preparation process. Regarding the repair function, we establish its feasibility up to the single-shot properties, using the alternative lifting properties proposed in Ref.~\cite{bravyi2020quantum} and considering a repair function based on the HDRG decoder. This shows that these logical states can be prepared for errors below a certain threshold with exponentially high confidence for increasing the lattice size. Finally, a detailed description of the code conditions and the full derivation of the solution described above can be found in SI Sec.\ C1.
\\

\paragraph{Non-adaptive Clifford circuits with magic state injection.}

To redesign the quantum circuits from \cref{thminf:HigherDim} for a noise-resilient architecture, we retain all Clifford operations, as they pose no challenges. However, addressing the non-Clifford rotations necessitates a new approach. The simplest solution is to use a gate teleportation device that incorporates the non-Clifford gate into the specified advice state. While these devices are adaptive, we must rely on non-adaptive gadgets to ensure noise resilience. Despite this constraint, we can implement these gadgets without adaptive correction. We account for the additional phases introduced by this absence and demonstrate that we can solve the ISMR problems through a more complex reduction based on the outcomes of these qupit circuits, leading to a significantly more intricate $\NC^0$  reduction. Thus, we obtain constant-depth non-adaptive Clifford circuits that use the advice state $\ket{T^{1/p}}^{\otimes n}$, consisting of magic states $\ket{T^{1/p}}:=\frac{1}{\sqrt{p}} \sum_{j=0}^{p-1} e^{i\frac{ 2\pi \cdot j }{p^2}}\ket{j}$, that solve the ISMR problem $\mathcal{R}_p$ with output strings of size $m=\mathcal{O}(n \cdot g^{p^3})$, on uniformly random input strings $x\in\mathbb{F}_2^n$ that satisfy the condition $\sum_{i=1}^n x_i\ \MOD\ p = 0$ with the exact correlation as in Eq.~\eqref{eq:correlation-4A}.

Furthermore, we demonstrate that a minimal universal gate set, which requires only the set of Clifford gates and a single magic gate, suffices. We prove that this magic gate is related to $T$-type qubit gates (see SI Sec.\ C2). Thus, we recover the standard noise-resilient circuit architecture, incorporating magic state injection gadgets for all qupit dimensions, thereby extending the scope of these results. 

As a corollary, we also obtain the same separation for qubits. However, since these are based on Clifford gates, we obtain a direct separation between noisy $\QNC^0$ circuits and $\BTC^0(k)$ circuits without any advice states.

\subsection{New switching lemma for biased PTF circuits}\label{new_switch}

As mentioned, we prove a new multi-switching lemma for $\BTC^0(k)$ circuits, obtaining tighter concentration bounds compared to prior work. We then use it to establish depth reduction for $\BTC^0(k)$ circuits with multiple output bits. Our approach parallels the strategies employed for multi-output $\AC^0$ circuits \cite{rossman2017entropy, Watts19}, suitably modified for $\BTC^0(k)$ circuits. Using this tool, we analyze how $\BTC^0(k)$ circuits reduce under random restrictions of the input variables, showing that we obtain \textit{decision forests} of low complexity with high probability. These more tractable decision forests can then be compared with quantum circuits via lightcone arguments.\\

\paragraph{Tighter multi-switching lemma for $\BTC^0(k)$.}
One of our technical contributions is a new multi-switching lemma, which shows that a finite set of depth-2 $\BTC^0(k)\circ\AND_w$ circuits reduces to a decision forest with high probability converging to unity as $1-O(2^{-t})$.

Our proof employs an inductive approach similar to the one used by Ref.~\cite{Hastad2016} for $\AC^0$ circuits. We conduct induction over the number of variables fixed by random restrictions and the circuits $f_i$ corresponding to each output bit, aiming to lower bound the probability with which the latter reduce to depth-$l$ decision trees. Simultaneously, when we encounter a circuit that does \textit{not} reduce to a decision tree of depth-$l$ with a set of variables fixed by random restrictions, we query the variables that remain ``alive" to forcefully simplify this circuit. These variables then become part of the global decision tree, sequentially growing a decision forest. For this approach to work, we address the issue of fixing variables and clauses that describe these circuits by combining our induction with canonical decision trees from the witness method (also used in Ref.~\cite{Kumar23}).

The above induction technique necessitates the use of \textit{downward-closed} random restrictions, which guarantees the monotonicity of decision tree size under random restrictions. The algorithm that constructs our canonical depth-$l$ decision trees corresponding to the leaves of the decision forest has two properties that ensure this: all random restrictions reducing the circuit to a fixed depth-$l$ decision tree overlap in the variables that they fix; consequently, for an arbitrary random restriction, fixing more variables does not increase the depth of the decision tree to which the initial circuit reduces.\\

\paragraph{Depth reduction for $\BTC^0(k)$.}
Adopting the proof technique of Ref.~\cite{rossman2017entropy} for $\AC^0$ circuits, we use our new multi-switching lemma to prove a new depth reduction lemma for $\BTC^0(k)$ circuits. The multi-switching lemma first reduces the depth by 1, from $d$ to $d-1$, for a computational object $\mathsf{DF}$ that is a decision forest feeding the inputs to a biased polynomial threshold circuit of depth $d$ with layers of the circuit having  $s_1,\ldots,s_d$ gates each, as demonstrated in the following lemma. 

\begin{lemma}\label{inf:gc_depth_opt}
\normalfont
A random restriction reduces $\mathsf{DF}$ to a decision forest, with a global decision tree of equal size as $\mathsf{DF}$ and decision trees of depth $l\geq \log s_1 + k + 2$ at the leaves feeding inputs to a $k$-biased polynomial threshold circuit of depth $d-1$ with layers of the circuit having  $s_2\ldots,s_d$ gates. In particular, for restriction probability $p$, this reduction in complexity happens with probability at least $1-s_1\cdot2^{k}(400wp)^{t/2}$. Here $t$ is the depth of the global decision tree, and $w$ the size of the decision trees at the leaves of the initial decision forest.
\end{lemma}

This key result enables us to reduce the depth of the circuit by one each time, replacing the layer that disappears with a decision forest. By iteratively applying this insight, we can completely reduce the circuit to a decision forest.

\begin{lemma}\label{inf:GClemma_opt}
\normalfont
Iterative application of \cref{inf:gc_depth_opt} completely reduces any $n$-input and $m$-output polynomial threshold circuit of depth $d$ and bias $k$ to a decision forest with global decision tree depth $2t-2$, and maximum depth $q$ of the decision trees at its leaves. If the restriction probability in step $i$ of the iteration is $p_i$, then the probability that the depth reduction succeeds is at least
\begin{equation}
1- \left(\sum_{i=2}^{d-1} s_i \cdot2^{k}\mathcal{O} (p_il_{i})^{t/2}  + 2^{k} m^{1/q} \mathcal{O}(p_d \cdot l_{d})^t\right).
\end{equation}
\vspace{1mm}
\end{lemma}

Finally, we choose the probabilities $p_1,\ldots,p_d$ for the sequence of random restrictions, and the depths $l_1,\ldots,l_d$ of the local decision trees, such that for global decision tree depth $t$, local decision tree depths $q$, initial circuit size $s$, and circuit type parameterized by $k$, we succeed in reducing the circuit to a decision forest with the probability defined informally in \cref{leminf:GCred2}. For the full derivation of the multi-output multi-switching lemma, refer to SI Sec.\ E; concurrent work of Ref.~\cite{grewal24} gives an independent derivation of a similar multi-switching lemma with different parameters.

\section*{Data availability}

No additional datasets were generated or analyzed in this study. All relevant numerical values, derived from the analytical formulas described in the text, are presented in the manuscript and available from the authors upon request.

\section*{Code availability}

Code sharing is not applicable to this article.

\vspace{0.15in}


%

\begin{acknowledgments} 
M.d.O. is supported by National Funds through the FCT - Fundação para a Ciência e a Tecnologia, I.P. (Portuguese Foundation for Science and Technology) within the project IBEX, with reference PTDC/CCI-COM/4280/2021, and via CEECINST/00062/2018 (EFG). S.S. is supported by a Royal Commission for the Exhibition of 1851 Research Fellowship. S.S. would like to thank Bruno Cavalar, Zhenjian Lu, and Ninad Rajgopal for insightful discussions.
\end{acknowledgments}

\section*{Author contributions}
The project was conceived by M.d.O., S.S., and M.-H.H. Theoretical results were proved by M.d.O. in discussion with S.S., with inputs from L.M. and M.-H.H. in parts of the proof of noise-robustness. Numerical implementations were performed by M.d.O. The first version of the manuscript was written by M.d.O. and subsequently improved by M.d.O. and S.S. with inputs from M.-H.H.

\section*{Competing Interests}
The authors declare no competing interests.


\captionsetup[table]{name=Supplementary table}
\captionsetup[figure]{name=Supplementary Fig.}
\setcounter{page}{1}
\renewcommand\thefigure{\arabic{figure}}
\setcounter{figure}{0} 
\renewcommand{\paragraph}[1]{\vspace{1em}\noindent\textbf{#1 }}
\setcounter{section}{0}
\setcounter{theorem}{0}
\setcounter{corollary}{0}
\setcounter{lemma}{0}
\setcounter{proposition}{0}
\setcounter{definition}{0}
\setcounter{remark}{0}
\setcounter{table}{0}

\onecolumngrid

\begin{center}
\large{ Supplementary Material: Unconditional advantage of noisy quantum circuits over biased threshold circuits in constant depth
}
\end{center}

Before beginning our main discussion, we introduce important notation that will be used throughout the Supplementary Material.

\section{Preliminaries}\label{app:prelims}

\medskip

In this work, we expand on the set of classical circuit classes against which we can establish a separation of computational power of $\mathsf{QNC^0}$, for a relational problem. Towards this goal, in \cref{subsec:prelims-circuits} we first set up the notation and definitions for the constant-depth classical circuit classes we study, and then in \cref{subsec:prelims-random-restrictions} introduce the most important techniques we use: random restrictions and switching lemmas. We then present some background on the quantum information theory of non-local games in \cref{subsec:prelims-nonlocal}, and noisy quantum circuits in \cref{subsec:prelims-noisy}.

Throughout the text we use the notation $[n]=\{1,2,\ldots,n\}$,  $\log$ denotes the logarithm to base two, and $\exp$ denotes the exponential function, with the base $e$ or $2$ that will be clear from context.

\subsection{Low-depth complexity classes}
\label{subsec:prelims-circuits}
We now explicitly introduce the circuit classes and computational models that we analyze in this work. For a detailed introduction to these topics, we refer to standard textbooks such as \cite{AroraBarak}. The \textit{depth} of a circuit is the maximum number of gates that are composed sequentially. The \textit{width} of a circuit refers to the largest number of gates applied simultaneously in any of the its layers. In terms of standard computational resources, depth is related to time, and width is related to the space of a computation. We say that a circuit is composed of bounded fan-in gates when the number of wires feeding logical values into each gate for processing is limited by a fixed constant. Conversely, unbounded fan-in refers to gates that can accommodate a polynomial number of input wires, relative to the input size, feeding into the logical gate for processing.

We assume the standard notation and definitions for classical notions such as $\NC^0$ and $\AC^0$. $\mathsf{NC}^0$ is the only circuit class that uses only bounded fan-in gates, but unbounded fan-out and arbitrary wiring. A family of circuits $\{\mathcal{C}_n\}_{n\geq 1}$ is said to be $\mathsf{X}$-uniform for a complexity class $\mathsf{X}$ if there is a Turing machine $T$ in $\mathsf{X}$ which on input $1^n$ outputs a description of $\mathcal{C}_n$ for each $n$. All of our classical circuits will be non-uniform, and admit randomised advice strings of polynomial length in the input size (denoted by $/\mathsf{rpoly}$), while our quantum circuits will always be uniform. 

\begin{definition}[$\mathsf{NC}^0$ circuits]
The class $\mathsf{NC^0}$ consists of classical constant-depth circuits composed of bounded fan-in $\AND, \OR$, and $\mathsf{NOT}$ gates, with a polynomial number of gates in the input size $n$. 
\end{definition}

The subsequent two circuit classes permit unbounded fan-in, which enhances their computational capabilities. They are distinguished by their computational power and the selection of gate sets.

\begin{definition}[$\mathsf{AC}^0$ circuits] 
The class $\mathsf{AC}^0$ consists of classical constant-depth circuits composed of unbounded fan-in $\AND, \OR$, and $\mathsf{NOT}$ gates, with a polynomial number of gates in the input size $n$. 
\end{definition}

\begin{definition}[$\mathsf{TC}^0$ circuits] 
The class $\mathsf{TC}^0$ consists of classical constant-depth circuits composed of unbounded fan-in threshold ($\mathsf{TH}$) gates and $\mathsf{NOT}$ gates, with a polynomial number of gates in the input size $n$. 
\end{definition}

Note that in fact circuits with sub-logarithmic depth, $d=o(\log n)$, are considered to be of constant depth. 

Our primary focus is $\QNC^0$, the class of constant-depth quantum circuits built from a universal and finite gate set, often regarded in the literature as the quantum analogue of $\NC^0$.
\begin{definition}[$\QNC^0$ circuits]
The class $\QNC^0$ consists of constant-depth quantum circuits composed of bounded fan-in quantum gates from a finite universal set, and a polynomial number of gates in the input size $n$.
\end{definition}

It is worth mentioning that $\mathsf{QAC}^0$ and $\mathsf{QTC}^0$, which are defined as the quantum analogs of $\AC^0$ and $\TC^0$ respectively, have also been actively studied, yielding novel and unexpected results when multi-qubit gates of unbounded fan-in are included \cite{Hoyer05,parham24,rosenthal2020,Anshu24,Morris24}. However, our discussion in this paper will focus on quantum classes without unbounded fan-in multi-qubit gates.

In addition to the standard classical circuit classes, our analysis will focus on a novel class featuring parameterized gates, first defined by \cite{Kumar23}. However, we will present these using a different notation, in the form commonly used in relation to polynomial threshold functions \cite{o2003new,kabanets2017polynomial,podolskii_et_al}. We first introduce the class of polynomial threshold gates restricted by a bias parameter. 

\begin{definition}[biased polynomial threshold gates]\label{def:PTF-full}  We denote by $\BT[k]$ the set of unbounded fan-in gates that implement two types of Polynomial Threshold Functions with bias $k$, defined as follows. The first type of gates implement $\OR$-type PTFs, of the form
\begin{align}
    f_{\OR}(x)=\begin{cases}
        P(x), &\sum_{i=0}^n x_i \leq k\\
        1, &\sum_{i=0}^n x_i > k
    \end{cases}; \, \, \, \hfill \text{ with } P:\mathbb{F}_2^n\to\mathbb{F}_2 \text{ a polynomial over }\mathbb{F}_2=\{0,1\}.
\end{align}
These gates permit arbitrary mappings of strings with a Hamming weight bounded by $\leq k$ and identically equal to $1$ for all inputs with a Hamming weight $> k$. The second type of gates includes all unbounded fan-in gates that implement the analogous $\AND$-type PTFs:
\begin{align}
    f_{\AND}(x)=\begin{cases}
        P(x), &\sum_{i=0}^n x_i \geq n-k\\
        0, &\sum_{i=0}^n x_i < n-k
    \end{cases}; \, \, \, \hfill \text{ with } P:\mathbb{F}_2^n\to\mathbb{F}_2 \text{ a polynomial over }\mathbb{F}_2=\{0,1\}.
\end{align}
$\AND$-type PTFs allow for arbitrary mappings of strings with Hamming weight $\geq n-k$ and are identically equal to $0$ for all inputs with a Hamming weight $< n-k$.
\end{definition}

Note that when $k=1$, we recover the usual unbounded fan-in $\AND$ and $\OR$ gates. Given this definition of biased polynomial threshold gates, we now introduce circuit classes composed of these gates, which will be the main objects of study in this work.

\begin{definition}[biased polynomial threshold circuits]
The class $\BTC^0(k)$ consists of classical circuits composed of unbounded fan-in $\BT[k]$ gates with constant depth, and a polynomial number of gates in the input size $n$.
\end{definition}

Note that for $k=\mathcal{O}(1)$, $\BTC^0(k)$ is equal to $\AC^0$. In addition to these circuit classes, we also work with decision trees, which are fundamental in the study of these classes due to their role in random restriction techniques. We denote by $\DT(t)$ the class of Boolean functions computed by depth-$t$ decision trees with a single output bit.

Furthermore, we will examine more complex decision trees, whose leaves do not just bear single binary values, but additional sets of decision trees. This structure allows for a hierarchical querying process: a global decision tree first addresses a set of variables, which in turn delineates a collection of local decision trees responsible for computing the bits of the final output string. To denote non-binary values at the leaves of a decision tree, we introduce the symbol $\leaves$ to signify that the entity to the left represents its leaves.

\begin{definition}[Decision forests] We use $\DT(w)^m \leaves  \DT(t)$ to denote the class of $(t,m,w)$-decision forests, defined as depth-$t$ decision trees whose leaves are labeled by $m$-tuples of depth-$w$ decision trees. 
\end{definition}

Building upon these objects, we introduce a new type of decision trees whose outputs become the inputs to a $\BTC^0(k)$ circuit. In the following, we write $\BTC[k;d; s_1, s_2, \ldots, s_d]$ to denote a $\BTC^0(k)$ circuit with depth $d$, and $s_i$ gates for each layer $i$ for $i \in \{1, \ldots, d\}$\footnote{These computational objects will be our $\BTC^0(k)$ equivalent of those introduced in \cite{rossman2017entropy} for $\AC^0$ and used to establish the respective depth reduction techniques.}.

\begin{definition}[biased threshold circuits with decision tree inputs] The object $\BTC[k;d;s_1,s_2,\hdots,s_d]\circ \DT(w)$ defines the class comprised of $\BTC^0(k)$ circuits with depth $d$, and $s_i$ gates per layer whose inputs are labeled by decision trees in $\DT(w)$.
\end{definition}

We also introduce a second object of this type, comprising a global decision tree whose leaves specify a collection of local decision trees; these, in turn, determine the input to a $\BTC^0(k)$ circuit.

\begin{definition}
$(\BTC[k;d;s_1,s_2,\hdots,s_d]\circ \DT(w)\leaves \DT(t))$. The object $\BTC[k;d;s_1,s_2,\hdots,s_d]\circ \DT(w)\leaves \DT(t)$ defines the class comprised of  depth-$t$ decision trees, whose leaves are objects in $\BTC[k;d;s_1,s_2,\hdots,s_d]\circ \DT(w)$.
\end{definition}

Finally, we also introduce the Algebraic Normal Form (ANF) which refers to polynomial representation of Boolean functions over $\mathbb{F}_2$.

\begin{definition}[Algebraic Normal Form.] Every Boolean function $f:\{0,1\}^n\mapsto \{0,1\}$ can be represented uniquely as, 
\begin{equation}
      f(x)= \bigoplus_{S \subseteq [n]} c_S \prod_{i \in S} x_i.
\end{equation}
\noindent with $c_S \in \{0,1\}$.
\end{definition}

\subsection{Random restrictions and Switching lemmas}
\label{subsec:prelims-random-restrictions}

Our analysis relies heavily on random restrictions, a technique that has been widely utilized in computational circuit complexity. 

\begin{definition}[$p$-random restriction] A $p$-random restriction $\rho \in R_p$ is a function mapping from a set of literals $\{x_i\}_{i\in I}$ to values $\{0, 1, *\}$. It is parameterized by the probability $p\in (0,1)$: each variable is independently kept `alive' using the symbol $*$ with probability $p$, or assigned the value $0$ or $1$ with equal probability $\frac{1-p}{2}$.
\end{definition}

Incorporating the concept of random restriction, the application of such constraints to a function $f$, denoted by $f\lceil \rho$, results in the definition of a new function evaluated according to $f(x \circ \rho)$. Put simply, the restrictions imposed by $\rho$ are applied to the inputs of $f$, where the surviving variables adopt the values from $x$ that are allowed to fluctuate. Thus, a restricted function can be viewed as a new function characterized by a reduced domain of input.

A random restriction can also be described by first choosing a ground truth $z \in \{0,1\}^n$ and a set of indices $\Lambda \subseteq [n]$. The random restriction, denoted as $\rho(z,\Lambda)$, is then defined such that $\Lambda$ determines the indices $x_i$ for which the corresponding bits are set to the symbol $*$. The remaining bits adopt the corresponding values from $z$. For instance, a $p$-random restriction in this context corresponds to a stochastic process where each index in $\Lambda$ is selected independently with probability $p$, and $z$ is chosen uniformly at random from the set $\{0,1\}^n$.

With the use of random restriction, it is possible to reduce the circuits from the classes defined previously to decision trees based on switching lemmas. In particular, we will discuss two important switching lemmas derived by \cite{Kumar23} for $\BTC^0(k)$ circuits, along with some additional useful lemmas from the same text.

\begin{theorem}\label{GC_switch}
\cite{Kumar23} Let $f$ be computable by a depth-2 circuit\footnote{We will refer to a circuit that computes a certain function using upper case letters, and reserve lower case letters for the corresponding abstract function.} $F$ comprised of a $\BT[k]$ gate which has as inputs $\AND$ or $\OR$ gates with maximal fan-in $w$, designated as $\BT[k]\circ \{\AND, \OR\}_w$. Then, 
\begin{equation}
\Pr_{\rho \in R_p}[F\lceil_{\rho} \notin \DT(t)] \leq (20pw)^t 2^k.
\end{equation}
\end{theorem}

The previous theorem indeed provides a switching lemma for $\BTC^0(k)$, exactly with the same structure as Hastad's switching lemma for $\AC^0$. Now, we will also consider a multi-switching lemma developed in the same work.

\begin{theorem}\label{multi_switch_K}
\cite{Kumar23} Let $\mathcal{F}=\{F_1,\hdots, F_m\}$ be a list of $\BT[k]\circ \AND_w$ circuits on $\{0,1\}^n$. Then, 
\begin{equation}
    \Pr_{\rho \in R_p} [ \mathcal{F}\lceil_\rho \notin  \DT(q)^m \leaves \DT(t)]\leq 4(64(2^km)^{1/q}pw)^t .
\end{equation}
\end{theorem}

Finally, we can leverage a useful lemma that facilitates the mapping between $\AND_w$ and $\OR_w$ clauses and decision trees.

\begin{lemma}\label{lem:tree_literal_red}
\cite{Kumar23} Any depth-2 circuit of the form $\BT[k]\circ \DT(w)^m$ comprised of a top (i.e., last) gate, which is a $\BT[k]$ gate with fan-in $m$, with $\DT(w)$ as inputs can be expressed as a circuit in $\BT[k]\circ {\AND,\OR}_w$ with a size of $m2^w$.
\end{lemma}

\subsection{Non-local games and Correlation}\label{subsec:prelims-nonlocal}

Non-local games are a fundamental element of study for quantum-classical separation in foundational terms. These games involve specific rules played by participants. It is possible to demonstrate that, for certain games where entanglement is a key resource, classical players cannot achieve the same winning probabilities as quantum players. Particularly of interest are XOR and generalized XOR games in the multi-party setting, as we relate these to computational relational problems. They are defined as follows.

\begin{definition}[$\mathsf{XOR}$ and generalized $\mathsf{XOR}$ games]\label{def:non_local} These classes of multiparty non-local games involve $n$ parties, denoted $P_1,\hdots,P_n$. In these games, each party $P_i$ receives information $x_i$, where $(x_1,\hdots,x_n)$ is drawn from a distribution $\mathcal{D}$. The parties, unable to communicate post-distribution, must each provide an individual response $y_i$, contributing to a collective output $y_1\otimes y_2 \otimes \hdots \otimes y_n$. The probability of winning the game is determined by the expression,
\begin{equation}
\Pr_{(x_1, \ldots, x_n) \sim \mathcal{D}} \left[\sum_{i=1}^n y_i\ \MOD\ p = f(x_1, \ldots, x_n)\right],
\end{equation}
where $f:\mathbb{F}_p^n\mapsto \mathbb{F}_p$ is a predetermined function. This function maps the input tuple from the finite field $\mathbb{F}_p^n$ to a single element in $\mathbb{F}_p$, thereby dictating the winning condition based on the collective responses $y_i$.
\end{definition}

These are typically associated with Bell inequalities \cite{collins02,Brunner14}, which define the statistics that can distinguish between local hidden variable theories (classical models) and quantum observations with quantum states. Specifically, we will utilize established definitions to determine the winning probabilities for games in this context. It is also worth mentioning that our primary interest lies in achieving the maximal classical winning probability, while aiming for quantum winning probabilities that surpass the classical limit, achievable through constant-depth quantum circuits.

Additionally, we need to establish a measure of efficiency for higher-dimensional non-local games and computational problems. Specifically, we consider that the correctness of the outcome of ISMR problems and the generalized $\mathsf{XOR}$ non-local games  depends solely on the modular residue of their $\ell_1$-norm. We can define proximity between relations based solely on that value. Furthermore, considering $\omega^{|C(x)|}$ again, with $\omega$ being the roots of unity, we encounter group arithmetic over the finite fields $\mathbb{F}_p$ with respect to the addition operation. Thus, we can assess the distance between two relations using the natural measure of the inner product for domains that form an Abelian group, defined as $\langle f,g\rangle=\sum_{|G|} f \overline{g}$ \cite{o2014analysis}. More precisely, we can use a general correlation measure for any circuit attempting to approximate an ISMR problem or a generalized $\mathsf{XOR}$ non-local game as follows,

\begin{definition}[Correlation between $\mathbb{F}_p$-valued functions]\label{def:correlation} We define the correlation between functions $f,g:\mathbb{F}_p^n\mapsto\mathbb{F}_p$ on a distribution $\mathcal{D}$ over $\mathbb{F}_p^n$ as,
\begin{equation}
\mathsf{Corr}_{\mathcal{D}}(f,g)=\Expec_{x \sim \mathcal{D}}\left[\mathsf{Re}\left(e^{-i\frac{2 \pi |f(x)|-|g(x)|}{p}}\right) \right].
\end{equation}
\end{definition}

More precisely, it will discount for deviation will be defined by its size and the respective real part of the roots of unity involved in each dimension (see \cref{fig:roots}).

\begin{figure}[htbp]
\centering
\includegraphics[scale=0.6]{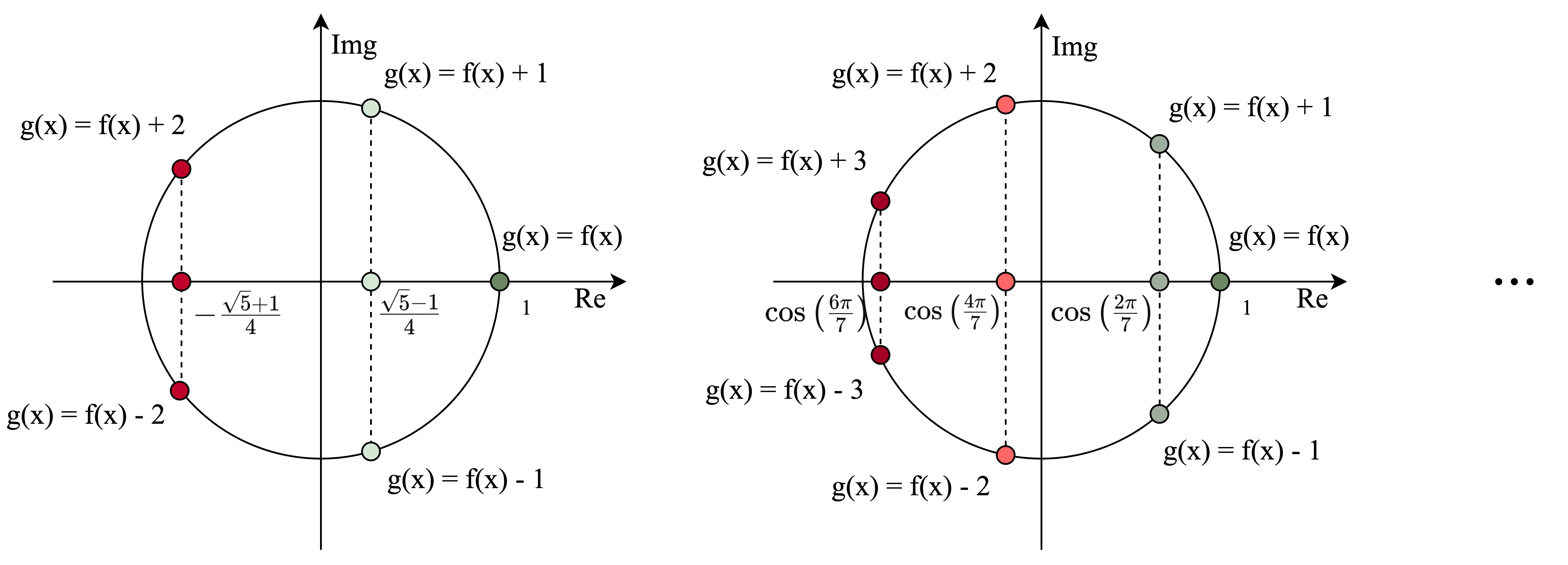}
\caption{Illustration of the respective contributions to the correlations function based on the deviation between two arbitrary function $f$ and $g$ over $\mathbb{F}_3$ and $\mathbb{F}_5$.}
\label{fig:roots}
\end{figure}

\subsection{Noise and quantum error correction in higher dimensions}
\label{subsec:prelims-noisy}
We will analyze noisy quantum circuits over qubits and qudits in \cref{sec:noisy-advantages}. To do this, we first introduce the noise model that we adopt and identify key properties for later use. We will then introduce the topological code that we use to attain noise-resilient quantum advantage.

\subsubsection*{Qupit operations and local stochastic noise}\label{qudit_operations}
We first establish preliminaries on the standard operations and gates used to characterize the impact of noise. In particular, for qupits, we need to introduce the generalized Pauli operators \cite{gottesman1998fault}, which are higher-dimensional versions of the single-qubit Pauli operators. Specifically, the Pauli $\mathsf{X}$ and $\mathsf{Z}$ operators generalize to the shift and clock operators, respectively
\begin{equation}
\mathsf{X} = \sum^{p-1}_{j=0}\ket{j}\bra{j{\oplus} 1}, \hspace{0.2cm} \mathsf{Z} = \sum^{p-1}_{j=0}\omega^{j}\ket{j}\bra{j}.
\end{equation} 
Both are non-Hermitian unitary operators, where $j \in \mathbb{F}_p$, $\omega = e^{\frac{2\pi i}{p}}$ and we will use ${\oplus}$ to denote addition mod $p$ respective to each qupit dimension under consideration. These operators have order $p$, and commutation relation
\begin{equation}\mathsf{Z}^{p}=\mathsf{X}^{p}=\mathds{1},\hspace{0.2cm} \mathsf{ZX} = \omega \mathsf{XZ}. 
\end{equation} 
For a system of $n$-qupits, we use the notation $\omega(a), \mathsf{X}(a), \mathsf{Z}(a)$ with $\mathsf{Z}(a) = \mathsf{Z}^{a_1}\otimes\cdots\otimes \mathsf{Z}^{a_n}$ for $a \in \mathbb{F}_p^n$. 

\begin{definition}[Generalized Pauli group]
The generalized Pauli group $\mathcal{P}_p$ is generated by the single qupit operators  $\mathsf{X}$ and $\mathsf{Z}$. The elements of $\mathcal{P}_p$ are given by $\omega^{k}\mathsf{Z}^a\mathsf{X}^b \text{ for }a,b,k\in \mathbb{F}_p$. The generalized Pauli group $\mathcal{P}_p^{\otimes n}$ over $n$ qupits is the $n$-fold tensor product of the single-qupit Pauli groups, and each operator can be expressed as $\omega(k)\mathsf{Z}(a)\mathsf{X}(b) \text{ for }a,b\in \mathbb{F}_p^n$ and $k\in \mathbb{F}_p$.
\end{definition}

In the rest of the paper we drop the subscript $p$ in $\mathcal{P}_p$ when it is clear from the context. Using the generalized Pauli operators, we can define the generalized version of the Clifford group.

\begin{definition}[Qupit Clifford group]\label{def_Cliff}
The Clifford group $\mathit{C}^n$ over $n$ qupits is the normalizer of the $n$-qupit Pauli group, defined as follows.
\begin{equation}
\mathcal{C}^n = \{C\in U(p^n)\hspace{0.2cm}| \hspace{0.2cm} \forall P\in \mathcal{P}^n, CPC^{\dagger} \in \mathcal{P}^n \}.
\end{equation} 
\end{definition}
\noindent Interesting elements of the generalized Clifford group include the qupit versions of the Hadamard and CNOT gates, known as the Fourier and Sum gates respectively, as well as the involution gate $\mathsf{INV}$. These can be defined as follows.
\begin{equation}
 \mathsf{F}=\frac{1}{\sqrt{p}}\sum^{p-1}_{x=0}\sum^{p-1}_{y=0} \omega^{xy}\ket{x} \bra{y},\ \mathsf{SUM} = \sum^{p-1}_{x=0} \ket{x} \bra{x}\otimes\ket{x \oplus y} \bra{y},\text{ and }\mathsf{INV}\ket{j}=\ket{-j}.
\end{equation}

We are now ready to introduce the error model we consider, along with some of its properties. In the following, for an $n$-qupit Pauli operator $\mathcal{E}$ we write $\Supp(\mathcal{E})$ to denote the subset $\mathcal{I}\subset[n]$ of indices of the qupits on which $\mathcal{E}$ acts non-trivially. That is, $\Supp(\mathcal{E})$ is the subset of qupits that are corrupted by $\mathcal{E}$. The local stochastic noise model is defined as follows \cite{gottesman2014faulttolerant, Fawzi_2018}.
\begin{definition}[Qupit local stochastic noise]\label{def:noise}
Let $0 \leq \tau \leq 1$. A random $n$-qupit Pauli error $\mathcal{E}$ with support $\Supp(\mathcal{E}) \subseteq [n]$ is said to be $\tau$-local stochastic if for all $F\subseteq [n]$, we have 
\begin{equation} 
\Pr[F \subseteq \Supp(\mathcal{E})] \leq \tau^{|F|}.
\end{equation}
\noindent We write $\mathcal{E} \sim \mathcal{N}(\tau)$ to denote random Pauli errors that follow a $\tau$-local stochastic noise model.
\end{definition} 

This is a very versatile noise model, being able to account for independent, correlated, and even adversarial noise, given that the probability of an error occurring decays as its support size increases. This means that local errors are more likely to occur than non-local ones. 

\begin{lemma} \label{stoch:proper}
The qupit local stochastic noise defined in \cref{def:noise}, over qupit dimension $p$, satisfies the following properties.
\begin{enumerate}
    \item Let $\mathcal{E} \sim \mathcal{N}(\tau)$. Then for any random error $\mathcal{E}'$ such that $\mathsf{Supp}(\mathcal{E}') \subseteq \mathsf{Supp}(\mathcal{E})$ with certainty, we have that $\mathcal{E}' \sim \mathcal{N}(\tau)$.
    \item Let $\mathcal{E} \sim \mathcal{N}(\tau)$ and $\mathcal{E}'\sim \mathcal{N}(\varrho)$ be (possibly dependent) random Paulis. Then, $\mathcal{E}\circ \mathcal{E}' \sim \mathcal{N}(\tau')$ where $\tau'=p\cdot\max(\sqrt{\tau},\sqrt{\varrho})$.
\end{enumerate} 
\end{lemma}
These properties are up to some constants equal across all qudit dimensions, and the proof for the qubit case presented in \cite{bravyi2020quantum} extends to the qupit case with only minor differences.

When a circuit is affected by the local stochastic noise model, each layer of the circuit, defining the support of the error, will be impacted by a random Pauli operator, adhering to the local stochastic property parameterized by the respective probability.

Importantly, this noise model allows for commuting the individual error operators through Clifford circuits. This enables the overall effect of noise, represented by random Pauli errors occurring after each layer of gates, to be described as amplified local stochastic noise occurring only at the end of the entire Clifford circuit.
\begin{lemma}\label{lemma:commute}
For any $n$-qupit Clifford circuit $U_d$ of depth $d$ composed of one- and two-qupit gates, the noisy implementation of $U_d$ described as $C_1\mathcal{E}_1C_2\mathcal{E}_2 \hdots C_d \mathcal{E}_d$ with $\mathcal{E}_i\sim \mathcal{N}(\tau)$ when applied to $\ket0$ and measured produces an outcome $z\in\{0,1\}^n$ with the following conditional probability.
\begin{equation}
\Pr(z|\mathcal{E})= |\bra{z}U_d\mathcal{E}_t\ket{0}|^2,
\end{equation}
\noindent with $\mathcal{E}_t\sim \mathcal{N}(p^2\tau^{4-d-1})$.
\end{lemma}
The proof in \cite{bravyi2020quantum} for qubit systems with local stochastic noise extends straightforwardly to the qupit case.

\subsubsection*{Qupit surface code}
To address noise, we must select a quantum error correcting code with an appropriate decoder to correct the errors based on the chosen noise model. We identify the surface code over qupits as a candidate for our general purpose.

Analogously to the qubit case, we have a square lattice of qupits of size $L\times L$, with smooth (square tiles) and rough (tiles without all the edges) boundaries. In particular, this topological code is characterized by the stabilizer subgroup $\mathcal{S}=\langle A_v,B_p \rangle$ defined as follows,

\begin{align} 
&A_v = X_{e} \otimes X_{e}^\dagger \otimes X_{e} \otimes X_{e}^\dagger \ \forall e \in V \\ 
&B_p = Z_{e} \otimes Z_{e} \otimes Z_{e}^\dagger \otimes Z_{e}^\dagger \ \forall e \in P,
\end{align}

\noindent where $V$ defines all the edges incident to a single vertex of the lattice $v$, and $P$ defines the plaquette constituted by all the edges that close a square/tile of the lattice. The entire  $L\times L$  lattice, containing $2L^2-2L-1$ qudits, defines a single logical qupit. Given this logical qupit, the logical operators $\overline{\mathsf{X}}$ and $\overline{\mathsf{Z}}$ are defined as follows
\begin{equation}
    \prod_{e\in \mathsf{Column}} X_{e},\ \  \prod_{e\in \mathsf{Line}} Z_{e},
\end{equation}
which also fixes the code distance to be $L$. Note that throughout the entire text, we will use the notation $\ket{\overline{a}}$ for the logical version of a state $\ket{a}$ and $\overline{A}$ for an operator $A$.

\begin{figure}[h!]
  \centering
  \begin{minipage}{.45\textwidth} 
    \centering
    \includegraphics[scale=0.45]{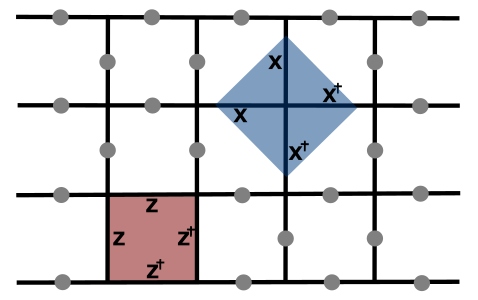}
    \caption{Illustration of $A_v$ and $B_p$ type surface code stabilizers on a $3 \times 3$ lattice.}
    \label{fig:surf_stab}
  \end{minipage}%
  \hspace{0.5cm} 
  \begin{minipage}{.45\textwidth}
    \centering
    \includegraphics[scale=0.45]{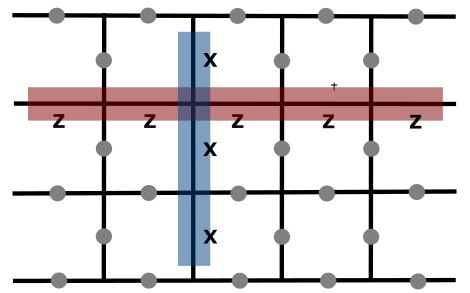}
    \caption{Representation of the logical operators $\overline{\mathsf{X}}$ and $\overline{\mathsf{Z}}$ for the surface code.}
    \label{fig:surf_opt}
  \end{minipage}
\end{figure}

Complementary to the code, we use a hard decision renormalization decoder \cite{Bravyi_2013}, which can be implemented for any stabilizer code with topological order, as is the case for the surface code that we will rely upon.

\section{Quantum advantage in the noiseless case}\label{noise_free}
We now proceed to formally state and prove our separations of noiseless shallow quantum circuits against biased polynomial threshold circuits. We divide the analysis into two sections, the first focusing on qubit $\QNC^0$ circuits, and the second on higher dimensional qupit $\QNC^0$ circuits. In each case we deal with the respective family of ISMR problems that demonstrates the separation from $\BTC^0(k)$. We also highlight qualititative and quantitative differences in the results achieved in the two cases, as well as the nature of the corresponding challenges and proof methodologies.

Let us first recall the ISMR family of problems. 

\begin{definition}\label{insec:defmod}
For any positive integer $p$, the Inverted Strict Modular relational problem \textnormal{(ISMRP)} $\mathcal{R}_{p}^m:\mathbb{F}_2^n \mapsto \mathbb{F}_2^m$ with input and output lengths $n$ and $m$ respectively is defined as follows. For any $x\in\mathbb{F}_2^n$ such that $|x|\ \MOD\ p = 0$ 
\begin{align}
\mathcal{R}_{p}^m (x) := &\left \{ y\ \big |\ y \in \mathbb{F}_2^{m} :\ |y|\ \MOD\ p = -\left(|x|/p\right)\ \MOD\ p
\right \},
\end{align}
where $|z|=\sum_{i=1}^n z_i$ is the $\ell_1$-norm of the bit string $z$. 
\end{definition} 

Here, we use the term ``strict" to denote that the problems require the outcome to be congruent modulo $p$ to the additive \textit{inverse} of $|x|/p$, rather than taking the interpretation of the Hamming weight of the input being congruent to zero modulo $p^2$.

The ISMR problems are defined for each prime $p$ and generalize the Parity Halving Problem introduced and studied by \cite{Watts19}, as they coincide when $p=2$. While several generalizations of PHP may be possible, ours was influenced by a specific family of non-local games for which we have derived an exponential separation between the success probabilities of the best classical and quantum winning strategies, as we shall see in \cref{subsec:qudit}. Note that we consider this binary setting as it provides a fair basis for a separation between the binary $\BTC^0(k)$ circuits and $\QNC^0$ circuits over qupits, not changing the arithmetic of the problem to dits, which could favor the quantum circuits and weaken the classical circuits class.

In the case of the qubit, we will prove that $\QNC^0$ circuits can solve the $\mathcal{R}_2^m$ problem exactly. Additionally, we establish exact-case and average-case hardness bounds for $\BTC^0(k)$. We use the term ``exact case'' to refer to the scenario where both quantum and classical circuits are required to solve the proposed problem with probability one (i.e.\ with certainty) for all inputs. Although this setting was implicitly considered in previous lower bounds, such as those determined in \cite{Watts19}, it was not explicitly explored. We address this gap and demonstrate that it is possible to establish stronger lower bounds on the size of the classical circuits for this scenario. Following this, we will present the average-case hardness for all other ISMR problems, $\mathcal{R}_p^m$ where $p>2$ is prime, along with the respective quantum (i.e.\ qupit $\QNC^0$) upper bounds.

\subsection{Separations for qubit cases}

The qubit $\QNC^0$ case is special because it enables us to solve one of the problems in $\mathcal{R}_2^m$ exactly, meaning it produces a correct outcome with probability 1 for all input strings. We leveraged this property to obtain larger bounds for the size of the first $\AC^0$ and then $\BTC^0(k)$ circuits required to solve the same problems, thereby achieving greater separations with $\QNC^0$ beyond the previous results for $\AC^0$ and simultaneosly completely new bounds with $\BTC^0(k)$. The following theorem formally establishes this separation for $\QNC^0$ with $3D$ and all-to-all connectivities.

\begin{theorem}\label{EsepQNC}
There exists a $\QNC^0$ circuit that solves the $\mathcal{R}_2^m$ exactly with a subquadratic number of gates, which an arbitrary $\BTC^0(k)/\mathsf{rpoly}$ circuit of depth $d$ with access to a random string $\mathsf{rpoly}$ and parameter $k$ that does require size $s$ circuit lower bounded as in \cref{tab_res1}, 
\vspace{-0.2cm}
\begin{center}
\begin{table}[!hbtp]
\begin{tabular}{|c|c|c|}
\hline
$\Omega(s)$ & \text{3D} & \text{All-to-all} \\  
\hline
$k=\mathcal{O}(1)$ $(\AC^0/\mathsf{rpoly})$ & $\exp\left({\left(\frac{n^{1/3}}{(\log n)^{1+\mathcal{O}(1)}}\right)^{\frac{1}{d-1}}}\right)$ & $\exp\left({\left(\frac{\sqrt{n}}{(\log n)^{3/2+\mathcal{O}(1)}}\right)^{\frac{1}{d-1}}}\right)$\\
\hline
 $k=n^{1/(5d)}$ $(\BTC^0(k)/\mathsf{rpoly})$&  $\exp\left (\left( \frac{n^{2/15} }{(\log{n})^{1+\mathcal{O}(1)}}\right)^{\frac{1}{d-1}}\right )$  & $\exp \left (\left( \frac{n^{3/10}}{\left(\log{n}\right)^{3/2+\mathcal{O}(1)}} \right)^{\frac{1}{d-1}}\right )$ \\
\hline
\end{tabular}
\caption{\justifying Size lower bounds for the circuit classes $\BTC^0(k)$ for different values of $k$ in solving the $\mathcal{R}_2$ problem.}
 \label{tab_res1}
\end{table}
\end{center}
\end{theorem}
\vspace{-0.6cm}

In particular, for this result, we will utilize the fact that we can reduce the $\AC^0$ and $\BTC^0(k)$ circuits to a first decision forest using \cref{lem:GCred2}, and then to forest lines with a second application of random restrictions. We define a forest line as a set of individual decision trees, where each computes one of the final outcome bits. We will then demonstrate that for this object to produce exact solutions, the decision trees comprising the forest line require a minimum depth with \cref{lem:max-tree-terms}, which consequently imposes a minimum size on the initial $\AC^0$ and $\BTC^0(k)$ circuits. More concretely, this is achieved by showing that the decision trees in this forest line require a minimum depth to produce terms of degree 2 with respect to the Algebraic Normal Form of the Boolean function they evaluate. Additionally, we establish a lower bound on the number of such degree 2 terms necessary to evaluate the relational problem exactly  \cref{lemma:term_lower_bound}. Therefore, by combining \cref{lem:GCred2} and \cref{lem:max-tree-terms}
we obtain our classical lower bound enunciated in  \cref{ElowerPHP}. 

In a subsequent subsection, we present quantum circuits that solve these problems exactly with $3D$ and all-to-all connectivities in \cref{3dqnc} and \cref{alltoallqnc} respectively. By integrating the quantum upper bounds, from which we obtain all the parameters, with the classical lower bounds, we derive the previously stated separation.

Our second result concerning the qubit $\QNC^0$ circuits is an average-case hardness separation from the $\BTC^0(k)$ circuits with a super-logarithmic size bias parameter $k$, addressing the same $\mathcal{R}_2^m$ problem, as stated in the following theorem.

\begin{theorem}\label{averagequbit}
There exists a $\QNC^0$ circuit that solves the $\mathcal{R}_2^m$ over a uniform distribution of even strings exactly with a subquadratic number of gates, which an arbitrary $\BTC^0(k)/\mathsf{rpoly}$ circuit of depth $d$ with access to a random string $rpoly$ and parameter $k$ solves the problem with probability bounded as in \cref{tab_res2}.
\vspace{-0.2cm}
\begin{center}
\begin{table}[!hbtp]
\renewcommand{\arraystretch}{1.5}
\begin{tabular}{|c|c|c|}
\hline
$\Omega(s)$ & \text{3D} & \text{All-to-all}  \\  
\hline
 $k=n^{1/(5d)}$ & $\frac{1}{2}+\exp\left(- \Omega \left(\frac{n^{4/15-\mathcal{O}(1)}}{ (\log s)^{2d-2}}\right)\right)$ & $\frac{1}{2}+\exp\left (-\Omega \left(\frac{n^{3/5 - \mathcal{O}(1)} }{(\log{s})^{2d-1}}\right) \right )$ \\
\hline
\end{tabular}
\caption{\justifying Upper bounds on the success probability for the circuit classes $\BTC^0(k)$ for different values of $k$ in solving the $\mathcal{R}_2$ problem.}
 \label{tab_res2}
\end{table}
\end{center}
\end{theorem}
\vspace{-0.6cm}

In this setting, we once again utilize \cref{lem:GCred2} to reduce the initial $\BTC^0(k)$ circuit to a decision forest and then, with a second application of random restrictions, to a forest line. However, in contrast to the previous case, we can directly apply \cref{NC_correlation_bound}, which determines the maximal average-case hardness of the forest line. This can then be correlated with the hardness of the $\BTC^0(k)$ circuits in solving the $\mathcal{R}_2^m$ problem, a relationship we formalize in \cref{lowerPHP}.

Afterward, we once again refer to the same quantum upper bounds \cref{3dqnc} and \cref{alltoallqnc} to characterize all the parameters and optimize the free parameters to achieve the largest possible separation.

\subsubsection{Exact and Average-case \texorpdfstring{$\BTC^0(k)$}{bPTF0(k)} hardness}

To solve exactly and deterministically an ISMR problem $\mathcal{R}_q^m$, in the sense that to each input $x$ a single valid output string $y$ is created, one can demonstrate that for any of these solutions, the outcome bits $y_i$ of $y$ are individual Boolean functions of the type $f_{y_i}:\{0,1\}^n\mapsto \{0,1\}$. Furthermore, these individual functions do not have to fulfill any specific conditions, some of these could even be constant functions. However, the parity of their outcomes must adhere to a property dictated by the relational problem. Specifically, the parity of these output bits computes another Boolean function.

Now we would like to specify this Boolean function specifically for the case of $\mathcal{R}_2^m$. In this case, for all the strings that are even we have, 
\begin{equation}
    \mathcal{D}_2(x)=\begin{cases}
        1,\ |x|\ \MOD\ 4=0\\
        0,\ |x|\ \MOD\ 4=2.
    \end{cases}
\end{equation}

This, however, does not define a total Boolean functions, as it only defines the images over even input strings. Moreover, this implies that any mapping of the odd input strings while respecting the previous mapping for the even input strings creates a valid Boolean function for the parity of the outcome bits. In particular, we can now derive which is exactly the set of valid Boolean functions.

\begin{lemma}\label{valid_set}
The parity of the outcome string $y$ of a $\BTC^0(k)$ circuit that does solve the $\mathcal{R}_2^m$ for arbitrary input $x\in \{0,1\}^n$ is contained in the following set
\begin{equation}\label{VanfSet}
S_{\mathcal{D}_2}=  \bigg \{   \underbrace{\bigoplus_{i_1=1}^{n-1} x_{i_1} \bigg(   \bigoplus_{i_2=i_1+1}^{n} x_{i_2} \bigg )}_{\text{Even mapping}} \oplus \underbrace{\bigg(  \bigoplus_{j=1}^{n} x_j \bigg) \wedge g}_{\text{Odd mapping}} \ \bigg | \  g\in \{0,1\}^n \mapsto \{0,1\}   \bigg \} . 
\end{equation}
\end{lemma}
\begin{proof}
The correctness of the previous functions can be analyzed simply by considering two cases: one where the input strings are even, which obliges it to map the same decision version of $\mathcal{D}_2$, and another where the input strings are odd, necessitating consideration of all possible mappings.

\paragraph{Even case.} In the case where the input strings are even, the expression $\left(\bigoplus_{j=1}^{n} x_j\right) \wedge g$ contributes nothing to the outcome. This means that only the first expression has any impact on the outcome. Furthermore, the first expression is known as the ``Second least significant bit" function, which computes the second bit of a binary representation of the Hamming weight of the initial bit strings \cite{Oliveira22}. More precisely, it has the following definition:
\begin{equation}
\mathsf{LSB}(x)=\begin{cases}
        1,\ |x|\ \MOD\ 4\in\{0,1\}\\
        0,\ |x|\ \MOD\ 4\in\{2,3\}.
    \end{cases}
\end{equation}
It immediately follows that for even input strings, all the elements of this set map the function to the same outcome as the parity of $\mathcal{R}_2^m$.

\vspace{0.2cm}
\noindent\textbf{Odd case.} For the odd case, we can be assured that the set contains all possible mappings of odd input strings. This is deduced by recognizing that the contribution of the second term is now equivalent to an arbitrary Boolean function $g$. Additionally, this allows for the mapping of any function such that it can be decomposed as $\bigoplus_{i_1=1}^{n-1} x_{i_1} \bigg(\bigoplus_{i_2=i_1+1}^{n} x_{i_2} \bigg ) \oplus h$, with $h$ being another arbitrary Boolean function. Thus, through this decomposition, we can cancel out the effect of the first term and obtain any arbitrary mapping for odd input strings.
\end{proof}

We now focus on a specific representation of these Boolean functions: their Algebraic Normal Form (ANF), which we previously used in \cref{valid_set}. This representation is crucial for capturing some common structure shared by all elements in the set of valid Boolean functions.  In particular, the subsequent lemma formalizes a key aspect of this structure - the number of degree-two terms present in the ANF.

\begin{lemma}\label{lemma:term_lower_bound}
Every Boolean function belonging to the set $S_{\mathcal{D}_2}$ contains at least $\Omega(n^2)$ degree-two terms in its Algebraic Normal Form (\textnormal{ANF}).
\end{lemma}
\begin{proof}
First, we will demonstrate that any function of the set $S_{\mathcal{D}_2}$ does have a $\Omega\big(n^2\big)$ term of degree two. We can conduct this analysis using the expression for $S_{\mathcal{D}_2}$ in \cref{VanfSet}. 

Consider the left-hand side, where the number of such terms is given by $\binom{n}{2}$. Now, let's examine how terms of degree-two can arise from the second part of the expression. Specifically, we consider terms generated by the operation $\bigoplus_{j=1}^{n} x_j \wedge g$, where $g$ represents a function of degree-one terms, as higher degree terms do not generate any degree-two terms. Note that degree-three terms either maintain or increase the degree, and any degree-two terms either increase the degree further or occur twice, cancelling each other out. It's crucial to note that if a degree-three term is generated in this process, it is subsequently cancelled out in the final ANF. Given this setup, we can generate the parity function from the first part of the expression. However, the degree-two terms that emerge from the $k$ degree-one terms in $f$ are bounded above by,
\begin{equation}
h(k) = \sum_{i=0}^k (n-i) - \binom{k}{2} = k(n-1) - \frac{k(k-1)}{2}.
\end{equation}

This expression accounts for the number of terms that are generated in total minus the ones that repeat and are immediately canceled out. Taking the derivative with respect to $k$, we find $\frac{dh}{dk} = n - 2k$. This derivative indicates that $h(k)$ reaches its maximum when $k = \frac{n}{2}$. At this point, the maximum number of degree 2 terms generated by $f$ can be computed. By examining the limit as $n$ approaches infinity, we find $\lim_{n \to \infty} \frac{g(n/2)}{\binom{n}{2}} = \frac{1}{2}$. This implies that for large $n$, approximately half of the potential degree 2 terms are effectively generated and not canceled out by the expression. Consequently, we conclude that any valid Algebraic Normal Form (ANF) must contain $\Omega(n^2)$ degree 2 terms, as intended.
\end{proof}

Simultaneously, we obtain that the ANF of the various $\BTC^0(k)$ circuits have to equal one element of the set $S_{\mathcal{D}_2}$ if the circuit does compute correctly the $\mathcal{R}_2^m$ problem. We can formalize the previous property as
\begin{equation}
\mathsf{ANF}\big(\underbrace{\BTC^0(k)(x)}_{y_1}\big) \oplus \mathsf{ANF}\big(\underbrace{\BTC^0(k)(x)}_{y_2}\big) \oplus \hdots \oplus \mathsf{ANF}\big(\underbrace{\BTC^0(k)(x)}_{y_m}\big) \in S_{\mathcal{D}_2}. \tag{P1}
\end{equation}

Subsequently, we will build a set of arguments to prove a lower bound on the size of these circuits to fulfill this property. However, we first solve a simpler instant, where we consider decision trees $\DT(q)^m$ in place of the $\BTC^0(k)$ circuits. This will not be a problem because later on we will be able to reduce these circuits to the previous objects. Moreover, each decision tree $\DT(q)$ of this forest line $\mathsf{FL}(m,d):=\bigoplus_{i=0}^m \DT(q_i)$ with $d=\max{q_i}$, has a unique description as an ANF over the variables that are part of its nodes. Obtaining with that a second property over this computational structures as follows, 
\begin{equation}
\mathsf{ANF}\big(\mathsf{FL}(m,q)\big)=\mathsf{ANF}\big(\underbrace{\DT(q)}_{y_1}\big) \oplus \mathsf{ANF}\big(\underbrace{\DT(q)}_{y_2}\big) \oplus \hdots \oplus \mathsf{ANF}\big(\underbrace{\DT(q)}_{y_m}\big) \in S_{\mathcal{D}_2}. \tag{P2}
\end{equation}

We want to prove the minimum decision tree depth concerning the largest decision tree depth in the previous object so that property $P2$ is fulfilled.

For that, we will consider that for any valid Boolean function, the number of degrees two terms in its ANF is lower bounded by a quadratic term, as shown in \cref{lemma:term_lower_bound}. Simultaneosly, we will consider \cref{alg:DT-to-ANF} which translates an arbitrary decision tree to its ANF. Here, again one can bind the number of terms of finite degree that a certain decision tree can have in its ANF. Combining these two ideas we obtain that there exists a minimum decision tree depth such that one can compute exactly one of the Boolean functions at hand.

\begin{lemma}\label{lem:max-tree-terms}
For any forest line $\mathsf{FL}(m,q)$ that computes a Boolean function from the set $S_{\mathcal{D}_2}$ does have a depth bounded by $q=\Omega\big(\frac{n}{m^{1/2}}\big)$.
\end{lemma}
\begin{proof}
This involves establishing an upper bound on the number of degree-2 terms generated by a decision tree, we examine the algorithm described in \cref{alg:DT-to-ANF} for converting a binary decision tree into its Algebraic Normal Form (ANF). Initially, the algorithm identifies paths that culminate in a Boolean true value, each represented as a clause. These clauses are subsequently merged using the logical $\OR$ operation, thereby formulating a polynomial in $\mathbb{F}_2$ that encapsulates the decision tree's Boolean function.

Focusing on the clauses formed with the logical $\AND$ operation during ANF creation. We note that these can be described by the path they follow in the tree that can be decomposed by the various left ($L$) and right ($R$) turns at each node assigned with a variable $x_i$.

\paragraph{Left Turn ($L$).} Incorporates the variable $x_i$ directly as
\begin{align}
path_{new}&=path_{old}\wedge x_i \\
&= \hdots \wedge x_{i-1}  \wedge x_i \ \OR \ = \hdots \wedge (x_{i-1}\oplus 1) \wedge x_i \\
&= \hdots x_{i-1} x_i \ \OR \ = \hdots x_{i-1} x_i \oplus x_i.
\end{align}
This case increases the degree of the new clause designated as $path_{new}$ by 1.

\paragraph{Right Turn ($R$).} Incorporates the complement of $x_i$ as
\begin{align}
path_{new}=path_{old}\wedge (x_i\oplus 1)\\ =(path_{old}\wedge x_i) \oplus path_{old}.
\end{align}

\noindent While the second case does not alter the degree of the new clause designated as $path_{new}$. 

Given these transformations, degree-2 terms arise from paths with specific sequences of turns. A single left turn surrounded by right turns $\mathsf{R}\ldots \mathsf{RLR} \ldots ,\mathsf{R}$ can contribute to a degree-2 term. Furthermore, a path with two left turns, each potentially followed by right turns $\mathsf{R}\ldots \mathsf{RLR} \ldots \mathsf{RLR} \ldots \mathsf{R}$, also contributes to degree-2 terms, as each left turn increases the degree by 1. The $\OR$ operations, which combine these paths, preserve the degrees of the terms in the resulting expression or increase it. Consequently, to enumerate the maximum number of degree-2 terms, it suffices to count the paths that can yield such terms. For a decision tree of depth $q$, the relevant paths are those with one or two left turns. The total count of such paths is given by the sum of paths with one left turn and those with two left turns $\binom{q}{2}+\binom{q}{1}$. This formulation provides a definitive upper bound on the number of degree-2 terms in the ANF representation of a Boolean function derived from a decision tree, predicated on the tree's depth.

In the end, we consider that in forest line $\mathsf{FL}(m,d)$, there are $m$ decision trees generating degree two terms and each can generate at most $\binom{q}{2}+\binom{q}{1}$ such terms\footnote{Notice, that we assume the existence of an efficient method to distribute the $n$ variables across $m$ decision trees in a way that generates all the corresponding degree-two terms. Specifically, this task aligns with the Steiner system problem, where we seek solutions for the $S(2,m/n,n)$ system. However, the existence of such a solution is not guaranteed, and the determination of the minimal set size remains an unresolved problem. Consequently, we will limit our consideration to the trivial upper bound.}. Finally, combining the previous bound with the minimum number of degree two terms necessary to prove in \cref{lemma:term_lower_bound} we obtain that the minimum depth will depend on the size of the forest line and the size of the input strings as follows $q \geq \frac{n}{m^{1/2}}$.

\end{proof}

\begin{algorithm}[ht]
\caption{Decision Tree to ANF Converter}\label{alg:DT-to-ANF}
\begin{algorithmic}[1] 
\Procedure{ANF}{$\DT$}
\State $S=\emptyset$, $\mathsf{ANF}_{\DT}=\mathsf{False}$
\State $S'=$RecCl$(\DT, root, \mathsf{True}, S)$
\For{$i=0$ \textbf{to} $|S'|$}  
\State $\mathsf{ANF}_{\DT} \gets S[i]\vee \mathsf{ANF}_{\DT} $.
\EndFor
\State \textbf{return} Reduce$(\mathsf{ANF}_{\DT})$.
\EndProcedure
\Statex
\Procedure{RecCl}{$\DT$, $node$, $path$, $set$}.
\If{$\DT[$node$]$=True}
\State \textbf{return} $set=set \uplus \{path\}$.
\Else
\If{$\DT[node]=\mathsf{False}$}
\State \textbf{return} set.
\Else
\State \textbf{return} RecCl$(\DT, \mathsf{Left}(\DT,node), \mathsf{Var}(node) \wedge path, set)$ 
\State $\uplus$ RecCl$(\DT, \mathsf{Right}(\DT,node), (1\oplus \mathsf{Var}(node)) \wedge path, set)$
\EndIf
\EndIf
\EndProcedure
\Statex
\Procedure{$\wedge$}{$\bigoplus_{i=0}^{k_i} x_i$, $\bigoplus_{j=0}^{k_j} y_j$}
\State \textbf{return} $\bigoplus_{i\in [k_i],j \in [k_j]} x_iy_j$
\EndProcedure
\Statex
\Procedure{$\vee$}{$x$, $y$}
\State \textbf{return} $x\oplus y \oplus xy$
\EndProcedure
\end{algorithmic}
\end{algorithm}

Now we will be able to prove a lower bound for the size of any $\BTC^0(k)$ circuit that does compute the $\mathcal{R}_2^m$ exactly. This will be followed by applying our new switching lemmas created in \cref{SI_new_switch} to reduce the initial circuits to a decision tree of the type $\DT(w)^m \leaves \DT(t)$. Subsequently, we will apply a second set of random restrictions, that reduce the same object to a forest line $\mathsf{FL}(m,q)$ for which we will be able to prove a minimum depth, and consequently link this depth to the minimum size of the initial circuit. 


\begin{lemma}\label{ElowerPHP}
For sufficiently large $n$, any $\BTC^0(k)/\mathsf{rpoly}$ circuit depth $d\geq 4$ and $k\leq n^{1/(5d)}$ that solves the $\mathcal{R}_{2}^m$  has size no smaller then $s \geq 2^{\big(\frac{n}{ k^d\cdot  q\cdot  m^{1/2+1/q}}\big)^{1/(d-1)}}$, with $q \in \mathbb{N}_{>0}$. 
\end{lemma}
\begin{proof}
This proof will be decomposed first into the application of a random restriction $\rho$ with a probability equal to $p=\frac{1}{m^{1/q}\cdot \mathcal{O}(\log(s)^{d-1}\cdot k^d}$. Notice that this is the largest probability one can use without disrespecting the various switching lemmas, which have the minimal conditional of generating probabilities below 1. We additionally choose $t=pn/8$ and that $s\leq \exp(n^{1/(2d-2)})$, which will not interfere with our result given that this value will be asymptotically larger than the lower bounds of $s$ derived with this assumption. With that, we first show that $s \leq 2^{t/2}$, this results from the following comparison and the previously defined values for the parameters, 
\begin{equation}
    2^{t/2}= exp\big(\Omega(pn)\big)= exp\Big (\Omega \Big ( \frac{n^{(d-1)/(2d-2)-\frac{1}{5}}}{m^{o(1)}} \Big) \Big ).
\end{equation}

\noindent Also, we have that for any $d\geq 4$ then $s= \exp(n^{1/(2d-2)}) \leq exp\Big (\Omega \Big ( \frac{n^{(d-1)/(2d-2)-\frac{1}{5}}}{m^{o(1)}} \Big) \Big )$. This ensures, as established by \cref{lem:GCred2}, that with a probability exceeding $1-s^2\cdot2^{-t}$, which is greater than $1-2^{-t/2}$ and consequently surpasses $\exp(-\Omega(pn))$, the initial $\BTC^0(k)$ circuit simplifies to a tree of the type $\DT(q-1)^m \leaves \DT(2t-2)$.

The second element is that at least $pn/2$ variables remain active with very high probability. This conclusion is reached directly through the application of the Chernoff bound, given that the probability of a variable staying active is $\Omega(pn)$. In conjunction with the aforementioned probability, this suggests that
\begin{equation}
\Pr\big [F\lceil_\rho \in \DT(q-1)^m \leaves \DT(2t-2), (x\circ \rho)_{S\subseteq [n]}=* \text{ for }|S|\geq pn/2\big ] = 1- \exp(-\Omega(pn)).
\end{equation}

A second random restriction $\tau$ will be applied to the $2t$ variables that are in the global decision tree of $\DT(q-1)^m \leaves \DT(2t-2)$ after applying the first random restriction $\rho$. This random restriction then reduces with high probability each one of the initial $\BTC^0(k)$ circuits that compute the resulting outcome bits to a forest line $\mathsf{FL}(m,q-1)$. Therefore, for all the variables that are kept alive, the $m$ output bits are computed with local decision tree $\DT(q-1)$, 
\begin{equation}
\Pr\left[F\lceil_{\tau\circ \rho} \in \mathsf{FL}(m,q-1), (x\circ(\tau\circ\rho))_{S'\subseteq [n]}=\text{ for }|S|\geq pn/4 \right] = 1- \exp(-\Omega(pn)).
\end{equation}

Now as we have been using $p$-random restrictions we obtained that sampling even strings consistent with the restrictions does provide us with a uniform distribution over the even input strings. Exactly, the same argument does work over the restriction $\tau$ given that the variables of the global decision tree are selected randomly in $\{0,1\}^{2t}$. This property of $p$-random restrictions will now be applied to the \cref{valid_set} defining the necessary outcome that these resulting forest line $\mathsf{FL}(m,q-1)$ has to produce. In particular, we can simply consider the effect of the random restrictions on the variables as,
\begin{equation}\label{VanfSet2}
S_{\mathcal{D}_2}\lceil_{\tau\circ \rho}=  \bigg \{  \bigoplus_{i_1,i_2\in [n]\setminus [\tau\circ \rho] } x_{i_1}x_{i_2} \bigg ) \oplus \bigg(  \bigoplus_{j\in [n]\setminus [\tau\circ \rho] } x_j \bigg) \wedge g\lceil_{\tau\circ \rho} \ \bigg | \  g\in \{0,1\}^n \mapsto \{0,1\}   \bigg \} 
\end{equation}

\noindent with $[\tau\circ \rho]$ representing the set of indexes of the variables that are assigned values by the respective random restriction. We obtain that the reduced expression does keep exactly the same structure. This demonstrates that over the $pn/4$ variables alive \cref{lemma:term_lower_bound} applies equally. Thus, we can apply \cref{lem:max-tree-terms} considering the variables alive and the resulting forest line. Consequently, such that this object computes correctly our relational problem we have that, 
\begin{align*}
q&\geq \frac{(p\cdot n)}{m^{1/2}}= \frac{n}{m^{1/2+1/q}\log(s)^{d-1}k^d} \\
\log(s)^{d-1} &\geq \frac{n}{m^{1/2+1/q} k^d q} \\
s &\geq 2^{\left(\frac{n}{m^{1/2+1/q} k^d \cdot  q}\right)^{1/(d-1)}}.
\end{align*}
The proof remains applicable in cases where the $\BTC^0(k)$ circuit receives a random string as advice. This is grounded in our demonstration that no deterministic circuit can accurately compute the function. Therefore, any probabilistic strategy relying on random advice is destined to fail, as it merely selects from the available deterministic solutions. As a result, $\BTC^0(k)/\mathsf{rpoly}$ circuits are equally susceptible to failure under these conditions.
\end{proof}

After obtaining a tight lower bound for the size of the $\BTC^0(k)$ circuits that compute exactly the modular relational problem that one can solve with $\QNC^0$ circuits on qubits. We intend to also create the average-case hardness scenario which will be of interest by itself as it does also serve to build a noise-resistance version of this quantum-classical separation. In particular, we will make use of the following theorem which will be applied to the forest line to obtain the correlation bound in contrast to our previous size bound for the exact case. 

\begin{theorem}\label{NC_correlation_bound}
\cite{Watts19} Let $C$ be an $\NC^0/\mathsf{rpoly}$
circuit with $n$ inputs, $m$ outputs, and locality $l$.
Then $C$ solves $\mathcal{R}_{2}^m$ on a random even-parity input with probability at most $\frac{1}{2} + 2^{-\Omega(\min(n,\frac{n^2}{l^2m})}$
\end{theorem}

The previous theorem does provide us with the tool to relate the resulting forest line, which we obtain from the set of random restrictions reducing the $\BTC^0(k)$ circuits, with the probability with which it computes correctly the $\mathcal{R}_2^m$ problem.

\begin{corollary}\label{lowerPHP}
For sufficiently large $n$ and $q \in \mathbb{N}_{>0}$, any $\BTC^0(k)/\mathsf{rpoly}$ circuit of depth $d\geq 4$, size $s \leq \exp(n^{1/(2d-2)})$ and parameter $k\leq n^{1/(5d)}$ has small correlation with $\mathcal{R}_{2}^m$, bounded by,
\begin{equation}
\frac{1}{2} +\exp\left({-\Omega\Big(\frac{n^2}{2^{2q}m^{1+2/q}(\log s)^{2d-1}k^{2d}}\Big)}\right).
\end{equation}
\end{corollary}
\begin{proof}
The initial part of the proof follows equally to \cref{ElowerPHP} with a first random restriction with  $\rho$ with a probability equal to $p=\frac{1}{m^{1/q}\cdot \mathcal{O}(\log(s)^{d-1}\cdot k^d}$ while choosing $t=pn/8$, and second random restriction $\tau$ will be applied to the $2t$ variables of the global decision tree of $\DT(q-1)^m $ resulting from the random restriction with probability $1- \exp(-\Omega(pn))$.

Thus, using $p$-random restrictions, we find that sampling even strings that are consistent with the restrictions provide a uniform distribution over even input strings. Similarly, this argument applies to the restriction $\tau$, considering that the variables for the global decision tree are randomly selected from $\{0,1\}^{2t}$. This characteristic of $p$-random restrictions enables the direct use of \cref{NC_correlation_bound}. Additionally, considering that there is a high probability that more than $pn/4$ variables remain alive, and the circuit is required to solve the $\mathcal{R}_{2}^m$ problem with local decision trees of depth $q-1$, it follows that these trees depend on at most $2^{q-1}$ of the active variables. Consequently,
\begin{align*}
\Pr[\BTC_d^0(k)\text{ solves } \mathcal{R}_{2}^m(x),\ |x|\ \MOD\ 2=0] \leq \frac{1}{2} + 2^{-\Omega\Big(\frac{(pn)^2}{2^{2q}m}\Big)}\\
\leq \frac{1}{2} +2^{-\Omega\Big(\frac{n^2}{2^{2q}m^{1+2/q}\mathcal{O}(\log(s)^{2d-2})k^{2d}}\Big)}.
\end{align*}
The proof also holds for scenarios where the $\BTC^0(k)$ circuit is provided with a random string as advice. This is because any $\BTC^0(k)$ circuit is reduced to the same entity after the application of a random restriction. Furthermore, as established in \cref{NC_correlation_bound}, we know that any probabilistic strategy informed by random advice will fail. Consequently, $\BTC^0(k)/\mathsf{rpoly}$ circuits are also subject to failure under these conditions.
\end{proof}

Both these classical lower bounds left the parameters $q$, $m$, and $k$ undefined as we intend to study the optimal values after the consideration of the quantum circuit solving the problem. This is motivated not only by the optimal asymptotic bound but yes by the smallest values of the input for which the quantum circuits demonstrate an advantage over the classical circuit classes. 
\subsubsection{Qubit upper bound}
We will consider the optimal $\QNC^0$ circuit for the quantum upper bound that exactly solves the $\mathcal{R}_{2}^m$  that maximizes the quantum-classical separation. For that, we rewrite the circuit in the measurement-based quantum computation description with 3D connectivity, while also considering all-to-all connectivity which is effectively a type of connectivity realizable by certain quantum hardware platforms \cite{bluvstein2023logical}.

\begin{lemma}\label{3dqnc}
There exists a $\QNC^0$ circuit with, 3D connectivity, of depth 4 and subquadratic size that does solve the $\mathcal{R}_{2}^m$ exactly with $m=\mathcal{O}(n^{4/3})$.
\end{lemma}
\begin{proof}
The proof will follow based on the demonstration that the subsequent circuit does compute the $\mathcal{R}_{2}^m$. More precisely, by describing and proving the correctness of the following 4 stages marked on it.
\begin{figure}[H]
\centering
\scalebox{0.75}{
\begin{quantikz}[classical gap=0.25cm]
\lstick{$x_1$} &\slice[style={blue},label style={inner sep=1pt,anchor=south west,rotate=45}]{Step 1}  & \slice[style={blue},label style={inner sep=1pt,anchor=south west,rotate=45}]{Step 2}&\slice[style={blue},label style={inner sep=1pt,anchor=south west,rotate=45}]{Step 3a} & \ctrl{6}& \slice[style={blue},label style={inner sep=1pt,anchor=south west,rotate=45}]{Step 3b}& \gate[10,disable auto height]{\begin{tabular}{c}\\$\NC^0$ \\ \\R\\[-1ex] e\\[-1ex]d\\[-1ex]u\\[-1ex]c\\[-1ex]t\\[-1ex]i\\[-1ex]o\\[-1ex]n\end{tabular}} \slice[style={blue},label style={inner sep=1pt,anchor=south west,rotate=45}]{Step 4} &\rstick{$y_1$}  \\
\lstick{$x_2$} & & & & & & &  \rstick{$y_2$} \\
\lstick{\vdots} &\wireoverride{n} &\wireoverride{n} &\wireoverride{n} &\wireoverride{n} &\wireoverride{n} &\wireoverride{n} &\wireoverride{n} \\
\lstick{$x_n$} & & & & &\ctrl{1} & \\
\lstick{$\ket{+}$} &\gate[6,disable auto height]{\begin{tabular}{c} $\mathsf{CZZZ}$ \\ \\R\\[-1.1ex] o\\[-1.1ex]u\\[-1.1ex]t\\[-1.1ex]i\\[-1.1ex]n\\[-1.1ex]e\\ \\[-1.1ex] A\end{tabular}}  & \gate[6,disable auto height]{\begin{tabular}{c} $\mathsf{CZZZ}$ \\ \\R\\[-1.1ex] o\\[-1.1ex]u\\[-1.1ex]t\\[-1.1ex]i\\[-1.1ex]n\\[-1.1ex]e\\ \\[-1.1ex] B\end{tabular}} & & & \meter{}  &  \\
\lstick{\vdots} & \wireoverride{n}& \wireoverride{n}& \wireoverride{n}& \wireoverride{n}\push{\qquad\qquad\rotatebox{135}{ $\vdots$}} & \wireoverride{n} & &  \wireoverride{n} \rstick{\vdots} \\
\lstick{$\ket{+}$}& & & & \meter{}  & & \\
\lstick{\vdots} &\wireoverride{n}&\wireoverride{n}&\wireoverride{n}&\wireoverride{n}&\wireoverride{n}&\wireoverride{n}&\wireoverride{n}\\
\lstick{$\ket{+}$}& & &  \meter{} & \setwiretype{b} & \push{\raisebox{.2cm}{\rotatebox{20}{$\vdots$ }}}& & \setwiretype{q}  \rstick{$y_{m-1}$}\\
\lstick{$\ket{+}$}& & &  \meter{} & \setwiretype{b} &  \push{\raisebox{.2cm}{\rotatebox{20}{$\vdots$ }}}& & \setwiretype{q}  \rstick{$y_m$}
\end{quantikz}}
\caption{Illustration of the MBQC-based constant-depth quantum circuit solving an instance of the inverted strict modular relational problem. Stages $1$ and $2$ involve creating the graph states. In Stage $3a$, the edge qubits of the graph are measured, producing the poor-man's cat state and fanning out the measurement outcomes of the edge measurements. In Stage $3b$, phase operations are applied based on the input, and the entire state is measured. Stage $4$ involves the classical post-processing.}
\label{fig:circ_qubit}
\end{figure}
\paragraph{Steps 1 $\&$ 2.} The circuit begins in the Hadamard basis, and both $\mathsf{CZZZ}$ Routine A and $\mathsf{CZZZ}$ Routine B aim to generate 3D graph states without cycles. For simplicity we will first show that we can create a more densely connected 3D graph state with a simple cubic unit cell in the lattice. To achieve this structure, one needs to apply a $\mathsf{CZ}$ gate between each vertex (acting as the control) and its neighboring edge qubit (acting as the target). Employing single $\mathsf{CZ}$ gates or finite multi-control $(\mathsf{C})^{\otimes t}\mathsf{Z}$ gates would necessitate a depth at least equal to the coloring number of this 3D graph, resulting in a minimum depth of $d=6$. However, from the perspective of the edges, each serves as the target for only two vertices. Hence, we suggest utilizing $\mathsf{CZZZ}$ gates, drawing inspiration from $\mathsf{CZZ}$ gates, which have been demonstrated to be native operations in certain platforms \cite{Ehsan16}. Specifically, we can leverage these gates to simultaneously target all vertices and edges. The primary challenge lies in ensuring a two-step selection of $\mathsf{CZZZ}$ gates, such that each vertex directs the target $\mathsf{ZZZ}$ gates to all adjacent edges. This can be effectively achieved using the following grid pattern in \cref{fig:mapping1} and afterward this using the following pattern in the second round the pattern in \cref{fig:mapping2}. 
\begin{figure}[h!]
  \centering
  \begin{minipage}{.45\textwidth} 
    \centering
    \includegraphics[scale=0.3]{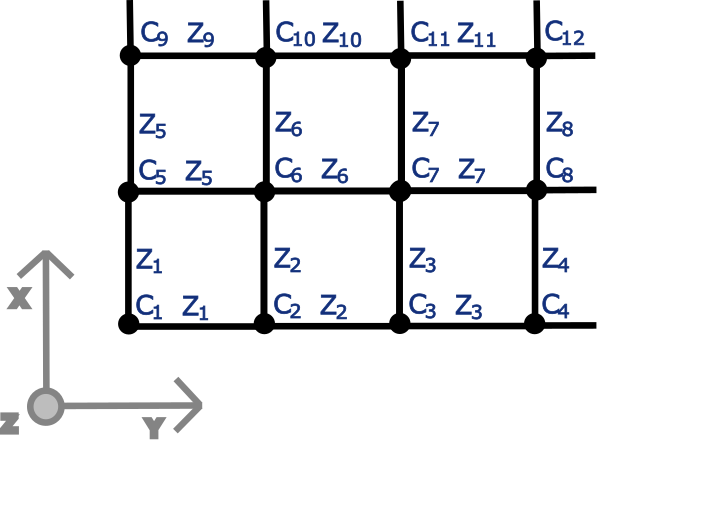}
    \caption{First horizontal entangling routine}
    \label{fig:mapping1}
  \end{minipage}%
  \hspace{0.5cm} 
  \begin{minipage}{.45\textwidth}
    \centering
    \includegraphics[scale=0.3]{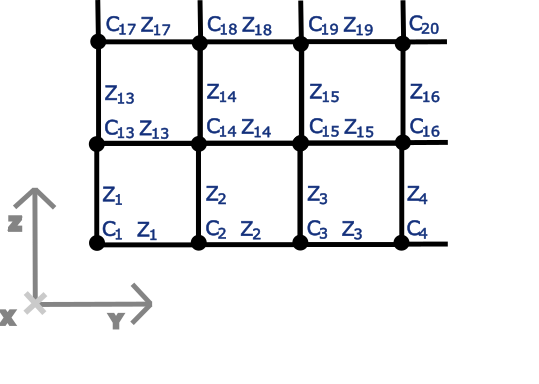}
    \caption{Second vertical entangling routine}
    \label{fig:mapping2}
  \end{minipage}
\end{figure}
These two sequential entangling layers effectively result in a depth-2 process that generates a larger resource state than desired. Now, one only has to remove some of the utilized gates to eliminate paths in the previously described resource state and, with that, obtaining the desired resource state, utilizing practically feasible and finite fan-in quantum gates.

\paragraph{Step 3a.} This step resorts to simply measuring the Hadamard basis of all the edge qubits. This is known to create on the vertex qubits an LU equivalent to the $\mathsf{GHZ}$ \cite{Hein04}. These local unitaries being all $\mathsf{X}$ gates do generate the state previously named poor-mans cat state, with the following description
\begin{equation}
\ket{\Psi_{3a}}= \frac{\ket{v,e_1,\hdots,e_{3n}}+\ket{\overline{v},e_1,\hdots,e_{3n}}}{\sqrt{2}}\text{, with } v\in\mathbb{F}_2^n.
\end{equation}

\noindent while the measurement outcomes from the edges $e_i$ do define the relations between the bits of the string $v$ as follows,
\begin{equation}
    v_i=v_j\oplus_{l\in path(v_i,v_j)} e_l.
\end{equation}
\noindent Thus, by considering $v_1=0$ one can compute the entire string $v$ by the previous relations, which would derive that to obtain effectively a $\mathsf{GHZ}$ state one requires to apply a $\mathsf{X}$ to all the vertexes with value $1$.

\paragraph{Step 3b.} Although the current step does occur simultaneously with the previously described, we can analyze it as independent, given the nature of the measurement operations. Also, we will decompose the represented controlled measurement as an input-dependent controlled $S^{x_i}$ gate based on the entire input strings $x$, and a set of measurements in the Hadamard computational basis. 

With that, we recover that before the measurements we have the state,
\begin{align}
\ket{\Psi_{3b}}&=\frac{i^{\langle v,x\rangle }\ket{v,e_1,\hdots,e_{3n}}+i^{\langle \overline{v},x\rangle }\ket{\overline{v},e_1,\hdots,e_{3n}}}{\sqrt{2}} \\
&=\frac{\ket{v,e_1,\hdots,e_{3n}}+i^{\langle v,x\rangle +|x|/2}\ket{\overline{v},e_1,\hdots,e_{3n}}}{\sqrt{2}}.
\end{align}

Afterward, depending on the value of $\langle v,x\rangle +|x|/2$ having parity either $0$ or $1$ does translate into a measurement of a superposition of even or odd strings respectively. Therefore, the parity of the string $v$ directly relates to the parity of $\langle v,x\rangle +|x|/2$.

\paragraph{Remark.} Before proceeding to the final step, it's important to note that although Steps 3a and 3b were initially considered to be sequential, they take place simultaneously. The realization that analyzing these steps as sequential yet observing that their parallel execution leads to the same outcome stems from the principle that the sequence of single-qubit measurements does not alter the outcome probabilities for any quantum state.

\paragraph{Step 4.} This last $\NC^0$ reduction does the post-processing that guarantees that the outcome strings are the result of the $\mathcal{R}_2^m$ problem. For that we will consider that one can describe $\langle v,x\rangle $  as follows, 
\begin{align}
\langle v,x\rangle  &= \bigoplus_{i=1}^n \AND(x_i,v_i)= \bigoplus_{i=1}^n \AND\bigg(x_i,v_{c}\oplus \bigoplus_{l\in  path(v_i,v_{c})} e_l \bigg)\\
&= \bigoplus_{i=1}^n  \bigoplus_{l\in  path(v_i,v_{c})} \AND(x_i,e_l),\label{corrections}
\end{align}
\noindent with $v_c$ being a central vertex that we use as reference. 

Now all the $\AND$ terms of this expression are computed in the routine, producing a single processor of depth $1$. The final string provided does contain the result of all these computations and the outcome of the string $v$ from the procedure described in stage $3b$. This guarantees, as the parity of this string is equal to $\langle x,v\rangle $, that the outcome string has a parity equal to $|x|/2$. This is exactly the correct outcome for $\mathcal{R}_2^n$, as for even input strings $|x|/2=\mathsf{LSB}(x)$.

Finally, we have only to consider the size of the outcome string. This resorts to account for the size of the paths considered in \cref{corrections}. These do depend on the distance of any vertex in the lattice to a central vertex in the 3D lattice. One can see that even for a cubic disposition this distance is upper bounded by $n^{1/3}$, and given that the number of paths is at most $n$ we obtain that the number of these terms is upper bounded by $\mathcal{O}(n^{4/3})$, as is the final outcome string.
\end{proof}

Now, we do additionally consider the use of $\QNC^0$ circuit with arbitrary all-to-all connectivity which will guarantee tighter bounds. 

\begin{corollary}\label{alltoallqnc}
There exists a $\QNC^0$ circuit with, all-to-all connectivity, of depth 4 and polynomial size that does solve the $\mathcal{R}_{2}^m$ exactly with $m=\mathcal{O}(n \log{n})$.  
\end{corollary}
\begin{proof}
The circuit for this new connectivity does not change in any manner from the circuit of figure \cref{fig:circ_qubit}, except for the two $\mathsf{CZZZ}$ routines. The graph state intended to be created will follow that of a binary decision tree. Again, taking advantage of $\mathsf{CZZZ}$ gates and considering that each edge of the decision tree relates to only two vertices, we can achieve a depth of $2$ in the circuit creating the resource state. In particular, this entanglement structure requires applying the $\mathsf{CZZZ}$ gate to all the vertices at even levels of the tree and their respective edges, and then in the second pass, performing the same procedure on the odd levels, as will be illustrated subsequently.
\begin{figure}[htbp]
    \centering
    \begin{minipage}{0.47\textwidth} 
        \centering
        \includegraphics[scale=0.3]{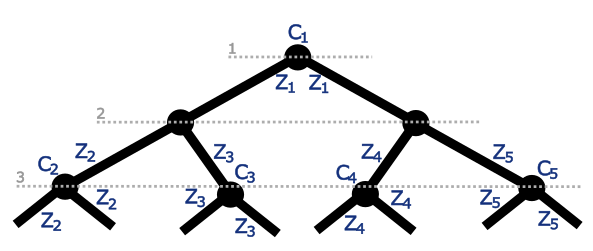} 
        \caption{$\mathsf{CZZZ}$ Routine A}
        \label{fig:sub1}
     \end{minipage}
    \hfill 
    \begin{minipage}{0.47\textwidth} 
        \centering
        \includegraphics[scale=0.3]{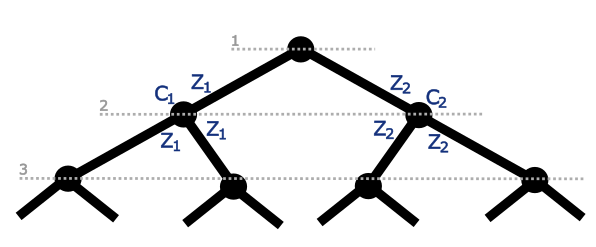} 
        \caption{$\mathsf{CZZZ}$ Routine B}
        \label{fig:sub2}
    \end{minipage}
    \label{fig:depth-4-all-to-all}
\end{figure}

Afterward, every other part of the circuit remains the same. Consequently, the length of the final outcome string is solely determined by the distance from any vertex to the root node vertex of the tree, as we designate this as the central vertex. From this, we deduce that the maximum distance is at most $\log n$, which means the length of the outcome string is upper bounded by $\mathcal{O}(n\log{n})$.
\end{proof}

This establishes two $\QNC^0$ circuits that solve the specified relational problem. We will now examine the lower bounds for $\BTC^0(k)$ derived in the preceding section alongside this quantum upper bound and present new separations between these circuit classes. To do this, we will consider the minimum value of $k$, which is $0$, and demonstrate a new exact-case separation between $\QNC^0$ and $\AC^0$, with the tightest bound yet established.

Also, in this context, we consider a different scenario where the last reduction is performed by an $\AC^0$ or a $\BTC^0(k)$ circuit, and we prove a quantum advantage for the problem being solved at step $3b$ of \cref{fig:circ_qubit} with a contrapositive argument against these circuit classes using a ISMR problem. This allows us to obtain even tighter quantum-classical separations. Furthermore, we find that such a separation is also feasible with the highest value of $k$ allowed by our classical lower bounds, indicating the existence of a problem within $\QNC^0$ that requires superpolynomial-sized $\BTC^0(k)$ circuits, where $k=n^{1/(5d)}$. It is important to note that the earlier class specifies a circuit class that is strictly larger than $\AC^0$.

\paragraph{Proof of \cref{EsepQNC}.} The determined expressions do follow from the fact that we have defined the values of $m$ for which there exists a $\QNC^0$ solving the $\mathcal{R}_2^m$ problem exactly in \cref{3dqnc} and \cref{alltoallqnc}. This allows, us to use the exact lower bound for the size $s$ of $\BTC^0(k)/\mathsf{rpoly}$ circuits of depth $d$ derived in \cref{ElowerPHP} with these values of $m$ and the values selected by us for $k$. For instance, by considering the outcomes of the 3D $\QNC^0$ circuit with $m=n^{4/3}$, we obtain that the size of an $\AC^0$ circuit is larger or equal to
\begin{equation}
s\geq 2^{\left(\frac{n^2 \left(n^{4/3}\right)^{-(1 + 2/q)}}{q^2}\right)^{\frac{1}{d-1}}}.
\end{equation}

The expression still has a non-circuit-dependent parameter $q$ that relates only to the random restrictions used to obtain this bound. As we intend to maximize the quantum-classical separation, we would like to obtain the largest value possible for the previous expression. Therefore, we search for the value of $q$ that archives exactly that. For that, we derive the expression in the function of $q$ first,
\begin{equation}
\frac{ds}{dq}=-\frac{ 2^{1 + \left(\frac{n \big(n^{\frac{4}{3}}\big)^{-\frac{2 + q}{2 q}}}{q}\right)^{\frac{1}{ d-1}}} \left(\frac{n \left(n^{\frac{4}{3}}\right)^{-\frac{2 + q}{2 q}}}{q}\right)^{\frac{1}{d-1}}\left(q - \log\left(n^{\frac{4}{3}}\right)\right)}{(d-1) q^2}.
\end{equation}

Subsequently, it requires only determining for which values of the variable of $q$, for which the expression is equal to 0, to determine the maximums and minimums of the size depending on $q$. Moreover, this results in the solution $q=\log\big(n^{4/3}\big)$ that we do requires only so substitute in our initial expression, as we know that for small integer values of $q$ and large values close to $n$ the expression is smaller, proving that this value a maximum as desired. Finally, we obtain that the size of an $\AC^0$ solving the $\mathcal{R}_2^m$ problem exactly requires a size no smaller then 
\begin{equation}
s\geq \exp\left({\left(\frac{n^{1/3}}{(\log n)^{1+\mathcal{O}(1)}}\right)^{\frac{1}{d-1}}}\right).
\end{equation}

All the other expressions follow equally with the respective parameters of $m$ and $k$.\qed
\vspace{0.2cm}

The final set of separations refers to the average case hardness of solving the same problem. For these, we do focus uniquely on the case where we have the largest value of $k$, so we drop the case where $k$ is zero as this was already studied previously. This again will refer to the case where we have either 3D or all-to-all connectivity in the corresponding $\QNC^0$ circuit.

\paragraph{Proof of \cref{averagequbit}.}
Again the expressions do follow from the specific values of $m$ for which there exists a $\QNC^0$ solving the $\mathcal{R}_2^m$ problem (\cref{3dqnc} and \cref{alltoallqnc}). This is now coupled with the average-case lower bound for the efficiency with a $\BTC^0(k)/\mathsf{rpoly}$ circuits of depth $d$ and size $s$ with maximum parameter $k$ does solve the $\mathcal{R}_2^m$ from \cref{lowerPHP} does provide us the enunciates bounds up to a value for $q$. To determine the value of $q$, we employ the same methodology used in the proof of \cref{EsepQNC}, thereby deriving the optimal values for the entire set of expressions. \qed

\subsection{Separation for qupit cases}

In this subsection, we will examine a larger set of ISMR problems, denoted as $\mathcal{R}_p^m$, and the correlation of our candidate circuit with the correct solutions. More concretely, we use these problems to establish distinctions between $\QNC^0$ and $\BTC^0(k)$ circuits. For that, we demonstrate that a $\QNC^0$ circuit can solve these problems with constant and non-zero correlation while generating outcome strings of sub-quadratic size, $m = o(n^2)$. In contrast, we will show that a $\BTC^0(k)$ circuit of polynomial size solves the problem with a correlation asymptotically approaching zero as the input size increases for equal-size output strings, even for large and nearly optimal parameter $k$, as formalized in the subsequent theorem.

\begin{theorem}\label{thm:qudit_average}
For all \textnormal{ISMRP} problems defined for a prime dimension $p$, denoted by $\mathcal{R}_p^m$, there exists a $\QNC^0$ circuit with all-to-all connectivity, considering a uniform distribution $\mathcal{D}_p$ over $\mathbb{F}_2^n$ of strings with a Hamming weight that satisfies $\left(\sum_{i=1}^n x_i\right)\ \MOD\ p = 0$ as the input distribution, that solves the problem for $m = n \cdot \log n^{p-1}$, achieving a correlation of $\frac{p-1}{p^2}$. In contrast, for any $C\in \BTC^0(k)/\mathsf{rpoly}$ circuit with fixed depth $d$, size polynomial $s$, and parameter $k \leq n^{1/(5d)}$ solves each one of these problems with a correlation that is bounded as in \cref{tab_res3}.
\vspace{-0.2cm}
\begin{center}
\begin{table}[!hbtp]
\renewcommand{\arraystretch}{2.5}
\begin{tabular}{|c|c|c|}
\hline
$\BTC^0(k)/\mathsf{rpoly}$ & $k=\mathcal{O}(1)$ $(\equiv \AC^0/\mathsf{rpoly})$ & $k=n^{1/(5d)}$ \\  
\hline
$\mathsf{Corr}_{\mathcal{D}_p} \left(C,\mathcal{R}_p^m \right)$ & $\exp\left(- \Omega\left(\frac{n^{1 - \mathcal{O}(1)}}{  (\log{n})^{p-1} (\log{s})^{2d-2}}
\right)\right)$ & $\exp\left(-\Omega \left(\frac{ n^{3/5 - \mathcal{O}(1)}}{(\log{n})^{ p-1} (\log{s})^{2d-2}}
\right)\right)$ \\
\hline
\end{tabular}.
\caption{\justifying Correlation upper bounds for the circuit classes $\BTC^0(k)$ for different values of $k$ in solving the $\mathcal{R}_p$ problems.}
\label{tab_res3}
\end{table}
\end{center}
\end{theorem}
\vspace{-0.6cm}

The previous theorem follows in three steps, delineated in \cref{subsec:games,subsec:qudit_lower,subsec:qudit}. In the first subsection, we introduce a family of non-local $\mathsf{XOR}$ games related to the ISMR problems and the respective mapping between bits and dits, which will later facilitate the translation between these two objects. Then, in \cref{lemma:entire_breakp}, we establish upper bounds on the efficacy of any classical strategy in solving this family of non-local $\mathsf{XOR}$ games. 

In the second subsection, we establish the classical lower bound for $\BTC^0(k)$ circuits that address the ISMR problems. We integrate the upper bounds determined for the non-local $\mathsf{XOR}$ games with light cone arguments, enabling us to constrain the maximum efficiency of $\NC^0$ circuits operating on bits in addressing the ISMR problems. Additionally, we revisit our random restriction techniques to simplify $\BTC^0(k)$ circuits into forest lines using \cref{lem:GCred2}, thereby connecting with the $\NC^0$ upper bounds. This approach allows us to define upper bounds for $\BTC^0(k)$ circuits concerning the correlation with which they resolve the problems, which equivalently translates to size lower bounds for $\BTC^0(k)$ circuits that achieve a constant positive correlation with the ISMR problems.

Finally, in \cref{subsec:qudit}, we define the quantum resource state that will be utilized for the quantum solutions. This involves a generalization of the poor man's cat state, and we also demonstrate that these states can be prepared using qupit $\QNC^0$ circuits as shown in \cref{poorconst}. Subsequently, in \cref{lemma:computcorre}, we demonstrate the existence of $\QNC^0$ circuits that address the aforementioned problems with a constant positive correlation. The synthesis of the classical lower bounds and the quantum upper bounds culminates in our previous theorem. Additionally, it is worth noting that \cref{thm:lower_qudit_average} serves as an intermediate result, delineating separations between $\NC^0$ and qupit $\QNC^0$ circuits for specific problem instances.

Simultaneously with the previous result, we establish a separation for the qutrit case that parallels the qubit case, with simple success probabilities of being correct or incorrect. This equivalence arises from the correlation measure for the qutrit case, which assesses deviations from the correct outcome by magnitudes of $0$ or $1$, mirroring the correlation measure in the qubit case that is translatable to the commonly used success probability. Furthermore, this problem is particularly intriguing because, while the quantum solution can address the bounded error probabilistic version of the $\mathcal{R}_p^m$—solving all inputs with a probability distinctly greater than $1/2$—the average success probability of a polynomial-size $\BTC^0(k)$ circuit decreases asymptotically to $1/3$ as the input size increases.

\begin{corollary}\label{cor:qutrit_average_sep}
For the \textnormal{ISMR} problems defined for a prime dimension $p=3$, denoted by $\mathcal{R}_3^m$, there exists a $\QNC^0$ circuit with 3D 	connectivity, considering a uniform distribution over $\mathbb{F}_2^n$ of strings with a Hamming weight that satisfies $\left(\sum_{i=1}^n x_i\right)\ \MOD\ 3 = 0$ as the input distribution, that solves the problem for $m = n^{4/3}$, achieving a sucess probability strictly larger then $\frac{1}{2}$. In contrast, for any $C\in \BTC^0(k)/\mathsf{rpoly}$ circuit with fixed depth $d$, size polynomial $s$, and parameter $k \leq n^{1/(5d)}$ solves each one of these problems with a success probability that is bounded as in \cref{tab_res4}.
\vspace{-0.2cm}
\begin{center}
\begin{table}[!hbtp]
\begin{tabular}{|c|c|c|}
\hline
$\BTC^0(k)/\mathsf{rpoly}$ & $k=\mathcal{O}(1)$ $(\equiv \AC^0/\mathsf{rpoly})$ & $k=n^{1/(5d)}$  \\  
\hline
$\Pr[\text{Success}]$ & $\frac{1}{3}+ \exp\left(-\Omega  \left( \frac{n^{2/3-\mathcal{O}(1)} }{ \log(s)^{ 2d-1}}
 \right )\right) $ &$\frac{1}{3}+ \exp\left(- \Omega  \left( \frac{n^{4/15-\mathcal{O}(1)}}{\log(s)^{2d-1}} \right )\right)$  \\
\hline
\end{tabular}.
\caption{\justifying Upper bounds on the success probability for the circuit classes $\BTC^0(k)$ for different values of $k$ in solving the $\mathcal{R}_3$ problem.}
\label{tab_res4}
\end{table}
\end{center}
\end{corollary}
\vspace{-0.6cm}

This corollary follows the same proof techniques as \cref{thm:qudit_average}, with the difference being that the correlation measure can always be directly translated into the respective success probabilities.

\subsubsection{Quantum non-local games in higher dimensions}
\label{subsec:games}

We will now define the precise family of $\mathsf{XOR}$ non-local games, enabling us to connect the correlation upper bounds of these games to the upper bounds of the ISMR problems for the previously described separations.

Recall that to define a family of non-local games of a particular type we need to describe only the function that maps the modular remainder of the responses by each player and the messages sent to the parties playing the game. In our specific game, the function will also be a modular remainder of the message sent to the parties. Additionally, we extend the $\mathsf{XOR}$ non-local games defined in \cref{def:non_local} for the correlation measure presented in \cref{subsec:prelims-nonlocal}, resulting in the following family of non-local games.

\begin{definition}[Modular $\mathsf{XOR}$ games]\label{def_modXOR} The class $\mathcal{G}_p$ of multiparty non-local $\mathsf{XOR}$ games involves $n$ parties, with each party $P_i$ receiving information $x_i\in \mathbb{F}_p$. After each party receives its input, without further communication, each party must provide an individual response $y_i\in \mathbb{F}_p$, contributing to a collective output which is the concatenation $y=(y_1,y_2,\hdots, y_n)$ of their outputs. This defines the classical strategy $w_{\mathcal{G}}$. The correlation with the game is determined by the following expression,
\begin{equation*}
\mathsf{Corr}(w_{\mathcal{G}_p},\mathcal{G}_p)=\Expec_{(x_1,x_2,\hdots,x_n) \sim \mathcal{D}}\left[\mathsf{Re}\left(e^{-i\frac{2 \pi \left(\sum_{i=1}^n y_i-\sum_{i=1}^n x_i/p \right)}{p}}\right) \right],
\end{equation*}
for an arbitrary given input distribution $\mathcal{D}$.
\end{definition}

These games are naturally related to the ISMR problems (\cref{insec:defmod}) because the messages to the parties and the responses follow a similar mapping between the input and output for the inverted strict modular relational problems. However, there is a technical misalignment since the non-local games have messages to the parties in $\mathbb{F}_p$, while the ISMR problems involve binary strings. Furthermore, we require uniform input distributions over $\mathbb{F}_2^n$ to achieve average-case hardness separation for the ISMR problems. This will be resolved by establishing correlation upper bounds for the Modular $\mathsf{XOR}$ non-local games with biased input distributions, allowing us to relate uniform binary input distributions for the ISMR problems to the corresponding biased distributions over $\mathbb{F}_p^n$ for the non-local games.

\paragraph{Encoding.} To bridge our techniques from \cref{SI_new_switch} with the higher-dimensional non-local games described in \cref{def_modXOR}, we will consider a specific encoding between binary and $p$-ary, for prime $p$. Subsequently, we will demonstrate that a Hamming weight-based encoding, defined as follows, meets all the necessary criteria.

\begin{definition}[Hamming encoding]\label{encoding} We define the Hamming encoding as a mapping of the form $H:\mathbb{F}_2^{n(p-1)}\mapsto \mathbb{F}_p^n$, defined as follows 
\begin{equation}
y=\bigotimes_{j=1}^n \left( \sum_{i=1}^{p-1} x_{j\cdot p+i} \right )
\end{equation}
\noindent with $x\in \mathbb{F}_2^{n(p-1)}$ and $y \in  \mathbb{F}_p^{n}$. 
\end{definition}

The previous encoding does work very well in the binary to the base-p, however, the inverse from the fact that it is not injective provides a nondetermined inverse. As we will consider later, this will not be a problem to handle the ISMR problems. However, the fact that a uniform distribution over $\mathbb{F}_2^{n(p-1)}$ generates non-uniform distributions over $\mathbb{F}_p^n$ has to be handled. A simple demonstration of this effect can be seen with the trit case, where we have
\begin{equation}
00 \mapsto 0;\ \ \  01\mapsto 1;\ \ \  10\mapsto 1 \text{   and }11\mapsto 2.  
\end{equation}

One does see that the frequency of $1$'s in the resulting strings will be larger than $0$ or $2$. We can account for this by defining the functions that map us the precise frequency of a term in $\mathbb{F}_p$ resulting in a binary string with the Hamming encoding, as follows,
\begin{equation}
\Delta_{2}^p(x)= \frac{\binom{x}{2}}{2^{p-1}}.
\end{equation}

\noindent This allows us to account for and determine how the winning probabilities of the respective non-local games are altered by this input distribution. Additionally, we will consider the uniform distribution over $\mathbb{F}_p$, which is defined simply as $\Delta_{p}^p(x)= \frac{1}{p}$. The latter distribution will allow us to prove stronger separations between $\QNC^0$ qupit circuits and $\NC^0$ circuits.

Now that we have created the necessary elements to bridge our techniques, we can start establishing the sequence of correlation upper bounds for the selected family of non-local games. However, as these demonstrations become increasingly complex when dealing with all prime dimensions $p$, we will first, for a better exposition, demonstrate the case for qutrits. In addition, this case is of particular interest because $p=3$ represents the first prime dimension where quantum circuits lack exact solutions to the ISMR problems, and the deviation of quantum-classical efficiencies captures the existing advantage.

\subsubsection*{Ternary correlation bound}
In the case of qutrits, we demonstrate that the correlation between any classical strategy, including the optimal classical strategy denoted by $w_{\mathcal{G}_3}^{\star}$, and the non-local game defined by $\mathcal{G}_3$ decreases exponentially fast as the number of participants increases. Furthermore, we aim to establish a corresponding bound for the following expression,
\begin{equation}
\mathsf{Corr}_{\mathcal{D}}\left(w_{\mathcal{G}_3}^{\star},\mathcal{G}_3 \right)=\Expec_{x\sim \mathcal{D}}\left( \mathsf{Re}\left(e^{i\frac{2\pi \left( \left |w_{\mathcal{G}_3}^{\star}(x)\right|-\sum_{i=1}^n x_i/3\right)}{3}} \right)\right).
\end{equation}

Additionally, this correlation measure can in the trinary case be used to compute the probability with which the $\mathsf{XOR}$ non-local game can be solved\footnote{Note that in the case of having a maximal mismatch between the functions we obtain that $\Pr[f(x)=g(x)]=0$ given that will be equal to $\mathsf{Corr}_{\mathcal{D}}(*,*)=-0.5$. Also, whenever $\Pr[f(x)=g(x)]=1$ we have that $\mathsf{Corr}_{\mathcal{D}}(*,*)=1$ which is again consistent, and all the values in between do follow equally.},
\begin{equation}
\Pr_{x \sim \mathcal{D}}[w_{\mathcal{G}_3}^{\star}(x)\in \mathcal{G}_3(x)]= \frac{1+2\cdot \mathsf{Corr}_{\mathcal{D}} \left (w_{\mathcal{G}_3}^{\star},\mathcal{G}_3 \right )}{3}.
\end{equation}

Building on the previous objects, we will show that the winning probability of any classical strategy for the non-local game defined by the $\mathcal{G}_3$ problem is bounded below $1/2$ for sufficiently large number of participants $n$ and asymptotically approaches $1/3$ for uniform distributions over messages in $\mathbb{F}_3^n$ and biased distribution over $\mathbb{F}_3^n$ resulting from the translation from a uniform distribution over $\mathbb{F}_2^n$ using the Hamming encoding.

\begin{lemma}\label{entire_break}
Any local probabilistic classical strategy $w_{\mathcal{G}_3}$ that solves the non-local $\mathsf{XOR}$ games denoted as $\mathcal{G}_3$, with messages over $\mathbb{F}_3^n$, achieves a winning probability, for uniform input distributions over $\mathbb{F}_3^n$ with $\ell_1$-norm satisfying $\big(\sum_{i=1}^n x_i\big)\ \MOD\ 3=0$ and $r$ restricted dits, bounded by $\frac{1}{3}+\big(\frac{17}{20}\big)^{n-r}$. Similarly, for a uniform distribution over $\mathbb{F}_2^{n}$, where the Hamming weight satisfies $\big(\sum_{i=1}^n x_i\big)\ \MOD\ 3=0$ and the inputs are translated using a bit-to-ternary encoding, the winning probability with $r$ restricted bits is bounded by $\frac{1}{3}+\big(\frac{9}{10}\big)^{\frac{n-r}{2}}$.
\end{lemma}
\begin{proof}
In this setting, classical strategies allow each participant, identified as $i$ within a set of $n$ parties, to generate an output of either $x_i$ or $x_i + b$, with $b$ in $\mathbb{F}_3$. Therefore, the efficiency of any classical strategy $w_{\mathcal{G}_3}$ in solving the $\mathsf{XOR}$ non-local game $\mathcal{G}_3$  depends on the linear function $l_a = \sum_{i=1}^n a_i \cdot x_i$ (with $a$ in $\mathbb{F}_3^n$) resulting from the messages send to the parties and the function $f(x) = (|x|/3)^{-1} \MOD\ p$ for all messages where $|x| \MOD\ 3 = 0$, which defines the non-local game.

With this defined, we would like to compute the maximal value for the correlation and consequently the probability $\Pr[f(x)=l_a(x)]$. In particular, we have the following expression for the correlation,
\begin{equation}\label{init}
\mathsf{Corr}_{*}\left(w_{\mathcal{G}_3}^{\star},\mathcal{G}_3\right) = \max_{b\in \mathbb{F}_3^n}\left( \sum_{{\substack{x \in \mathbb{F}_3^n \\  |x|\ \MOD\ 3= 0}}} \frac{\mathsf{Re} \left(e^{i\frac{2\pi 2(|x|/3)}{3}} \prod_{i=1}^n e^{i\frac{2\pi x_ib_i}{3}} \right )}{\prod_{i=1}^n \Delta_{*}^p(x_i)}  \right).
\end{equation}
\noindent without yet completely defining a distribution over the delivered messages. 

The correctness of the previous expression follows from the fact that $|x|/3\ \MOD\ 3$ and $\sum_{i \in \mathbb{F}_3^n}x_ib_i\ \MOD\ 3$ are up to sign changes respectively representative of the two functions $f$ and $l_a$ that we intended to compare. Also, we can simply drop the $\MOD\ 3$ operation on $f$ since this is inherently carried out by the arithmetic associated with the roots of unity under consideration.

With that, we will write the correlation  from \cref{init} as,
\begin{align}
 \sum_{{\substack{x \in \mathbb{F}_3^n \\  |x|\ \MOD\ 3= 0}}} \frac{\mathsf{Re} \left (e^{i\frac{4\pi (|x|/3)}{3}} \prod_{i=1}^n e^{i\frac{2\pi x_ib_i}{3}} \right ) }{\prod_{i=1}^n \Delta_{*}^p(x_i)} = \mathsf{Re} \left( \sum_{{\substack{x \in \mathbb{F}_3^n \\  |x|\ \MOD\ 3= 0}}}  \left( \prod_{i=1}^n \frac{e^{i\frac{2\pi x_i(2+3b_i)}{9}}}{\Delta_{*}^p(x_i)}  \right) \right)\\
=  \mathsf{Re} \left ( \left( \left (\sum_{x_1\in \mathbb{F}_3}  \frac{e^{i\frac{2\pi x_1 (2+3b_1)}{9}}}{\Delta_{*}^p(x_i)}  \right ) \hdots \left (\sum_{x_{n-1}\in \mathbb{F}_3} \frac{e^{i\frac{2\pi x_{n-1} (2+3b_{n-1})}{9}}}{\Delta_{*}^p(x_i)} \right )\right)\bullet\frac{e^{i\frac{2\pi x_{n} (2+3b_{n})}{9}}}{\Delta_{*}^p(x_n)}  \right ) \label{exprtrit}
\end{align}
\noindent with $x_n=\Big(\sum_{i=1}^{n-1}x_i \Big)^{-1}$ and $\bullet$ being a conditional multiplication operation depending on the values of $x_1$ to $x_{n-1}$ in the first expression, such that the total sum of the inputs fulfills the condition of being congruent with zero modulo 3.

Now that we have reduced the correlation expression to a sum over fixed and equivalent expressions, we can establish bounds for the maximum value of the entire expression to obtain an upper limit. In particular, we will now apply the method proposed for qubits in \cite{brassard2004recasting}, and use the fact that the last term $e^{i\frac{2\pi x_{n} (2+3b_{n-1})}{9}}$ takes each value of $x_n$ over $\mathbb{F}_3^n$,

\begin{align}
 &\leq \max_{b\in \mathbb{F}_3^n}\left(\mathsf{Re}  \left(\left(\sum_{x_1\in \mathbb{F}_3} \frac{  e^{i\frac{2\pi x_1 (2+3b_1)}{9}}}{\Delta_{*}^p(x_1)} \right )  \hdots \left (\sum_{x_{n-1}\in \mathbb{F}_3} \frac{ e^{i\frac{2\pi x_{n-1} (2+3b_{n-1})}{9}}}{\Delta_{*}^p(x_{n-1})} \right ) \left (\sum_{x_{n}\in \mathbb{F}_3} \frac{e^{i\frac{2\pi x_{n} (2+3b_{n})}{9}}}{\Delta_{*}^p(x_n) \cdot (p-1)} \right ) \right ) \right ) \label{eq:singlesum}\\
&\leq \max_{b\in \mathbb{F}_3^n}\left(\mathsf{Re} \left (  \frac{\left (\frac{1}{\Delta_{*}^p(1)} + \frac{e^{i\frac{2\pi (2+3b_i)}{9}}}{\Delta_{*}^p(1)} +\frac{e^{i\frac{2\pi (4+6b_i)}{9}}}{\Delta_{*}^p(2)} \right )^{n} }{ p-1} \right) \right ) \label{eq:Realbound}.
\end{align}

Considering the resulting expression and $\Delta_{3}^3$ fully defining the input distribution, we can show that each term in \cref{eq:Realbound} has its norm bounded for the three different assignments of $b_i$ by the following values,
\begin{equation}
 \bigg | \frac{1 + e^{i\frac{2\pi (2+3b_i)}{9}} +e^{i\frac{2\pi (4+6b_i)}{9}}}{3} \bigg |=
\begin{cases}
\frac{1}{\sqrt{3 \left(\cos^2\left(\frac{\pi}{18}\right) + (-1 + \sin\left(\frac{\pi}{18}\right))^2\right)}}
<0.45 ,\ b_i=0 \\
\frac{1}{\sqrt{3 \left((-1 - \cos\left(\frac{\pi}{9}\right))^2 + \sin^2\left(\frac{\pi}{9}\right)\right)}}
<0.3,\ b_i=1\\ 
\frac{1}{\sqrt{3 \left((-1 + \cos\left(\frac{2\pi}{9}\right))^2 + \sin^2\left(\frac{2\pi}{9}\right)\right)}}
<0.85,\ b_i=2
\end{cases}.
\end{equation}

Obtaining with that the bound $\mathsf{Corr}_{\mathcal{D}_3'}\left(w_{\mathcal{G}_3}^{\star},\mathcal{G}_3 \right)\leq \frac{(0.85)^n}{2}$, with $\mathcal{D}_3'$ being the uniform input distributions over $\mathbb{F}_3^n$ with $\ell_1$-norm satisfying $\big(\sum_{i=1}^n x_i\big)\ \MOD\ 3=0$. Additionally, the probability $\Pr_{x \sim \mathcal{D}_3'}[f(x)=g(x)] \leq \frac{1}{3}+\big(\frac{17}{20}\big)^n$, which asymptotically approaches $1/3$ and for $n\geq 12$ is already strictly smaller the $1/2$.

Equally, we will consider $\Delta_2^3$ which defines the distribution $\mathcal{D}_3$ generated by considering a uniform distribution of strings over $\mathbb{F}_2^{2n}$ converted to strings in $\mathbb{F}_3^n$ with the Hamming encoding. In particular, in this case, we obtain that for the three different assignments of $b_i$, the term in \cref{eq:Realbound} is bounded by
\begin{equation}
 \bigg | \frac{1}{4} + \frac{e^{i\frac{2\pi (2+3b_i)}{9}}}{2} +\frac{e^{i\frac{2\pi (4+6b_i)}{9}}}{4}\bigg |=
\begin{cases}
\frac{1}{2} \left(1 + \sin\left(\frac{\pi}{18}\right)\right)<0.59 ,\ b_i=0 \\
\sin^2\left(\frac{\pi}{18}\right)
<0.04,\ b_i=1\\ 
\frac{1}{2} \sqrt{\frac{1}{2} \left(3 + 4\cos\left(\frac{2\pi}{9}\right) + \sin\left(\frac{\pi}{18}\right)\right)}
<0.89,\ b_i=2
\end{cases}.
\end{equation}
\noindent In this case, we obtain the following bound, $\mathsf{Corr}_{\mathcal{D}_3}\left(w_{\mathcal{G}_3}^{\star},\mathcal{G}_3 \right) \leq \frac{(0.89)^n}{2}$. Additionally, the probability $\Pr_{x \sim \mathcal{D}_3}[f(x)=g(x)] \leq \frac{1}{3}+\big(\frac{9}{10}\big)^n$, which asymptotically approaches $1/3$ and for $n\geq 18$ is already smaller then $1/2$.

The effect of restricted dits in \cref{eq:singlesum} for the respective sum of that variables does make it have a constant normal equal to $1$. Therefore, these cases do not contribute to increasing the correlation and the overall expression. Equally, the space of inputs is reduced with the same factor. This implies that restricting dits makes the overall problem equivalent to the non-restricted, non-local game with a smaller input. Also, as there is no structure of the game limiting the rearrangement of messages, we can reorganize the fixed variables such that they form a $r/(p-1)$ dits, and the previous analysis holds equally in this case.

Finally, there is no probabilistic classical strategy that surpasses the two upper bounds determined for the correlation measures by Yao's min-max principle.
\end{proof}

\subsubsection*{Higher order correlation bounds}
Now, we will generalize the upper bound on the efficiency of any classical strategy, including the optimal classical strategy $w_{\mathcal{G}_p}^{\star}$, for all the $\mathsf{XOR}$ non-local games $\mathcal{G}_p$ defined over higher prime dimensions. Therefore, we now intend to constrain the following expression,
\begin{equation}
\mathsf{Corr}_{\mathcal{D}}\left(w_{\mathcal{G}_p}^{\star},\mathcal{G}_p\right)=\Expec_{x\sim \mathcal{D}}\left( \mathsf{Re}\left(e^{i\frac{2\pi \left(\left|w_{\mathcal{G}_p}^{\star}(x)\right|-\sum_{i=1}^n x_i/p\right)}{p}} \right)\right).
\end{equation}

Note that previously, for $p$ equal to $2$ and $3$, due to the symmetry of the roots of unity, the correlations were still binary discriminators relating to the probability of two functions being exactly equal modulo these prime numbers. However, it is important to note that this measure no longer directly relates to the probability of the outcome being correct in binary discrimination; rather, it is a distance measure between the vectors formed by the outcome space of the two functions. Although a similar measure of distances can be created, it generates non-linear terms for which the lower bounding techniques no longer apply. Nevertheless, we can resort to this linear distance measure natural to higher dimensions and use it to prove a quantum-classical efficiency separation for the relational problems at each prime dimension.

\begin{lemma}\label{lemma:entire_breakp}
Any local probabilistic classical strategy $w_{\mathcal{G}_p}$ that solves one of the $\mathsf{XOR}$ non-local games  $\mathcal{G}_p$, with messages over $\mathbb{F}_p^n$, taken from a uniform distribution $\mathcal{D}_p$ over $\mathbb{F}_2^{n(p-1)}$ with Hamming weight satisfying $\big(\sum_{i=1}^n x_i\big)\  \MOD\ p=0$ and a bit to base-p encoding, the maximal correlation with $r$ restricted bits is bounded by,
\begin{equation}
\mathsf{Corr}_{\mathcal{D}_p}\left(w_{\mathcal{G}_p}^{\star},\mathcal{G}_p \right) \leq \mathsf{Re}\left ( \left (\frac{2^{1 - p} e^{\frac{2\pi i}{p^2}} \left(e^{-\frac{2\pi i}{p^2}} \left(1 + e^{\frac{2\pi i}{p^2}}\right)\right)^p}{1 + e^{\frac{2\pi i}{p^2}}} \right )^{\frac{n-r}{p-1}}   \mathbin{\bigg/} (p-1)\right ) \leq (c_p)^{\frac{n-r}{p-1}}  
\end{equation} 
\noindent with $c_p \in (0, 1)$. Simultaneously, for uniform input distribution $\mathcal{D}_p'$ over $\mathbb{F}_p^n$ with their $\ell_1$-norm satisfying $\big(\sum_{i=1}^n x_i\big)\  \MOD\ p=0$ and $r$ restricted dits, we have 
\begin{equation}
\mathsf{Corr}_{\mathcal{D}_p'}\left(w_{\mathcal{G}_p}^{\star},\mathcal{G}_p \right) \leq \frac{1}{p^{n-1}\left( p-1\right)}\mathsf{Re} \left( - \frac{e^{2\pi i/p^2} \left(-1 + e^{2\pi i \frac{-1 + p^2}{p} }\right)}{-1 + e^{i\frac{2\pi}{p^2}} } \right )^{n-r} \leq \frac{p}{p-1}\left(\mathsf{Re}\left( e^{\frac{2\pi i}{p^2}}\right)\right)^{n-r}.
\end{equation}
\end{lemma}
\begin{proof}
To obtain this bound on the correlation, we again consider that classical strategies $w_{\mathcal{G}_p}$ allow each participant, identified as $i$ within a set of $n$ parties, to generate an output of either $x_i$ or $x_i + b$, with $b$ in $\mathbb{F}_p$ as defined by the generalized $\mathsf{XOR}$ non-local games. Thus, the capability of approximating the ideal outcome will depend on the distance of the closest linear function $l_a = \sum_{i=1}^{n} a_i \cdot x_i$ with $a$ in $\mathbb{F}_p^n$, and the function defined by the respective non-local game $\mathcal{G}_p$ over $\mathbb{F}_p$, taking the form of a function of the type $f:\mathbb{F}_p^n \mapsto \mathbb{F}_p$, and specifically defined as $f(x) = (|x|/p)^{-1}\ \MOD\ p$.

Our generalized correlation measure will then be defined as follows for arbitrary $p$ and distribution $\Delta_{*}^p$,
\begin{equation}\label{eq:correlation}
\mathsf{Corr}_{*}\left(w_{\mathcal{G}_p}^{\star},\mathcal{G}_p\right) = \max_{b\in \mathbb{F}_p^n}\left( \sum_{{\substack{x \in \mathbb{F}_p^n \\  |x|\ \MOD\ p= 0}}}  \mathsf{Re} \left( \frac{e^{i\frac{2\pi (p-1)(|x|/p)}{p}} \prod_{i=1}^n e^{i\frac{2\pi x_ib_i}{p}}}{\Delta_{*}^p(x)} \right ) \right).
\end{equation}

The previous expression is algebraically equivalent to the one in \cref{exprtrit}, and the terms in the exponents are linear, in the sense that all have a degree of 1. Therefore, we can bound the previous expression similarly as follows, 
\begin{align}
 &\leq \max_{b\in \mathbb{F}_p^n}\left(\mathsf{Re}  \left( \sum_{x_1\in \mathbb{F}_p} \left (\frac{ e^{i\frac{2\pi x_1 ((p-1)+pb_1)}{p^2}}}{\Delta_{*}^p(x_1)} \right ) \hdots \sum_{x_{n-1}\in \mathbb{F}_p} \left (\frac{e^{i\frac{2\pi x_{n-1} ((p-1)+pb_{n-1})}{p^2}}}{\Delta_{*}^p(x_{n-1})}  \right ) \sum_{x_{n}\in \mathbb{F}_p} \left ( \frac{e^{i\frac{2\pi x_{n} ((p-1)+pb_{n})}{p^2}}}{\Delta_{*}^p(x_n)\cdot (p-1)} \right) \right )   \right )\\
&\leq \max_{b\in \mathbb{F}_p^n}\left(\mathsf{Re} \left (  \left(\sum_{x_j\in \mathbb{F}_p} \frac{e^{i\frac{2\pi x_j ((p-1)+pb_j)}{p^2}}}{\Delta_{*}^p(x_j)} \right )^{n} \mathbin{\bigg/} (p-1)  \right) \right ). \label{sumexpdit}
\end{align} 

Unfortunately, at this point, it is no longer feasible to compute all the values for $b$ in $\mathbb{F}_p$ to bind the expression for a general proof, as previously done for the qutrit case. Nevertheless, there is a simple geometric intuition for the values of the expression $\sum_{x_j \in \mathbb{F}_p} e^{i\frac{2\pi x_j ((p-1)+pb_j)}{p^2}}$. For any value of $b_j$ selected, the sequential product with the values of $x_j$ moves the angle of each term involved in the sum by a step equivalent to a multiple of $2\pi i / p^2$, equal to $(p-1)+pb_j$. Ideally, the maximum for the expression would be reached with a step equal to zero, so we always have the maximum value of $1$ and perfect correlation. Since this case with optimal correlation is not possible, as $(p-1)+pb_j$ is never zero for any value of $b_j$, we have that the subsequent best option is to have that expression equal to either $1$ or $p^2-1$ to maximize either the real or the absolute values.

In addition, the same analysis holds equally for any of the distributions of the inputs $\Delta_{*}^p(x)$, as the same arguments apply for the selection of the value $b_j$. Therefore, by considering $b_j=p-1$ and $\Delta_{p}^p$, we obtain the distribution $\mathcal{D}_p'$ and define the optimal correlation as follows,
\begin{equation}
\mathsf{Corr}_{\mathcal{D}_p'}\left(w_{\mathcal{G}_p}^{\star},\mathcal{G}_p\right)\leq \frac{1}{p^{n-1}\left( p-1\right)}\mathsf{Re} \left( - \frac{e^{2\pi i/p^2} \left(-1 + e^{2\pi i \frac{-1 + p^2}{p} }\right)}{-1 + e^{i\frac{2\pi}{p^2}} } \right )^n.
\end{equation} 

To obtain the exact value, one is only required to obtain a numerical bound for the previous expression depending on $p$. From an analytical point of view, we know that all the terms in the sum contained in \cref{sumexpdit}, except for the first term, are smaller than $1$, even when we use the best possible value for $b_i$. Thus, we have that,
\begin{equation}
\mathsf{Corr}_{\mathcal{D}_p'}\left(w_{\mathcal{G}_p}^{\star},\mathcal{G}_p \right) \leq \frac{p}{p-1}(c_p)^n,\text{ with }c_p\in (0,1).
\end{equation} 
Also, as $c_p$ is no larger than $1-\mathsf{Re}\left( e^{\frac{2\pi i}{p^2}}\right)$, we have that it is guaranteed that it is infinitely close to 1. Thus, the correlation of any classical strategy goes exponentially fast close to zero with increasing input size.

Now we need to consider the case where we have $\mathcal{D}_p$ the uniform distribution over $\mathbb{F}_2$, whose distribution over $\mathbb{F}_p$ we capture with $\Delta_{2}^p$. This case then defines the following bound on the correlation,
\begin{equation}
\mathsf{Corr}_{\mathcal{D}_p}\left(w_{\mathcal{G}_p}^{\star},\mathcal{G}_p\right) \leq \mathsf{Re}\left ( \left (\frac{2^{1 - p} e^{\frac{2\pi i}{p^2}} \left(e^{-\frac{2\pi i}{p^2}} \left(1 + e^{\frac{2\pi i}{p^2}}\right)\right)^p}{1 + e^{\frac{2\pi i}{p^2}}} \right )^n   \mathbin{\bigg/} (p-1)\right ).
\end{equation} 

The previous expression should asymptotically approach zero with increasing input size $n$ for arbitrary $p$. To verify that this holds, we need to show that the expression
\begin{equation}
 \varkappa_p=\bigg |\frac{2^{1 - p} e^{\frac{2\pi i}{p^2}} \left(e^{-\frac{2\pi i}{p^2}} \left(1 + e^{\frac{2\pi i}{p^2}}\right)\right)^p}{1 + e^{\frac{2\pi i}{p^2}}}    \bigg |,
\end{equation}
\noindent is smaller than $1$ for arbitrary $p$. We will prove this inductively, informed by the limit of the expression $\varkappa_p$ with an increasing value of $p$ being 1, given that 
\begin{equation}
\lim_{p \to \infty} \left (\frac{2^{1 - p} e^{\frac{2\pi i}{p^2}} \left(e^{-\frac{2\pi i}{p^2}} \left(1 + e^{\frac{2\pi i}{p^2}}\right)\right)^p}{1 + e^{\frac{2\pi i}{p^2}}} \right )  =1 .
\end{equation}

\noindent Combining this limit with the fact that for $p=3$, the expression $\varkappa_3$ is smaller than $1$, it is enough to prove that the same expression $\varkappa_p$ reaches from below its asymptotic bound. More precisely, we need to prove that this expression is monotonically increasing. For that, we only need to prove that for two values $k$ and $k+1$, the following expression is negative,

\begin{equation}\label{eq:assymptpositive}
\varkappa_k-\varkappa_{k+1}\propto \left( \frac{2 e^{\frac{2\pi i}{k^2}} \left(1 + e^{-\frac{2\pi i}{k^2}}\right)^k}{1 + e^{\frac{2\pi i}{k^2}}} - \left(1 + e^{-\frac{2\pi i}{(k+1)^2}}\right)^k \right). 
\end{equation}

With respect to this expression, we can show that the real part of $\big(1 + e^{-\frac{2\pi i}{(k+1)^2}}\big)^k$ is larger than $\big(1 + e^{-\frac{2\pi i}{k^2}}\big)^k$. For that, we consider the norms of these values without the exponent $k$, and as for the first expression, the corresponding angle $\theta_1=\frac{-2\pi i}{(k+1)^2}$ is larger than $\theta_2=\frac{-2\pi i}{k^2}$, the angle of the second expression, while both being in the fourth quadrant. Hence, the norm of the first term $r_1$ is larger than that of the second term $r_2$. Then, with the exponent $k$, we look into their respective polar representations and observe that the respective real parts of these are equal to $(r_1)^k\cos(k\theta_1)$ and $(r_2)^k\cos(k\theta_2)$. Now, as mentioned previously, $r_1$ is larger than $r_2$, and also both the angles $k \theta_1$ and $k \theta_2$ are in the fourth quadrant, with $k \theta_1$ being larger. Thus, having a larger value of $\cos(k\theta_1)$ than $\cos(k\theta_2)$, obtaining our initial assertion.

The final step is simply to consider that the real part of $\big(1 + e^{\frac{-2\pi i}{k^2}}\big)^k$ does not get larger than $\big(1 + e^{\frac{-2\pi i}{(k+1)^2}}\big)^k$ with the multiplication with $\frac{2 e^{\frac{2\pi i}{k^2}}}{1 + e^{\frac{2\pi i}{k^2}}}$. In particular, this value can be simplified to $1 + i\tan(\frac{\pi}{k^2})$, which, when multiplied with the previous value, makes the real part of $(1 + e^{\frac{-2\pi i}{k^2}})^k$ equal to $(r_2)^k \cos(k\theta_2)(1-\tan(\frac{\pi}{k^2}))$, which is bounded by $2(r_2)^k \cos(k\theta_2)$ for $k\geq 2$. Also, the absolute value of the first term is smaller than the second for $k\geq 2$. Thus, for $k\geq 2$, this expression in \cref{eq:assymptpositive} is negative as we intended. Thus, we obtain that the value of $\mathsf{Corr}_{\mathcal{D}_p}\left(w_{\mathcal{G}_p}^{\star},\mathcal{G}_p \right)$ with the distributed messages being sampled from a uniform distribution over $\mathbb{F}_2^{n(p-1)}$ and encoded with the Hamming encoding to strings in $\mathbb{F}_p^n$ is bounded by $(c_p)^n$, with $c_p$ being a constant strictly smaller than $1$.

Finally, the effect of restricting bits or dits simply results in a version of the same $\mathsf{XOR}$ non-local game with fewer variables, analogous to the case in \cref{entire_break}. Additionally, no probabilistic classical strategy can surpass the two upper bounds established for the correlation measures according to Yao’s min-max principle.
\end{proof}

\subsubsection{Average-case \texorpdfstring{$\BTC^0(k)$}{bPTF0(k)} lower hardness}\label{subsec:qudit_lower}

Here, we will demonstrate that the classical circuit class $\BTC^0(k)$ cannot solve the ISMR problems with a sufficiently significant correlation. In other words, exponentially large circuits will be required to achieve constant positive correlations with these problems. To accomplish this, we will utilize the classical correlations identified for the Modular $\mathsf{XOR}$ non-local games, along with the novel tools developed in \cref{SI_new_switch}.

\subsubsection*{$\NC^0$ lower bounds}
In the first step, we will use the previously determined correlation bounds for the $\mathsf{XOR}$ non-local games and lightcone arguments to demonstrate lower bounds on the efficiency of $\NC^0$ circuits solving the class of ISMR problems. This can be done given that through the Hamming encoding, we can consider each set of $p-1$ bits as a representation of a base-$p$ value in the non-local game for both the messages to the parties and their responses. Consequently, any valid solution to the $\mathcal{R}_p^{m(p-1)}$ problem would equally represent a valid solution to the respective $\mathsf{XOR}$ non-local game played over messages of the type $\mathbb{F}_p^n$, and vice versa.

In addition to the previous connection, we will also show that $\NC^0$ circuits exhibit bounded locality in terms of their causal influence. This implies that each outcome depends on a fixed number of inputs that can never exceed $n$. More precisely, through some combinatorial arguments, it is possible to demonstrate that producing a correct outcome string with these circuits is equivalent to solving the aforementioned non-local games on a subset of the inputs with the same capacity as any of the optimal classical strategies considered before. This will allow us to show that unless the circuits have sufficient depth, they will fail with the same probability as the best classical strategies for the specific non-local games in question determined in \cref{entire_break} and \cref{lemma:entire_breakp}, and this directly defines the efficiency with which they solve the ISMR problems on binary inputs and outputs.

\begin{lemma}\label{block_locality}
Let $C$ be a circuit with $n(p-1)$ binary inputs, arranged in $n$ blocks of size $p-1$, and $m(p-1)$ binary outputs, arranged in $m$ blocks of size $p-1$, with locality $l$ for each outcome bit. Then, for any arrangement of the blocks, there exists a subset of input blocks $S$ of size $\Omega(\min(n,\frac{n^2}{l^2m\cdot(p-1)^4}))$ such that each output block depends on at most one block from $S$.
\end{lemma}

\begin{proof}
To prove this we first consider that each of the output blocks does result from at most $(p-1)\cdot f_{in}$ bits from the previous layer, with gates of fanin equal to $f_{in}$. For the subsequent layers, each bit depends on at most $f_{in}$ elements of the previous layer, except for the last layer. Here as there are $p-1$ input bits per block, each element can depend on $(p-1)\cdot f_{in}$ input blocks. This means that each output block has a locality equal to $l'=(p-1)^2\cdot f_{in}^d=(p-1)^2\cdot l$, respective to the input blocks.

At this point, we can turn to the intersection graph that includes all the output blocks as vertices and connects vertices with an edge if they share an input block. This graph will have $m$ vertices. Given that each output block can depend on $(p-1)^2\cdot l$ input blocks, and each of these inputs can influence up to $(p-1)^2\cdot l$ output blocks, it follows that each vertex is connected to at most $(p-1)^4\cdot l^2$ other vertices. Therefore, this graph will have at most $\mathcal{O}(m\cdot(p-1)^4\cdot l^2)$ edges.

The final step involves applying Turán's theorem, which provides an upper bound for the number of edges $E$ in a graph $G$ to avoid having with certainty a complete subgraph $K_S$, given by the inequality $\big(1-\frac{1}{S}\big)\frac{n^2}{2}\leq E$. Our goal, however, is to ascertain the minimum size of the independent set. To this end, we consider the complement graph $\overline{G}$, where an edge exists in $G$ wherever it was absent in $\overline{G}$, and vice versa. This ensures that if the original graph contained a complete subgraph of size $S$, then the complement graph $\overline{G}$ will contain an independent set of the same size. Consequently, we can establish a lower bound for the size of this independent set as $\big (1-\frac{1}{S}\big)\frac{n^2}{2}\geq \binom{n}{2}-E$. By simplifying the previous expression, we find that $S\geq \frac{n^2}{E}$. Substituting in the number of edges we have, we obtain $S\geq \frac{n^2}{m(p-1)^4 l^2}$.
\end{proof}

Subsequently, by integrating the previous lemma with the upper bounds for the correlations of classical strategies with the Modular $\mathsf{XOR}$ non-local games, we demonstrate that the extent to which an $\NC^0$ circuit successfully solves the ISMR problems is directly influenced by the circuit's locality, its depth, and the length of the desired output string.

In particular, we will determine these upper bounds on the correlation for uniform input distributions over $\mathbb{F}_2^n$ and, subsequently, for uniform distributions over $\mathbb{F}_p^n$.

\begin{lemma}\label{lemma:uniform_qupit_clower}
Let $C$ be a circuit in $\NC^0/\mathsf{rpoly}$ with locality $l$. Then, $C$ fails to solve each of the $\mathcal{R}_{p}^m$ problems, with a uniform distribution $\mathcal{D}_p$ over $\mathbb{F}_2^n$ and Hamming weight satisfying $\big(\sum_{i=1}^n x_i\big)\  \MOD\ p=0$, with a correlation larger than,
\begin{equation}
\mathsf{Corr}_{\mathcal{D}_p} \left(C,\mathcal{R}_p^m \right)=  \mathcal{O} \left (\mathsf{Re} \left (\frac{2^{1 - p} e^{\frac{2\pi i}{p^2}} \left(e^{-\frac{2\pi i}{p^2}} \left(1 + e^{\frac{2\pi i}{p^2}}\right)\right)^p}{1 + e^{\frac{2\pi i}{p^2}}} \right )^{min\big(n,\frac{n^2}{m(p-1)^4l^2}\big)} \right) .
\end{equation}
In particular, for $p=3$ we can bind the success probability by 
\begin{equation}
\Pr_{{\substack{x \sim \mathbb{F}_2^n \\  |x|\ \MOD\ 3= 0}}}[C(x) \in R_3^m(x)] \leq\frac{1}{3}+\left (\frac{9}{10} \right)^{min\big(n,\frac{n^2}{16ml^2}\big)}.
\end{equation}  
\end{lemma}
\begin{proof}
This follows directly from the combination of \cref{block_locality} with \cref{entire_break} and \cref{lemma:entire_breakp}, provided that one considers all the blocks in the independent set $S$ derived by \cref{block_locality} as outcome bits producing an outcome equivalent to linear functions in $\mathbb{F}_p$ over a set of input bits of equal size. Therefore, for this input distribution which translates as a uniform distribution over $\mathbb{F}_2^n$ to a biased distribution over $\mathbb{F}_p^n$ with the respective bias determined by $\Delta_2^p$, the efficiency cannot exceed the respective probability or correlation determined in \cref{entire_break} and \cref{lemma:entire_breakp}. Thereby establishing a lower bound on the efficiency of an $\NC^0$ circuit based on its locality and size.

In addition, we can consider that a random advice string allows for the probabilistic mixture of $\NC^0$ circuits. However, since all these circuits can be reduced to the efficiency of a single classical strategy for the ISMRP at hand, and since there is no probabilistic mixture with a larger correlation with the optimal outcomes, the same applies to $\NC^0/\mathsf{rpoly}$ circuits.
\end{proof}

Considering a biased binary input distribution, further reducing the respective classical correlation bound for $\NC^0$ circuits is possible. 

\begin{corollary}
\label{lemma:Trit_NC_correlation_bound}
Let $C$ be a circuit in $\NC^0/\mathsf{rpoly}$ with locality $l$. Then, $C$ fails to solve each of the $\mathcal{R}_{p}^m$ problems with a biased binary input distribution $\mathcal{D}_p'$, which is defined by a uniform distribution over $\mathbb{F}_p^n$ that is converted into a binary distribution using one of the potential inverse mappings of the Hamming encoding, with an average correlation larger than,
\begin{align}
\mathsf{Corr}_{\mathcal{D}_p'}\left(C,\mathcal{R}_p^n \right)&= \Expec_{{\substack{x \sim H^{-1}(\mathbb{F}_p^{n/\log(p)}) \\  |x|\ \MOD\ p= 0}}}\left ( \mathsf{Re}\left(e^{i\frac{2\pi \left(\left|C(x)\right|-\left|\mathcal{R}_p^n(x)\right|\right)}{p}} \right)\right) \\ &\leq  \frac{p}{p-1}\left(\mathsf{Re}\left( e^{\frac{2\pi i}{p^2}}\right)\right)^{min\big(n,\frac{n^2}{m(p-1)^4l^2}\big)}.
\end{align}

In particular, for $p=3$ we can with the same setting bind the success probability by 
\begin{equation}
\Pr_{{\substack{x \sim H^{-1}(\mathbb{F}_3^{n/\log(3)}) \\  |x|\ \MOD\ 3= 0}}}[C(x)\in R_3^m(x) ] \leq \frac{1}{3}+\left(\frac{17}{20}\right )^{min\big(n,\frac{n^2}{16ml^2}\big)}.
\end{equation} 
\end{corollary}
\begin{proof}
The corollary follows exactly the same steps as the proof of \cref{lemma:uniform_qupit_clower} with the respective use of the bias consideration of $\Delta_p^p$ in \cref{entire_break} and \cref{lemma:entire_breakp}.
\end{proof}

\subsubsection*{$\BTC^0(k)$ lower bounds}
The last lower bound to be determined for the higher-dimensional analysis pertains to the $\BTC^0(k)$ circuit. This will follow, once again, from the use of our switching lemma developed in \cref{lem:GCred2}, in combination with the previous lower bounds for $\NC^0$ circuits to which the initial circuit is reduced. Here again one of the main points for consideration is the goal to determine super-polylogarithmic values of $k$, for which no polynomial-size $\BTC^0(k)$ circuit solves any of the $\mathcal{R}_p^m$ problems with fixed and constant correlation value.

\begin{lemma}\label{lemma:lowerqudit}
For sufficiently large $n$ and $q \in \mathbb{N}_{>0}$, any $\BTC^0(k)/\mathsf{rpoly}$ circuit of depth $d\geq 4$, size $s \leq \exp(n^{1/(2d-2)})$ and parameter $k\leq n^{1/(5d)}$ solves the $\mathcal{R}_{p}^m$ with correlation bounded by,
\begin{equation}
\mathsf{Corr}_{\mathcal{D}_p} \left(|C|,|\mathcal{R}_p^m|\right)
= \mathcal{O} \left (  \mathsf{Re}\left (\left (\frac{2^{1 - p} e^{\frac{2\pi i}{p^2}} \left(e^{-\frac{2\pi i}{p^2}} \left(1 + e^{\frac{2\pi i}{p^2}}\right)\right)^p}{1 + e^{\frac{2\pi i}{p^2}}} \right )^{\frac{n^2}{2^{2q}m^{1+2/q}\mathcal{O}(\log(s)^{2d-1})k^{2d}}} \right )  \right ),
\end{equation}
\noindent with $\mathcal{D}_p$ being the uniform distribution over string in $\mathbb{F}_2^n$ that satisfy $|x|\ \MOD\ p= 0$.
\end{lemma}

\begin{proof}
Beginning with a $\BTC^0(k)$ circuit, we intend to reduce this object, enabling the application of appropriate lemmas to the problem at hand. The initial segment of our proof parallels the approach in \cref{ElowerPHP}, initiating with a random restriction denoted as $\rho$. The probability of this restriction is set to $p=\frac{1}{m^{1/q}} \cdot \mathcal{O}(\log(s)^{d-1} \cdot k^d)$, where $t$ is chosen to be $\frac{pn}{8}$. Subsequently, a second random restriction, $\tau$, is applied to the $2t$ variables within the global decision tree, which, after the initial random restriction, is reduced to $\DT(q-1)^m$. This second restriction occurs with a probability of $1- \exp(-\Omega(pn))$. 

Thus, using $p$-random restrictions we obtained that sampling $|x| \ \MOD\ p$ strings consistent with the restrictions does provide us with a uniform distribution over the $|x| \ \MOD\ p$ strings, given that over any non-total subset of the input bits these are uniform distributions in $\mathbb{F}_2$. Exactly, the same argument does work over the restriction $\tau$ given that the variables of the global decision tree are selected randomly in $\{0,1\}^{2t}$. This property of $p$-random restrictions, combined with the fact that we can reconstitute the surviving variables into a new string for which we need to solve the initial problem, allows us to apply \cref{lemma:uniform_qupit_clower} directly without any limitations on the specific locations of the variables that remain alive. Combining this with the fact that with a high probability more $pn/4$ variables are alive, and the circuit does have to evaluate the $\mathcal{R}_{p}^m$ problem with local decision trees of depth $q-1$, we obtain that these depend at most on $2^q$ of the variables alive. Consequently, we have that any $\BTC^0(k)$ circuit does compute the ISMR problems with a correlation bounded by,
\begin{align*}
\mathsf{Corr}_{\mathcal{D}_p} \left(C,\mathcal{R}_p^m \right)
\leq  \mathsf{Re}\left ( \left (\frac{2^{1 - p} e^{\frac{2\pi i}{p^2}} \left(e^{-\frac{2\pi i}{p^2}} \left(1 + e^{\frac{2\pi i}{p^2}}\right)\right)^p}{1 + e^{\frac{2\pi i}{p^2}}} \right )^{\Omega\Big(\frac{(pn)^2}{2^{2q}m}\Big)}   \mathbin{\bigg/} (p-1)\right )\\
= \mathcal{O} \left ( \mathsf{Re}\left ( \left (\frac{2^{1 - p} e^{\frac{2\pi i}{p^2}} \left(e^{-\frac{2\pi i}{p^2}} \left(1 + e^{\frac{2\pi i}{p^2}}\right)\right)^p}{1 + e^{\frac{2\pi i}{p^2}}} \right )^{\frac{n^2}{2^{2q}m^{1+2/q}\mathcal{O}(\log(s)^{2d-1})k^{2d}}}  \right ) \right ),
\end{align*}
\noindent with $\mathcal{D}_p$ being the uniform distribution over string in $\mathbb{F}_2^n$ that satisfy $|x|\ \MOD\ p= 0$.

Finally, the same proof does hold for the case where the $\BTC^0(k)$ circuit is given a random string as advice. Since any $\BTC^0(k)$ reduces to the same object after the use of random restriction, and that from \cref{lemma:uniform_qupit_clower} we do now that any probabilistic strategy defined by the random advice does fail, so do $\BTC^0(k)/\mathsf{rpoly}$ circuits.
\end{proof}

This finishes determining the classical upper bounds for the correlations achieved by the $\BTC^0(k)$ circuit for the entire set of ISMR problems. Finally, we will also enunciate a bound for the maximum probability with which the same circuit class can solve the specific case of $\mathcal{R}_3^m$, as in this case, it follows directly from the correlation bound.

\begin{corollary}\label{lowerqutrit}
For sufficiently large $n$ and $q \in \mathbb{N}_{>0}$, any $\BTC^0(k)/\mathsf{rpoly}$ circuit of depth $d\geq 4$, size $s \leq \exp(n^{1/(2d-2)})$ and parameter $k\leq n^{1/(5d)}$ solves the $\mathcal{R}_{3}^m$ with a probability bounded by,
\begin{equation}
\Pr_{{\substack{x \sim \mathbb{F}_2^n \\  |x|\ \MOD\ 3= 0}}}[C(x)\in R_3^m(x)] \leq\frac{1}{3}+\left (\frac{9}{10} \right)^{\Omega\Big(\frac{n^2}{2^{2q}m^{1+2/q}\mathcal{O}(\log(s)^{2d-1})k^{2d}}\Big)}.
\end{equation}
\end{corollary}
\begin{proof}
The proof follows in the same manner as in \cref{lemma:lowerqudit} by executing all the steps up to the use of the lower bound for the success probability with $\NC^0/\mathsf{rpoly}$ circuits. In particular, at this point, we consider the upper bound determined for the success probability in the specific case of $\mathcal{R}_{3}^m$ as outlined in \cref{lemma:uniform_qupit_clower}.
\end{proof}

\subsubsection{Qupit upper bound}\label{subsec:qudit}

Here, we will describe $\QNC^0$ circuits over qupits to solve probabilistically the ISMR problems. This will use the qupit generalizations of poor man's cat states from \cite{Oliveira24}. However, we apply some slight modifications such that it simplifies the following set of proofs as it also contributes to the self-containment of the document.

\begin{definition}[Generalized poor-man's qudit state]\label{poorghz}
Let $p$ be a prime and $z^{+0}\in \mathbb{F}_p^n$.
We define a generalized poor-mans qudit state $\ket{\mathsf{GPM}_p^n}=\frac{1}{\sqrt{p^{n}}} \sum_{i \in \mathbb{F}_p} \ket{z^{+i}}$ with $z^{+i}$ being the string that takes a initial defined string $z^{+0}$ and applies the bitwise sum of $i$ to each bit of the string obtaining with that,
\begin{equation}
z^{+i}=\big ( (z_1+i)\ \MOD\ p\big ) || \big ( (z_2+i)\ \MOD\ p \big )  || \hdots || \big ((z_n+i)\ \MOD\ p \big ).
\end{equation}
\end{definition}

Now it will be important to prove these types of states can be created with a $\QNC^0$ circuit.

\begin{lemma}\label{poorconst}
For $n \in \mathbb{N}$, $\ket{\mathsf{GPM}_p^n}$ can be prepared with $\QNC^0$ circuit based on a fully connected and without cycles graph $G=(V,E)$ over the respective qudit prime dimension $p$.  The respective string $v$ defining the precise state is described by the following recurrence relation between vertex qudits and the measurement outcomes from the edge qudits,
\begin{equation}
 v_u = \left(\sum_{e_{i,j} \in Path(u,w)} (-1)^{|Path(u,j)|+|Path(u,w)|} \ e_{i,j} + v_w\right)\ \MOD\ p    
\end{equation}
for any two vertices $u,w \in V$, and the shortest path $Path(u, w)$ from $u$ to $w$ in the graph $G$.
\end{lemma}

\begin{proof}
We will provide proof by describing concretely the circuit that does create the  $\ket{\mathsf{GPM}_p^n}$ states. In particular, we will describe how the various stages of the circuit illustrated in \cref{fig:poor_creation} compute the correct state. 
\begin{figure}[htbp]
\centering
\scalebox{0.75}{
\begin{quantikz}[classical gap=0.25cm]
\lstick{$\ket{0}_{v_1}$} & \gate{\mathsf{SUM}}\slice[style={blue},label style={inner sep=1pt,anchor=south west,rotate=45}]{Step 1}  &\gate[8,disable auto height]{\begin{tabular}{c}$\mathsf{CSUM}$ \\ \\R\\[-1.1ex] o\\[-1.1ex]u\\[-1.1ex]t\\[-1.1ex]i\\[-1.1ex]n\\[-1.1ex]e\\ \\[-1.1ex] \end{tabular}}\slice[style={blue},label style={inner sep=1pt,anchor=south west,rotate=45}]{Step 2} & \slice[style={blue},label style={inner sep=1pt,anchor=south west,rotate=45}]{Step 3} &\gate[5,disable auto height]{\begin{tabular}{c} \\$\mathsf{INV}$\\ \\R\\[-1.1ex] o\\[-1.1ex]u\\[-1.1ex]t\\[-1.1ex]i\\[-1.1ex]n\\[-1.1ex]e\\ \\[-1.1ex] \end{tabular}} \slice[style={blue},label style={inner sep=1pt,anchor=south west,rotate=45}]{Step 4}& \rstick{$\ket{z_1}$} \\
 & \wireoverride{n} & \wireoverride{n} & \wireoverride{n} & \wireoverride{n}  &  \wireoverride{n} \\
\lstick{\vdots} & \wireoverride{n} & \wireoverride{n} & \wireoverride{n}\vdots   &  \wireoverride{n} & \wireoverride{n} \rstick{\vdots} \\
 & \wireoverride{n} & \wireoverride{n} & \wireoverride{n}  & \wireoverride{n} &  \wireoverride{n}  \\
\lstick{$\ket{0}_{v_n}$}&  \gate{\mathsf{SUM}} & &   & & \rstick{$\ket{z_n}$}\\
\lstick{$\ket{0}_{e_{i,j}}$}& & &  \meter{} &  & \rstick{$e_{i,j}$}  \\
\lstick{\vdots}& \wireoverride{n}&\wireoverride{n}&\wireoverride{n}&\wireoverride{n}&\wireoverride{n} \rstick{\vdots}\\
\lstick{$\ket{0}_{e_{l,k}}$}& & &  \meter{} &  &   \rstick{$e_{l,k}$}
\end{quantikz}}
\caption{Illustration of a constant-depth quantum circuit for generating a generalized poor man's cat state, accompanied by the classical information needed to describe the specific instance created.}
\label{fig:poor_creation}
\end{figure}
\paragraph{Step 1.} The circuit does start with a qudit of prime dimension $p$ for each vertex and edge in the graph $G=(V,E)$. Based on this labeling, we apply a Fourier gate or higher dimensional Hadamard gate $\mathsf{F}$ to each vertex qudit. This does create the state $\ket{\Psi_1}= \frac{1}{\sqrt{p^n}}\sum_{x\in \mathbb{F}_n^d} \ket{x}\ket{0}^{\otimes l}$.

\paragraph{Step 2.} Then, the $\mathsf{CSUM}$ routine applies a $\mathsf{CSUM}$ gate between every edge $e_i$ and vertex $v_i$ qudit that is connected in the graph $G$, with the former being the target and the latter the control. Creating with that the state $\ket{\Psi_2}= \frac{1}{\sqrt{p^n}}\sum_{x\in \mathbb{F}_n^d} \ket{x, e_{i,j},\hdots, e_{l,k}}$, with $e_{i,j}=v_j+v_i\ \MOD\ p$, given that $e_{i,j}$ is an edge between the vertexes $v_j$ and $v_i$ in $G$.

\paragraph{Step 3.} Subsequently, one does measure all the qudits with the edge labeling. Previously, these were in a superposition of possible values for the sums of the computational bases of the respective vertexes. Now, these are fixed and this does fixes strictly that the quantum state over the vertex qudits fulfill the following relation, 
\begin{equation}\label{rel1}
 v_u  \equiv \left(\sum_{e_{i,j} \in Path(u,w)} (-1)^{|Path(u,j)|} \ e_{i,j} + (-1)^{|Path(u,w)|} v_w\right)\ \MOD\ p  .
\end{equation}

Considering a certain vertex in the graph, we immediately obtain that we have for this vertexes the superposition $\sum_{i\in \mathbb{F}_p} \frac{1}{\sqrt{p}}\ket{i}$, and that for each one of this computational basis states all the other computational basis states are fixed, such that, 
$\ket{\Psi_3}= \frac{1}{\sqrt{p}}\sum_{i\in \mathbb{F}^d} \ket{x_i,e_{i,j},\hdots, e_{l,k}}$, with $x_1,\hdots,x_p \in \mathbb{F}_p^n$ and dictated by the relation described in \cref{rel1}. In particular, if we select the first listed vertex $v_0$ as our reference, and with, we can describe the previous states as
\begin{equation}
\ket{\Psi_3}= \frac{1}{\sqrt{p}}\sum_{i\in \mathbb{F}_p} \ket{v_0^{+i},\hdots, v_j^{(-1)^{Path(0,j)}i},\hdots, v_n^{(-1)^{Path(0,n)}i},e_{i,j},\hdots, e_{l,k}},
\end{equation}
\noindent with $v^{+0} \in \mathbb{F}_p^n$.

\paragraph{Step 4.} The state in the previous stage is very close to fulfilling the definition of a generalized qudit poor-mans cat state. The only issue is that for every vertex the path that is of odd size to the reference vertex, has the computational basis changing by the additive inverses of $i$ of the reference qudit. In the final step, This can be simply solved by applying a $\mathsf{INV}$ gate to all these qudits in the ``$\mathsf{INV}$ Routine". The final state, then is of the form $\ket{\Psi_4}= \frac{1}{\sqrt{p}}\sum_{i\in \mathbb{F}^d} \ket{v^{+i},\hdots, v_n^{+i},e_{i,j},\hdots, e_{l,k}}$, with $v^{+0}\in \mathbb{F}_p^n$. This exactly fulfills the definition of the poor-mans qudits states, $\ket{\mathsf{GPM}_p^n}$. Also, the individual $z_i$ bits fulfill the following recurrence relation, 
\begin{equation}
 v_u  \equiv \left(\sum_{e_{i,j} \in Path(u,w)} (-1)^{|Path(u,j)|+|Path(u,w)|} \ e_{i,j} + v_w\right)\ \MOD\ p .
\end{equation}
\end{proof}

Now we will show that there exists a $\QNC^0$ circuit that can solve the probabilistic $R_3^m$ ISMR problem, with a probability strictly larger than the upper bound determined for the classical solutions. We do the qutrit case first as it works as a specific instance and is easier to follow, and then we generalize for the arbitrary qudit prime dimension.

\begin{lemma}\label{lemma:comput_trit}
There exists a $\QNC^0$ circuit that given a $\ket{\mathsf{GPM}_3^n}$ state, along with the corresponding graph $G$ and a set of measured edges $E=\{e_{i,j},\ldots,e_{l,k}\}$, is capable of solving the modular relational problem denoted as $\mathcal{R}_{3}^m$, over a uniform distribution of strings satisfying $\sum_{i=1}^n x_i\ \MOD\ 3=0$ in $\mathbb{F}_2^n$,  with a success probability exceeding $1/2$ for any given input. Here, $m=n\cdot g ^2$, the size of the outcome string, depends on the length $g$ of the largest path in $G$ from any vertex to a particular vertex that minimizes the maximum distance to any other vertex.
\end{lemma}

\begin{proof}
For this proof, we consider the quantum circuit that archives this and is represented in \cref{fig:circ_qudit}. Also, we consider the existence of a $\QNC^0$ based on \cref{poorconst} that prepares the resource state, which will be in this case the $\ket{\mathsf{GPM}_3^n}$ state. Also, we consider the additional access to the corresponding graph $G$ and string $E$ for the measured edges in the state preparation.
\begin{figure}[htbp]
\centering
\scalebox{0.75}{
\begin{quantikz}[classical gap=0.25cm]
\lstick{$x_1$} & \slice[style={blue},label style={inner sep=1pt,anchor=south west,rotate=45}]{State preparation}  & \ctrl{11}  &  & & \slice[style={blue},label style={inner sep=1pt,anchor=south west,rotate=45}]{Step 1} & \slice[style={blue},label style={inner sep=1pt,anchor=south west,rotate=45}]{Step 2a}&  \slice[style={blue},label style={inner sep=1pt,anchor=south west,rotate=45}]{Step 2b}& &\gate[13,disable auto height]{\begin{tabular}{c} $\NC^0$ \\ \\C\\[-1.1ex] o\\[-1.1ex]r\\[-1.1ex]r\\[-1.1ex]e\\[-1.1ex]c\\[-1.1ex]t\\[-1.1ex]i  \\[-1.1ex]o\\ [-1.1ex]n\\[-1.1ex] \end{tabular}} \slice[style={blue},label style={inner sep=1pt,anchor=south west,rotate=45}]{Step 3}& \rstick{$y_1$} \\
\lstick{$x_2$} & &  & \ctrl{9} &  &  &  & & & & \rstick{$y_2$}  \\
\lstick{\vdots} & \wireoverride{n} & \wireoverride{n} & \wireoverride{n} & \wireoverride{n}  &  \wireoverride{n}& \wireoverride{n}&\wireoverride{n}& \wireoverride{n}&\wireoverride{n} \\
  & \wireoverride{n} & \wireoverride{n} & \wireoverride{n} &  \wireoverride{n} &\wireoverride{n}&  \wireoverride{n}  &  \wireoverride{n} \\
\lstick{$x_n$}  &  &  & & & \ctrl{2} & &  & &   \\
& \wireoverride{n} & \wireoverride{n} & \wireoverride{n} & \wireoverride{n}  &  \wireoverride{n}& \wireoverride{n} & \wireoverride{n}\\
\lstick{$\ket{0}_1$} &  \gate[7,disable auto height]{\begin{tabular}{c}$\mathsf{GPM}$ \\ \\S\\[-1.1ex] t\\[-1.1ex]a\\[-1.1ex]t\\[-1.1ex]e\\[-1.1ex]\\[-1.1ex]c\\[-1.1ex]r  \\[-1.1ex]e\\ [-1.1ex]a\\[-1.1ex]t\\ [-1.1ex]o\\ [-1.1ex]n\\ \end{tabular}}  &  & & & \gate{R_{Z}(\phi)}  & \gate{\mathsf{F}} &\meter{}  & \gate[7,disable auto height]{\begin{tabular}{c} \\ $\QNC^0$ \\\\[-1.1ex] \\[-1.1ex]D\\[-1.1ex]e\\[-1.1ex]c\\[-1.1ex]o\\[-1.1ex]d\\[-1.1ex] e \\[-1.1ex]r\\ [-1.1ex]\\[-1.1ex]\\ [-1.1ex]\\ [-1.1ex]\\ \end{tabular}} & & \wireoverride{n} \rstick{\vdots} \\
\wireoverride{n}  &\wireoverride{n}   &\wireoverride{n}  &\wireoverride{n}  & \wireoverride{n}  & \wireoverride{n} &  \wireoverride{n} &  \wireoverride{n}&  \wireoverride{n}   \\
\lstick{\vdots}\wireoverride{n}&\wireoverride{n}& \wireoverride{n} &\wireoverride{n} & \wireoverride{n}\rotatebox{145}{ $\vdots$}  &  \wireoverride{n}  &\wireoverride{n} & \wireoverride{n}& \wireoverride{n}  &   \wireoverride{n} &   \wireoverride{n}\\
\wireoverride{n}& \wireoverride{n}&\wireoverride{n} &  \wireoverride{n}   & \wireoverride{n}   & \wireoverride{n} & \wireoverride{n}& \wireoverride{n} & \wireoverride{n}    \\
\lstick{$\ket{0}_{n-1}$}& & & \gate{R_{Z}(\phi)}&  & & \gate{\mathsf{F}}&\meter{}& & \\
\lstick{$\ket{0}_n$}& & \gate{R_{Z}(\phi)} &   & & & \gate{\mathsf{F}} & \meter{}  & &\\
\lstick{$\ket{0}^{\otimes |E|}$} &  &  & \wire[l][1]["{E=\{e_{i,j},\hdots,e_{l,k}\}}" {above,pos=0.8}]{a}  & & &  & &  & &  \rstick{$y_m$}
\end{quantikz}}
\caption{Illustration of a constant-depth quantum circuit solving all instances of the ISMR problems $\mathcal{R}_p^m$ with parameter $p$ prime.}
\label{fig:circ_qudit}
\end{figure}
\paragraph{Step 1.} In this step, controlled qudit $\mathsf{Z}$ rotations will be applied to each one of the qudits of the resource state. The control of each one of these rotations will be a bit of input string $x$. Obtaining with that the following state, 
\begin{equation}
 \ket{\Psi_1}= \bigotimes_{i=1}^n R_z \Big(\frac{2\pi x_i}{9}\Big) \ket{\mathsf{GPM}_3^n} = \frac{1}{\sqrt{3}}\sum_{i=0}^{2} e^{\frac{2\pi i \langle x,z^{+i} \rangle }{9}}\ket{z^{+i}}.
\end{equation}

Although the previous expression accurately represents the resulting state, it would be beneficial to reformulate the state in a manner more conducive to further analysis. Specifically, we aim to show that the resulting phases of these states, when rewritten as follows, more effectively capture the output strings obtained after applying step 2,
\begin{align}
\ket{\Psi_1}&=   e^{\frac{2\pi i\cdot \langle x,z \rangle}{9}} \frac{1}{\sqrt{3}}\sum_{i=0}^{2} e^{\frac{2\pi i\cdot \langle x,z^{+i} \rangle-\langle x,z \rangle }{9}}\ket{z^{+i}}\\
&= e^{\frac{2\pi i\cdot \langle x,z \rangle}{9}} \frac{1}{\sqrt{3}} \Big ( \ket{z} + \sum_{i=1}^{2} e^{\frac{2\pi i\cdot \big (2 \langle z^{+i},x\rangle + (\sum_{j \in \{1,2\}\setminus i} \langle z^{+j},x\rangle) \big ) }{9}}\ket{z^{+i}} \Big ) \\
&= e^{\frac{2\pi i\cdot \langle x,z \rangle}{9}} \frac{1}{\sqrt{3}} \Big ( \ket{z} +  e^{\frac{2\pi i\cdot (|x|/3+ \langle x,(z^{+1})^2 \rangle) } {3}}\ket{z^{+1}} +e^{\frac{2\pi i\cdot (2|x|/3- \langle x,z^2 \rangle) } {3}}\ket{z^{+2}} \Big ).
\end{align}

\noindent We begin by highlighting the initial phase component of $z^{+0}$, or equivalently $z$, by treating it as a universal phase factor. Then, the transition to the second expression is achieved by considering that $\sum_{i=0}^{2} \langle x,z^{+i} \rangle = 3|x|$, and rewriting the various phase terms. Likewise, given our assumption that $|x|\ \MOD\ 3=0$, the term $\frac{3|x|}{9}$ will be matching to a multiple of $2\pi$ and can thus be disregarded.

For the last step, we need to examine the inner products present in the phases. For simplicity, we will analyze the ideal cases where $z$ is uniformly equal,
\begin{center}
\begin{tabular}{|c|c|c|c|}
\hline
Phases & $z=0^{\otimes n}$ & $z=1^{\otimes n}$ & $z=2^{\otimes n}$  \\  
\hline
$2 \langle z^{+1},x \rangle + \sum_{j \in \{1,2\}\setminus 1} \langle z^{+j},x \rangle$ & $4 |x|$ & $4|x|$ & $|x|$ \\
\hline
$2 \langle z^{+2},x \rangle + \sum_{j \in \{1,2\}\setminus 2} \langle z^{+j},x \rangle$  & $5|x|$ & $2|x|$ & $2|x|$ \\
\hline
\end{tabular}.
\end{center}
\vspace{0.2cm}

Although $z$ can assume any value in $\mathbb{F}_3^n$, we can divide the string into three sub-components corresponding to the aforementioned strings. Specifically, the first term contributes at least $|x|$ and the second term with $2|x|$, while one can introduce additional terms that determine the correct phases based on the specific random string $z$. This determination is straightforward due to an additional contribution of magnitude 3 in cases where the bits  $z_i$ are either $0$ or $1$ for the first phase term, which can be computed using the expression $\langle x,(z^{+1})^2 \rangle$. The same principle applies to the second phase, with the additional contribution determinable via the expression $\langle x,z^2 \rangle$, resulting in the final expression.

\paragraph{Step 2.} In this step, we will first apply the respective Fourier gate $\mathsf{F}_3$ to each one of the qudits of the state, and analyze the resulting support of the outcome state,

\begin{align}
\ket{\Psi_{2a}}=&\bigotimes_{l=1}^n F_3 \left( e^{\frac{2\pi i\cdot \langle x,z \rangle}{9}} \frac{1}{\sqrt{3}} \left( \ket{z} +  e^{\frac{2\pi i\cdot (|x|/3+ \langle x,(z^{+1})^2 \rangle) } {3}}\ket{z^{+1}} +e^{\frac{2\pi i\cdot (2|x|/3+ \langle x,z^2 \rangle) } {3}}\ket{z^{+2}} \right) \right) \\
= &e^{\frac{2\pi i\cdot \langle x,z \rangle}{9}}  \frac{1}{\sqrt{3}} \Big(\sum_{y\in \mathbb{F}_3^n }  e^{\frac{2\pi i\langle y,z \rangle }{3}}\ket{y} + e^{\frac{2\pi i\cdot (|x|/3+ \langle x,(z^{+1})^2 \rangle) } {3} +\frac{2\pi i \langle y,z^{+1}\rangle}{3}}\ket{y}\\  &+e^{\frac{2\pi i\cdot (2|x|/3+ \langle x,z^2 \rangle) }{3} +\frac{2\pi i \langle y,z^{+2}\rangle}{3}} \ket{y} \Big) \\
=& e^{\frac{2\pi i\cdot \langle x,z \rangle}{9}+\frac{2\pi i \langle y,z\rangle }{3}} \frac{1}{\sqrt{3}}  \sum_{y\in \mathbb{F}_3^n } \ket{y} + e^{\frac{2\pi i\cdot (|x|/3+|y|+ \langle x,(z^{+1})^2 \rangle) }{3} } \ket{y} +
e^{\frac{2\pi i\cdot (2(|x|/3+|y|)+ \langle x,z^2 \rangle }{3}} \ket{y} .\label{final_phases}
\end{align}

Now, we can analyze the outcomes of the measurements after step 2b. In particular, let's consider the ideal case where the inner products $\langle x,(z^{+1})^2 \rangle$ and $\langle x,z^2 \rangle$ are both zero. In this scenario, we find that for $|x|/3+|y|\ \MOD\ p =0$, the three phases are equal to 1, and all the base states $y$ meeting this criterion are measured with equal probability. Specifically, the strings $y$, observed after measurement, are all congruent with the additive inverse of $|x|/p$. Also, we do obtain that, whenever the $|x|/3+|y|\ \MOD\ p$ equals either 1 or 2, it generates all three roots of unity as one can see in \cref{final_phases}. These basis states then have a zero probability of being observed. This proves that, for these ideal values of the inner products $\langle x,(z^{+1})^2 \rangle$ and $\langle x,z^2 \rangle$, we can perfectly solve the inverse strict modular relational problem $\mathcal{R}_3^n$. 

Unfortunately, these two additional terms in the phase can take any value in $\mathbb{F}_3$, and therefore scramble the outcomes that we intended to observe. Forcing us to analyze all possible cases.

\paragraph{Step 3.} The final step, is intended to handle the randomness resulting from the inner products in the phases of \cref{final_phases}. In particular, we can process the outcome such that we increase the probability of having the correct outcomes to the desired relational problem at hand. Moreover, we will list all possible values, and divide their analysis into three cases.

\vspace{0.3cm}

\begin{minipage}{0.32\textwidth}
\centering
\begin{tabular}{|c|c|c|c|}
\hline
Case 1 & $a_1$ & $b_1$ & $c_1$  \\  
\hline
$\langle x,(z^{+1})^2 \rangle$ & 1 & 1 & 0 \\
\hline
$\langle x,z^2 \rangle$ & 2 & 0 & 2 \\
\hline
\end{tabular}
\label{tab:case1}
\end{minipage}%
\hfill
\begin{minipage}{0.32\textwidth}
\centering
\begin{tabular}{|c|c|c|c|}
\hline
Case 2 & $a_2$ & $b_2$ & $c_2$ \\  
\hline
$\langle x,(z^{+1})^2 \rangle$ & 2 & 2 & 0 \\
\hline
$\langle x,z^2 \rangle$ & 1 & 0 & 1 \\
\hline
\end{tabular}
\label{tab:case2}
\end{minipage}%
\hfill
\begin{minipage}{0.32\textwidth}
\centering
\begin{tabular}{|c|c|c|}
\hline
Case 3 & $a_3$ & $b_3$    \\  
\hline
$\langle x,(z^{+1})^2 \rangle$ & 2 & 1  \\
\hline
$\langle x,z^2 \rangle$ & 2 & 1  \\
\hline
\end{tabular}
\label{tab:case3}
\end{minipage}

\paragraph{Case 1.} For the assignment of values to the inner product displayed in $a_1$, the measured strings $y$ are congruent with the additive inverse of $|x|/3+1$, due to the phases being related to $|x|/3 + |y| + 1\ \MOD\ 3  = 0$. Furthermore, in cases $b_1$ and $c_1$, we obtain strings $y$ that fulfill the same property with a probability of $1/2$.

In these scenarios, more than half of the subcases map to an incorrect string. However, more than half of the subcases also map to a string that deviates by the addition of two to its $\ell_1$-norm. This effect can be simply countered by adding the value $1$ (it's additive inverse). Nevertheless, this addition must be informed by the specific cases under consideration. Therefore, we will use the term $\langle x,(z^{+1})^2 \rangle$ for this purpose, and the method to construct this value will be described later. Ultimately, given that the probability of event $a$ is non-zero, all these cases map to the correct string with a probability greater than $1/2$.

\paragraph{Case 2.} For $a_2$, we encounter the inverse scenario of $a_1$. In this case, the states are projected onto strings that are equal to the additive inverse of $(|x|/3+2)^{-1}$, resulting in a shift of $1$ from the correct $\ell_1$-norm. The subcases described in $b_2$ and $c_2$ exhibit a similar effect, mapping to the correct string half of the time, while in the other instances, the strings have their correct Hamming weights increased by one.

It's noteworthy that, by applying the same correction as in case 1, specifically the term $\langle x,(z^2{+1})^2\rangle$, the strings in $a_2$ which were all previously shifted, are now corrected. Furthermore, the two subcases $b_2$ and $c_2$ continue to map correctly half of the time. This, following the same line of reasoning presented previously, demonstrates that for all these cases, the input strings are mapped to the correct result with a probability greater than $1/2$.

\paragraph{Case 3.} The last two cases, $a_3$ and $b_3$ are instances where both map with a probability of $1/2$ to the two incorrect strings. These strings are congruent to $(|x|/3+1)^{-1}$ and $(|x|/3+2)^{-1}$. Therefore, by applying our correction term $\langle x,(z^2{+1})^2 \rangle$, we ensure that one of them is correctly mapped in both scenarios, thus correctly solving half of these instances. Additionally, by considering the ideal case discussed initially, which is unaffected by the corrections, more than half of the outcome strings are correct.

\paragraph{Correction string.} Finally, as previously demonstrated, our correction is effective in all cases, ensuring consistency across all computed probabilities. As a result, we obtain final measured strings capable of solving the problem with a probability greater than $1/2$, given any $\ket{\mathsf{GPM}}$ characterized by arbitrary $z$ and $G$. However, to fully substantiate this conclusion, one must demonstrate that the term $\langle x,(z^{+1})^2 \rangle$ can be computed using an $\NC^0$ circuit. To achieve this, we can compute this value using the following expression,
\begin{equation}
\langle x,(z^{+1})^2 \rangle= \sum_{i\in [n]} x_i \cdot \left(\sum_{e_{i,j} \in Path(i,c)} (-1)^{|Path(i,j)|+|Path(i,c)|} \ e_{i,j}+ z_c +1\right)^2  \ \MOD\ 3.
\end{equation}

The inner sum $\sum_{e_{i,j} \in Path(i,c)} (-1)^{|Path(i,j)|+|Path(i,c)|} \ e_{i,j}$ can be constructed by enumerating the values of the edge measurements along the respective paths. Then, achieving the sum of the two values $z_0$ and $1$ can be straightforwardly accomplished by appending these digits to the string. To realize the square, we only need to generate all the terms that would result from squaring this expression and perform the individual multiplication operations over $\mathbb{F}_3$ for these elements. Subsequently, all these bits are multiplied with the respective input bit $x_i$, and this process is repeated for all values of $i$ in the outer sum. The total number of sums is $n$ while using the largest path $g$ we have that the exponentiation creates $g^{2}$ terms at most and we end up with $ng^{2}$ size correction string. 

We have been describing these operations over $\mathbb{F}_3$, but one can perform these operations over $\mathbb{F}_2$ using an inverse of the Hamming encoding and the respective logical operations considering our encoding. This approach allows us to use an $\NC^0$ circuit for this reduction after using a $\QNC^0$ decoder to map the measured dits to bits. Finally, the binary version of the string $\langle x,(z^{+1})^2 \rangle$ is added to the measurement outcomes from stage 3, as previously described, thereby concluding the proof.
\end{proof}

The generalization to arbitrary qupits does follow the same ideas as previously. However, the states generated by the randomness of the qupit poor-mans cat states are more evolved. This will be encapsulated in the following lemmas, with the added consideration that one can solve the modular inverted strict relational problems closer to the correct outcomes within the Abelian domains and with a probability that exceeds mere random guessing.  

\begin{lemma}\label{lemma:computcorre}
A $\QNC^0$ circuit, denoted as $C_q$, is given a state $\ket{\mathsf{GPM}_p^n}$, the corresponding graph $G$, and a set of measured edges $E = \{e_{i,j}, \ldots, e_{l,k}\}$. This circuit is capable of solving the modular relational problem, represented as $\mathcal{R}_p^m$, across a uniform distribution $\mathcal{D}_p$ of strings that satisfy the condition $\sum_{i=1}^n x_i\ \MOD\ p = 0$ within $\mathbb{F}_2^n$. It achieves this with a success probability of $\frac{2p - 2}{p^2}$ for any given input. Moreover, it maintains an average correlation with the correct outcome within the Abelian domain,
\begin{equation}
\mathsf{Corr}_{\mathcal{D}_p} \left(C,\mathcal{R}_p^m \right) = \frac{p-1}{p^2}.
\end{equation}
In this context, $m = n \cdot g^{p-1}$, where $m$ represents the size of the outcome string. This size is influenced by $g$, which is the length of the longest path within $G$ from any vertex to a designated vertex that minimizes the greatest distance to all other vertices.
\end{lemma}
\begin{proof}
To prove this, we will once again consider the circuit that effectively achieves this. In particular, we will revisit the circuit presented in \cref{fig:circ_qudit}, as it possesses the structure necessary to solve all ISMR problems by simply adapting the correct prime-dimensional Hilbert space for each stage.  Starting, with the consideration that a resource state for the form of a $\ket{\mathsf{GPM}_p^n}$ state will be provided in addition to the corresponding graph $G$ and string $E$ for the measured edges in the state preparation.

\paragraph{Step 1.} The first step applies controlled qudit $\mathsf{Z}$ rotations parameterized by the angle $\phi=\frac{2\pi}{p^2}$ to the resource state, while the control bits of each one of these rotations will refer to the input string $x$. Obtaining with that the following state, 
\begin{align}
 \ket{\Psi_1}&= \bigotimes_{i=1}^n R_z \left(\frac{2\pi x_i}{p^2}\right) \ket{\mathsf{GPM}_p^n} = \frac{1}{\sqrt{p}}\sum_{i=0}^{p-1} e^{\frac{2\pi i \langle x,z^{+i} \rangle }{p^2}}\ket{z^{+i}}\\
&= e^{\frac{2\pi i\cdot \langle x,z \rangle}{p^2}} \frac{1}{\sqrt{p}} \left( \ket{z} + \sum_{i=1}^{p-1} e^{\frac{2\pi i\cdot \left (2 \langle z^{+i},x\rangle + (\sum_{j \in \{1,\hdots,p-1\}\setminus i} \langle x,z^{+j} \rangle) \right) }{p^2}}\ket{z^{+i}} \right )\label{eq:quditg1} \\
&= e^{\frac{2\pi i\cdot \langle x,z \rangle}{p^2}} \frac{1}{\sqrt{p}} \left ( \ket{z} + \sum_{i=1}^{p-1} e^{\frac{2\pi i\cdot \left (i|x|/p + (\sum_{j \in {1,\hdots,i}}  \left(\langle x,z^{+j} \rangle\right)^{p-1} \right ) }{p}}\ket{z^{+i}} \right ).\label{eq:quditg2}
\end{align}

All the reductions until the last transition from \cref{eq:quditg1} to \cref{eq:quditg2} follow exactly equally as the reduction from \cref{lemma:comput_trit}. For the last, one does need to redo the analysis of the contribution to the resulting phases by each string of the type $z=0^n,1^n, \hdots, (p-1)^n$. In particular, one can determine that for these values of $z$, we have that, 
\begin{table}[H]
\refstepcounter{table} 
\label{mytable}
\begin{center}
\begin{tabular}{|c|c|c|c|c|}
\hline
Phases & $z=0^{\otimes n}$ & $\hdots $ & $z=(p-2)^{\otimes n}$ & $z=(p-1)^{\otimes n}$ \\  
\hline
$2 \langle z^{+1},x \rangle + \hdots$ & $(p\Pi(p)-p+1)|x|$ & $\xleftrightarrow{r}$ & $(p\Pi(p)-p+1)|x|$ & $(p\Pi(p)-2p+1)|x|$\\
\hline
$2 \langle z^{+2},x \rangle + \hdots$  & $(p\Pi(p)-p+2)|x|$  & $\xrightarrow{r}$ &  $(p\Pi(p)-2p+2)|x|$& $(p\Pi(p)-2p+2)|x|$\\
\hline
$\vdots$ &$\vdots$ & $\hdots$ & $\vdots$ & $\vdots$ \\
\hline $2 \langle z^{+(p-1)},x \rangle + \hdots$ & $(p\Pi(p)-1)|x|$ & $\xleftarrow{r}$ & $(p\Pi(p)-p-1)|x|$& $(p\Pi(p)-p-1)|x|$\\
\hline
\end{tabular},
\end{center}
\end{table}
\vspace{0.2cm}

\noindent with $\Pi(p)$ representing the prime counting function, which determines the number of primes up to and including $p$. Additionally, the symbol $\xleftrightarrow{r}$ indicates that the value on the left and right repeats for all strings $z$ in between. In contrast, $\xrightarrow{r}$ signifies that the value from the preceding string is repeated for all intermediate strings, while $\xleftarrow{r}$ implies that the value from the string on the right is repeated in between.

Using these phases it is easy to account for all the terms contributions in the phase by the internal product of the input with the random string $z$ in \cref{eq:quditg2}.

\paragraph{Step 2.} Subsequently, we will consider the effect of applying the respective Fourier gates $\mathsf{F}_p$ to the qudits of the resource state. More precisely we obtain the following states, 

\begin{align}
\ket{\Psi_{2a}}=&\bigotimes_{l=1}^n F_p \left( e^{\frac{2\pi i\cdot \langle x,z \rangle}{p^2}} \frac{1}{\sqrt{p}} \left ( \ket{z} + \sum_{i=1}^{p-1} e^{\frac{2\pi i\cdot \left (i|x|/p + (\sum_{j \in {1,\hdots,i}}  \langle x,(z^{+j})^{p-1} \rangle \right ) }{p}}\ket{z^{+i}} \right ) \right) \\
= &e^{\frac{2\pi i\cdot \langle x,z \rangle}{p^2}}  \frac{1}{\sqrt{p}} \left(\sum_{y\in \mathbb{F}_p^n }  e^{\frac{2\pi i\langle y,z\rangle }{p}}\ket{y}+ \sum_{i=1}^{p-1} e^{\frac{2\pi i\cdot \left (i|x|/p + (\sum_{j \in {1,\hdots,i}}  \langle x,(z^{+j}\right)^{p-1}  \rangle) }{p}+\frac{2\pi i \langle y,z^{+j}\rangle}{p}}\ket{y}\right) \\  
=& e^{\frac{2\pi i\cdot \langle x,z \rangle}{p^2}+\frac{2\pi i\langle y,z \rangle }{p}} \frac{1}{\sqrt{p}} \left(\sum_{y\in \mathbb{F}_p^n }  e^{\frac{2\pi i\langle y,z \rangle }{p}}\ket{y}+ \sum_{i=1}^{p-1} e^{\frac{2\pi i\cdot \left (i|x|/p+|y| + (\sum_{j \in {1,\hdots,i}} \langle x, (z^{+j})^{p-1 }\rangle \right ) }{p}}\ket{y}\right) .\label{eq:final_phases_dit}
\end{align}

The resulting state from step 2a will then be measured. Once again, if all the inner products of the type $\langle x,(z^{+i})^{p-1} \rangle$ are equal to zero, the final string $y$ measured will fulfill $|x|/p+|y|\ \MOD\ p =0$. Therefore, it would produce the correct outcome for each of the inverse strict modular relational problems.

Regrettably, the values of these terms vary across the field $\mathbb{F}_p$. However, they assume uniformly random values within $\mathbb{F}_p$, facilitating the prediction of precise outcomes. More precisely, this allows us to consider the support of the outcome strings and their respective probabilities based on these values and the input strings.

In particular, for this, we will consider all possible vectors such as $(0,0,\ldots,0)$ and $(0,1,2,\ldots,p-1)$ to $(p-1,0,1,\ldots,p-2)$ with a simple inline shift, which represents the values of the terms $\langle x,(z^{+i})^{p-1} \rangle$ that induce a shift in the correct outcome by an increment of $0$ and $1$ to $p-1$, respectively.  Following this representation, we can assert that it is possible to decompose any of the possible values that these terms might adopt into vectors of the form $(a_1, a_2, \ldots, a_{p-1})$. The degree of overlap between these vectors and the aforementioned basis vectors directly determines the probability of measuring a string offset related to the corresponding basis vector. Given this description and the uniform distribution of all $a_i$ within $\mathbb{F}_p$, the probability of yielding each shift relative to the accurate outcome string is uniform. This ensures that a correct string is produced with a probability of $1/p$ and every incorrect string equally with a probability of $1/p$.

\paragraph{Step 3.} This process now aims to use the information provided by the state's creation, which defines the random string $z$, to increase the probability of accurately computing the solutions to inverse strict modular relational problems.

To achieve this, we will begin by analyzing the vectorization that determines the probability of obtaining the outcome with a shift in its $\ell_1$-norm originating from the randomness in the creation of the generalized poor man's cat state. More specifically, we consider the use of a single term representative of one of the inner product terms $\langle x, (z^{+i})^{p-1} \rangle$. Thus, by selecting the first $a_1$ from $(a_1, a_2, \ldots, a_n)$ and then adding the inverse shift associated with $(a_1, a_1+1, a_1+2, \ldots, a_1+(p-1))$, we ensure that the outcome string is corrected for this component. In particular, we find that the outcome strings are shifted based on the inner product with a vector of the type $(0, a_2*, a_3*, \ldots, a_{p-1}*)$ and all the basis vectors, with all the values $a_2*, a_3*, \ldots, a_{p-1}*$ remaining uniformly random over $\mathbb{F}_p$.

This ensures that the overlap with a zero shift is $\frac{2p-2}{p^2}$, and all the other shifts from $1$ to $p-1$ have probability $\frac{p-1}{p^2}$. Also, the average overlap with the correct outcomes within the Abelian domains can be determined as being 
\begin{align}
\mathsf{Corr}_{\mathcal{D}_p} \left(C,\mathcal{R}_p^m \right)&=\frac{p}{2^{n}} \sum_{{\substack{x \sim \mathbb{F}_2^n \\  |x|\ \MOD\ p= 0}}}  \mathsf{Re}\left(e^{-i\frac{2 \pi |C_q(x)|}{p}} e^{i\frac{2 \pi |\mathcal{R}_p^m(x)|}{p}}\right)\\
&= \frac{p-1}{p^2},
\end{align}
\noindent with $\mathcal{D}_p$ being the uniform distribution over string in $\mathbb{F}_2^n$ that satisfy $|x|\ \MOD\ p= 0$.

This value can be easily obtained by dividing the probabilities of having a shift as follows: with a probability of $\frac{p-1}{p^2}$, there is no shift at all, and with a probability of $1-\frac{p-1}{p^2}$, there is a uniform distribution of shifts. Therefore, the second fraction of shifts does not contribute to the value, as this represents a sum over all the values for the roots of unity with equal probability and equals zero. Simultaneously, a $\frac{p-1}{p^2}$ fraction of the inputs contributes to the value $1$. 

\paragraph{Correction string.} For the correction, we only need to guarantee that the term $\langle x,(z^{+1})^{p-1} \rangle$ is computationally possible to produce within the considered class such that we can sum it to the outcome string $y$. To achieve this, we consider that we can determine this value using the following expression,
\begin{equation}
\langle x,(z^{+1})^{p-1} \rangle= \sum_{i\in [n]} x_i \cdot \left(\sum_{e_{i,j} \in Path(i,c)} (-1)^{|Path(i,j)|+|Path(i,c)|} \ e_{i,j} + z_c +1\right)^{p-1}\ \MOD\ p.
\end{equation}
The inner sum $\sum_{e_{i,j} \in \text{Path}(i,c)} (-1)^{|\text{Path}(i,j)|+|\text{Path}(i,c)|} \ e_{i,j}$ can be constructed by enumerating the edge measurement results along the respective paths. Additionally, the sum of the values $z$ and $1$ can be straightforwardly achieved by appending these values to the string. To exponentiate, it is sufficient to create all the terms that result from this sequence of products, and since the exponent is finite, each term is a product involving at most $p-1$ terms, making them all efficiently computable. These terms are then multiplied in $\mathbb{F}_p$ with the respective input bit $x_i$. The total number of sums is $n$ while using the largest path $g$ we have that the exponentiation creates $g^{p-1}$ terms at most and we end up with $ng^{p-1}$ size correction string.

Similar to the qutrit case, one can use a $\QNC^0$ decoder to map the measured dits to bits based on an inverse of the Hamming encoding. Therefore, all the previously described operations over $\mathbb{F}_p$ can be performed over $\mathbb{F}_2$ using the respective logical operations with an $\NC^0$ circuit. In conclusion, the string $\langle x, (z^{+1})^{p-1} \rangle$ in bits is concatenated with the measurement outcomes from stage 3, thereby concluding the proof.
\end{proof}

\subsubsection*{Quantum vs. Classical circuit separations}

We will now combine the lower bounds determined for $\NC^0$ and $\BTC^0(k)$ circuits in \cref{subsec:qudit_lower} with the upper bounds established for $\QNC^0$ circuits in this section to derive new quantum-classical separations. We begin by stating explicit separations of qupit $\QNC^0$ circuits against $\NC^0$ circuits.

\begin{theorem}\label{thm:lower_qudit_average}
For every \textnormal{ISMR} problem defined for a prime dimension $p$, denoted by $\mathcal{R}_{p}^m$, considering a uniform distribution $\mathcal{D}_p$ over $\mathbb{F}_2^n$ of strings with Hamming weight that satisfies $\left(\sum_{i=1}^n x_i\right)\ \MOD\ p = 0$ as the input distribution, there exists a $\QNC^0$ circuit with all-to-all connectivity that solves the problem for $m = n \cdot \log n^{p-1}$ archiving a correlation with the correct outcome of $\frac{p-1}{p^2}$. In contrast, any $\NC^0/\mathsf{rpoly}$ circuit with locality $l$ fails to solve this problem with a correlation larger than, 
\begin{equation}
\mathsf{Corr}_{\mathcal{D}_p} \left( C, \mathcal{R}_p^m \right) =    \mathcal{O}\left ( \mathsf{Re} \left (\frac{2^{1 - p} e^{\frac{2\pi i}{p^2}} \left(e^{-\frac{2\pi i}{p^2}} \left(1 + e^{\frac{2\pi i}{p^2}}\right)\right)^p}{1 + e^{\frac{2\pi i}{p^2}}} 
 \right )^{\frac{n}{\log n(p-1)^4l^2}}  \right ).
\end{equation}
In particular, for $p=3$ there exist a $\QNC^0$ circuit with 3D connectivity that solves this specific \textnormal{ISMRP} instance for $m=n^{4/3}$ with an efficiency larger the $1/2$ on all inputs, while the success probability for any $\NC^0/\mathsf{rpoly}$ is bound by 
\begin{equation}
\Pr_{{\substack{x \sim \mathbb{F}_2^n \\  |x|\ \MOD\ 3= 0}}}[C(x) \in R_3^m(x)] \leq\frac{1}{3}+\left (\frac{9}{10} \right)^{\frac{n^{2/3}}{16l^2}}.
\end{equation}  
\end{theorem}
\begin{proof}
We consider the upper bound determined in \cref{lemma:computcorre} to ascertain the correlation with which one can compute each of the ISMR problems using a $\QNC^0$ circuit with all-to-all connectivity. Concurrently, we refer to \cref{lemma:uniform_qupit_clower} for the classical lower bound of the same quantity, utilizing the values for $m = n \cdot \log n^{p-1}$ obtained by the quantum solution and the considered input distribution.

For the specific case where $p = 3$, we examine the lower bound on the success probability determined in \cref{lemma:comput_trit} for a qutrit $\QNC^0$ circuit with 3D connectivity, alongside the classical lower bound in \cref{lemma:uniform_qupit_clower} for the respective input distribution and output size of $m=n^{4/3}$.  
\end{proof}

The previous separation can be extended using the techniques from the works of \cite{le2019average,coudron2021} for parallel repetitions, such that the success probability decreases exponentially fast to zero. This results in an exponential separation between the quantum and classical efficiencies, specifically for $\QNC^0$ and $\NC^0$ circuits, across the entire class of inverted strict modular relational problems. Additionally, due to our findings in \cref{lemma:Trit_NC_correlation_bound}, it is possible to accelerate this convergence. A different input distribution allows for larger worst-case separations. We do not explicitly determine these results, as they are not the main focus of this document, and we intend to further explore the separation against the $\BTC^0(k)$ class. 

Finally, we combine the lower bounds for $\BTC^0(k)$ and quantum upper bounds for the same problems. We determine that one can achieve separation between $\QNC^0$ and $\BTC^0(k)$ using any prime instance of the inverted strict modular relational problems with the asymptotically largest possible value of $k$.

\paragraph{Proof of \cref{thm:qudit_average}.} The quantum upper bound is derived by combining \cref{poorconst} and \cref{lemma:computcorre}, which determines the precise correlation $\frac{p-1}{p^2}$. The output size $m$, which allows us to solve each of the problems using a $\QNC^0$ circuit over qupits, follows from considering the all-to-all connectivity. Specifically, it is defined that the minimum-maximum path $g$ in the graph will have a size of $\log n$, which in turn defines $m$ to be $n \cdot \log n^{p-1}$.

Simultaneously, the classical lower bounds are derived by considering \cref{lemma:lowerqudit}, with all parameters being defined by the quantum solution. Furthermore, all additional values for the expressions, as well as the optimal values of $q$, are determined in the same manner as in \cref{EsepQNC}.\qed
\vspace{0.2cm}

Finally, as one can determine success probabilities for the ISMRP  $\mathcal{R}_p^m$ where $p=3$, just as one can for the binary case where $p=2$, we can establish a separation in the success probabilities for this problem. Additionally, this case is particularly intriguing as it demonstrates that all instances of the $\mathcal{R}_3^m$ problem can be solved with a probability strictly greater than $1/2$ in the quantum case, whereas any classical circuit solves the problem with at most a probability asymptotically close to random guessing, which for this problem corresponds to a success probability of $1/3$. This indicates that $\QNC^0$ effectively addresses the probabilistic version of $\mathcal{R}_3^m$, adhering to the standard definition for bounded-error probabilistic problems, while any classical circuit within the $\BTC^0(k)$ class fails to solve the same problem. Furthermore, the quantum circuit can exhibit this advantage while maintaining geometric locality in a 3D geometry.

\paragraph{Proof of \cref{cor:qutrit_average_sep}.}
We derive the lower bounds for the probabilities of the quantum solution using \cref{lemma:comput_trit} and the parameters $m=n^{4/3}$ from the specific 3D connectivity.

Then, we consider the lower bounds for the probability of solving the $\mathcal{R}_3^{n^{4/3}}$ defined by \cref{lemma:entire_breakp}. Finally, the determination of the values of $q$ for both values of $k$ follows equally as in \cref{EsepQNC}. \qed
\vspace{0.2cm}

In conclusion, all the newly derived bounds, including those applicable to the qubit scenario, could potentially be extended with parallel repetition games, thereby enhancing the quantum advantage even further. We have not yet pursued this path, as our focus has been on expanding the problem set and tightening the bounds. Our priority is to establish unconditional separations against broader and more powerful classes of classical circuits, aiming to reduce the resources needed for a quantum advantage. This is because parallel repetition games only become valuable once the initial problem set, which is to be repeated, has achieved a quantum advantage.

\section{Noise-resilient quantum advantage}
\label{sec:noisy-advantages}

In this section, we demonstrate that the separations we proved against $\BTC^0(k)$ in the preceding section can be made noise-robust. Specifically, we show that a separation can still be achieved even when our quantum circuits are noisy and classical circuits are noiseless. As before, we present the qubit and qupit cases separately. 

Both separations are proven for new noise-tolerant relations $\mathfrak{R}_p$, defined based on the ISMR problems as follows. 
\begin{definition}\label{qnoise:relation} 
Let $\mathcal{R}_{p}:\mathbb{F}_p^n\mapsto\mathbb{F}_p^{n'}$ be the original $\mathsf{ISMR}$ problem (\cref{insec:defmod}). Then, the noise tolerant extension $\mathfrak{R}_p: \mathbb{F}_p^n \times \mathbb{F}_p^{n'\cdot m}\mapsto \mathbb{F}_p$ is defined as follows, 
\begin{equation}
\mathfrak{R}_p(x,y) =  \left (\left |\mathcal{R}^{n'}_{p}(x)\right |- \Big| \Dec^{*}(y)\Big| \right ).
\end{equation}
Here $\Dec^{*}$ is an arbitrary functions of the form $\Dec^{*}:\mathbb{F}_p^{n'\cdot m}\mapsto  \mathbb{F}_p^{n'}$ which must be computable by an $\AC^0$ circuit.
\end{definition}

We begin by presenting our noise-robust separation for qupits for arbitrary prime $p$, which is the more intricate case. Notably, for qupits with $p\geq 3$, the quantum circuits solving the candidate relational problems are not Clifford circuits. This necessity arises because Bell violations with stabilizer states alone do not generalize from qubits to qupits with $p\geq 3$ \cite{gross2006hudson,howard2013quantum,meyer2024bell}. This indicates that quantum strategies can not have larger winning probabilities than classical ones if only Clifford circuits are considered for general qupits. Therefore, for this separation, we require quantum advice states, specifically magic states in their logical form, with the same code distance under which the remaining quantum circuit will operate. This requirement stems from the additional conjecture that any magic state factory that creates a logical magic state with at least $\poly(\log n)$ code distance is not realizable within $\QNC^0$. Simultaneously, we must resolve the incompatibilities between the presented quantum circuits, solving the ISMR problems from \cref{subsec:qudit}, and the quantum circuit architecture for non-Clifford circuits that we describe and prove to be noise-resilient in \cref{subsec:cond,subsec:dec,subsec:single_shot}. Solving the previous incompatibility involves demonstrating that quantum circuits equipped with non-adaptive magic state injection gadgets can solve all the ISMR problems with equal effectiveness, as shown in \cref{subsec:inject}.

\begin{theorem}\label{thm:noisy}
Let $x$ be an input drawn uniformly at random from the subset of $n$-bit strings with Hamming weight satisfying $\left(\sum_{i=1}^n x_i\right)\ \MOD\ p = 0$. Consider the local stochastic noise $\mathcal{E}\sim \mathcal{N}(\tau)$ with bounded probability $\tau< \tau_{th}$. Then, there is a $\QNC^0/\ket{\overline{T^{1/p}}}$ circuit that solves the relation $\mathfrak{R}_p(x,y)$ for output strings $y=o(n^2)$ with a constant positive correlation. 

Further, any circuit $C \in \BTC^0(k)/\mathsf{rpoly}$ with fixed depth $d$ and size $s$, for large enough $d'\in \mathbb{N}^+$ has exponentially small correlation with $\mathfrak{R}_p$ bounded as in \cref{tab_res5}.
\begin{center}
\begin{table}[!hbtp]
\vspace{-0.2cm}
\renewcommand{\arraystretch}{2.5}
\begin{tabular}{|c|c|c|}
\hline
 & $k=\mathcal{O}(1)$ $(\equiv \AC^0/\mathsf{rpoly})$ & $k=n^{1/(5d)}$ \\  
\hline
$\mathsf{Corr}_{\mathcal{D}_p} \left(\mathfrak{R}_p(x,C(x))\right)$ & $\exp\left(- \Omega\left(\frac{n^{1 - \mathcal{O}(1)}}{  (\log{n})^{p-1} (\log{s})^{2(d+d')-2}}
\right)\right)$ & $\exp\left(-\Omega \left(\frac{ n^{3/5 - \mathcal{O}(1)}}{(\log{n})^{ p-1} (\log{s})^{2(d+d')-2}}
\right)\right)$ \\
\hline
\end{tabular}.
\caption{\justifying Correlation upper bounds for the circuit classes $\BTC^0(k)$ for different values of $k$ in solving the $\mathfrak{R}_p$ problems.}
\label{tab_res5}
\end{table}
\end{center}
\end{theorem}
\vspace{-0.6cm}

We also obtain a separation for the qubit case. Later in this section, we highlight the distinctions between the two cases. However, the main difference is that the qubit case allows for an advice-free separation between the noisy $\QNC^0$ and noiseless $\BTC^0(k)$ circuits. Therefore, we will with our qubit separation lift the previously proven noise-resilient separations from $\NC^0$ and $\AC^0$ to $\BTC^0(k)$ for optimal $k$ \cite{bravyi2020quantum,caha2023colossal}.

\begin{corollary}\label{cor:Thm_noisy_qubit}
Let $x$ be an input drawn uniformly at random from the set of binary even strings. Consider the local stochastic noise $\mathcal{E}\sim \mathcal{N}(\tau)$ with bounded probability $\tau< \tau_{th}$. Then, there is a $\QNC^0$ circuit that solves the relation $\mathfrak{R}_2(x,y)$ for output string $y=o(n^2)$ with a constant probability bounded away from 1/2. However, any circuit $C \in \BTC^0(k)/\mathsf{rpoly}$ with fixed depth $d$ and size $s$, for parameter $k=1/(5d)$ large enough $d'\in \mathbb{N}^+$ has exponentially small correlation with $\mathfrak{R}_2$ bounded by
\begin{equation}
\Pr[\mathfrak{R}_p(x,C(x))=0]=\frac{1}{2}+\exp\left (-\Omega \left(\frac{n^{3/5 - \mathcal{O}(1)} }{(\log{s})^{2(d+d')-1}}\right) \right ).
\end{equation}
\end{corollary}

The correlation bounds previously described differ from the noise-free bounds in \cref{thm:qudit_average}, particularly in the exponent of the $\log(s)$ term, which arises from the additional overhead for error correction. However, this does not compromise the exponential deviation between the classical and quantum correlations achieved by the respective circuit classes.

To demonstrate these separations with higher-dimensional systems, we show that the roadmap established in \cite{bravyi2020quantum} for the qubit case can be adapted to extend this advantage in the presence of noise. We will proceed in several steps. First, we outline the requirements for the fault-tolerant construction in constant depth for the higher-dimensional versions of the code, including the additional state injection. Next, we introduce a new decoder for the surface code over qudits, highlighting the differences between the two cases. Then, we generalize the single-shot state preparation for the two-qudit states known as $\mathsf{GHZ}_2$ (generalized qudit Bell states). Finally, we identify new non-adaptive qupit circuits based on the magic state injection gadget, solving the new noise-resilient relational problems to establish our bounds against $\BTC^0(k)$.

\subsection{Noise-resilient qupit Clifford circuits with quantum advice}\label{subsec:cond}

To demonstrate our noise-robust separations, we introduce constant-depth quantum circuits that incorporate logical quantum states as advice, as illustrated in \cref{fig:advice_noise_resilient}. This model adopts the structure used in standard error-corrected quantum circuits with magic state injection, presenting broader interest. It is important to note, however, that while our circuit architecture is noise-robust, it does not support the noise-resilient realization of adaptive quantum circuits that include gates controlled by prior measurement outcomes. Instead, it is limited to noise-resilient, non-adaptive constant-depth Clifford circuits with magic state injection. Thus, we restrict our analysis to a subset of $\QNC^0$ circuits for a noise-resilient version of these circuits.

The difficulty in realizing adaptive measurements in a noise-resilient manner arises because the code distance necessary for good error correction properties unfortunately also requires the execution of the decoding function, which is beyond the computational capabilities of $\QNC^0$ circuits with the considered codes and decoder. This creates a conundrum where either the noise levels render the outcome useless or the information is recoverable but not by the circuit class $\QNC^0$ itself. Therefore, we focus on non-adaptive constant-depth Clifford circuits with magic state injection, which we will show in \cref{subsec:inject} are sufficient to obtain noise-robustness for our qupit separations.
\begin{figure}[htbp]
\centering
\scalebox{0.75}{
\begin{quantikz}[classical gap=0.03cm]
\lstick{$x\in\{0,1\}^n$}  & \cw \wireoverride{n} & \cw \wireoverride{n} & \ctrl[vertical wire=c]{1} \cw \wireoverride{n}\\
\lstick{$\ket{0}^{m}$} & \gate[2,disable auto height]{\begin{tabular}{c}$\mathcal{Z}$ \\  \end{tabular}}  & & \gate[9,disable auto height]{\begin{tabular}{c}$\overline{\mathsf{Cliff}}$ \\  \end{tabular}} & \gate[9,disable auto height, ps=meter ]{}  &  \wireoverride{n}\cw \\
\lstick{$\ket{0}^{m_{a}}$} &  & \meter{s_1} & \wireoverride{n}  &\wireoverride{n}  & \wireoverride{n}  \\
\lstick{\vdots} & \wireoverride{n} & \wireoverride{n} & \wireoverride{n}\vdots   &   & \wireoverride{n} \cw\\
\lstick{$\ket{0}^{m}$} & \gate[2,disable auto height]{\begin{tabular}{c}$\mathcal{Z}$ \\  \end{tabular}}  & &   & &  \wireoverride{n} \cw\\
\lstick{$\ket{0}^{m_{a}}$} &  & \meter{s_n} & \wireoverride{n} & \wireoverride{n}  & \wireoverride{n}  \\
\gategroup[4,steps=2,style={dashed,rounded
corners,fill=blue!20, inner
xsep=12pt,xshift=-0.4cm},background,label style={label
position=below,anchor=north,yshift=-0.2cm}]{{Logical Advice state $\ket{\overline{A}}$}}\lstick{$\ket{0}^m$} & \gate[4,disable auto height]{\begin{tabular}{c}$\mathsf{Prep}$ \\  \end{tabular}} & \wire[l][1]["\ \ \ \ket{\overline{A_1}}"{above,pos=0.2}]{a} & &  & \wireoverride{n}\cw\\
\lstick{\vdots} & \wireoverride{n}  & \wireoverride{n} \vdots & \wireoverride{n} \\ &\wireoverride{n} &\wireoverride{n}  &\wireoverride{n} &\wireoverride{n}\\
\lstick{$\ket{0}^m$} & & \wire[l][1]["\ \ \ \ket{\overline{A_l}}"{above,pos=0.2}]{a} & & & \wireoverride{n}\cw
\end{quantikz}}
\caption{Constant-depth non-adaptive Clifford quantum circuit with quantum advice states.}
\label{fig:advice_noise_resilient}
\end{figure}
In particular, we will consider the conditions a quantum error correcting code must satisfy, so that we may use it to obtain a noise-resilient version of circuits with this structure. We will prove that these conditions are sufficient, and subsequently demonstrate that the qupit surface code, along with the hard normalization decoder, fulfills all the requirements for a noise-resilient execution of a non-adaptive Clifford circuit with advice states.

\begin{definition}[noise-resilient constant-depth quantum code conditions]\label{cond:Ecode} 
A CSS-type code $Q_m$ exhibiting the following properties gives constant depth fault tolerant circuit constructions for qudit local stochastic noise $\mathcal{E}\sim \mathcal{N}(\varrho)$.

\begin{enumerate}
\item Any qudit advice state $\ket{A}$ required must be provided in its logical form $\ket{\overline{\mathcal{A}}}$, while being affected by local stochastic noise $\mathcal{E}_A\sim \mathcal{N}(\tau)$ at most as follows, 
\begin{equation}
    \ket{\overline{\mathcal{A}}}\propto  \mathcal{E}_A \ \Prep(\ket{A})\ket{0}^{\otimes(|A|\cdot m)}.
\end{equation}
\noindent where $\textsc{prep}(\mathcal{\ket{A}})$ represents the noise-free quantum circuit that prepares the logical state $\ket{\overline{\mathcal{A}}}$ over a code of distance $m$.

\item The logical qudit Clifford gates can be implemented in constant depth (as defined in \cref{qudit_operations}).

\item There exist recovery and repair functions 
\begin{align} 
\Rec&: \mathbb{F}_p^{m_a} \mapsto   \mathsf{Pauli}(m) \nonumber \\  
\Rep&: \mathsf{Pauli}(m+m_{\text{a}}) \mapsto  \mathsf{Pauli}(m) \nonumber 
\end{align}
\noindent respectively, where $\mathsf{Pauli}(n)$ respresents a Pauli operator over $n$ qupits, such that the application of a constant depth Clifford circuit $\mathcal{Z}$  on the state $\ket{0}^{\otimes m}\otimes \ket{0}^{m_{\text{a}}}$ followed by the $\mathsf{Z}$-measurement of its ancillary qudits, leading to an outcome $s \in \mathbb{F}_p^{m_a}$, can be represented by the channel  
\begin{equation}\left(\ketbra{s}{s}\otimes \Rec(s)\right) \mathcal{Z} \left(\ket{0}^{\otimes m} \otimes \ket{0}^{\otimes m_{\text{a}}}\right) \propto \ket{\bar{0}}\otimes \ket{s},
\end{equation}
\noindent where $\ket{\overline{0}}$ represents the logical zero state. In the case of local stochastic noise $\mathcal{E}\sim \mathcal{N}(\varrho)$, this becomes
\begin{equation}\left(\ketbra{s}{s}\otimes \Rec(s)\right)\mathcal{E}\ \mathcal{Z} \left(\ket{0}^{\otimes m} \otimes \ket{0}^{\otimes m_{\text{a}}}\right) \propto \Rep(\mathcal{E})\ket{\bar{0}}\otimes \ket{s},
\end{equation}
\noindent for a local stochastic function $\Rep(\mathcal{E})\sim \mathcal{N}((c'\varrho)^{c''})$.

\item There exists a decoding function $\Dec: \mathbb{F}_p^{m} \mapsto  \mathbb{F}_p$ such that 
\begin{equation}
    \textsc{dec}(y) \coloneqq \mathsf{Mod}_p(y), 
\end{equation}
\noindent for $y\in \mathbb{F}_p^{m}$ such that $y$ belongs to the space of all basis vector obeying $Z-$type stabilizers. When subject to noise, for error rate $q< q_{th}$  i.e. a threshold parameter, and error string $v\in \mathbb{F}_p^{m}$ it is the case that

\begin{equation} 
    \Pr[\Dec(x \oplus v) = \mathsf{Mod}_p(x)]\geq 1 -e^{(c'm^c)},
\end{equation} 

\noindent for some $c$ and $c' \in \mathbb{R}_{>0}$.
\end{enumerate}
\end{definition}

We have that any error correction code that fulfills the previous condition, will allow us to realize error resilient version of the any logical quantum circuit with the structure of \cref{fig:advice_noise_resilient}.

\begin{lemma}\label{thrm:cnst_noise_res}
Let $\mathcal{EC}_{C}$ represent the error-corrected version of $C$, a logical classically controlled constant-depth Clifford circuit with advice states $\ket{A}$. Respecting the structure depicted in \cref{fig:advice_noise_resilient} and utilizes a quantum error correction code featuring a decoding function $\Dec$, single-shot state preparation circuit $\mathcal{Z}$, a recovery function $\Rec$, and a repair function $\Rep$, all of which meet the conditions specified in \cref{cond:Ecode}. Then, for an arbitrary input $x\in \mathbb{F}_2^{n}$ and denote the outputs of the noisy implementation $s$ and $y$ such that $\mathcal{EC}_{C}(x)=(s,y)$, with local stochastic noise $\mathcal{E}\sim\mathcal{N}(\varrho)$ affecting the circuit execution and $\mathcal{E}_{A}\sim \mathcal{N}(\tau)$ the advice state. Then, for $m=\mathcal{O}(\poly\log n)$ and all local stochastic noise with $\varrho_{th}=2^{-2^{\mathcal{O}(d)}}$, where $d$ is a constant representing the depth of the circuit, the following holds for all $\tau$ and $\varrho<\varrho_{th}$,
\begin{equation}
    \Pr[f(s,y)\equiv C(x)]>0.99,
\end{equation}
\noindent with $f(s,y)=\otimes_{i=0}^{n-1} \Dec \left(y_i \oplus w_i \right)$ for $w$ such that $\bra{w} \overline{C}(x)\Rec(s)^\dagger\overline{C}(x)^\dagger\ket{0}^{\otimes n\cdot m}=1$. 
\end{lemma}
\begin{proof}
We start by using the first condition of \cref{cond:Ecode}, which guarantees that we can describe the advice states as $\ket{\overline{\mathcal{A}}}\propto  \mathcal{E}_A \ \Prep(A)\ket{0}^{\otimes |A|\cdot m}$, with $\mathcal{E}_A\sim \mathcal{N}(\tau)$. Additionally, given that the remainder of the logical circuit is a classically controlled Clifford circuit of constant depth $C(x)$, we have, according to condition 2 of \cref{cond:Ecode}, that its logical version also is a Clifford circuit of constant depth. Therefore, all the Pauli errors that can occur due to the local stochastic noise under consideration can be commuted to the end using \cref{lemma:commute}. Thus, the distribution $\mathcal{D}_{\mathcal{EC}_C}(x)$ over strings that result from the $\mathcal{EC}_C(x)$ circuit can be described as follows, 
\begin{equation}
    \mathcal{D}_{\mathcal{EC}_C}(x)= \bra{s,y} \mathcal{E}_t\ \overline{C}(x)(\mathcal{Z}\otimes I)\ket{0}^{\otimes n\cdot m}\ket{0}^{\otimes n\cdot m_{anc}}\ket{\overline{A}}.
\end{equation}
\noindent with $\mathcal{E}_t=\mathcal{N}(\tau')$ for $\tau'=\mathcal{O}(\max\left(\tau^{-2^d},\varrho^{-2} \right))$ and $\varrho$ being the noise level of the local stochastic noise affecting the state preparation circuit $\mathcal{Z}$ as well as $C(x)$ and the measurements.

Now, we need to describe the effect of the measurements in the single-shot state preparation to define the support of the circuit's outcome more precisely. For that, we commute the total error back through the Clifford circuit $\overline{C}(x)$, which we describe as $\mathcal{E}_t'=\overline{C}(x) \mathcal{E}_t \overline{C}(x)^\dagger$ and divide $\mathcal{E}_t'=\mathcal{E}_{t_1}'\otimes \mathcal{E}_{t_2}'$ with $\mathcal{E}_{t_1}'$ operating over the first $n\cdot(m+m_{anc})$ qudits, while $\mathcal{E}_{t_2}'$ operates over the last $|A|$ qudits of the advice state. Thus, we obtain that for each logical qupit that is affected by $\mathcal{E}_{t_1}'$ the following equality based on condition 3 of \cref{cond:Ecode},
\begin{align}
 (I \otimes \ket{s}\bra{s})\overline{C}(x)\ (\mathcal{E}_{t_1}'\otimes \mathcal{E}_{t_2}')&(\mathcal{Z}\otimes I)\ket{0}^{\otimes n \cdot m}\ket{0}^{\otimes  n \cdot m_{anc}}\ket{\overline{A}}\\
&\propto  \overline{C}(x)(I\otimes \mathcal{E}_{t_2}')\Rec(s)^\dagger\Rep(\mathcal{E}_{t_1}')\ket{\overline{0}}^{\otimes n}\ket{s}\ket{\overline{A}}\\
 &\propto  \overline{C}(x)(\mathcal{E}_{t_1}''\otimes \mathcal{E}_{t_2}')\Rec(s)^\dagger\ket{\overline{0}}^{\otimes n}\ket{s}\ket{\overline{A}}\\
 &\propto  \mathcal{E}_t''' \overline{C}(x) \Rec(s)^\dagger\ket{\overline{0}}^{\otimes n}\ket{s}\ket{\overline{A}},
\end{align}
\noindent with $\mathcal{E}_{t_1}''=\Rec(s)^\dagger \Rep(\mathcal{E}_{t_1}') \Rec(s)$, and consequently using the properties of the local stochastic noise we obtain that $\mathcal{E}_{t}'''\sim \mathcal{N}(\tau'')$ with $\tau''=\mathcal{O}\left (\max\left((c'\tau^{-2^d\cdot c''})^{-2^{d}},(c'\varrho^{-2c''})^{-2^d}\right)\right )$.

Therefore, we now are able to rewrite the distribution $\mathcal{D}_{\mathcal{EC}_U}(x)$ as follows, 
\begin{equation}
    \mathcal{D}_{\mathcal{EC}_C}(x)= \bra{s,y} \mathcal{E}_t'''\ \overline{C}(x)\Rec(s)^\dagger\ket{\overline{0}}^{\otimes n}\ket{s}\ket{\overline{A}}.
\end{equation}

Subsequently, we have that $\mathcal{E}_t'''=\omega \mathsf{X}(a)\mathsf{Z}(b)$. At this point, we can use condition 4 of \cref{cond:Ecode} to define a noise threshold $q<q_{th}$ based on our local stochastic noise to prove that for $p<p_{th}$, given $p_{th}=\Omega(q_{th}^{\mathcal{O}(1)})$, we have with exponentially high confidence,
\begin{equation}
 \mathcal{D}_{\mathcal{EC}_C}(x)=\bra{s,y\oplus a}\overline{C}(x)\Rec(s)^\dagger\ket{\overline{0}}^{\otimes n}\ket{s}\ket{\overline{A}},
\end{equation}
\noindent where $a$ is the vector that defines the $\mathsf{X}(a)$ operator, derived solely from the Pauli $\mathsf{X}$ component of $\mathcal{E}_t'''$. The previous assumption on the error threshold then implies that after the error correction process we obtain with exponential high confidence that $\Dec\left(y\oplus a \right)=\Dec(y)$ such that the resulting error corrected distribution is equal to, 
\begin{align}
\mathcal{D}_{\mathcal{EC}_C}'(x)&=\bra{s,y}\overline{C}(x)\ \Rec(s)^\dagger\ket{\overline{0}}^{\otimes n}\ket{s}\ket{\overline{A}}\\
&= \bra{s,y \oplus w}\overline{C}(x)\ket{\overline{0}}^{\otimes n}\ket{s}\ket{\overline{A}}
\end{align}
\noindent for $w$ such that $\overline{C}(x) \Rec(s)^\dagger \overline{C}(x)^\dagger \ket{0}^{\otimes n\cdot m}=\ket{w}$.

From the last expression, after applying the function $f$ defined previously—which corrects possible errors introduced during single-shot state preparations and decodes the error-corrected quantum circuit—we find that the outcome of our circuit $C$ can be described by $\bra{y}C(x)\ket{0}^{\otimes n}\ket{A}$. Finally, it is important to note that the bounds allow us to achieve this with a probability close to 1.
\end{proof}

We will demonstrate that the surface code over qupits fulfills all the necessary conditions of \cref{cond:Ecode}. To begin, we consider the first two conditions. The first one follows trivially, as it is possible to create any quantum state $\ket{A}$ in its logical version within the surface code, while the local stochastic noise can be seen as a layer of errors occurring after the noise-free state preparation. For condition 2, we need to demonstrate that the Clifford operators are realizable in constant depth. Fortunately, this was shown previously for both the qubit and qudit cases in the work of \cite{Moussa_2016}.

\begin{lemma}[\cite{Moussa_2016}] 
Any qudit Clifford operators $C$ of prime dimension can be implemented in constant depth with the surface code.
\end{lemma}

We turn to proving conditions 3 and 4 in the next two subsections. We first introduce the new decoder under consideration and prove that, when applied to the surface code, it fulfills conditions 3 and 4. Afterwards, we will show that the generalization of single-shot state preparation for qupits, as demonstrated for qubits in \cite{bravyi2020quantum}, holds equally well. We have changed the order of presentation for these two conditions for clarifying the details more coherently.

\subsubsection{Qupit error threshold}\label{subsec:dec}

As remarked in \cref{cond:Ecode}, we need to show that the qupit surface code, along with the selected decoder, achieves the required noise resilience properties. In particular, the higher dimensional case requires several changes, the most important of which is the decoder used. Prior work \cite{bravyi2020quantum} uses the \textit{minimum weight perfect matching} (MWPM) decoder for the qubit case. We now discuss below why the qupit case needs a different decoder and present our choice, the hard decision renormalization group decoder. 

A single physical (Pauli) error on the surface code over qubits generates two defects (i.e.\ $\mathsf{X}$ syndromes), appearing on the vertices $v_1$ and $v_2$ that share the edge $(v_1,v_2)$ where the error was detected, which form the boundary of the error. If we have another error on an adjacent edge $(v_2,v_3)$, the defect on the common vertex $v_2$ cancels out, leaving only the defects on $v_1$ and $v_3$. Hence a succession of adjacent errors generates a path on the surface code with defects on its boundary. One method to tackle such errors is to find the most likely one and correct it. An efficient algorithm for this is to find minimum (Pauli) weight operators representing the errors that connect the defects. This problem can be mapped to a weighted graph where the goal is to find a perfect matching with minimum weight. The MWPM in this case can be solved efficiently \cite{iolius2024decodingalgorithmssurfacecodes,Edmonds_1965}.

This error pattern happens because the Pauli operators in two dimensions (i.e.\ $p=2$) are self-inverse, leading to cancellations. However, this is not the case for generalized Pauli operators with $p>2$; these operators are in general neither self-inverse nor self-adjoint. Hence the intermediate defects do not cancel along a (edge-)path of errors, instead leaving scattered defects throughout the lattice. For $p>2$, the equivalent problem would be to find the perfect matching of \textit{hyperedges} in a $p$-uniform hypergraph, a well known NP-hard problem. The MWPM decoder is hence ill-suited to deal with qupits due to efficiency issues.  

\subsubsection*{Hard-Decision Renormalization Group decoder} 
A more suitable choice of decoder for our case is the use of the Hard-Decision Renormalization Group (HDRG) decoder \cite{Bravyi_2013}. This decoder operates iteratively. In each iteration or level, it considers all the defects (syndromes) in the lattice and groups them into subsets called clusters. The defects are clustered by comparing their pairwise distances with a threshold distance for each level, known as the search distance. All defects with pairwise distance below the threshold are merged into a cluster. Formally, at the $l$th level, the search distance is based on a function $D(l)$, chosen to be the Chebyshev distance $D(l) = 2^l$ in \cite{Bravyi_2013}. However, other choices and subsequent improvements have used different metrics or combinations thereof, such as the Manhattan distance and a linear search distance $D(l) = l + 1$.

In the first iteration, the decoder considers all the syndrome information and creates the first set of clusters, which then become the inputs for the second iteration. For each cluster at any level, it checks if the sum of the syndrome measurements equals zero modulo $p$, the dimension of the qupit system. If the sum is zero, the cluster is called neutral, indicating that a Pauli operator exists that can turn all the syndromes trivial (i.e., remove the defects). If the cluster is neutral, a valid correction operator is recorded, and the cluster is removed from the next iteration. If a cluster is charged (not neutral), it is added to the set of clusters to be grouped in the next iteration or level, where the search distance (defining which clusters are merged) will increase, and the process is repeated. If no cluster is left after the last iteration, which is defined by the largest possible value of the search distance that does not exceed the maximum distance between any two lattice points, the decoder returns the composition of all the Pauli operators determined for all levels so far. Otherwise, it returns failure.

Here, for simplicity, we use the original proposal for the HDRG due to \cite{Bravyi_2013}, but the results can be improved by considering subsequent refinements \cite{Anwar_2014,hutter2015improved}. We start off with the essential definitions required to describe the operational behaviour of this decoder, which will help us prove the required error threshold.

First, we define the vertices and edges in the 2D lattice as sites, and consider the following distance measure between sites $x=(x_1,x_2)$ and $y=(y_1,y_2)$: $\dist(x,y)=\max\{|x_1-y_1|,|x_2-y_2|\}$. Suppose the lattice is affected by a Pauli error $\mathcal{E}$. Then we define the diameter of a subset of $\mathcal{E}'$ errors (where the notion of subset means that $\Supp(\mathcal{E}')\subseteq \Supp(\mathcal{E})$) by the maximum pairwise distance between the sites within this subset, $\diam(\mathcal{E}')=\max\{\dist(a,b)~|~a,b\in \Supp(\mathcal{E}')\}$. To understand how the decoder merges error clusters, we define a notion of connectivity for subsets of errors. In particular, a subset of errors $\mathcal{E}'$ is $r$-connected if there does not exists subsets $A \in \Supp(\mathcal{E}')$ and $B =\Supp(\mathcal{E}'\setminus A)$ such that the distance between $a\in A$ and $b\in B$ is larger than $r$.

Next, we recursively define a notion of level-$k$ errors in terms of subsets of a Pauli error $\mathcal{E}$ that can occur. These are distinct from the level-$k$ \textit{clusters}; the clusters over which the decoder operates are formed from the defects generated by $\mathcal{E}$. Level-$k$ errors ultimately determine the syndromes and influence the behavior and efficiency of the decoder, and will therefore be the object of analysis.

\begin{definition}\label{def:Edecomp}
A subset $\mathcal{E}' \subseteq \mathcal{E}$ of an error $\mathcal{E}$ is called a level-$k$ error if it satisfies the following conditions, for a fixed integer $l>>1$.
\begin{itemize}
    \item It contains at least two disjoint level-$(k-1)$ errors.
    \item The maximum distance between single Pauli errors within $\mathcal{E}'$ is $\frac{l^k}{2}$.
    \item The level-$0$ errors are single Pauli errors.
\end{itemize}
The union of all level-$k$ errors of $\mathcal{E}$ is denoted by $\mathcal{E}_k$.
\end{definition}

With this definition of level-$k$ errors, any error $\mathcal{E}$ can be partitioned into a disjoint union of indexed subsets, where each subset consists of lattice sites in the support of $\mathcal{E}$ that appear only in errors of a level equal to the index of that subset. This implies that all subsets  are mutually exclusive, denoted as $F_i = \mathcal{E}_i\backslash \mathcal{E}_{i+1}$. We refer to this as the decomposition of the error $\mathcal{E}$ and write
\begin{equation}
    \mathcal{E} = \bigsqcup^{h}_{i=0}F_i,
\end{equation} 
where the length of the decomposition $h$ is closely related to the structure of the error  $\mathcal{E}$. This decomposition can uniquely determine whether the decoder will succeed or fail. In particular, a decoder that runs for less than $h$ iterations will abort.

\begin{lemma}[\cite{Bravyi_2013}]\label{hdrg:correc}
Let $l \geq 10$ and $m$ be the lattice size. For any Pauli error $\mathcal{E}$ with a decomposition of length $h$ that satisfies $l^{h+1}<m$, $\mathcal{E}$ is corrected by the \textnormal{HDRG} decoder.
\end{lemma}

It is also of interest to understand the sizes of and distances between the errors in the last iteration of the decoder, which precisely determine the conditions under which the decoder fails. The following lemma addresses this issue.

\begin{lemma}[\cite{Bravyi_2013}]\label{lemma:enclosing}
Let $l\geq 6$ and $\mathcal{E}'$ be $l^k$-connected subset of $F_k$. Then $\mathcal{E}'$ has diameter $\leq l^k$ and is separated from the set $\mathcal{E}_k\setminus \mathcal{E}'$ by distance $>\frac{1}{3}l^{k+1}$.
\end{lemma}

Subsequently, considering the local stochastic nature of the errors, we must determine the threshold below which an error occurring according to our noise model has a sufficiently small decomposition with very high probability. Simultaneously, we obtain the error threshold below which the decoder operates with exponential accuracy. The proof follows as in \cite{Bravyi_2013}, with some additional considerations for the local stochastic noise model that we use in this work.

\begin{lemma}\label{error:thr1}
There exists a constant threshold $\tau_{th} > 0$ such that for any $\tau < \tau_{th}$ the \textnormal{HDRG} decoder over distance-$m$ qupit surface code corrects local stochastic errors $\mathcal{E}\sim \mathcal{N}(\tau)$ with failure probability bounded by
\begin{equation}
\mathrm{Pr}[\mathrm{Fail}] \leq \exp\left(-\Omega({m^{\eta}})\right),  
\end{equation}
\noindent for some constant $\eta > 0$.
\end{lemma}
\begin{proof}
The HDRG decoder continues to iterate as long as the area over which it attempts to merge clusters is still within the 2D lattice of the surface code. Therefore, the decoder fails if the decomposition of the error is longer than $h \approx \log(m)/\log(l)$ as given by \cref{hdrg:correc}. We need to bound the probability $p$ that an error $\mathcal{E}\sim \mathcal{N}(\tau)$ has a decomposition longer than $h$. This event implies the existence of at least one $h$-level error, which causes the decoder to fail.

The 2D lattice that defines our surface code and the associated distance measures imply that the HDRG decoder will consider squares over the lattice as the site-enclosing objects, to analyze $k$-level errors. For a level-$h$ error to exist, it must overlap with a square $\Box_h$ of side length $l^h$, which can be considered independently at any position on the 2D lattice, as the following properties are translation-invariant. Furthermore, from \cref{lemma:enclosing}, we obtain that a larger square $\Box_{h+1}'$ of size $9l^{h+1}$ centered on a box $\Box_l$  must fully contain a level-$k$ error if it overlaps with $\Box_h$. This allows us to decompose the probability of failure as follows
\begin{align}
    \Pr[\textnormal{HDRG fails}]&\leq \Pr[ \mathcal{E}_h\in \Box']\\
    &\leq \Pr[(\mathcal{E}_{h-1}\sqcup \mathcal{E}_{h-1}')\in \Box']
\end{align}

\noindent Since the failure probability of the HDRG decoder is bounded by the probability of $\Box'$ containing a level-$h$ error, we can bound the same probability in terms of the probability of the same square containing two disjoint level-$(h-1)$ errors by the definition of a level-$h$ error.

Finally, the probability of having two disjoint level-$(h-1)$ errors in the same square can be further bounded. The events whose probabilities we are bounding here are known as \textit{increasing events}, as defined in percolation theory. In simple terms, increasing events are those events that remain true when the state space leading to the event is extended or augmented. For instance, rain occurs whenever a certain threshold of water falls from the sky (precipitation rate), and it remains a true event if the precipitation rate increases and it rains more intensively. This implies that these events occur under varying conditions, as in our case. Most importantly, with that, we can apply the van den Berg and Kesten inequality \cite{van1985inequalities} to bind the probability as follows,
 \begin{align}
    \Pr[\mathcal{E}_{h-1}\in \Box' \sqcup \mathcal{E}_{h-1}\in \Box'] &\leq  \Pr[\mathcal{E}_{h-1}\in \Box' ]\cdot \Pr[\mathcal{E}_{h-1}\in \Box']\\
    &\leq  \Pr[\mathcal{E}_{h-1}\in \Box' ]^2 .
\end{align}

Subsequently, we have that there are $(3l)^2$ squares $\Box''$ of size $l^{h-1}$ in $\Box'$, such that with a simple union bound we can write
\begin{equation}
\Pr[\mathcal{E}_{h-1}\in \Box'] \leq  (3l)^2\cdot \Pr[\mathcal{E}_{h-1}\in \Box''].
\end{equation}

This allows us to write the probability of the failure recursively in terms of smaller squares until we reach a single site, which coincides with the probability that a single qupit is affected by noise; this in turn is definitionally equal to $\tau$ for local stochastic errors\footnote{Notice that under this reduction, the HDRG decoder has the same error correction success probability for the local stochastic error model considered in this text as for the error model considered in \cite{Bravyi_2013}.}. Hence we have that 
\begin{align}
     \Pr[\textnormal{HDRG fails}]&\leq (3l)^{-4}(3l^{4}\tau)^{2^h}\\ &= (3l^{4}\tau)^{\Omega(m^{\eta})},
\end{align}
\noindent for $\eta>0$ using again that $h \approx \log(m)/\log(l)$. We thus obtain the stated asymptotic bound on the failure probability of the HDRG decoder. Conversely, it decodes correctly with high probability as required.
\end{proof}

The last element needed to prove the error threshold from our conditions in \cref{cond:Ecode} is to define the decoding function $\Dec$ in terms of the string that the HDRG decoder outputs.

\begin{corollary}\label{error:thr2}
There is a decoding function $\Dec$ based on the \textnormal{HDRG} decoder that for a distance-$m$ qupit surface code subjected to local stochastic noise $\mathcal{E}\sim \mathcal{N}(\tau)$ with $\tau<\tau_{th}$ satisfies
\begin{equation}
\Dec(x) = \mathsf{Mod}_p(x)
\end{equation}
and
\begin{equation}
\Pr[\Dec(x\oplus v)=\mathsf{Mod}_p(x)] \geq 1 - e^{-\Omega({m^{\eta}})},
\end{equation}
\noindent for every $v\in \mathbb{F}_p^m$ defined by $\ket{v}=\mathcal{E}_{X}\ket{0}^{\otimes m}$, where $\mathcal{E}_{X}$ denotes the restriction of the error $\mathcal{E}$ to those qupits where it acts as a tensor of Pauli $\mathsf{X}$ operators.
\end{corollary} 
\begin{proof}

Recall that the logical measurement operation of a logical qupit in the surface code is performed by single qupit $\mathsf{Z}$ measurements of the physical qupits along a diagonal of the lattice. The value of the logical measurement result is determined by the sum of the physical measurement results modulo $p$, with $p$ being the dimension of the qupit system. In particular, in the noise-free case when $\tau=0$, the result of this operation effectively produces the correct value for the logical qupit.

For the noisy case, we can model the error that occurs just before the measurement operation as $\mathcal{E}=\omega^i\mathsf{X}(a)\mathsf{Z}(b)$. This error introduces deviations to the measurement results through the $X(a)$ component. For convenience, we represent these deviations from the noise-free measurement results with an additional string $v\in \mathbb{F}_p^n$, which follows from the $\mathsf{X}(a)$ component. We then use the fact that the HDRG decoder, with exponentially high probability, can guess the string $v$ from the error syndromes based on \cref{error:thr1}. Therefore, if we consider the decoding function $\Dec$ as the modular remainder of the measured string combined with the outcome of the HDRG decoder $v_{Guess}$, we have, with exponential confidence, that
\begin{equation}
\mathsf{Mod}_p(x\oplus v \oplus v_{Guess})=\mathsf{Mod}_p(x).
\end{equation}

This finishes the proof of the corollary establishing the exponential confidence for our choice of the decoding function.
\end{proof}

With this proof of \cref{error:thr2} we have completed the proof of the condition which was required to demonstrate the effectiveness of our selected decoder for the qupit surface code. We now proceed to the next crucial ingredient in our noise-resilient construction, single-shot state preparation.

\subsubsection{Qupit single-shot state preparation}\label{subsec:single_shot}

In what follows, we will demonstrate how to achieve single-shot state preparation of the GHZ$_2$ state for qupit surface codes. Considering adaptive state preparation processes, we will outline the predetermined state preparation operations and the adaptive correction mechanisms. Specifically, the latter involves a designated recovery function used to correct errors from the fixed state preparation operations in creating our logical resource state and a repair function that addresses additional noise that may occur in the preceding circuit.

We then construct a 3D cubic lattice into which the logical GHZ$_2$ state will be encoded. To this end, we define all the required stabilizer generators on the 3D lattice and determine the measurement pattern that allows us to create our resource state. The additional Pauli operators that will be required, controlled by the measurement outcomes, define our recovery function. Lastly, we prove that there exists a valid repair function that takes as input any local stochastic error affecting single-shot state preparation and produces as output a global correction in the form of another local stochastic error. Recall that this type of correction implies that our decoder can recover the encoded information with a higher likelihood. This is achieved by demonstrating that the selected repair function satisfies the lifting property, as defined in \cite{bravyi2020quantum}. Similar to the qubit case, it ensures the correct behavior of the repair function across all qupit cases. By combining these two demonstrations, we achieve a single-shot state preparation that fulfills the condition specified in \cref{cond:Ecode} required for fault-tolerant executions of shallow-depth quantum circuits.

\subsubsection*{Qupit recovery and repair functions}

The construction of the resource state needed to demonstrate an advantage in the qupit case involves two main stages, each comprising several steps. In the first stage, a constant-depth Clifford circuit $\mathcal{Z}$ is applied to $2m$ physical qupits and $m_{\text{a}}$ ancilla qupits. The former qupits will be designated for the two logical qupits in the $\ket{\overline{\mathsf{GHZ}_2}}$ state, while the latter are uniquely employed in the state preparation process. Specifically, the correction applied to the state after the measurement process in $\mathcal{Z}$ is described as the recovery operation $\Rec(s)$, where $s$ is a string resulting from measuring the $m_{\text{a}}$ ancilla qupits. More precisely, after executing the constant-depth circuit $\mathcal{Z}$ and the recovery operation $\Rec(s)$, we obtain the following state, 
\begin{equation}\label{singleshot:ess}
\left(\ketbra{s}{s}\otimes \textsc{rec}(s)\right)\mathcal{Z} \left(\ket{0}^{\otimes 2m} \otimes \ket{0}^{ \otimes m_{\text{a}}}\right) \propto \ket{\overline{\mathsf{GHZ}_2}}\otimes \ket{s}. 
\end{equation} 
However, when this process is subject to local stochastic errors, it changes the LHS of the above equation to 
\begin{equation}
    \left(\ketbra{s}{s}\otimes \textsc{rec}(s)\right)\mathcal{E}\mathcal{Z} \left(\ket{0}^{\otimes 2m} \otimes \ket{0}^{\otimes m_{\text{a}}}\right), 
\end{equation} where $\mathcal{E}$ accounts for all the errors occurring in the circuit $\mathcal{Z}$. This is possible because $\mathcal{Z}$ is a constant-depth Clifford circuit, and we can commute the individual errors occurring in each layer to the end (using again \cref{lemma:commute}) . However, this error will be carried to the output state; hence, in the presence of noise, \cref{singleshot:ess} would not be valid. To obtain the valid expression we need to resort to an additional degree of freedom that it will correct for noise, this will be expressed in the form of a repair function $\textsc{rep}(\mathcal{E})$, and the updated condition for the single-shot state preparation,
\begin{equation}
\left(\ketbra{s}{s}\otimes \textsc{rec}(s)\right)\mathcal{E}\mathcal{Z} \left(\ket{0}^{\otimes 2m} \otimes \ket{0}^{\otimes m_{\text{a}}}\right) \propto \textsc{rep}(\mathcal{E})\ket{\overline{\mathsf{GHZ}_2}}\otimes \ket{s}. 
\end{equation} 
\noindent Subsequently, we will describe these steps in more detail, showing how both the recovery and repair functions are going to be defined for the higher dimensional case. The arguments will follow similarly to those in the qubit cases in \cite{bravyi2020quantum, Raussendorf_2005}. 

Our process will utilize a 3D lattice, employing  $m$ qupits on two opposing faces for the logical qupits of the GHZ$_2$ state, with all $m_a$ qupits in the bulk of this lattice serving as ancilla qupits. Specifically, we will define stabilizers over this lattice as $\mathcal{S}_0$ whenever they are uniquely associated with the ancilla qupits and $\mathcal{S}_1$ when they are defined over the entire lattice. These will be detailed further in the next subsection. However, to provide a high-level description, consider that elements of $\mathcal{S}_0$ are the qupits measured during the process of creating the $\ket{\overline{\mathsf{GHZ}_2}}$. Moreover, $S^{i}_{0}$ will correspond to a trivial stabilizer of the $\ket{\overline{\mathsf{GHZ}_2}}$ state. Consequently, the syndrome of the error $\mathcal{E}$  associated with it will only depend on $s$. We will define this as a string $\text{syn}_0(\mathcal{E}) \in \mathbb{F}_d^{k}$, where $k$ represents the number of $\mathcal{S}_0$ generators. The syndrome elements will thus depend on the commutation between the stabilizer elements of $\mathcal{S}_0$ and $\mathcal{E}$. Since we are using qupits, this relationship can be defined through the Mod function, analogous to the parity function in the qubit case,
\begin{equation}
\text{syn}_0(\mathcal{E})_i=\text{Mod}(s,\Supp(S_\alpha^{i})) = \sum_{u \in \Supp(S_\alpha^{i})} s_u \ \MOD\ p, 
\end{equation} 
with the index $i$ used to describe the element of the stabilizer i.e one of $S^{1}_{0}, \ldots, S^{k}_{0}$ and the lower index $\alpha$ describes which subgroup of stabilizers we are referring to. 

In contrast, elements of $\mathcal{S}_1$, will act with $\mathsf{Z}$ type Paulis on the ancilla qupits, but they will also act with both $\mathsf{Z}$ and $\mathsf{X}$ type Paulis on the qupits that will encode our resource state. Furthermore, all elements $\mathcal{S}^{i}_{1}$ will stabilize $\mathcal{Z}\ket{0}^{\otimes 2m} \otimes \ket{0}^{\otimes m_{\text{a}}}$  such that we have,
\begin{equation}
\omega^{b(s)}\left(\ketbra{s}{s}\otimes \mathds{1}_{B})\right)\mathcal{Z} \left(\ket{0}^{\otimes 2m} \otimes \ket{0}^{\otimes m_{\text{a}}}\right), 
\end{equation} 
\noindent with $b(s)$ depending only on the modular sum of the bulk measurements for each stabilizer. Therefore, from $s$ alone, one can infer the error that appears on the facets of the 3D structure with the logical qupits. If not for the existence of extra errors, one could easily obtain the intended logical state from the measurements of the auxiliary qupits. However, the actual phase that appears will be
\begin{equation}
\omega^{b(s) \oplus \text{syn}_{1}(\mathcal{E})}\left(\ketbra{s}{s}\otimes \mathds{1}_{B})\right)\mathcal{Z} \left(\ket{0}^{2m} \otimes \ket{0}^{m_{\text{a}}}\right),
\end{equation} 
\noindent with $\text{syn}_{1}(\mathcal{E})$ accounting for the additional error that we cannot infer from $s$ alone, but detected trough the stabilisers in the set $\mathcal{S}_1$. From here on, for simplicity, we will represent the phases by their powers alone. Additionally, in the ideal implementation, $0$ will be considered the correct phase. Therefore, the phase accumulated, represented as $b(s)\oplus \text{syn}_{1}(\mathcal{E})$, requires a correction, $c$, such that $c \oplus b(s)\oplus \text{syn}_{1}(\mathcal{E}) = 0$. 

Unfortunately, the syndrome of the error on the surface cannot be determined from the measurement outcomes performed on the bulk (we do not have access to the values of the stabilizers $S_1^i$). This problem is addressed in \cite{bravyi2020quantum} by defining an appropriate proxy $P$ for the error, chosen to be a function of the syndrome measurements, $P=f(s)$, in such a way that it will have the same syndrome as the error on the bulk. Then, since the effects of the measurements on the bulk, which generate additional phases on the surfaces and $b(s)$, are known, one can consider the following correction $c \oplus b(s)\oplus \text{syn}_{1}(P) = 0$. 

Recognising $c$ as the recovery function with syndrome $\text{syn}_{1}(\textsc{rec}(s))$ over the set $\mathcal{S}_1$ and using the remaining degrees of freedom that we have, in the form of a repair function, we correct for the difference between the recovery employed, based on the proxy, and the actual error present on the surface, such that \begin{equation}
\text{syn}_{1}(\textsc{rec}(s)) \oplus b(s)\oplus \text{syn}_{1}(\mathcal{E}) \oplus \text{syn}_1(\textsc{rep}(\mathcal{E})) = 0,
\end{equation} leading us to a recovery of the form $\text{syn}_{1}(\textsc{rec}(s)) = -b(s)\oplus \text{syn}_{1}(P)$, and a repair function of the form $\text{syn}_1(\textsc{rep}(\mathcal{E})) = -\text{syn}_{1}(P) \oplus -\text{syn}_{1}(\mathcal{E})$. 

This concludes the high-level description of the repair and recovery functions. Subsequently, we will show how this high-level behavior is concretized in our 3D constructions and demonstrate that these functions fulfill the required conditions from \cref{cond:Ecode} with their precise definitions.

\subsubsection*{Logical $\mathsf{GHZ}_2$ states from the 3D block construction}

We now proceed to define the lattice used for single-shot state preparation of the logical $\mathsf{GHZ}_2$ state. Our lattice retains the general structure described in \cite{Raussendorf_2005} but replaces qubits with qupits at each site. Due to the asymmetries observed in the qupit surface code for $p>2$, which influence how plaquette and vertex stabilizers must be oriented (see \cref{fig:surf_stab}), we must carefully incorporate these asymmetries into the design of the constant-depth quantum circuit $\mathcal{Z}$. Specifically, we ensure the proper alignment of each stabilizer generator in the graph states produced by the quantum circuit. Moreover, we will meticulously verify the geometric placement of each stabilizer generator, ensuring that a valid measurement pattern is achieved for the preparation of a logical $\mathsf{GHZ}_2$ state.

The three-dimensional lattice $\mathcal{C}$, used for the creation of the logical $\mathsf{GHZ}_2$ states with a qupit surface code of distance $d$, has a side length of $r = 2d + 1$ and is defined as follows, 
\begin{equation}\label{eq:lattice}
\mathcal{C} = \{(u_1,u_2,u_3): 0\leq u_1,u_2\leq r-1 \text{, } 1\leq u_3\leq r\}.
\end{equation} 
In particular, the two opposite surfaces, corresponding to $u_3 =1$ and $u_3=r$, are defined such that their first two coordinates $u_1$ and $u_2$ have different parities. We will refer to this set of qupits as the surface\footnote{Throughout this section, we use the symbols $o$ and $e$ inside set notation to mean variables that range over all arbitrary odd and even indices respectively}
\begin{equation}
\sur =\{(e,o,u_3),(o,e,u_3)\in \mathcal{C}: u_3\in \{1,r\}\}.
\end{equation} 
When we use this lattice $\mathcal{C}$  to build the noise-resilient circuit, qupits that lie on $\sur$ will encode the logical GHZ$_2$ state, while we regard all other qupits as ancillary, contained in the \textit{bulk} of this 3D lattice, defined by $\bulk= \mathcal{C}\backslash \sur$.

Similarly to the initial proposal in \cite{Raussendorf_2005} for qubits, we will utilize four graphs—the odd graph, even graph, surface code graph, and dual surface code graph, respectively labeled $T_o$, $T_e$, $T_{sc}$, and $T^{*}_{sc}$. These graphs help us describe the measurement pattern required for preparing the logical GHZ$_2$ state. More specifically, the set of vertices of these graphs will help us in defining the stabilizers of $\mathcal{S}_0$ and $\mathcal{S}_1$ and are given by
\begin{align*} 
V(T_e) &= \{(e,e,e) \in \mathcal{C}\}, & V(T_o) = \{u=(o,o,o) \in \mathcal{C}:u_3 \notin \{1,r\} \}, \\
V(T_{sc}) &= \{u=(e,e,o) \in \mathcal{C}:u_3 \in \{1,r\}\}, & V(T^{*}_{sc}) = \{u=(o,o,o) \in \mathcal{C}:u_3 \in \{1,r\}\}, 
\end{align*}
while the edge sets of $T_o$ and $T_e$ will aid in defining the stabilizers of the resulting logical state, which are given by
\begin{equation*} 
E(T_e) = \{(e,e,o), (e,o,e), (o,e,e) \in \mathcal{C}\},\ 
E(T_o) = \{u=(o,o,e), (o,e,o), (e,o,o) \in \mathcal{C}: u_3 \notin \{1,r\} \}.
\end{equation*} 

Now, to address the previously mentioned asymmetry and achieve the intended logical GHZ$_2$ states, we define the constant-depth quantum circuit $\mathcal{Z}$ over the qupits on the lattice $\mathcal{C}$. For simplicity, we will divide the circuit into operations over the 2D layers running along $u_1$ and $u_2$, for even and odd values of $u_3$, designated as $\mathcal{Z}_{\mathsf{even}}$ and $\mathcal{Z}_{\mathsf{odd}}$, respectively. Additionally, as we will discuss, we will require stabilizers with four different directionalities. To keep track, we define the following function $dir(u)=(-u_1+u_2+u_3)\ \mathsf{mod}\ 4$, which will help us position each of these four elements. Having considered all these details about the circuit, we obtain the following definitions,
\begin{align}\label{eq:zeven}
\mathcal{Z}_{\mathsf{even}}= \left(\prod_{dir(u)= \{0,1\}} \mathsf{X}_u^\dagger \right) \mathsf{F}_u \left(\prod_{v\in n(u),dir(u)=\{2,3\}}  \mathsf{CZ}_v^{-(v_1-u_1+v_2-u_2+v_3-u_3)}\right)\\ 
\left(\prod_{v\in n(u),dir=\{0,1\}}  \mathsf{CZ}_v^{-(v_1-u_1-(v_2-u_2)+v_3-u_3)}\right)
\prod_{u\in E(T_o)\cup E(T_e)\setminus \sur} \mathsf{F}^{\otimes n} 
\end{align}
\noindent forall $u=(\cdot,\cdot,e)\in\mathcal{C}$ with the nearest neighbors of each site $u$ given by $n(u) = \{v \in \mathcal{C}: |u_1-v_1|+|u_2-v_2|+|u_3-v_3| = 1\}$, and respectively
\begin{align}\label{eq:zodd}
\mathcal{Z}_{\mathsf{odd}}= \left(\prod_{dir(u)= \{2,3\}} \mathsf{X}_u^\dagger \right) \mathsf{F}_u \left(\prod_{v\in n(u),dir(u)=\{0,1\}}  \mathsf{CZ}_v^{-(v_1-u_1+v_2-u_2+v_3-u_3)}\right)\\ 
\left(\prod_{v\in n(u),dir=\{2,3\}}  \mathsf{CZ}_v^{-(v_1-u_1-(v_2-u_2)+v_3-u_3)}\right)
\prod_{u\in E(T_o)\cup E(T_e)\setminus \sur} \mathsf{F}^{\otimes n} 
\end{align}
\noindent forall $u=(\cdot,\cdot,o)\in\mathcal{C}$.

Then, the state $\mathcal{Z}\ket{0}^{\otimes |\mathcal{C}|}$, with $\mathcal{Z}=\mathcal{Z}_{\mathsf{even}}\mathcal{Z}_{\mathsf{odd}}$, resulting from our constant-depth circuit, will be measured in the $Z$ basis on all the qupits of the lattice $\mathcal{C}$ such that $u\notin E(T_o)\cup E(T_e)\setminus \sur$. These sites can be directly removed from our analysis for simplicity, as such measurements over graph states correspond to the simple removal of the vertex from the graph. Therefore, for the remaining qupits, the stabilizers of the state are transformed with the Fourier gate to the $X$ basis, resulting in the stabilizers,
\begin{align}
\label{eq:stabilizersX}
    K_u&=S_0^a=\mathsf{Z}_u\left(\prod_{v\in n(u)} \mathsf{X}_v^{(v_1-u_1+v_2-u_2+v_3-u_3)}\right)\text{ for } u=(\cdot,\cdot,e),dir(u)=3\ \text{or}\ u=(\cdot,\cdot,o),dir(u)=0; \\ 
    &=S_0^b= \mathsf{Z}_u\left(\prod_{v\in n(u)} \mathsf{X}_v^{(v_1-u_1-(v_2-u_2)+v_3-u_3)}\right)\ \ \ \ \ \ \ \ \ \hdots\ \ \ \ \ \ dir(u)=2  \ \ \ \ \ \ \ \ \ \hdots\ \ \ \ \ \   dir(u)=1;\\ 
    &= \left(S_0^a\right)^\dagger\ \ \hdots\ \ \ dir(u)=1  \ \ \hdots\ \ \ dir(u)=2;\ \ \ \ \ \ \ \ \ \ \ \ = \left(S_0^b\right)^\dagger \ \ \hdots\ \  dir(u)=0 \ \ \hdots \ \ dir(u)=3;
\end{align}
\noindent with all the vertex are of the type $u\in E(T_o)\cup E(T_e)\setminus \sur$. See \cref{fig:stab_lattive} for a spacial representation of the placement of the stabilisers over $\mathcal{C}$)
\begin{figure}[h]
\begin{center}
\includegraphics[scale=0.55]{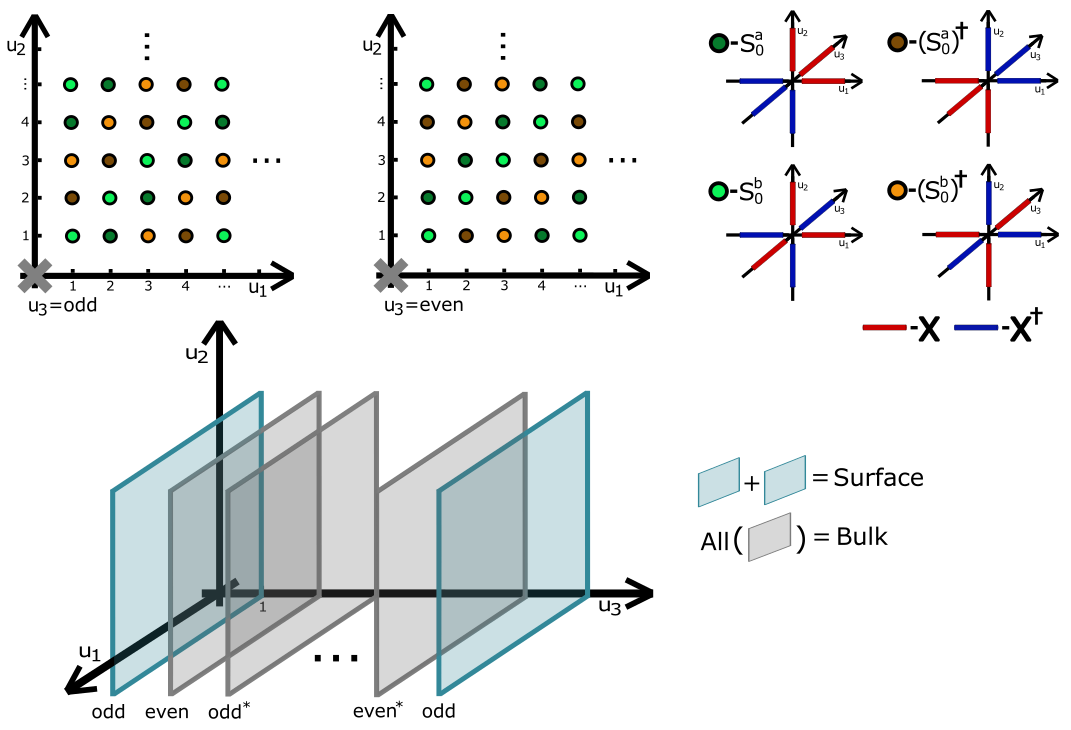}
\caption{\justifying On the top left, we have a representation of the distribution of the stabilizer generators over layers (planes) parallel to the $u_1$ and $u_2$ axes, corresponding to even and odd values of $u_3$ in the lattice $\mathcal{C}$. The layers termed "even$^*$" and "odd$^*$" are equivalent, having the same pattern defined by the distribution of stabilizers, to the "even" and "odd" layers, respectively, except for the translation by two units in the same plane formed by $u_1$ and $u_2$. The cross (×) at the origin of the 2D graphs indicates that the $u_3$ axis is inward-directed. The top right panel displays a representation of the stabilizer generators $K_u$ from \cref{eq:stabilizersX} with their respective spatial orientation. At the bottom, we have a depiction of the entire lattice $\mathcal{C}$, indicating the layers and defining stabilizer positions after the application of the constant-depth circuit $\mathcal{Z}$ to the lattice. Note that the first and last layers constitute the surface ($\sur$), where the logical qupits reside, while the remaining middle layers define the bulk ($\bulk$).}
\label{fig:stab_lattive}
\end{center}
\end{figure}
Having defined the stabilizer group $\mathcal{S}$ over the remaining qupits of the lattice $\mathcal{C}$, which arises from the generators $K_u$ and the graphs $T_o$, $T_e$, $T_{sc}$, and $T^{*}_{sc}$, we can now specify the two special subgroups of this stabilizer group, $\mathcal{S}_0 < \mathcal{S}_1 < \mathcal{S}$. As discussed in the definition of the recovery and repair functions in the previous subsection, these subgroups facilitate the retrieval of the intended outcome state following the measurements of the ancillary qupits. More concretely, within our lattice $\mathcal{C}$, elements of $\mathcal{S}_0$ are designed to act trivially on $\sur$ and to affect $\bulk$ solely through $\mathsf{Z}$ operations. In contrast, elements of $\mathcal{S}_1$ perform $\mathsf{X}$ or $\mathsf{Z}$ operations on $\sur$ and only $\mathsf{Z}$ operations on $\bulk$. These elements are specifically required to be stabilizers of the logical GHZ$_2$ state. Additionally, we intend to have a subset of the stabilizers of $\mathcal{S}_1$, which function as plaquette and vertex stabilizers of the surface qudit code.

Starting with the generators of $\mathcal{S}_0$, which are associated with the qupits corresponding to vertices of $T_o$ and $T_e$, we can describe these based on our stabilizers referenced in \cref{eq:stabilizersX} as follows,
\begin{equation}\label{eq:stab0a}
S^{u}_0 = \prod_{v\in n(u)}K_v=  \prod_{v\in n(u)} \mathsf{Z}_v^{(v_1-u_1+v_2-u_2-(v_3+u_3))},\text{ with }\ u\in V(T_e),
\end{equation}
\noindent and 
\begin{equation}\label{eq:stab0b}
S^{u}_0 = \prod_{v\in n(u)}K_v=  \prod_{v\in n(u)} \mathsf{Z}_v^{(v_1-u_1-(v_2-u_2)+v_3-u_3)},\text{ with }\ u\in V(T_o).
\end{equation}
\noindent In both cases, all the $\mathsf{X}$ operators from the $K_v$ stabilizers overlap with some other $\mathsf{X}^\dagger$ operators within the same stabilizer set and completely cancel out. Thus, we obtain exclusively $\mathsf{Z}$ operators for this stabilizer subgroup, as intended.

The second subgroup, $\mathcal{S}_1$, utilizes the graphs $T_{sc}$ and $T^{*}_{sc}$ to define two types of stabilizers; one functions as the vertex stabilizer and the other as the plaquette stabilizer of the qudit surface code, as we will see subsequently demonstrate.
\begin{align}\label{eq:stab1a}
S^{u}_1 = K_u=&  \prod_{v\in n(u)} \mathsf{Z}_u^\dagger\left(\prod_{v\in n(u)} \mathsf{X}_v^{-(v_1-u_1-(v_2-u_2)+v_3-u_3)}\right),\text{ with }\ dir(u)=1, u\in V(T_{sc}),\\
 &\prod_{v\in n(u)} \mathsf{Z}_u\left(\prod_{v\in n(u)} \mathsf{X}_v^{(v_1-u_1-(v_2-u_2)+v_3-u_3)}\right)\ ,\ \text{ with }\ dir(u)=3, u\in V(T_{sc}).
\end{align}
\noindent These operators act exclusively on the qupits in the planes of the logical qupits where $u_3 = {1, r}$, since operators in adjacent layers affect qupits that have already been measured out. Additionally, the vertices $u \in V(T_{sc})$ are not located on $\sur$, the 2D lattice containing the logical qupit. Consequently, these aforementioned operators precisely function as the vertex stabilizer operators within the qupit surface code for any logical states on $\sur$.

Subsequently, we will define the subgroup of $\mathcal{S}_1$ generating plaquette operators stabilizers,
\begin{align}\label{eq:stab1b}
S^{u}_1 =K_{u\pm (0,0,1)}\prod_{v\in (n(u)\cap \sur)} K_v =&   \prod_{v\in n(u)} \mathsf{Z}_v^{-(-(v_1-u_1)+v_2-u_2+v_3-u_3)} ,\text{ with }\ dir(u)=1, u\in V(T^{*}_{sc}),\\
 &\prod_{v\in n(u)}  \mathsf{Z}_v^{(-(v_1-u_1)+v_2-u_2+v_3-u_3)}\ \ ,\text{ with }\ dir(u)=3, u\in V(T^{*}_{sc}).
\end{align}
\noindent This equality follows because all the $\mathsf{X}$ operators for the selected elementary stabilizer generators $K_u$ cancel each other out completely, leaving us with only the previously mentioned set of $\mathcal{Z}$ operators. In particular, these two stabilizers, each operating on half of the vertices defined by $T^{*}_{sc}$, generate the precisely the plaquette operators of the qupit surface code when restricted to $\sur$, with the correct orientations.

Finally, we have defined all the state-independent stabilizers for the logical qupits on the surface $\sur$. However, we still need to identify the final one for these states to determine the specific logical stabilizer state they represent. Therefore, given our goal to obtain the logical $\mathsf{GHZ}_2$ state, ideally, the logical stabilizers $\overline{\mathsf{X}^\dagger} \otimes \overline{\mathsf{X}^\dagger}$ and $\overline{\mathsf{Z}} \otimes \overline{\mathsf{Z}^\dagger}$ would be generated from the graph state stabilizers in $\mathcal{S}_1$. Furthermore, the previous logical operators are defined over the lattice as follows,
\begin{equation}
\overline{\mathsf{X}^\dagger}\otimes \overline{\mathsf{X}^\dagger} = \prod_{u=(1,e,1) \in \sur}\mathsf{X}_u^\dagger \prod_{u=(1,e,r) \in \sur}\mathsf{X}_u^\dagger\ \text{  and  }\ \overline{\mathsf{Z}}\otimes \overline{\mathsf{Z}^\dagger} = \prod_{u=(0,o,1) \in \sur}\mathsf{Z}_u \prod_{u=(0,o,r) \in \sur}\mathsf{Z}^{\dagger}_u.
\end{equation}

We will demonstrate that our resulting logical state, up to some local Clifford operations, is equivalent to the intended GHZ$_2$ state. To begin, we obtain that for an even-dimensional code distance $d$, the stabilizer $S_1^\mathsf{X}$ is defined as follows,
\begin{align}\label{eq:stab1c}
   S_1^{\mathsf{X}^\dagger}= \prod_{(u_1,u_1-1,e) \in \mathcal{C}} K_u&= \prod_{u=(u_1,u_1-1,1)\in \mathcal{C}} \mathsf{X}_u^\dagger \prod_{(u_1,u_1-1,r)\in \mathcal{C}} \mathsf{X}_u^\dagger \prod_{(u_1,u_1-1,e) \in \mathcal{C}, dir(u)=3} \mathsf{Z}_u \prod_{(u_1,u_1-1,e) \in \mathcal{C}, dir(u)=1} \mathsf{Z}_u^\dagger\\
    &\equiv \left(\overline{\mathsf{X}^\dagger} \otimes \overline{\mathsf{X}^\dagger}\right)_{\sur} \prod_{(u_1,u_1-1,e) \in \mathcal{C}, dir(u)=3} \mathsf{Z}_u \prod_{(u_1,u_1-1,e) \in \mathcal{C}, dir(u)=1} \mathsf{Z}_u^\dagger,
\end{align}
\noindent with $(\cdot)_{\sur}$ designating that this operator is applied on the qupits of the surface $\sur$, which encode our logical qupits. In particular, this ensures that the logical states we generate through our single-shot state preparation $\ket{\phi}$ are transformed by the application of the identified stabilizers as follows $\overline{\mathsf{X}^\dagger} \otimes \overline{\mathsf{X}^\dagger}$ to produce $\ket{\phi} = \lambda_{\mathsf{X}^\dagger\mathsf{X}^\dagger}\ket{\phi}$.

Similarly, we will also define another stabilizer for the resulting state, $S_1^{\mathsf{Z}^\dagger}$, as described below,
\begin{equation}\label{eq:stab1d}
   S_1^{\mathsf{Z}^\dagger}= \prod_{(u_1,u_1-1,o) \in \mathcal{C}} K_u= \prod_{u=(u_1,u_1-1,1)\in \mathcal{C}} \mathsf{Z}_u^\dagger \prod_{(u_1,u_1-1,r)\in \mathcal{C}} \mathsf{Z}_u^\dagger \equiv \left(\overline{\mathsf{Z}^\dagger} \otimes \overline{\mathsf{Z}^\dagger}\right)_{\sur} .
\end{equation}
\noindent This second stabilizer ensures that our state satisfies $\overline{\mathsf{Z}^\dagger}\otimes\overline{\mathsf{Z}^\dagger}\ket{\phi} = \lambda_{\mathsf{Z}^\dagger\mathsf{Z}^\dagger}\ket{\phi}$.  Consequently, following a correction based on the measurement outcomes from elements in $\bulk$, the resultant state $\ket{\phi'}$ will, up to some known logical Clifford operations, be equivalent to the GHZ$_2$ state for two qupits, represented as $\ket{\overline{\mathsf{GHZ}}_2} = \frac{1}{\sqrt{p}}\sum^{p-1}_{i=0}\ket{\overline{ii}}$.

We note that one could also consider the case where the code distance $d$ is of odd dimension; however, the stabilizers would become slightly more complex. In addition, unlike the previous solutions and the qubit setting, the qupit surface code will exhibit directionality. Consequently, the copy of the surface code at $u_3=r$ will serve as a mirror image of the surface code at $u_3=1$.
 
\subsubsection*{Properties of the Repair function}

We will now show that the repair function continues to be manageable even in the higher dimensional case; in particular, it continues to output a local stochastic correction. For a local stochastic error $\mathcal{E}=\mathcal{E}_\bulk \otimes \mathcal{E}_\sur\sim\mathcal{N}(\tau)$ with $\tau\leq 1$ acting on our lattice $\mathcal{C}$, this condition translates to 
\begin{equation}
\label{prob} 
    \Pr_\mathcal{E}[V \subseteq \text{Supp}(\Rep(\mathcal{E_\bulk}))] \leq \tau^{|V|}. 
\end{equation}
\noindent This can be inherited from the analysis of the qubit cases presented in \cite{bravyi2020quantum}. However, we must define our recovery and repair functions more concretely. Unlike previous works that use the MWPM decoder, these functions will be based on the HDRG decoder. For example, the recovery function is defined as the Pauli correction operator produced by the HDRG decoder, parameterised on a fixed distance value $k$, which uses as input the syndromes generated by our proxy $P$ and the additional Pauli corrections from the bulk measurements $b(s)$. Therefore, the repair function will be the Pauli operator derived from the difference between the two Pauli operators obtained from the HDRG, corresponding to the concrete error $\mathcal{E}_\bulk$ and our guessed proxy $P$,
\begin{equation}
\text{syn}_{1}(\Rep(\mathcal{E}_\bulk)) = \text{syn}_{1}(\mathcal{E}_\bulk)\oplus \text{syn}_{1}(P).
\end{equation}

Now, more importantly, we need to demonstrate that under this precise definition, the repair function retains the property expressed in \cref{prob}. To achieve this, we can leverage two properties, collectively introduced as the \textit{lifting property} in \cite{bravyi2020quantum}, to substantiate that \cref{prob} holds for our $p$-dimensional repair function. In particular, there must be a lifting function $f: V_K \to V_L \subseteq 2^{\mathcal{C}}$ that maps a subset $V_K$ of the surface $\sur$ to a set of subsets within $\mathcal{C}$. Then, using these functions, the first property stipulates that for each $V_K \subseteq \text{Supp(\Rep}(\mathcal{E}))$ there exists a set $V_{\mathcal{E}'} \in f$($V_K$) such that  
\begin{equation}
\label{a:sub}
    |V_{\mathcal{E}'} \cap \text{Supp}(\mathcal{E})|\geq \alpha|V_{\mathcal{E}'}|. 
\end{equation}
\noindent Secondly, for every $\lambda > 0$, there must be a corresponding $\lambda_1 > 0$ such that
\begin{equation}
\label{const}
    \sum_{V_L \in f(V_K)}\lambda^{|V_L|} \leq \lambda^{|V_K|}_1.
\end{equation}
 
Previously, these properties were demonstrated for the qubit case using a repair and lifting function based on the MWPM decoder. Next, we will demonstrate that this property applies to a repair and lifting function constructed using the HDRG decoder. The lifting function based on the MWPM decoder was defined to output an edge-disjoint union of trees (also called a forest) with boundaries at the endpoints of a subset of the $\bulk$. In this paper, we define the lifting function to output the set of all subsets $V_L$, which are $r$-connected components that intersect $V_K$ at each connected component. Recall that these connected components, according to our definition of $r$-connectivity, are subsets of the error at a distance bounded by $r$ between each other.

Now for the first condition in \cref{a:sub}, consider that errors $\mathcal{E}_\bulk$ on the $\bulk$ that anti-commute with the stabilisers of $\mathcal{S}_1$ and therefore, activate an $S_1^i$ syndrome, are localised single errors on one of the vertexes of $V(T_{sc})$ or $V(T_{sc*})$. This follows from the fact that these stabilisers are related to single $Z$ stabilisers when we are considering unique errors on the $\bulk$ and, therefore, do not cause any degeneracy. Furthermore, we will observe that our repair function will have an error on the surface $\sur$, localised on one of the stabilisers $S_1^i$, to correct an error of our proxy $P$ that generates a different syndrome from the error $\mathcal{E}_\bulk$ on the same vertex.

We consider the non-trivial case where $\text{syn}_1(\mathcal{E}_\bulk)$ is a non-zero value for some stabilizer $S_1^i$, indicating the presence of an error at the corresponding vertex associated with $S_1^i$. If, in this situation, the Pauli operator resulting from our repair function, defined on the sites $V_K$, intersects the support of $S_1^i$, we can conclude that the lift of $V_K$ certainly includes an error set with the connected component containing the error on the vertices on the bulk defined by $S_1^i$. Moreover, $V_K$ is defined exclusively on the surface $\sur$ and includes only the sites where the previous case repeats with the $S_1^i$ type stabilizers. Specifically, the Pauli operators on the sites of $V_K$ address errors at vertices in the $\bulk$, which are signaled by the syndromes of $S_1^i$ and were not accounted for by our proxy $P$. Therefore, the lift of $V_K$ generates certainly a set of $r$-connected components of the error $\mathcal{E}_\bulk$. This and the fact that the size of the $r$-connected components is finite implies that the size of the set $V_K$ correlates with the number of $r$-connected components of $\mathcal{E}_\bulk$, ensuring that the overlap between one of the elements in lift and the original error $\mathcal{E}_\bulk$ is proportional, thereby fulfilling the condition in \cref{a:sub} for some constant $\alpha$.

Secondly, we upper bound \cref{const} by upper bounding the number and size of the $r$-connected components that can be generated on the lattice $\mathcal{C}$ from a specific set $V_K$. 
\begin{lemma}\label{lemma:lift2}
There exists a lifting function $f$ defined using the HDRG decoder such that for any subset $V_K$ of $\sur$ and arbitrary value of $\lambda>1$ we have that
\begin{equation}\label{eq:bound_connected_comp}
 \sum_{V_L \in f(V_K)}\lambda^{|V_L|} \leq \lambda^{|V_K|}_1,  
\end{equation}
\noindent with $\lambda_1>1$.
\end{lemma}
\begin{proof}
We use the previously described lifting function $f$ based on the HDRG decoder and employ two bounds to constrain the sum on the left-hand side of \cref{eq:bound_connected_comp}. We start to bind the size of each connected component. To do this, we observe that at most $|V_K|$ connected components exist in $|V_L|$ since $V_L$ must intersect all the vertices of $V_K$. Next, we bound the size of each component by taking each vertex in $V_K$ as a root node and consider respectively the largest tree that can be spanned given the graph's connectivity $D$. Given that we are considering $r$-connected components, the maximum distance between our root element, a vertex in $V_K$, and the leaf nodes is $r$.

In the second step, we need to bind the number of possible sets $V_l$ in the lift. Since these are simply all the possible sets one can form, we can consider the power set of the vertices in the largest set $V_L$ previously determined. By sequentially applying these steps, we obtain that
\begin{align}
     \sum_{V_L \in f(V_K)}\lambda^{|V_L|} &\leq \sum_{V_L \in f(V_K)} \lambda^{|V_K|\cdot D^r}\\
     &\leq \sum_{i=0}^{|V_K|\cdot D^r}\binom{|V_K|\cdot  D^r}{i} \lambda^{|V_K|\cdot  D^r}
     \leq \left(\left(2\lambda\right)^{D^r}\right)^{|V_K|}
\end{align}
taking $\left(2\lambda\right)^{D^r}$ as $\lambda_1$ completes the proof.
\end{proof}

With this, we conclude that there exists a lifting function with the necessary properties to ensure that our the repair function for single-shot state preparation satisfies the local stochastic property defined in \cref{prob}.

Finally, we obtain that with the 3D construction described in the previous section, as well as the respective recovery and repair functions, we obtain a noise-resilient single-shot GHZ$_2$ state preparation process for the qupit surface code as required by our conditions in \cref{cond:Ecode}.

\begin{lemma}\label{lem:single_shot}
Let $\mathcal{E}\sim \mathcal{N}(\tau)$ be a local stochastic error and consider a surface code with code distance $d$. Then, there exists a constant-depth Clifford circuit $\mathcal{Z}$, along with corresponding recovery and repair functions, such that the conditions in \cref{cond:Ecode} are satisfied, ensuring noise-resilient preparation of the logical $\textnormal{GHZ}_2$ state.
\end{lemma}
\begin{proof}
This culminates with all the descriptions in the proofs of this subsection. We consider for that construction a 3D lattice composed of qupits as described in \cref{eq:lattice}, along with the corresponding constant depth state preparation circuit $\mathcal{Z}$, defined by \cref{eq:zeven,eq:zodd}. After performing the respective measurements on the $\bulk$, we use the precise repair and recovery functions defined based on the stabilisers of $\mathcal{S}_0$ and $\mathcal{S}_1$ from \cref{eq:stab0a,eq:stab0b} and \cref{eq:stab1a,eq:stab1b,eq:stab1c,eq:stab1d} respectively. The application of these functions, as shown previously, results in a logical GHZ$_2$ state on the qupit surface code, with a code distance of $d$.
\end{proof}

This provides the last ingredient for the noise-resilient implementation of any quantum circuit over qupits that has with the structure illustrated in \cref{fig:advice_noise_resilient} using the qupit surface code. Next, we will discuss the specific quantum circuits that will resort to that circuit structure.

\subsection{Noise-resilient quantum advantage with magic state injection}\label{sec:injection}

In this subsection, we address the incompatibilities between the logical quantum circuits proposed in section \cref{noise_free} and the circuit structure (\cref{fig:advice_noise_resilient}) required for proving noise resilience. This primarily involves managing non-Clifford operations in the quantum circuits, for which we have demonstrated computational separations for qupits with dimensions $p\geq 3$. These operations are essential and necessary because the separations observed with Clifford circuits using qubits do not generalize to qupits \cite{gross2006hudson, howard2013quantum}. Additionally, we face the constraint that the circuit cannot be adaptive, preventing the use of standard gate teleportation gadgets. To resolve this, we introduce a new quantum circuit that can solve the ISMRP family of problems using magic states and a non-adaptive constant-depth Clifford circuit.

We start off by demonstrating the existence of such a logical quantum circuit. We will then combine all the elements from the previous sections with our new logical quantum circuit, which has the necessary structure for noise-resilient execution. Finally, we will prove our noise-robust separations as outlined at the beginning of the section.

\subsubsection{Non-adaptive Clifford circuits with magic state injection}\label{subsec:inject}

To simplify the design of our new quantum circuits with the structure of \cref{fig:advice_noise_resilient}, we focus exclusively on the non-Clifford operators. In particular, we will address the implementation of the $\mathsf{Z}$ rotations for all dimensions $p \geq 3$. Our first step is to replace this gate with an equivalent realization using the following qupit gate teleportation gadget \cite{de2021efficient,Zeng08}, allowing us to offload the resource requirement for the execution of this gate into an advice state, 
\begin{equation}\label{eq:telep}
\begin{quantikz}
\lstick{$\ket{\psi}$} & \gate{R_Z(\phi)} & &\wireoverride{n} \equiv &\wireoverride{n} & \gategroup[wires=1,steps=3,style={dashed, rounded corners, fill=gray!20, inner xsep=15pt},background]{{Advice $\ket{T^{1/p}}$ state}} \lstick{$\ket{0}$}\wireoverride{n}  & \gate{\mathsf{F}} &  \gate{R_Z(\phi)}& &  \ctrl{1} & \gate{R_Z(\phi)\left(\mathsf{X}^{c_i}\right)^\dagger \left(R_Z(\phi)\right )^\dagger}& \rstick{$R_Z(\phi) \ket{\psi}$} \\
\wireoverride{n} & \wireoverride{n} & \wireoverride{n} & \wireoverride{n} & \wireoverride{n} & \wireoverride{n} \lstick{$\ket{\psi}$} &   \gate{\mathsf{F}^2} & & &  \targ{} & \meter{}  \wire[u][1]{c} & \rstick{$c_i$}
\end{quantikz}.
\end{equation}

Obtaining advice states defined as follows,
\begin{equation}
  \ket{T^{1/p}}:= R_Z\left (\frac{2\pi}{p^2}\right )\mathsf{F}\ket{0} =\frac{1}{\sqrt{p}} \sum_{j=0}^{p-1} e^{i\frac{ 2\pi \cdot j }{p^2}}\ket{j}.
\end{equation}

Understanding the state $\ket{T^{1/p}}$ in more detail is crucial, as it will define the type of corrections required in \cref{eq:telep}. Therefore, defining the new quantum circuit we intend to execute, as well as the complexity of the advice states themselves. For instance, if $R_Z(\phi)$ is an operation from the third level of the Clifford hierarchy, the required correction within the gadget would be a Pauli operator. This correction is particularly convenient, as we can trivially commute it to the end of the circuit without significant complications.

To build intuition, notice that in the qubit case these advice states are related to the $T$ operator. In particular, $\ket{T^{1/2}}=TH\ket{0}$. These states are called $T$ magic states and are sufficient to perform universal quantum computation when combined with Clifford gates and adaptive measurements. Additionally, when provided as resource states in their logical version encoded with a quantum error correction code capable of realizing Clifford gates transversally, these states allow for universal \textit{fault-tolerant} quantum computation. We now show that a similar type of relationship continues to hold between $\ket{T^{1/p}}$ states and the corresponding $T$-type magic states in higher dimensional qupit systems.

\begin{lemma}\label{lem:magic_state}
For all prime $p\geq 3$, the state $\ket{T^{1/p}}$ is derived from the  $p$th root operator that creates the $T$-type qupit magic state.  
\end{lemma}
\begin{proof}
For universal quantum computation, it is known that any operator at the third level of the Clifford hierarchy or above, when combined with the Clifford group, allows for universal quantum computation \cite{Cui2017}. The operation considered for implementing through-the-gate teleportation is a diagonal gate, which we intend to show is sufficient for universal computation along with Clifford operations and is related to standard magic states. As described in \cite{Howard12}, the generalization of $T$ gates for qupits, in combination with Clifford gates, allows for universal quantum computation over these higher-dimensional quantum systems. This generalization includes any operator in the family ($\omega=e^{\frac{2\pi i}{p}}$)
\begin{equation}
    U(v_0,v_1,\hdots,v_{p-1})= \sum_{k=0}^{p-1} \omega^{v_k} \ket{k}\bra{k},
\end{equation}
\noindent parameterized by $\{v_0,v_1,\hdots,v_{p-1}\}$ satisfying $v_k=\frac{k}{12}(\beta+k(6\gamma+(2k+3)\beta))+k\zeta$ for each $k\in [p-1]$ for some $\beta,\gamma,\zeta\in \mathbb{F}_p$. Furthermore, our advice state is created from the $p$th root of the operator that yields the following stat
\begin{equation}
\ket{T}= \frac{1}{\sqrt{p}} \sum_{k=0}^{p-1} w^k \ket{k}.
\end{equation}
\noindent Now, we will prove the previously defined $\ket{T}$ states are magic states under the given conditions for the qupit generalised $T$ gates. For that, we need to show that there exists a solution for the values of $\beta,\gamma$ and $\zeta$ such that we obtain the vector $v_k=(0,1,\hdots,p-1)$ for all primes.
For instance, with $\beta=0, \gamma=2$, and $\zeta=0$, we obtain the vector $v_k=(0,1,2)$ for the qutrit case. For $p=5$, we obtain the vector $v_k=(0,1,2,3,4)$ with $\beta=5, \gamma=5$, and $\zeta=6$. More generally, we have a computer-assisted proof, 
verifying analytically, that for all primes $p$, the following equation has solutions (We resort to Wolfram Mathematica, specifically using the ``\textit{reduce}" function, to analytically check if there exist integer solutions for the free variables in all prime dimensions).
\begin{equation}
\frac{k}{12}(\beta+k(6\gamma+(2k+3)\beta))+k\zeta  \equiv k\ (\MOD\ p ).
\end{equation}

This condition guarantees that
$\ket{T}$ is indeed a $T$-magic state for all prime $p$ as defined in \cite{Howard12,Gidney2019efficientmagicstate,Beverland_2020}, thus concluding the proof.
\end{proof}

We have characterized our advice states with respect to the $T$-type qupit magic state. On the other hand, we know that the latter is a magic state for a gate in the third level of the Clifford hierarchy, but any $p$th root of this gate must be of a higher level\footnote{By the recursive nature with which the Clifford hierarchy is defined (see \cref{def_Cliff}), we obtain, for instance, that the root of any operator at level $l$ must be an element of level $l+1$.}. Thus we find that our advice state is not a magic state for the third level of the Clifford hierarchy, but for some gate of level $3+p$. Consequently, the correction operator $R_Z(\phi)\left(\mathsf{X}^{c_i}\right)^\dagger \left(R_Z(\phi)\right)^\dagger$ for the associated gate teleportation gadget using the magic state is no longer a Pauli operator. This means we cannot simply commute it through to the end of the Clifford circuit. Since our $p$th root operator is in the $p+3$th level, this forces the correction operation to be an operator above the second level. Thus, it cannot be a Clifford operator either and cannot be implemented by any composition of the logical and transversal gates of the qupit surface code.

One potential solution is to use additional advice states to perform a cascade of gate teleportations, where each correction comes down one level in the Clifford hierarchy, ultimately making the final correction a Clifford operator \cite{de2021efficient, Zeng08}. However, this approach would require a number of sequential adaptive steps which is beyond our reach in the shallow-depth $\QNC^0$ model of interest to us. More precisely, we would need to decode the control qupit, which is a logical qupit, necessitating computations over a string of size  $\Omega(\log n)$. Each output of a constant depth bounded fan-in circuit can only depend on a constant number of inputs, making such computations as required above infeasible in constant depth. 

Nevertheless, whenever we intend to implement a rotation through a teleportation gadget as in \cref{eq:telep} but cannot realize the correction operations, we can overcome this issue by accounting for the effects on the output string due to the unrealized corrections, in addition to the rotations implemented by the respective gadget. More precisely, we will show that a circuit using the gadget from \cref{eq:telep} \textit{without} the correction steps can still solve the target ISMRP with an exponentially larger correlation than any classical circuit. The proof for this statement will be based mainly on a different $\NC^0$ reduction as the ones in \cref{lemma:comput_trit} and \cref{lemma:computcorre}, producing a different correction string from those analyzed so far. However, this is only possible if the correction operators have a specific form: they need to be a $\mathsf{Z}$ rotation followed by a Pauli operator. The Pauli operator can be commuted to the end and its effect inverted in a simple manner within the $\NC^0$ reduction step. At the same time, the diagonal $\mathsf{Z}$ rotations will be accounted for with an additional correction string in the $\NC^0$ reduction, as we will show. In the next lemma, we demonstrate that the correction operations have exactly the form that we require.

\begin{lemma}\label{lemma:corrections}
For any qupit input state $\ket{\psi}$, the output state in the first register of the teleportation gadget, as shown on the right-hand side of \cref{eq:telep}, can be written based on the measurement result $c\in \mathbb{F}_p$ in the second register as,
\begin{equation}
\mathsf{X}^{c} GR_Z(\theta_1) GR_Z(\theta_2,S_{c}) R_Z(\phi)^{c} \ket{\psi},
\end{equation}
\noindent where $GR_Z(\theta)$ is the diagonal matrix with all diagonal elements equal to $\theta$ and $\theta_1=-(2\pi c)/p^2$. The operator $GR_Z(\theta,S)$ is defined as follows
\begin{equation}
    GR_Z(\theta,S)= \sum_{x\in S}e^{i\theta}\ket{x}\bra{x} + \sum_{x\notin S} \ket{x}\bra{x},
\end{equation}
\noindent and we have $\theta_2=\frac{2\pi}{p}$ and the set $S_{c}$ defined for each $c$ as $[p-c,\hdots,p-1]$.
\end{lemma}
\begin{proof}
The output state of the teleportation gadget without applying the correction of the teleportation gadget shown in \cref{eq:telep}, can be described based on the result of the measurement in the second register $c\in \mathbb{F}_p$ as, 
\begin{equation}
\left(R_Z(\phi)\left(\mathsf{X}^{c}\right)^\dagger \left(R_Z(\phi)\right )^\dagger\right)^\dagger R_Z(\phi)\ket{\psi}.
\end{equation}
\noindent Thus, we only need to show that the operator $\left(R_Z(\phi)\left(\mathsf{X}^{c}\right)^\dagger \left(R_Z(\phi)\right )^\dagger\right)^\dagger$ is equal to $\mathsf{X}^{c}  GR_Z(\theta_1) GR_Z(\theta_2,S)$. This follows from a simple analysis of how the non-zero terms in these matrix operations are composed. First, observe that $(\mathsf{X}^{c})^\dagger=\mathsf{X}^{p-c}$ and that $(R_Z(\phi))^\dagger=R_Z(-\phi)$. From this, we know that if $c_i=0$, then we have $R_Z(\phi)R_Z(-\phi)=I$. This follows simply from the fact that for all bases $\ket{j}$ with $j\in\mathbb{F}_p$, the first operation introduces phase term $j\cdot \phi$, while the second introduced $-j\cdot \phi$, canceling each other out. However, if $c\neq 0$ the operator $\mathsf{X}^c$ is not the identity. This implies that the phases which operate are shifted according to the exponent $p-c$, and we get the following phases,
\begin{equation}
e^{j(c-j)\phi},\text{if }j< p-c\text{ and }e^{j(-(p-c)-j)\phi},\text{if }j\geq p-c.
\end{equation}

Therefore, we have that $c$ many bases state obtain a phase term with the value $-(p-c)\phi$ and $p-c$ many bases with a phase them having the value $c\phi$. The particular set of bases to which each one of these phases will be applied will be ignored for now as these follow from a simple analysis of the $\mathsf{X}^c$ operation for any $c\in \mathbb{F}_p\setminus 0$. Considering now the transposed conjugated operation, we obtain $c$ many terms with the phase parameterized by $(p-c)\phi$ and $p-c$ terms with the phase $-c\phi$. Interestingly, one can write this operator simply as $\mathsf{X}^a$ operator with $a\in \mathbb{F}_p\setminus 0$ multiplied with a diagonal matrix defined by the previously determined phase terms.
\begin{equation}
D = \mathrm{diag}\Big(\underbrace{e^{-c\phi},e^{-c\phi}, \ldots, e^{-c\phi}}_{(p-c) \text{ times}},\underbrace{e^{(p-c)\phi}, e^{(p-c)\phi}, \ldots, e^{(p-c)\phi}}_{c \text{ times}}\Big).
\end{equation}

Subsequently, we can factorize this diagonal matrix into two diagonal matrices as follows, 

\begin{equation}
D_1D_2 = \mathrm{diag}\Big(e^{-c\phi}, e^{-c\phi}, \ldots, e^{-c\phi} \Big) 
\mathrm{diag}\Big( \underbrace{1, 1, \ldots, 1}_{(p-c) \text{ times}},\underbrace{e^{p\phi}, e^{p\phi}, \ldots, e^{p\phi}}_{c \text{ times}}\Big).
\end{equation}
\noindent We obtain that $D_1$ is exactly $GR_Z(\theta_1)$, as $-c\phi$ is equal to $-(2\pi c)/p^2$. Similarly, $D_2$ is exactly equal to $GR_Z(\theta_2, S_{c})$, since $S_{c}$ corresponds to the elements over which $D_2$ operates with the phase $p\phi$, which is exactly $2\pi/p$. In our final step, we only need to determine the exponent of the operator $\mathsf{X}^a$. Before the transposition, the phase terms were shifted by the operator $\mathsf{X}^{p-c}$; thus, we obtain that $a = c$ after the transposition and complete the proof.
\end{proof}

Having defined the additional rotations that result from the teleportation gadget, we will now show that there exists a Clifford circuit which when provided with the quantum advice state $\ket{T^{1/p}}^{\otimes n}$ solves the ISMR problems with an exponentially larger correlation than $\BTC^0(k)$ circuits with $k=n^{1/(5d)}$, as in \cref{thm:qudit_average}.

\begin{lemma}\label{lem:clifford}
There exist a family of constant-depth Clifford circuits with the advice state $\ket{T^{1/p}}^{\otimes n}$, denoted as $\mathsf{Cliff}_{+T}$, which can solve the ISMR problems $\mathcal{R}_p^m$, for a uniform distribution $\mathcal{D}_p$ of input strings that satisfy the condition $\sum_{i=1}^n x_i\ \MOD\ p = 0$ within $\mathbb{F}_2^n$. For any given input, these circuits achieve a success probability of $\frac{2p - 2}{p^2}$. Additionally, they maintain an average correlation with the correct outcome within the Abelian domain, given by
\begin{equation} \mathsf{Corr}_{\mathcal{D}_p} \left(\mathsf{Cliff}_{+T},\mathcal{R}_p^m \right) = \frac{p-1}{p^2}, 
\end{equation} 
where $m = \mathcal{O}(n \cdot g^{p^3})$.  
\end{lemma}
\begin{proof}
This circuit closely follows the quantum circuit considered in \cref{fig:circ_qudit}. However, we substitute the classically controlled\footnote{Whenever we refer to a classical controlled gate, we mean that the gate will be implemented or not based on a classical binary bit defining each possible case. This contrasts with standard controlled gates that can be controlled by a qubit or qupit, which might be in a superposition} rotations $R_Z$, which are non-Clifford operators, with the teleportation gadget as represented on the right-hand side of \cref{eq:telep} without the correction operations, obtaining the following classically controlled Clifford circuit,
\begin{figure}[H]
\begin{center}
\begin{quantikz}
\gategroup[wires=1,steps=3,style={dashed, rounded corners, fill=gray!20, inner xsep=15pt},background]{{Advice $\ket{T^{1/p}}$ state}} \lstick{$\ket{\overline{0}}$} & \gate{\mathsf{F}} &  \gate{R_Z(\phi)}& &  \ctrl{1} & & \rstick{$\mathsf{X}^{c_i} GR_Z(\theta_1) GR_Z(\theta_2,S_{c_i}) R_Z(\phi)^{c_i} \ket{\psi}$} \\
\lstick{$\ket{\overline{\psi}}$} &   \gate{\mathsf{F}^2} & & &  \targ{} & \meter{}  & \rstick{$c_i$}
\end{quantikz}.
\end{center}
\label{fig:non_adpative_gad}
\end{figure}
Furthermore, we will consider the correction string $c\in \mathbb{F}_p^n$, which describes the measurement outcomes in all the teleportation gadgets and determines the additional rotations applied to the states. Note that $c_i=0$ if $x_i=0$, as no gadget is implemented for input bits, coming from the ISMR problems, with this value.

Having defined the setting, we need to redo the analysis of step 1 in \cref{lemma:computcorre}, as this is the first step where these two quantum circuits attempting to solve the same problems differ intrinsically. Specifically, we must describe the effect of the additional rotations resulting from the inability to perform corrections from the teleportation gadget, as detailed in \cref{lemma:corrections}.

As we have three $\mathsf{Z}$-rotations, they commute and we can implement these in any arbitrary order. Thus, we will consider the order most convenient for our analysis. First, we will analyze the rotations $GR_Z(\theta,S)$, which apply a $\mathsf{Z}$  rotation of magnitude $\theta$ to all the basis states in the set $S$. These rotations cannot be ignored, as they alter the measurement outcomes. To better describe this operation, we first define a vectorized version of the delta function $\overrightarrow{\delta}(x,i):\mathbb{F}_p^n \mapsto  \mathbb{F}_2^n$, by
\begin{equation}
    \overrightarrow{\delta}(x,i)=\delta_i(x_1),\delta_i(x_2),\hdots,\delta_i(x_n),
\end{equation}
\noindent with $\delta_i(x)=1$ if $x=i$, and $0$ otherwise. 

Now, we can analyze the effect of these rotations applied to the generalized poor-mans qupit states independently, given that these effects will add to the rotations in the ideal quantum circuit,

\begin{align}
\bigotimes_{i=0}^{n-1} GR_Z\left(\frac{2\pi}{p},S_{c_i}\right)\ket{\mathsf{GPM}_p^n}= \frac{1}{\sqrt{p}}\sum_{i=0}^{p-1}  e^{i\frac{ 2\pi\sum_{j=0}^{p-1} \sum_{k=0}^{j}  \langle \overrightarrow{\delta}(c,j), \overrightarrow{\delta}(z^{+i},n-k-1) \rangle}{p}}\ket{z^{+i}} \\
=\frac{\phi_g}{\sqrt{p}}\left (\ket{z}+\sum_{i=1}^{p-1}  e^{i\frac{ 2\pi\left(\sum_{j=0}^{p-1} \sum_{k=0}^{j}  \langle \overrightarrow{\delta}(c,j), \overrightarrow{\delta}(z^{+i},n-k-1) \rangle -  \langle \overrightarrow{\delta}(c,j), \overrightarrow{\delta}(z,n-k-1) \rangle\right)}{p}}\ket{z^{+i}}\right )
\end{align}
\noindent with the global phase $\phi_g=e^{\frac{2 \pi \sum_{j=0}^{p-1} \sum_{k=0}^{j}  \langle \overrightarrow{\delta}(c,j), \overrightarrow{\delta}(z,n-k-1) \rangle}{p}}$ and $S_{c_i}=[p-c_i,\hdots,p-1]$. Additionally, to simplify, we will express the terms in the phases as follows,
\begin{equation}
\varphi(z^{+i},c)=\sum_{j=0}^{p-1} \sum_{k=0}^{j}  \langle \overrightarrow{\delta}(c,j), \overrightarrow{\delta}(z^{+i},n-k-1) \rangle. 
\end{equation}

\noindent This will allow us to rewrite the previous states and describing the state with all the rotations applied as follows,
\begin{align}
 &\bigotimes_{i=0}^{n-1} R_Z\left(\frac{2\pi x_i}{p^2}\right)GR_Z\left(-\frac{2\pi\cdot c_i}{p^2}\right)GR_Z\left(\frac{2\pi}{p},S_{c_i}\right)\ket{\mathsf{GPM}_p^n}\\
 &= \frac{e^{i\frac{2\pi \varphi(z,c)}{p}}}{\sqrt{p}} \left(\bigotimes_{i=0}^{n-1} R_Z\left(\frac{2\pi x_i}{p^2}\right)GR_Z\left(-\frac{2\pi\cdot c_i}{p^2}\right)\left (\ket{z}+\sum_{i=1}^{p-1}  e^{i\frac{ 2\pi\left(\varphi(z^{+i},c)-\varphi(z,c)\right)}{p}}\ket{z^{+i}}\right ) \right )\\
&=  \frac{\phi_g'}{\sqrt{p}} \left ( \ket{z} + \sum_{i=1}^{p-1} e^{\frac{2\pi i\cdot \left (i|x|/p + (\sum_{j \in {1,\hdots,i}}  \left(\langle x,z^{+j} \rangle\right)^{p-1} +\varphi(z^{+i},c)-\varphi(z,c)\right ) }{p}}\ket{z^{+i}} \right ),
\end{align}

\noindent with $\phi_g'=\exp\left(\frac{2\pi i\cdot \langle x,z \rangle+ p\cdot \varphi(z,c)-|c|}{p^2}\right)$.

Subsequently, we repeat the same analysis as in \cref{lemma:computcorre}. We know that the measurement outcome depends on the values of  $\sum_{j \in [i]}  \left(\langle x,z^{+j} \rangle\right)^{p-1} +\varphi(z^{+i},c)-\varphi(z,c)$, which once again are uniformly random values in $\mathbb{F}_p$ as both strings $z$ and $c$ are uniformly random in $\mathbb{F}_p^n$. Specifically, we can use the same analysis as before, given that we can focus uniquely on the overlap with the basis states of the following form $\ket{\psi_p}= \ket{z} + \sum_{j=1}^{p-1} e^{\frac{2\pi i\cdot \left (j|k|/p \right ) }{p}}\ket{z^{+j}}$, as these define are the states that will project into strings that are congruent $\MOD\ p$ with $|k|$, and thus infer the probability that our state projects into a superposition of strings congruent with the correct outcome to the respective ISMR problem.

Once again, due to the random and uniformly distributed form of the terms $\sum_{j \in [i]} \left(\langle x, z^{+j} \rangle\right)^{p-1} + \varphi(z^{+i}, c) - \varphi(z, c)$, the overlap with this state is equal to $1/p$, thus having zero correlation with ISMR problems. However, we can apply the same solution as in \cref{lemma:computcorre} and concatenate a correction string to the outcome in the form of one of these terms. We then obtain a positive correlation with the correct outcome string. In particular, we obtain the same correlation if we use the first term,
\begin{equation}
\langle x, (z^{+1})^{p-1} \rangle + \varphi(z^{+1}, c) - \varphi(z, c).
\end{equation}

The term $\langle x, (z^{+1})^{p-1} \rangle$ we already know how to compute with a $\NC^0$ circuit, but we are missing to demonstrate that a string congruent modulo $p$ with $\varphi(z^{+1}, c) - \varphi(z, c)$ can be constructed with asymptotically small size using a $\NC^0$ circuit. Given the string $c$, the values of the terms $\varphi(\cdot,\cdot)$ reduce to computing delta functions $\overrightarrow{\delta}(\cdot,\cdot)$. Moreover, the total number of such terms in the $\overrightarrow{\delta}(\cdot,\cdot)$ function does not pose a problem, as the number of these terms needing computation is upper-biased by $\mathcal{O}(p^2)$.

Now we need only to show then how the $\overrightarrow{\delta}(a,b)$ for any input string $a$ and dit $b$ can be computed by an $\NC^0$ circuit. For that we show that, 
\begin{equation}
    \overrightarrow{\delta}(a,b)=  \left(\left(a\oplus (-b)^{\otimes n}\right)^{p-1}-1\right)^{p-1}.
\end{equation}

Each of the terms $(a_i\oplus -b)^{p-1}=1$ if $a_i\neq b$, and equal to $0$ otherwise. This means that we already have an inverse of the desired string. Now, we can subtract $1$ so that all the $1$s turn into $0$s and the $0$s to $p-1$. By applying the exponentiation of $p-1$, we ensure that all the $0$s stay as zeros and the $p-1$ indices turn $1$ as desired, effectively inverting the $0$s and $1$s. 

These terms must be efficiently computable for both the strings $c$ and $z^{+i}$. For $c$, this is trivially true as we have access to the individual bits, and thus, each one of these operations is over a single dit and can be computed trivially with an $\NC^0$ circuit. For $z^{+i}$, this again is efficiently computable as the $\oplus$ operation can be realized by adding these values to the string, and the same method can be applied for the $-1$ operation. For the exponentiation, as described before, we can simply list all the products of terms, which are still bounded in number by $\mathcal{O}(n\cdot g^{p^2})$ and perform these operations dit-wise. Finally, the inner product between these terms $\langle \overrightarrow{\delta}(c,\cdot), \overrightarrow{\delta}(z^{+i},\cdot) \rangle$ can be evaluated with the same approach as we did for $\langle x, (z^{+1})^{p-1} \rangle$ in \cref{lemma:computcorre}. With this correction string of size $\mathcal{O}(n\cdot g^{p^3})$ computed by an $\NC^0$ circuit, we solve the respective ISMR problems with the intended correlation.

The last element to consider is that we can decompose the previously considered circuit into a constant-depth Clifford circuit if we are given the magic states $\ket{T^{1/p}}$ (as proven in \cref{lem:magic_state}).
\end{proof}

With this, we have completed our construction of a family of quantum circuits over qupit gate sets composed of qupit Clifford gates with an additional $\textnormal{T}^{1/p}$ gate that can solve the ISMR problem for the respective prime dimension. This demonstrates that all the problems can be solved with minimal, finite, and standard univeral gate sets. Additionally, we will subsequently show that the ability to rewrite these circuits as non-adaptive constant-depth Clifford circuits with access to magic states allows for error resilience while maintaining constant depth.

\subsubsection{Noisy quantum circuits retain a computational advantage against \texorpdfstring{$\BTC^0(k)$}{bPTF0[k]}.}

Armed with all these ingredients, we are now finally ready to combine the conditions for the fault-tolerant realization of constant-depth quantum circuits with the alternative family of non-adaptive quantum circuits proposed in \cref{sec:injection} to demonstrate that noisy qupit $\QNC^0/\ket{\overline{T^{1/p}}}$ circuits can solve a family of relational problems beyond the capabilities of $\BTC^0(k)$ circuits. 

For this purpose, we use the relational problems from \cref{qnoise:relation}, which are closely related to the original ISMR problems, and the noise-resilient implementation of $\QNC^0/\ket{\overline{T^{1/p}}}$ circuits solving the these problems from \cref{subsec:inject}. More precisely, we will show that $\QNC^0$ circuits with $\ket{\overline{T^{1/p}}}$ provided as advice states, despite suffering from local stochastic noise, can solve a family of relational problems that are $\AC^0$-reducible to the ISMR problems (see \cref{fig:final-advice_noise_resilient}). To this end, we will exploit the fact that the functions associated with the noise-resilient implementation---$\Dec$ and $\textsc{rec}$---are computable in $\AC^0$, a necessary requirement for proving classical hardness against $\BTC^0(k)$.

We first prove that there exist $\QNC^0/\ket{\overline{T^{1/p}}}$ circuits that solve the respective $\mathfrak{R}_p$ problem for all prime $p$ with a constant positive correlation, bounded away from zero.
\begin{figure}[ht]
\centering
\scalebox{0.75}{
\begin{quantikz}[classical gap=0.03cm]
\lstick{$x\in\{0,1\}^n$}  & \cw \wireoverride{n} & \cw \wireoverride{n} & \ctrl[vertical wire=c]{1}\cw \wireoverride{n}& \cw \wireoverride{n} & \cw \wireoverride{n} &\gategroup[wires=11,steps=6,style={dashed, rounded corners, fill=gray!20},background]{{$\AC^0-reduction$}} \cw \wireoverride{n}& \cw \wireoverride{n}& \cw \wireoverride{n}& \cw \wireoverride{n}& \cw \wireoverride{n}&\gate[10,disable auto height]{$\Red$} \cw \wireoverride{n}& \cw \wireoverride{n}
\\
\lstick{$\ket{0}^{m}$} & \gate[2,disable auto height]{\begin{tabular}{c}$\mathcal{Z}$ \\  \end{tabular}}  & & \gate[9,disable auto height]{\begin{tabular}{c}$\overline{\mathsf{Cliff}}$ \\  \end{tabular}} & \gate[9,disable auto height, ps=meter ]{}\slice[style={blue},label style={inner sep=1pt,anchor=south west,rotate=25}]{$\mathfrak{R}_p$}  &\cw \wireoverride{n} &  \cw \wireoverride{n} &\cw \wireoverride{n}  & \gate[2,disable auto height]{\begin{tabular}{c}$\Rec_1^{\ast-1}$ \\  \end{tabular}} \cw \wireoverride{n} & \gate{\Dec}\cw \wireoverride{n} &\cw \wireoverride{n}& 
\slice[style={blue},label style={inner sep=1pt,anchor=south west,rotate=25}]{$\mathcal{R}_p$}\cw \wireoverride{n} & \cw \wireoverride{n}  \\
\lstick{$\ket{0}^{m_{a}}$} &  & \meter{s_1} & \wireoverride{n}  &\wireoverride{n}  & \wireoverride{n} & \wireoverride{n} &\lstick{$s$}\wireoverride{n}  &\cw \wireoverride{n} &  \wireoverride{n}  &  \wireoverride{n}  & \wireoverride{n}  & \cw \wireoverride{n}\\
\lstick{\vdots} & \wireoverride{n} & \wireoverride{n} & \wireoverride{n}&\wireoverride{n}\vdots   &  \wireoverride{n}\vdots & \wireoverride{n} & \wireoverride{n} & \wireoverride{n}\vdots  & \wireoverride{n} \vdots & \wireoverride{n} & \cw \wireoverride{n} & \cw \wireoverride{n}  \\
&\wireoverride{n} &\wireoverride{n}  &\wireoverride{n}  \\
\lstick{$\ket{0}^{m}$} & \gate[2,disable auto height]{\begin{tabular}{c}$\mathcal{Z}$ \\  \end{tabular}}  & &   & &\cw \wireoverride{n} &\cw \wireoverride{n} &\cw \wireoverride{n} & \gate[2,disable auto height]{\begin{tabular}{c}$\Rec_{n'}^{\ast-1}$ \\  \end{tabular}}\cw \wireoverride{n} & \gate{\Dec}\cw \wireoverride{n} & \cw \wireoverride{n} & \cw \wireoverride{n} &\cw \wireoverride{n} \\
\lstick{$\ket{0}^{m_{a}}$} &  & \meter{s_n} & \wireoverride{n} & \wireoverride{n}  & \wireoverride{n}& \wireoverride{n} &\lstick{$s$}\wireoverride{n}  & \cw \wireoverride{n}  &  \wireoverride{n}  &  \wireoverride{n}  &  \wireoverride{n}&\cw \wireoverride{n} \\ 
&   \wireoverride{n}\gategroup[wires=3,steps=1,style={dashed, rounded corners, fill=blue!20,xshift=+0.6cm},background,label style={label position=below, yshift=-1.5em}]{{Advice state}}& \wireoverride{n} \lstick{$\ket{\overline{T^{1/p}}}_1$} & &  &\cw \wireoverride{n} &\cw \wireoverride{n}  & \cw \wireoverride{n} & \gate[2,disable auto height]{\begin{tabular}{c}$\Rec_{n'+1}^{\ast-1}$ \\  \end{tabular}}\cw \wireoverride{n}& \gate{\Dec}\cw \wireoverride{n} & \cw \wireoverride{n}&\cw \wireoverride{n} &\cw \wireoverride{n} \\
\wireoverride{n} & \wireoverride{n}  &\vdots \wireoverride{n}  & \wireoverride{n} & \wireoverride{n} & \wireoverride{n} & \wireoverride{n} &\lstick{$s$}\wireoverride{n}  &\cw \wireoverride{n} &  \wireoverride{n}  &  \wireoverride{n}   &  \wireoverride{n}  &\cw \wireoverride{n} \\ 
&   \wireoverride{n}& \wireoverride{n}\lstick{$\ket{\overline{T^{1/p}}}_n$} & & &\cw \wireoverride{n} &\cw \wireoverride{n}  & \cw \wireoverride{n}& \gate[2,disable auto height]{\begin{tabular}{c}$\Rec_{n'+n}^{\ast-1}$ \\  \end{tabular}}\cw \wireoverride{n}  &  \gate{\Dec}\cw \wireoverride{n} & \cw \wireoverride{n} & \cw \wireoverride{n} & \cw \wireoverride{n}\\
\wireoverride{n} & \wireoverride{n}  &\wireoverride{n}  & \wireoverride{n} & \wireoverride{n} & \wireoverride{n} & \wireoverride{n} & \lstick{$s$}\wireoverride{n}   & \cw \wireoverride{n} &\wireoverride{n}  &  \wireoverride{n}  &  \wireoverride{n} &\wireoverride{n} &\wireoverride{n}  &\wireoverride{n} 
\end{quantikz}}
\caption{noise-resilient constant-depth Clifford quantum circuit with quantum advice states.}
\label{fig:final-advice_noise_resilient}
\end{figure}
\begin{lemma}\label{noise:tol} 
    For all primes $p$, there exists a family of classically controlled constant-depth Clifford circuits with all-to-all connectivity in the class $\QNC^0/\ket{\overline{T^{1/p}}}$, denoted as $QC$, that for any input $x$ drawn uniformly at random from the subset of $n$-bit strings with Hamming weight satisfying $\left(\sum_{i=1}^n x_i\right)\ \MOD\ p = 0$, and affected by noise $\mathcal{E}\sim\mathcal{N}(\tau)$ with a fixed probability $\tau < \tau_{th}$, $QC$ solves the respective relation $\mathfrak{R}_p$ for output string of length $y = o(n^2)$ with a constant positive correlation,
     \begin{equation}
    \mathsf{Corr}_{\mathcal{D}_p} \Big(\mathfrak{R}_p\big(x,QC(x)\big)\Big)= 0.99\cdot \left ( \frac{p-1}{p^2} \right).
    \end{equation}
\end{lemma} 
\begin{proof}
We will prove this claim by considering the noiseless versions of the circuit family described in \cref{lem:clifford}, which solve the ISMR problems with a circuit architecture that can be noise-resilient based on \cref{thrm:cnst_noise_res} and \cref{cond:Ecode}. Subsequently, we consider their logical version in the qudit surface code and show that the implementation of any element of this circuit family is noise-resilient with this code. Finally, we will demonstrate that decoding the logical outcome of any of these quantum circuits can be done efficiently by polynomial-size $\AC^0$ circuits, and thus their outcomes are $\AC^0$-reducible to solutions to the ISMR problems (see \cref{fig:final-advice_noise_resilient} for a representation of the described circuit).

Now we will repeat the previous steps with all the necessary details. For that, we consider non-adaptive constant-depth Clifford circuits described in \cref{lem:clifford}, and translating them into their logical version within the qupit surface code for a fixed distance $d=\poly(\log n)$. In addition, we consider $n$ repetitions of $\ket{\overline{T ^{1/p}}}$, the logical version of the magic state $T^{1/p}$ in the surface code with distance $d$, as our advice state, which can be affected at most by local stochastic noise $\mathcal{E}\sim \mathcal{N}(\tau)$ below a fixed threshold $\tau_{th}$. Simultaneously, we consider the single-shot state preparation process in \cref{subsec:single_shot} for creating the logical $\ket{\overline{\mathsf{GHZ}_2}}$ states within the surface code with the same distance $d$, without applying the recovery function $\Rec$ but still producing the string $s$. Next, we apply the logical version of the classically controlled Clifford circuit described in \cref{lem:clifford}, which we designate $\overline{\mathsf{Cliff}}(x)$, except we omit the execution of certain logical $\mathsf{Sum}$ gates that are unnecessary because we start with the $\ket{\overline{\mathsf{GHZ}_2}}$ states and not logical $\ket{\overline{0}}$ states. The quantum circuit then concludes with a layer of $\mathsf{Z}$ measurements, as displayed in \cref{fig:advice_noise_resilient}.

Given the outcome of this circuit, we will consider the string $s$ from the single-shot state preparation and the outcome of the logical circuit $y$. We can define can then define the functions $\Rec_i^{\ast\dagger}$ (displayed in \cref{fig:final-advice_noise_resilient}) as the inverse of the recovery function $\overline{\mathsf{Cliff}}(x)\Rec(s)^\dagger\overline{\mathsf{Cliff}}(x)^\dagger$, restricted to the single logical qupit $i$. This undoes the effects of the errors occurring during the state preparation. Additionally, we compose to this operation the inversion of the logical $\mathsf{X}^a$ gate from the teleportation gadget introduced in \cref{lem:clifford}, completing the definition of each $\Rec_i^{\ast \dagger}$. Furthermore, each one of these functions $\Rec_i^{\ast-1}$ can be computed by an $\AC^0$ circuit using the arguments presented in \cite{bravyi2020quantum,grier2021interactivenoisy}.

Now, given the assumption that the entire circuit is affected by local stochastic noise $\mathcal{E}\sim\mathcal{N}(\tau)$ bounded below $\tau_{th}$, the noise-resilient single-shot GHZ$_2$ state preparation from \cref{subsec:single_shot}, and choosing the decoding function $\Dec$ to be a function implemented by the HDRG decoder, which can be computed by an $\AC^0$ circuit for our code distance, we derive from the combination of \cref{lem:single_shot}, \cref{error:thr2}, and \cref{thrm:cnst_noise_res} that after the decoding function, we have the exact outcome of the quantum circuit described in \cref{lem:clifford} with a probability of 0.99. By the same lemma, we have an $\NC^0$ reduction, (represented as $\Red$ in \cref{fig:final-advice_noise_resilient}), which provides outcome strings that correlate as described with the $\mathfrak{R}_p$ problems, completing the proof since it is trivially computed by an $\AC^0$ circuit.
\end{proof}

Simultaneously, we can prove that any $\BTC^0(k)$ circuit has close to zero correlation with the $\mathfrak{R}_p^{n'}$ problems for the size of the outcome strings that the $\QNC^0/\ket{\overline{T^{1/p}}}$ under consideration produces.

\paragraph{Proof of \cref{thm:noisy}.}
From \cref{noise:tol}, we obtain that there exists a $\QNC^0/\ket{\overline{T^{1/p}}}$ circuit that solves the $\mathfrak{R}_p^{n'}$ for all prime $p$ with the specified outcome strings. Now, we intend to demonstrate that, for the same parameters, $\BTC^0(k)$ circuits have close to zero correlation with the same relational problems.

To link the correlation of solving $\mathfrak{R}_p^{n'}$ with $\mathcal{R}_p^{n'}$, we use that by definition they are $\AC^0$-reducible to each other. Thus, a solution with a high correlation to $\mathfrak{R}_p^{n'}$ by a $\BTC^0(k)$ implies the existence of a slightly larger $\BTC^0(k)$ with the same correlation for $\mathcal{R}_p^{n'}$. Therefore, by considering a constant depth $d'$ for this $\AC^0$ reduction, we obtain the bounds described in \cref{thm:noisy} based on the lower bound proven in \cref{thm:qudit_average}, concluding the proof.
\qed
\vspace{0.2cm}

We conclude by considering the isolated qubit case, which has two main distinguishing features compared to the previously proven qudit cases. First, it does not require magic states; unlike those for all other primes $p \geq 3$, the qubit circuit will be a Clifford circuit. Second, it allows for separation with a fixed probability due to the one-to-one relationship between the correlation function and the likelihood of producing a correct or incorrect outcome for the $\mathfrak{R}_2^{n'}$ problem.

\paragraph{Proof of \cref{cor:Thm_noisy_qubit}.} This corollary follows more easily, as we consider exactly the same ideas for the qupit circuit in \cref{noise:tol}, but we do use as the underlying logical Clifford circuit the circuit from \cref{lemma:computcorre} for qubits, without additional advice state, nor teleportation gadgets, given that the necessary rotation is a Clifford operator. 

For the $\BTC^0(k)$ lower bound, we use the same lower bound as in \cref{averagequbit}, which allows us to bound the success probability of any $\BTC^0(k)$ solving the $\mathfrak{R}_2$ problem. 
\qed
\vspace{0.2cm}

\section{Discussion}
\label{sec:discussion}
In this work, we have advanced the nascent line of work showing unconditional separations between the computational power of classical and quantum shallow-depth circuits. We have shown that linear-size local $\QNC^0$ circuits can solve multi-output search problems (relational problems) that polynomial-size circuits of $k$-biased polynomial threshold gates fail to solve with appreciable probability even on average. We have also developed a family of non-local games over qupits of each prime dimension and used the separation in winning probabilities between classical and quantum strategies to show that the computational separation we establish for qubit $\QNC^0$ circuits extends to  $\QNC^0$ circuits over higher dimensional quantum systems. This is especially significant in the light of our belief that under finite universal gate sets, the corresponding $\QNC^0$ circuit classes over $p$ and $q$ dimensional systems are incomparable for $p\neq q$. Finally, we have also shown that the computational separations we establish are noise-robust, in that even quantum circuits affected by local stochastic noise continue to exhibit computational advantages over noiseless classical circuits, using techniques that draw on magic state injection and hard decision renormalization group decoders. 

We wrap up our analysis exploring the applicability of our results to practical developments in quantum computing and machine learning. As we rapidly progress through the NISQ era, with several noisy intermediate scale quantum devices with different underlying architectures and technologies on the horizon, precise estimates of the resource requirements for demonstrations of quantum computational advantage become increasingly relevant. Such resource estimation and optimization is of particular significance to work such as ours that focuses on shallow-depth circuits that may be within the reach of both theoretical optimization and practical implementation. In \cref{subsec:resource-estim-main}, we present the best known order-of-magnitude circuit width, depth and size estimates for witnessing unconditionally provable separations of computational power between $\QNC^0$ and the largest known classical circuit classes to date.

In learning theory, the expressivity of hypothesis classes such as neural network architectures commands as much attention as computational complexity. Decades of work have approached this problem with tools from circuit complexity. Proving that a forward pass through a neural network architecture can be simulated by a circuit class such as $\AC^0$ establishes upper bounds that allow us to carry over known results about circuit complexity lower bounds to understand the expressivity and generality of the neural network. In \cref{subsec:neural-nets}, we discuss some relations between our findings on circuit complexity and recent results on the simulation of popular neural network models by shallow depth classical circuits.
\subsection{Resource estimation}\label{subsec:resource-estim-main}

In physical implementations that test computational separations, it is key to pin down at what values of circuit depth $d$ and input size $n$ (i.e. number of input qubits) we observe a transition in the circuit size. That is, at what depths and input sizes do the quantum advantages kick in? 

To answer this question, we can make preliminary estimates by solving for the parameter values at which the asymptotic lower bounds for the size of the best classical circuit become equal to the quantum upper bounds. Using the best estimates for the constant factors hidden in the asymptotics, we solve this equality for values of the input size $n$ and depth $d$ at which the quantum circuit size becomes smaller than the classical circuit size. 

We can leverage our exact-case hardness bound in \cref{EsepQNC} to derive tighter estimates compared to what is possible using prior work, bringing theoretical predictions closer to the capabilities of current quantum devices. For context, the transition point for Shor's factoring algorithm is estimated to be $\sim$1,700 qubits, $10^{36}$ Toffoli gates, and a circuit depth of $10^{25}$ \cite{chevignard2024reducing}, while for the HHL quantum matrix inversion algorithm it is roughly $10^8$ qubits and a depth of $10^{29}$ \cite{scherer2017concrete}. Recent advancements in quantum hardware favor larger devices (i.e., more qubits) over those with prolonged coherence, and so there is a push towards shallower circuits \cite{lubinski2023application,bluvstein2023logical}. Thus, separations against $\NC^0$ represent a promising avenue for near-term quantum devices, requiring only hundreds to thousands of qubits to demonstrate classical non-realizability results, such as Bell violations \cite{Shalm15,Rauch18,Bravyi17}. While much work in this direction has focused on such information theoretic demonstrations, a \textit{computational} example of this is solving the 2D Hidden Linear Function problem using a noise-free quantum circuit. In this case, the classical model being compared---circuits of bounded fan-in---require a depth of at least $d=\Omega(\log n)$ \cite{Bravyi17}. This provides a demonstrable quantum advantage when the quantum circuit resolves an instance with a depth strictly smaller than the minimum possible depth for any classical circuit. In a more realistic scenario that includes noise, the depth of the quantum circuit may increase by a constant factor, but the minimal classical circuit depth will still be at least $d=\Omega(\frac{\log n}{\log\log n})$ \cite{bravyi2020quantum}. This adjustment, due to the presence of noise, slightly increases the previous estimates of the input sizes required for the problem to witness a quantum advantage. Additionally, imposing connectivity restrictions on the classical circuit—which are not part of the $\NC^0$ definition—the size of the input strings at which the crossover to the quantum advantage regime happens may be effectively reduced \cite{bharti2023power}.

Moving up the ladder of potential demonstrations of unconditional quantum advantage, the challenge is to outperform larger classical constant-depth circuit classes such as $\AC^0$. Once again we may ask what the minimal size input size is at which which classical circuits require more gates to solve the instance than quantum circuits. As the quantum circuits for the Parity Halving Problem (PHP) have a linear number of gates $\Theta(n)$, our main tasks are to bound the constant factor hidden by the $\Theta$ notation, and to determine the minimum classical circuit size. As before we equate the PHP lower bound on the size of the $\AC^0$ circuits and the upper bound on the size of the $\QNC^0$ circuits, and solve for the input size. It is crucial that we do this after fixing the minimal allowed depth, because otherwise the subexponential dependence of the classical circuit size on the depth results in astronomically large input sizes. Specifically, using the local 2D $\QNC^0$ circuit proposed in \cite{Watts19} that achieves a separation against $\AC^0$ under average-case scenarios, avoiding asymptotic simplifications and using the precise bounds and parameters ($q=\sqrt{\log n}$) provided by the authors, and using the depth $d=5$ version of the circuit rewritten in the MBQC setting, we find that even with these choices we would require approximately $10^{97}$ qubits to observe a quantum advantage. 

However, by moving from local 2D circuits to all-to-all qubit connectivity and optimizing the parameters in the random restrictions technique used to lower bound the classical circuit size specifically for this connectivity, we can drop this requirement to $10^{21}$ qubits with depth-$3$ quantum circuits, a considerably lower value than the previous rough estimates. Similar estimates can be obtained for increasingly larger classical circuit classes by leveraging our average-case hardness bound in \cref{averagequbit}. For instance, for $\BTC^0(n^{1/(5d)})$ circuits, we obtain that one would require $10^{40}$ qubits for a demonstration of quantum advantage, once again using suitably optimized quantum circuits and parameter values (see \cref{tab:newtab}). 

Towards this goal, with our exact-case hardness bound from \cref{EsepQNC}, without any simplification as shown in \cref{tab:newtab}, we are able to reduce by several orders of magnitude the estimated minimal quantum resources required for demonstrations of quantum advantage. Specifically, we find that depth-$4$ circuits with $10^{13}$ and $10^{26}$ qubits respectively can exhibit quantum advantages over $\AC^0$ and $\BTC^0(n^{1/(5d)})$ circuits. This substantial reduction in resource requirements marks a step forward in the progression of quantum advantage experiment design. 

We hope our estimates spur further improvements in proof methods and parameter optimization to bring unconditional quantum advantage demonstrations closer to current hardware capabilities.

\begin{table}
\begin{center}
\renewcommand{\arraystretch}{2.5}
\begin{tabular}{|c|c|}
\hline
$\BTC^0(k)/\mathsf{rpoly}$ & \textnormal{All-to-all}\\  
\hline\hline \textnormal{Exact, }$k=\mathcal{O}(1)$ & $\log_2 (s)=\Omega\left(e^{\frac{1}{1-d}} \left(\frac{n}{\sqrt{n \log{n}} \log{(n \log{n})}}\right)^{\frac{1}{-1 + d}}\right) $ \\
\hline 
 \textnormal{Exact, }$k=n^{1/(5d)}$ &  $\log_2( s)= \Omega\left(e^{\frac{1}{1-d}} \left(\frac{n^{\frac{3}{10}}}{\sqrt{\log{n}} \log{(n \log{n})}}\right)^{\frac{1}{-1 + d}}\right)$\\
\hline
 \textnormal{Average, }$k=n^{1/(5d)}$ &  $\Pr[\text{Success}]\leq\frac{1}{2}+2^{-\Omega\left(\frac{ n^{8/5} (n \log{n})^{-1 - \frac{2}{\sqrt{\log{n} + \log{\log{n}}}}} }{2^{\left(2 \sqrt{\log{n} + \log{\log{n}}}\right)}(\log{s})^{2d-2} }\right)}$\\
\hline
\end{tabular}
\end{center}
\caption{Size lower bounds for the circuit classes $\BTC^0(k)$ without simplifications and optimal values for the parameter $q$. }
\label{tab:newtab}
\end{table}

\subsection{Neural networks as biased polynomial threshold circuits}
\label{subsec:neural-nets}

In this concluding subsection, we turn our attention to the connections between our main results and active fields of research in classical and quantum machine learning. We especially intend to (loosely) interpret how our techniques and conclusions relate to classical neural networks. We shall also identify potential implications of computational separations between shallow depth circuit classes for studying quantum over classical advantages in machine learning tasks.

\subsubsection*{Connection to classical neural networks}

The class of classical constant depth circuits $\TC^0$ was first defined in the quest for greater computational power beyond Boolean $\AC^0$ circuits \cite{Parberry1988}. Furthermore, it is also well motivated as a model for neural networks composed of Boolean neurons, inspired by biological neural networks \cite{Parberry1994book}. This class notably includes the threshold function $\mathsf{Th}_w: x \in \{0,1\}^n \mapsto 1$ iff $|x| > w$ and 0 otherwise, which captures the most basic activation behavior that can be exhibited by a neuron. Interestingly, this theoretical relationship has typically been analyzed retrospectively: new neural network models first prove themselves as practically useful, and only subsequently is their standing relative to the circuit classes studied in complexity theory understood, both interms o expressivity and computational capability. For instance, neural networks with fan-in bounded by the logarithm of the input size fall within the $\AC^0$ class \cite{ShaweTaylor1992}, and thus they cannot even compute the parity function. On the other hand, neural networks whose nodes can have unbounded fanout are known to be simulatable in $\TC^0$ for a wide variety of weight types associated with the neuron's activation function \cite{Smolensky2013}. Based on the widely held conjecture that $\TC^0 \subsetneq \NC^1$, these neural networks are therefore unlikely to be able to solve problems considered 'simple,' such as iterated matrix multiplication and solving linear equations, which are known to be solvable in $\NC^1$ and $\mathsf{P}$ respectively \cite{Mereghetti_2000,greenlaw1991compendium}.

More recently the transformer architectures of neural networks \cite{Vaswani17}, well-known for their spectacularly successful and widespread use in large language models (LLMs) such as ChatGPT, have been analyzed through the lens of constant-depth circuit classes. For instance, depending on the type of attention function used, the expressivity of transformer models may be limited; one example is the upper bound on hard attention in terms of languages computable by $\AC^0$ circuits \cite{Hahn2020HardAttention}. This suggests that any such model is generally weak, as we know of various relevant and simple computational functions, such as parity and majority, that are not included in this circuit class. Nevertheless, it is also worth noting that even this potentially weak attention model still finds powerful applications in computer vision models \cite{Xu2015hardAttention}, and continues to be an area of ative development \cite{Gamaleldin2019HardAttention}.  Transformers with the less restrictive saturated attention exhibit greater expressivity, and in a sense their power coincides with $\TC^0$ circuits \cite{Merrill2022,strobl2023averagehard}. This matches up with our empirical understanding that transformers employing more complex attention functions demonstrate better performance. The transformer architecture can be further extended to be able to capture computations in $\mathsf{P}$, and indeed even transformers restricted to hard attention are known to be Turing-complete if they are allowed to perform computations with arbitrary precision and run for an unbounded number of time-steps \cite{Perez2021}. However, this is not a realistic setting in practice, and it has been pointed out that practical cases with the precision of their internal weights and variables restricted to be logarithmic in the input length (a generous allowance) correspond exactly to uniform $\TC^0$ \cite{Merrill2023}. We refer the reader to \cite{strobl2024formal} for a survey of studies on the expressivity of neural network models in terms of circuit complexity. \\

We now proceed to explore the connection between the analysis and techniques in our work and the existing literature. We will demonstrate how one can compute commonly used activation functions using $\BTC^0(k)$ circuits, and examine how the bias parameter $k$ influences this. Notably, this shows that a substantial subclass of neural networks, including transformers with non-trivial attention functions which extend beyond $\AC^0$ circuits \cite{Merrill2022} can be analyzed using our techniques and switching lemmas. Furthermore, we remark that this analysis is applicable not only to decision problems but also to relational problems. This is of significance since neural networks, especially LLMs, should primarily be considered sequence-to-sequence models, that take (bit) strings as input and produce (bit) strings as output.

Consider discretizing a function of the type $f:\{0,1\}^n \mapsto  [0,b]\subseteq\mathbb{R}$ by choosing parameter $w\in\mathbb{Z}$ and discretizing the range $[0,b]$ into steps of size $w$, to obtain $\hat{f}:\{0,1\}^n \mapsto \mathbb{Z}_B$ where $B=w\cdot\lfloor\frac{b}{w}\rfloor$. Discretizing an activation function allows us to implement it by $\BTC^0(k)$ circuits, generating a $k$-biased function. This means that the function exhibits a `tail' that always outputs 1 for inputs with a Hamming weight greater than $k$, akin to an $\AND$ type activation function. In \cref{alg:decomposer}, we demonstrate that a discretized version of any activation function can be implemented accurately by $\BTC^0(k)$ circuits for input bitstrings of Hamming weight less than $k$.

It is easy to implement the discretized version for input bitstrings with a Hamming weight greater than $n-k$ by a related circuit. It includes an initial segment that always outputs 0 for all inputs with a Hamming weight less than $n-k$, resembling an $\OR$ type activation function. Specifically, our algorithm requires only a single layer of PTF gates in parallel, thereby achieving a depth-1 $\BTC^0(k)$ circuit that implements $\hat{f}$. \cref{fig:equivRelu} presents a notable example of this discretization for a specific parameter $k$ of the $\BTC^0(k)$ circuits for the widely used $\mathsf{ReLU}$ activation function, which is defined as $f(x)=\max\{0,x-c\}$ (where we have shifted the centre from $0$ to $c$). Our scheme considers $\mathsf{ReLU}:\{0,1\}^n\to[0,n-c]$ taking $n$-bit strings as input and interpreting their Hamming weight as the input $x$.

\begin{remark}
Note that functions with larger activation regions (i.e., those with higher values of the bias parameter $k$) can be achieved by composing deeper $\BTC^0(k)$ circuits. Also, by combining $\AND$ and $\OR$  type decompositions within a single depth-1 layer, one can obtain non-trivial activation regions for both Hamming weights smaller than $k$ and larger than $n-k$. This approach results in only a central region of fixed Hamming weight rather than the entire tail as in each case, $\AND$ and $\OR$ type decompositions.
\end{remark}

\begin{algorithm}
\caption{Hamming-Weight Biased Activation Function Decomposer}
\label{alg:decomposer}
\begin{algorithmic}[1] 
\Procedure{$\mathsf{Decomp}$}{Bias parameter $k$, activation function $f:\{0,1\}^n\mapsto  \mathbb{R}$, resolution parameter $w\in \mathbb{N}$} 
\State $gates\leftarrow[ ] , i \leftarrow 0$.
\State $l \leftarrow  \lfloor \displaystyle\max_{x}( f(x)/w )\rfloor\text{ for } |x|\leq k$. \Comment{Number discrete steps/gates in the decomposition.}
\While{ $\exists {x}\ s.t\ f(x)-i\cdot w>0$ with $|x|\leq k$ }
\State Find $S_i\subseteq [n]$ s.t.  $\forall{j\in S_i}$ $\lfloor\left(f(x_j)-i\cdot w\right)\rfloor=0$.
\State Define \[\BT_i[k](x)= \begin{cases}
    1, \text{ for all }x_j \text{ with }j\in S_i\text{ and }|x_j|\leq k \\
    0, \text{ otherwise}
\end{cases}.\]
\State Concatenate $(gates, \BT_i(k))$.
\State $i=i+1$.
\EndWhile \Comment{The loop ends when $w\cdot i > \displaystyle\max_x f(x)$.}
\State \textbf{return} $(l,gates)$
\EndProcedure
\vspace{0.2cm}

\hspace{-0.7cm}\textbf{High-level description:} The algorithm sequentially calculates the set of points of the activation function at discrete intervals defined by $w\cdot i$, where $i$ is an integer. It then groups these points into sets of indices corresponding to where the discretized function matches the value $w\cdot i$. This process continues iteratively until the final discretized step at $w\cdot k'$ closely approximates the maximum value of $f$.
\end{algorithmic}
\end{algorithm}

\begin{figure}
    \centering
    \begin{minipage}{0.60\textwidth} 
        \centering
        \includegraphics[scale=0.5]{Images/algReLUup3.png} 
        \caption{\justifying Representation of a $k$-$\mathsf{ReLU}$ gate, which is equivalent to a $\mathsf{ReLU}$ gate, up to precision $w$, for every input string with a Hamming weight bounded by $k$. The $k$-$\mathsf{ReLU}$ gate is contained in $\BTC^0(k)$; the precise circuit can be obtained by \cref{alg:decomposer}, which defines $k'$ the number and the particular $\BT[k]$ gates used in the $k$-$\mathsf{ReLU}$ layer along with the respective fan-out wires. These elements are depicted on the left side of the figure, and the composite gate, which is the outcome of the circuit, is displayed on the right.}
        \label{fig:equivRelu}
    \end{minipage}
    \hfill 
    \begin{minipage}{0.35\textwidth} 
        \centering
        \includegraphics[scale=0.6]{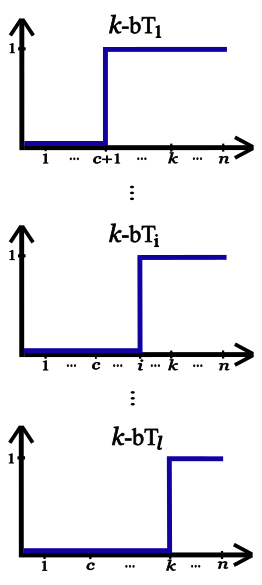} 
        \caption{\justifying Illustration of the unbounded region of each one of the $k'$ $\BT[k]$ gates employed in the $k$-$\mathsf{ReLU}$ gate. These functions are obtained from \cref{alg:decomposer}, based on the parameter $k$, the activation function $f$ and $w$ the size of the discretization steps.}
        \label{fig:funcRelu}
    \end{minipage}
    \label{fig:unbiased-region-BT-gate}
\end{figure}

\subsubsection*{Connection to quantum neural networks}
As we have seen, our computational separations demonstrate a quantum advantage over a model that encompasses many interesting classical machine learning models. From another perspective, this contributes to the extensively studied field of quantum machine learning. In particular, several studies have identified learning tasks that might benefit from a quantum advantage. Much work has also focused on the fundamental problem of learning a robust quantum solution, with a significant discussion based on the premise that models which are easy to train tend to be classically simulatable and thus are unlikely to offer quantum advantage. Conversely, models that are very general and potentially offer a quantum advantage tend to be difficult to train, from a theoretical perspective.

It is of great interest to understand which models emerge as good candidates for quantum advantage in the first place. Some notable classical-quantum separations rely on encoding a problem known to be classically hard into a learning task with an efficient quantum solution \cite{liu2021rigorous}. Others resort to quantum phenomena to explain enhanced expressivity and performance by explicitly pointing out where non-locality or contextuality clearly distinguishes between quantum and classical models at the operational level \cite{anschuetz2024arbitrary,anschuetz2023interpretable,bowles2023contextuality,Gao22}.

Our computational separation resembles the latter, where non-classical phenomena are translated into computational complexity. Furthermore, separations like ours have been utilized in a manner that combines both of the previously mentioned ideas: they demonstrate that certain problems that intrinsically exhibit quantum phenomena can be efficiently solved by quantum models, while remaining intractable for classical models. Interestingly, these separations enable us to explore potential quantum advantages in machine learning problems and the benefits of quantum models over classical ones. These analyses could be extended to cover standard figures of merit used in machine learning, such as the Kullback-Leibler divergence, which was shown to differ between classical and quantum models \cite{zhang2024quantum}, opening up avenues for further research.

Finally, we note that the line of work we advance demonstrates that even simple and potentially easy-to-learn quantum circuits can outperform their classical counterparts. This supports other findings that even efficiently learnable quantum circuits could provide a quantum advantage in practical information processing tasks \cite{anschuetz2024arbitrary,huang2024learning, gao2017efficient,anschuetz2023interpretable,zhang2024quantum,bowles2023contextuality,abbas2021power,Du21}.



\section{Proof of the tight multi-switching lemma for biased PTF circuits}\label{SI_new_switch}

In this section, we begin by proving our tighter multi-switching lemma for $\BTC^0(k)$ and then proceed to a depth-reduction lemma which allows us to obtain our multi-output multi-switching lemma for $\BTC^0(k)$ circuits.

\subsection{Tight multi-switching lemma}

The first lemma of this type will be a multi-switching lemma that reduces a finite set of depth-2 $\BTC^0(k)\circ\AND_w$ circuits, which are defined by $\BT[k]$ gate with $\AND$ gates of bounded fan-in $w$ as input, to a $\DT(w)^m \leaves \DT(t)$ decision forest. It is worth noting that \cref{multi_switch_K} serves as a multi-switching lemma of this type. However, a higher probability of success is required to reduce the initial circuit to a tree-like object, enabling the subsequent proof of quantum advantage with respect to the $\QNC^0$ circuit class.

In particular, better parameters can be obtained with additional assumptions on the decision trees to which the initial circuit will be reduced. This is demonstrated in the subsequent lemma, where we combine the inductive technique of \cite{haastad2014correlation} with the canonical decision trees proposed by \cite{Kumar23} for $\BTC^0(k)$ circuits restricted by a random restriction $\rho$.

\begin{lemma}\label{InductiveGC}
Let $\mathcal{F}=\{F_1,F_2,F_3,\hdots,F_m\}$ be a list of $\BT[k]\circ \AND_w$ circuits on $\{0,1\}^n$, $t>0$, $l\geq \log(m)+k+2$ and $p<\frac{1}{40w}$, then,
\begin{equation}
\Pr_{\rho\in R_p}[\mathcal{F}\lceil_\rho \notin  \DT(l)^m \leaves \DT(t) ] \leq 
m\cdot 2^k(80wp)^t.  
\end{equation}
\end{lemma}

\begin{proof}
Our proof employs an inductive approach similar to the one used by \cite{haastad2014correlation} for his multi-switching lemma. However, this necessitates the use of downward closed random restrictions. This attribute ensures that the size of the canonical decision tree does not expand when additional variables from the initial random restriction are restricted. It guarantees the monotonicity of the canonical decision tree's size under random restriction, which is crucial for the inductive reasoning process. Therefore, we will first establish some properties of the random restrictions and then utilize them in the proof of the new multi-switching lemma using inductive techniques.

\vspace{0.3cm}
\paragraph{Downward closedness.}We start with defining the downward closeness of random restrictions. 

\begin{definition}
Let $\mathcal{P}$ be a set of downward closed restrictions, and consider any $\rho\in\mathcal{P}$. Then for any other restriction $\rho'$ which satisfies $\rho'(x_i)=\rho(x_i)$ for all $i\in I=\{i\ |\ \rho(x_i)\in\{0,1\} \land \rho'(x_i)\in \{0,1\}\}$, we have that $\rho' \in \mathcal{P}$.
\end{definition}

Building on this definition, our interest lies in combining restrictions, necessitating the set of restrictions used in this process to be downwards closed. This requirement can be ensured with the following lemma.

\begin{lemma}\label{downcomp}
Let $\mathcal{P}$ and $\mathcal{P}'$ be two downward-closed sets, then $\mathcal{P} \cap \mathcal{P}'$ is also a downward-closed set.
\end{lemma}
\begin{proof}
This follows directly by the definition of the downward-closed set.
\end{proof}

Now, we need to demonstrate that the random restrictions chosen to bind the entire probability for the multi-switching lemmas are also downwards closed. To achieve this, we must analyze the canonical decision trees generated from $F\lceil_\rho$, where $F$ is an initial circuit of the type $\BT[k]\circ \{\AND,\OR\}_w$, and $\rho$ represents the random restrictions of interest. We will consider two objects defined in \cite{Kumar23}. The first is the algorithmic process that constructs the canonical decision tree from $F$ and $\rho$, while the other is the resulting description of the canonical decision tree in a non-standard format as considered by the author.

\begin{algorithm}
\caption{Canonical Decision Tree}
\label{alg:canonical}
\begin{algorithmic}[1] 
\Procedure{$CDT_F$}{$F=G(C_1,\hdots,C_n)$+ black-box access to a string $x \in \{0,1\}^n$.} 
\State $j* \leftarrow 0;x \leftarrow (*)^n; ctr \leftarrow 0$
\While{ $j* \leq m$ }
\State  Find $j\geq j*$ such that $C_j\not\equiv 0$. If no such $j$ exists, exit the loop.   
\State $B_j \leftarrow$ the set of variables not used by $C_j(x)$ 
\State Query $\alpha_{x_{B_j}}$
\State Set $x_{B_j}\leftarrow \alpha_{B_j}$.
\If{$C_j(x)=1$}
\State $ctr \leftarrow ctr+1$
\If{$ctr=k$}
\State return $G(1^m)$
\EndIf
\EndIf
\State $j* \leftarrow j$
\EndWhile
\State \textbf{return} $F(x\circ 0^n)$
\EndProcedure
\end{algorithmic}
\end{algorithm}

Given \cref{alg:canonical} that produces the canonical decision tree $T_{F|_\rho}$ for $F\lceil_\rho$, we can define the following object known as an $s$-witness. At an information-theoretical level, an $s$-witness describes the depth $s$ subtree of the canonical decision tree of $F\lceil_\rho$.

\begin{definition}[t-witness \cite{Kumar23}]\label{def:twitness}Let $F$ be a $\BT[k]\circ \AND_w$ circuit and $\rho$ a restriction. Let $t\geq 1$. Consider the tuple $(r,l_i,s_i,B_i,\alpha_i)$ where 
\begin{itemize}
    \item $r \in [1,t+k]$ is an integer;
    \item $(l_1, \hdots, l_r) \in [m]^r$ is an increasing list of indices;
    \item $(s_1, \hdots, s_r)$ is a list of non-negative integers, at most $k$ of which are allowed to be $0$, such that $s:= \sum_{i=1}^r s_i \in [t,t+w-1]$;
    \item  $(B_1, \hdots, B_r)$ is a list of subsets of $[w]$ satisfying $|B_i|=s_i$;
    \item $(\alpha_1, \hdots, \alpha_r)$ is a list of potential bit strings satisfying $|\alpha_i|=s_i$.
    \item $(r,l_i,s_i,B_i,\alpha_i)$ is called a $t$-witness for $\rho$ if there exists an $\alpha \in \{0,1\}^n$ such that 
    \begin{itemize}
        \item When we run $T_{F|_\rho}$ on a $\alpha$, $C_{l_i}$ is the $i$-th term queried by $T_{F|_\rho}$.
        \item $T_{F|_\rho}$ queries $s_i$ variables in $C_{l_i}$, and the relative location of those variables within $C_{l_i}$ are specified by set $B_i$.
        \item $T_{F|\rho}$ receives $\alpha_i$ in response to its $i$-th batch query. 
    \end{itemize}
\end{itemize}
The size of the witness $(r,l_i,s_i,B_i,\alpha_i)$ is defined to be $s:=\sum_{i=1}^r s_i$. We may denote the size of a witness $W$ as $size(W)$.
\end{definition}

Given \cref{alg:canonical} and a t-witness as defined in \cref{def:twitness}, the previous statements can be made precise with the following fact.

\begin{fact}
\cite{Kumar23} For every $\rho$ such that $F\lceil_\rho \geq \DT(t)$, there exists a t-witness for $\rho$.
\end{fact}

Now, we can begin to establish precise statements about the random restriction of interest. Our first property is as follows, which will aid in establishing downward closedness for two types of random restrictions.

\begin{lemma}\label{tree}
For all random restrictions $\rho$ and $F=\BT\circ\{\AND,\OR\}_w$, for which $F\lceil_\rho$ has an $l$-witness, let $L$ be the set of variables in the canonical decision tree $\DT(l)$. If $\rho'$ is a random restriction such that $\rho'(x_i)=\rho(x_i)$ $\forall i\notin T$, where $T=\{*\}^t$ and $L\cap T= \emptyset$, then $F\lceil_{\rho'}$ has the same $l$-witness. 
\end{lemma}
\begin{proof}
If a variable within $T$ is assigned a Boolean value, it implies that all clauses preceding the first non-fixed clause, denoted by $l_1$ in the $l$-witness, retain their assigned values. Crucially, any `alive' variables that could influence the outcome of this clause are encompassed within $B_1$, which are, by definition, elements of $L$. This assertion is supported by the ``Canonical Decision Tree" algorithm and the definition of an $s$-witness. Thus, altering any variable in $T$ leaves the initial segment of the $l$-witness unchanged. By inductively applying this reasoning to each element of the $l$-witness, we establish that variable assignments within $T$ do not alter the overall structure of the $l$-witness. Therefore, $F\lceil_{\rho'}$ has exactly the same $l$-witness as $F\lceil_{\rho}$.
\end{proof}

Now we will prove that a particular set of random restrictions is downward closed.

\begin{lemma}\label{downSwit}
The set of restrictions $\rho$ for which a circuit $F=\BT\circ\{\AND,\OR\}_w$ reduces to a decision tree of the type $\DT(l)$, with $l \in \mathbb{N}$, is downward closed.
\end{lemma}
\begin{proof}
We will demonstrate the lemma by considering a general scenario where any restriction $\rho$ satisfies the condition $F\lceil_\rho \in \DT(l)$. For such a restriction $\rho$, we observe that it designates $*$ to all variables within the decision tree $\DT(l)$, which we will refer to as the set $L$. Additionally, there exists a distinct set $T$ of variables also assigned $*$ that do not intersect with $L$, i.e., $T\cap L=\emptyset$.

When a variable in $L$ is set to a Boolean value, the downward closedness property is maintained. This is because the application of the ``Canonical Decision Tree" algorithm would yield an identical tree, albeit with the fixed variables omitted or potentially with further reductions. Such reductions occur if the fixed variables directly satisfy a complete clause $C_{l_i}$, with $l_i$ being an index of the $l$-witness. Consequently, the new $s$-witness encompasses at most a subset of the elements from the previous $l$-witness.

Now, by \cref{tree}, we can conclude that assigning a value to any variable in $T$ generates a new random restriction that does not change the $s$-witness of the initial function under this new restriction. Therefore, the resulting canonical decision tree is the same as for $\rho$. 

We can hence deduce that the set of restrictions for which $F\lceil_\rho \in \DT(l)$ is indeed downward closed.
\end{proof}

\paragraph{Induction.}
As mentioned previously, the precise proof of the multi-switching lemma will follow the inductive idea of the initial multi-switching theorem by \cite{haastad2014correlation}. Specifically, we will apply induction to the circuits $\Rec_i^{\ast-1}$ and the input size. Whenever we encounter a circuit that does not reduce to a decision tree of the type $\DT(l)$, we make a query on the variables that are alive to reduce this circuit and use these as variables of the global decision tree $\DT(t)$.

The theorem trivially holds when $m=0$ and $n=0$, where $m$ represents the number of $\BT[k]\circ \{\AND,\OR\}_w$ circuits considered, and $n$ denotes the input size. Now, we can divide the required induction step into two cases: one where the first $F_1\lceil\rho$ does have a decision tree of depth $l$, and another where it does not.
\begin{align}
\Pr\big[\mathcal{F}\lceil_\rho \notin \DT(l)^m \leaves  \DT(t)\big]&=
\Pr\big[\mathcal{F}\lceil_\rho \notin  \DT(l)^m \leaves \DT(t)\ \big|\ F_1\lceil_\rho \notin \DT(l) \label{decomp} \big]\cdot\Pr\big[F_1\lceil_\rho \notin \DT(l)\big] \\&+ \Pr\big[\mathcal{F}\lceil_\rho \notin \DT(l)^m \leaves \DT(t)\ \big|\ F_1\lceil_\rho \in \DT(l) \big]\cdot\Pr\big[\ F_1\lceil_\rho \in \DT(l)\big]. \label{hold_ind}
\end{align}

For $F_1\lceil_\rho \in \DT(l)$, we obtain that we can simply induct on the $\mathcal{F}\setminus F_1$. By doing that we obtain the following probability,
\begin{align*}
\Pr\big[\mathcal{F}\lceil\rho \notin \DT(l)^m \leaves  \DT(t)\ \big|\ F_1\lceil_\rho \in \DT(l) \big] &= \Pr\big[(\mathcal{F}\setminus F_1 )\lceil\rho \notin \DT(l)^m \leaves  \DT(t)\ \big|\ F_1\lceil_\rho \in \DT(l) \big]\\
&\leq 
(m-1)\cdot 2^k(80wp)^t.  
\end{align*}
 
To completely establish \cref{hold_ind}, we would need to bound $\Pr[F_1\lceil_\rho \in \DT(l)]$. However, this is unnecessary as we will observe when determining the conditional probability in \cref{decomp}. In addition to the previous argument, we need to ensure that the random restrictions guaranteeing $F_1$ to reduce to a decision tree of depth $l$ are downward closed, i.e., $F_1\lceil_\rho \in \DT(l) \land \rho \in \mathcal{F}$. This has already been proven in \cref{downSwit}, thus ensuring \cref{downcomp} that the same property holds for all random restrictions used in the inductive step of \cref{hold_ind}.

Now, we need to analyze the case where the first circuit does not have a reduction to a small decision tree. For that, we will use the following fact. 
\begin{fact}
For every $\rho$ such that $\DT(F\lceil_\rho) \geq t$ there exists a $t'>t$, such that $\DT(F\lceil_\rho)=t'$.
\end{fact}

This allows us to rewrite the expression which facilitates the application of the intended inductive step.
\begin{align*}&\Pr\big[\mathcal{F}\lceil\rho \notin \DT(l)^m \leaves \DT(t)\ \big|\ F_1\lceil\rho \notin \DT(l) \big]\cdot \Pr\big[F_1\lceil\rho \notin \DT(l)\big] \leq \\
&\sum_{l'>l} \Pr\big[\mathcal{F}\lceil\rho \notin  \DT(l)^m\leaves \DT(t) \ \big|F_1\lceil_\rho \in \DT(l')\big] \cdot \Pr\big[F_1\lceil_\rho \in \DT(l')\big].
\end{align*}

Now, we will use precisely the size of the canonical decision tree resulting from the restriction on the first circuit $F_1\lceil_\rho$ to provide us with a second restriction $\tau$. The idea is that we can rewrite the same probability $\Pr\big[\mathcal{F}\lceil\rho \notin \DT(t)\circ \DT(l)^m \big|\ F_1\lceil\rho \in \DT(l')\big]$ as $\Pr\big[\mathcal{F}\lceil_{\rho\circ \tau} \notin \DT(t-l')\circ \DT(l)^m |\ F_1\lceil\rho \in \DT(l')\big]$ when $l'<t$ , where $\tau$ is the assignment of the $l'$ variables of the canonical decision tree resulting from $F_1\lceil_\rho$. 
\begin{align}
\sum_{l'>l} \Pr\big[\mathcal{F}\lceil\rho \notin  \DT(l)^m\leaves \DT(t) \ \big|\ F_1\lceil_\rho \in \DT(l')\big] \cdot \Pr\big[F_1\lceil_\rho \in \DT(l')\big]=\\
\sum_{l'>l}\sum_{\tau \in \{0,1\}^{l'}}\Pr\big[\mathcal{F}\lceil_{\rho\circ \tau} \notin \DT(l)^m \leaves \DT(t-l') \ \big|\ F_1\lceil_\rho \in \DT(l')\big] \cdot \Pr\big[F_1\lceil_\rho \in \DT(l')\big].  \label{induct} 
\end{align}

In particular, when $l'<t$ we can bind the first term of \cref{induct} inductively, while the second can be very largely bound by the probability that $F_1\lceil_\rho \notin \DT(l')$, 
\begin{align*}
\sum_{l'>l}\sum_{\tau \in \{0,1\}^{l'}}\Pr\big[\mathcal{F}\lceil_{\rho\circ \tau} \notin \DT(t-l')\circ \DT(l)^m\ \big|&\ F_1\lceil_\rho \in \DT(l')\big] \cdot \Pr\big[F_1\lceil_\rho \in \DT(l')\big]\\
&\leq \sum_{l'>l} (20pw)^{l'}2^k  2^{l'}\Big( m\cdot(80wp)^{t-l'}2^k \Big ) \\
&\leq (80wp)^{t}2^{k-1} \sum_{l'>l} 2^{k+1} 2^{-l'} m. 
\end{align*}

Using our initial assumption that $l\geq \log(m)+k+2$ yields that,
\begin{equation}
 \sum_{l'>l} 2^{-l}\cdot m \cdot 2^{k} \leq 1 . 
\end{equation}
\noindent Therefore, the second term in \cref{decomp} in our initial division is inductively bounded above by $(80wp)^t2^{k-1}$ for cases where $l'<t$. When we consider this first case for the second term in \cref{decomp}, it's imperative to ensure that the random restrictions in question maintain the property of downward closedness. To this end, we invoke \cref{downcomp} inductively, but we must also affirm the same property for the composition of restrictions $\rho\circ\tau$ when $F_1\lceil_\rho \in \DT(l')$.

Since $\tau$ assigns definitive values to all variables within $\DT(l')$, our task simplifies to confirming that all variables not specifically assigned by the composition $\rho\circ\tau$ do not alter the conditional probability that our induction is predicated upon. This condition is indeed satisfied, as established by \cref{induct}, thus ensuring the downward closedness of the random restrictions we are considering. With this confirmation, our inductive argument remains robust and consistent.

Now, to show that an equal bound holds for the case that $l'\geq t$, we obtain that the first conditional probability of \cref{decomp} is fixed and equal to $1$, while second, we bound again with $F_1\lceil_\rho \notin \DT(l')$. 
\begin{align*}
\sum_{l'>l}\Pr\big[\mathcal{F}\lceil_{\rho} \notin \DT(t)\circ \DT(l)^m\ \big| F_1\lceil_\rho \in \DT(l')\big] \cdot \Pr\big[F_1\lceil_\rho \in \DT(l')\big] \leq 2^k \sum_{l'>t} (20pw)^{l'}\\
\leq 2^k \cdot 2(20pw)^t= 2^k \cdot 2^{2t-1}(80pw)^t.
\end{align*}
\noindent For any $t>0$ we obtain that this case is smaller or equal to $2^{k-1}(80pw)^t$. Thus, if we use the two values that build up \cref{decomp} and the value for \cref{hold_ind} obtained, we obtain that,
\begin{align*}
\Pr[\mathcal{F}\lceil\rho \notin \DT(t)\circ \DT(l)^m]
\leq (m-1)2^k(80wp)^t+2^k(80wp)^t \\
\leq m\cdot 2^k(80pw)^t,
\end{align*}
completing the proof.
\end{proof}
 
\subsection{Depth reduction lemmas}

Provided with the previous multi-switching lemma we can construct a new depth reduction lemma for $\BTC^0(k)$ circuits. This follows the proof technique from \cite{rossman2017entropy}.

\begin{lemma}\label{gc_depth_opt}
If $d,t\geq 1$ and $F \in \BTC[k;d;s_1,s_2,\hdots,s_d]\circ \DT(w)\leaves \DT(t-1)$, $l\geq \log s_1+ k +2$ and $p<\frac{1}{40w}$, then
\begin{equation}
\Pr[F\lceil_{\rho} \notin \BTC[k;d-1;s_2,\hdots,s_d]\circ \DT(l)\leaves \DT(t-1) ] \leq s_1 \cdot 2^{k}(400wp)^{t/2}.
\end{equation}
\end{lemma}

\begin{proof}
Suppose that $F$ is computed by a depth $t-1$ decision tree $T$, each of whose leaves $\lambda$ is labelled by a circuit $C_\lambda \in \BTC[k;d;s_1,s_2,\hdots,s_d]\circ \DT(w)$, we will assume the following two events,
\begin{align}
&\mathcal{A} \iff T\lceil_{p}\text{ has depth }\leq t/2-1\\
&\mathcal{B} \iff C_\lambda\lceil_{p} \in \BTC[k;d-1;s_2,\hdots,s_d]\circ \DT(l)\leaves \DT(t/2-1)
\text{ for every leaf } \lambda \text{ of T} 
\end{align}

\noindent such that,
\begin{equation}
    \mathcal{A} \land \mathcal{B} \implies F\lceil_{\rho} \in   \BTC[k;d-1;s_2,\hdots,s_d]\circ \DT(l)\leaves \DT(t-1).
\end{equation}

Now, the probability that neither of the events occurs will be considered. For the first event, we obtain the value this event to be
\begin{equation}
    \Pr[\neg \mathcal{A}] = \Pr[T\lceil_{p} \text{ has depth }\geq t/2] \leq (4ep)^{t/2}.
\end{equation}

For the second, we sort to \cref{lem:tree_literal_red} to translate $\BTC[k;d;s_1,s_2,\hdots,s_d]\circ \DT(w)$ to $\BTC[k;d;s_1,s_2,\hdots,s_d]\circ \AND_w$, and with \cref{InductiveGC} we obtain
\begin{align}
    \Pr[\neg \mathcal{B}] & \leq \sum_\lambda \Pr[C_\lambda\lceil_{p} \notin \BTC[k;d-1;s_2,\hdots,s_d]\circ \DT(l)\leaves \DT(t/2-1)] .\nonumber \\
    & \leq \sum_\lambda s_1\cdot 2^{k}(80wp)^{t/2}  \\
    & \leq 2^{t-1}\big(s_1\cdot 2^{k}(80wp)^{t/2}\big).
\end{align}

Combining both, we obtain, 
\begin{align}\label{lem12proof}
\Pr[F|_{R_p} \notin \BTC[k;d;s_2,\hdots,s_d]\circ \DT(l)\leaves \DT(t-1) ] 
&\leq \Pr[\neg \mathcal{A}] + \Pr[\neg \mathcal{B}] \nonumber \\
 \leq (4ep)^{t/2}+ 2^{t-1}\big(s_1\cdot 2^{k}(80wp)^{t/2}\big) \\
 \leq (4ep)^{t/2}+ \big(s_1 \cdot 2^{k-1}(400wp)^{t/2}\big)\\ \leq  s_1 \cdot 2^{k}(400wp)^{t/2}.
\end{align}
\end{proof}

The previous lemma allows one to reduce the depth of a $\BTC^0(k)$ circuit. Now, this lemma can be employed sequentially to reduce the entire circuit.

\begin{lemma}\label{GClemma_opt}
Let $d,t,k,q,m,s_1,l_1,\hdots s_{d-1},l_{d-1} \in \mathbb{N}$; $d,t>0$; $l_i\geq \log{s_{i}}+k+2$; $p_1,\hdots,p_d\in (0,1)$; $p_i<\frac{1}{40l_i}$. Let $F\in \BTC^0(k;d;s_1,\hdots,s_{d-1},m)$  with $n$ inputs and $m$ outputs. Let $s=s_1+\hdots+s_{d-1}+m$. Let $p=p_1\cdot p_2\cdot \hdots \cdot p_d$. Then, 
\begin{align*}
\Pr_{\rho \sim R_p }\big[F\lceil_\rho \notin \DT(q-1)^m \leaves &\ \DT(2t-2)\big] 
\leq s_1 \cdot 2^{k}\mathcal{O}(p_1)^{t/2} \\  &+\bigg ( \sum_{i=2}^{d-1} s_i \cdot2^{k}\mathcal{O} (p_il_{i})^{t/2} \bigg ) + 2^{k} m^{1/q} \mathcal{O}(p_d \cdot l_{d})^t.
\end{align*}
\end{lemma}

\begin{proof}
The first step is to consider our function $f$ as a $\BTC^0(k;d;s_1,\hdots,s_{d-1},m) \circ \DT(1) \leaves \DT(t-1)$ circuit. Then, to reduce this object sequentially through random restrictions we use \cref{gc_depth_opt} and $p_1$, such that  
\begin{align*}
\Pr[F\lceil_{\rho} \notin  \BTC_{d-1}^0(k;d-1;s_2,\hdots, s_n) \circ \DT(l_1)\leaves \DT(t-1) ] &\leq s_1 \cdot 2^{k}(400p_1)^{t/2} \\
&\leq s_1 \cdot 2^{k}\mathcal{O}(p_1)^{t/2}. 
\end{align*}
We will denote $F\lceil_{\rho} \notin \BTC_{d-1}^0(k;d-1;s_2,\hdots, s_n) \circ \DT(l_1) \leaves  \DT(t-1)$ as event $E_1$, and subsequent depth reductions equally for $E_i$. Building on this and based on the same \cref{gc_depth_opt}, we obtain that,
\begin{align}
\Pr[\neg E_i| E_1 \land \hdots \land E_{i-1}] &\leq  s_i \cdot 2^{k}(400l_{i}p_i)^{t/2} \\
&\leq s_i \cdot 2^{k}\mathcal{O} (p_i \cdot l_{i} )^{t/2}.
\end{align}
Now after applying $d-1$ times \cref{gc_depth_opt} we obtain that  $G$ defined as $G=F|_{p_1\circ \hdots \circ p_{d-1}}$ is
\begin{equation}
G \in  (\BT[k] \circ \DT(l_{d-1}))^m \leaves \DT(t-1),
\end{equation}
\noindent with a probability equal to the probability with the following sequence of event $E_1 \land E_2 \land \hdots \land E_{d-1}$ hold. Now, to the leaves of each of the global decision tree $\DT(t-1)$, we apply \cref{multi_switch_K}, such that for each leave, we have that,
\begin{align}
\Pr[G_{\lambda|_{\rho_d}} \notin \DT(q-1)^m \leaves  \DT(t-1) ] \leq 4\Big(64 (2^{k}m)^{1/q} p_d\cdot l_d \Big)^t \\
\leq 2^{k} m^{1/q} \mathcal{O}(p_d\cdot l_d)^t.
\end{align}
Then, we apply the union bound to all the leaves, which multiplies these values by $2^{t-1}$. However, this value will be upper bounded by the same value previously defined. Thus, we obtain that, 
\begin{align}
\Pr\big[\neg E_1 \land \neg E_2 \land \hdots \land E_d\big] = \sum_{i=1}^d Pr[\neg E_i | E_1 \land \hdots \land E_{i-1}] \\
\leq s_1 \cdot 2^{k}\mathcal{O}(p_1)^{t/2}  +\bigg ( \sum_{i=2}^{d-1} s_i *2^{k}\mathcal{O} (p_i\cdot l_i)^{t/2} \bigg ) + 2^{k} m^{1/q} \mathcal{O}(p_d\cdot l_d)^t.
\end{align}
\end{proof}

In the final step, we provide a value for the probabilities of the various random restrictions and dimensions for the local decision trees to derive a new lemma. This describes based on the previous values, and the size of the global decision tree $t$, the size of the initial circuit $s$, and the type of circuit parametrized by $k$ the probability of the success of reducing the circuit to decision tree of the type $\DT(q-1)^m \leaves \DT(2t-2)$.

\begin{lemma}\label{lem:GCred2}
 Let $F:\{0,1\}^n\mapsto  \{0,1\}^m$  be an $\BTC^0(k)$ circuit of size $s$, depth $d$. Let $p=1/(s_d^{1/q}\cdot\mathcal{O}(\log(s)^{d-1}\cdot k^{d-1}))$. Then, 
\begin{equation}
    \Pr_{\rho \sim R_p}[F\lceil_\rho \notin  \DT(q-1)^m \leaves \DT(2t-2)]\leq s \cdot 2^{-t+k} \ .
\end{equation}
\end{lemma}
\begin{proof}
This follows by applying \cref{GClemma_opt} with $p_1=1/\mathcal{O}(1)$ and $p_2= \hdots = p_{d-1}=1/\mathcal{O}(\log(s)+k)$ and $p_d=1/\mathcal{O}(s_d^{1/q}\cdot(\log(s)+k))$.
\end{proof}

%

\end{document}